\documentclass[a4paper,11pt]{article}
\expandafter\let\csname equation*\endcsname\relax
\expandafter\let\csname endequation*\endcsname\relax

\usepackage{algorithm}      
\usepackage{algpseudocode}
\usepackage{amsmath}
\usepackage{amsthm}
\usepackage{amssymb}
\usepackage{graphicx}
\usepackage{epstopdf}
\usepackage{multirow}
\usepackage{appendix}
\usepackage{color}
\usepackage{hyperref}
\usepackage{tikz}
\usepackage{rotating}
\usepackage{graphicx}
\usepackage{booktabs}
\usepackage{array}
\usepackage{subcaption}
\usepackage{makecell}
\usepackage{bm}
\usepackage{adjustbox}
\usepackage{array}
\usepackage{url}

\renewcommand{\varepsilon}{\epsilon}
\usepackage{mathtools}
\mathtoolsset{showonlyrefs}

\newcommand{\R}{{\mathbb R}}

\usepackage{mathrsfs}
\usepackage{xcolor}
\usepackage{enumerate}

\def\Tau{\mathscr{T}}

\theoremstyle{definition}

\theoremstyle{plain}

\newtheorem{lemma}{Lemma}

\newtheorem{proposition}{Proposition}

\theoremstyle{remark}
\newtheorem{rem}{Remark}

\usepackage[a4paper,left=2.cm,right=2.cm,top=3.cm,bottom=3.cm]{geometry}

\usepackage[a4paper,left=2.cm,right=2.cm,top=3.cm,bottom=3.cm]{geometry}

\newenvironment{keywords}
{\par\medskip\noindent\textbf{Keywords:}\ }
{\par\medskip}
\newcommand{\MSC}[1]{%
\par\medskip\noindent\textbf{MSC: }#1\par\medskip
}
\begin{document}

\title{Diffusion Graph Posterior Sampling for Nonlinear Inverse Problems with Application to Electrical Impedance Tomography}

\date{Last edit:  \today}

\author{
 Giovanni S. Alberti\thanks{ORCID ID 0000-0002-8612-3663. MaLGa Center, Department of Mathematics, University of Genova, Italy. 
    Email: giovanni.alberti@unige.it }
\and
Damiana Lazzaro\thanks{ORCID ID 0000-0002-2029-9842. Department of Mathematics, University of Bologna, Bologna, Italy. Email: damiana.lazzaro@unibo.it. }
      \and
    Serena Morigi\thanks{ORCID ID 0000-0001-8334-8798. Department of Mathematics, University of Bologna, Bologna, Italy.
   Email: serena.morigi@unibo.it.}
  \and
   Matteo Santacesaria\thanks{ORCID ID 0000-0001-7257-604X. MaLGa Center, Department of Mathematics, University of Genova, Italy. 
  	Email: matteo.santacesaria@unige.it}
     \and
    Shibo Wang\thanks{Department of Mathematics, Harbin Institute of Technology, China, Department of Mathematics, University of Bologna, Bologna, Italy. 
     	Email: wangbobo1012@gmail.com }
}
\date{ }

\maketitle

\begin{abstract}
Deep generative models have emerged as state-of-the-art for solving inverse problems, but applying them to inverse problems for PDEs, like electrical impedance tomography (EIT) remains challenging. Because physical domains are naturally discretized as unstructured meshes rather than regular grids, standard convolutional architectures are often inadequate. In this paper, we propose a novel framework that extends diffusion posterior sampling (DPS) to graph-structured data. We develop an unconditional score-based diffusion model directly on a 2D triangular mesh to learn an accurate prior over the physical solution space. Furthermore, we introduce a regularized variant, RDPS, which incorporates explicit regularization terms, such as total variation and generalized Tikhonov, to complement the implicit diffusion prior and mitigate severe ill-posedness. Extensive experiments on synthetic and real 2D EIT datasets demonstrate that RDPS produces stable, physically plausible reconstructions. Our approach generalizes well to out-of-distribution inclusion geometries, is highly robust to measurement noise, and outperforms current state-of-the-art solvers (e.g., GPnP-BM3D, DP-SGS) in reconstruction accuracy and artifact reduction.
\end{abstract} 

\begin{keywords}
Electrical Impedance Tomography,
Regularized Diffusion Posterior Sampling,
PDE-constrained Inverse Problems,
Graph-based Diffusion Models,
Finite Element Meshes.
\end{keywords}

\MSC{65N21; 65N75; 35R30; 62F15; 68T07}

\section{Introduction}
Inverse problems in imaging science are inherently characterized by severe ill-posedness leading to extreme numerical instabilities during reconstruction. To address this challenge, traditional methodologies rely on classical variational optimization, enforcing stability through handcrafted regularization terms. Conversely, recent data-drive paradigms—including purely generative methods---learn complex data priors directly from examples \cite{arridge2019solving,bubba-book}. Both approaches, however, exhibit critical limitations. While classical regularization often fails to capture the rich complexity and high-dimensional features of real-world data distributions, data-driven strategies inherently risk introducing `hallucinations', producing outputs that appear visually plausible but remain physically incorrect.

Among generative models, diffusion models currently represent the state-of-the-art for learning complex data distributions and have been successfully employed as expressive implicit priors \cite{ho2020denoising, song2020score}.  
For inverse problems, the appeal of diffusion posterior sampling (DPS) \cite{DPS24}  extends far beyond its exceptional generative accuracy. The true breakthrough of this methodology lies in its ability to integrate the physical forward model of the measurement operator directly into the solution generation process. In severely ill-posed systems, where multiple distinct configurations can explain the exact same observed data, DPS stabilizes inference by enforcing physical constraints.
This continuous guidance prevents artificial artifacts, yielding solutions that are both inherently stable and strictly faithful to the underlying physics.

This intrinsic stability is coupled with two profound operational advantages. First, DPS can encode rich, highly non-Gaussian prior information, capturing intricate geometries and textures that elude traditional mathematical frameworks. Second, it offers flexibility through zero-shot generalization. Because the generative model learns the underlying data manifold independently of the measurement operator, there is no need to train a fully supervised, end-to-end reconstruction network for each specific task. The exact same pre-trained model can be immediately deployed for inpainting, deblurring, or super-resolution simply by swapping the physical forward model during inference, effectively uniting statistical power with physical integrity.

In this work, we are interested in exploiting diffusion models to solve nonlinear inverse problems governed by partial differential equations (PDEs), with a focus on electrical impedance tomography (EIT) \cite{cheney-isaacson-newell-1999,borcea-2002}. Beyond the inherent nonlinearity and severe ill-posedness of EIT, a primary challenge is that the physical parameters of interest, such as the electrical conductivity, are naturally defined on continuous domains discretized via unstructured triangular meshes or graphs. This structural peculiarity renders standard pixel-based convolutional methods suboptimal.
In particular, mapping finite-element quantities from an unstructured mesh to a regular image grid by interpolations may introduce geometric inaccuracies, loss of mesh-level information, and artifacts that are not consistent with the numerical discretization used by the forward PDE solver.

Despite the contemporary success of DPS methods, a general framework that operates natively on unstructured finite element (FE) meshes for nonlinear PDE-constrained ill-posed inverse problems, such as EIT, is missing from the literature. This is precisely the gap addressed in this work: we combine an unconditional graph-based diffusion prior, learned directly on the FE computational mesh, with a physics-driven posterior guidance step based on the nonlinear forward model.
Our starting point DPS \cite{DPS24}. While highly successful at sampling from the posterior distribution by incorporating likelihood scores into the reverse-time process, the method was originally designed to operate in a grid-based, pixel setting.  
By employing multi-scale graph neural networkss (GNNs), we model the denoising process directly on the triangular FE mesh used for domain discretization.

Furthermore, we extend the standard DPS formulation by introducing an explicit regularization term. This extension is crucial for highly ill-posed, underdetermined nonlinear inverse problems like EIT, where incomplete data often cause distinct conductivity distributions to produce nearly identical boundary measurements.
From a unified Bayesian perspective, this explicit regularizer acts as an additional prior that complements the learned implicit diffusion prior. This regularized diffusion posterior sampling (RDPS) effectively suppresses spurious oscillations and guides the sampling process toward physically meaningful reconstructions when the likelihood provides insufficient information.

The main contributions of this paper are summarized as follows.
\begin{enumerate}
    \item We develop an unconditional diffusion model operating directly on a 2D triangular mesh to capture the prior distribution of PDE parameters over the solution space. This allows the learned prior to be expressed on the same geometric representation used by the finite-element forward solver.
    \item We introduce RDPS, a regularized extension of the diffusion posterior sampling (DPS) framework for solving ill-posed nonlinear inverse problems on graph-structured domains. 
    The new sampling algorithm extends the DPS framework with explicit regularization terms (such as total variation and Tikhonov).
    \item We extensively validate the framework on the inverse problem of EIT with numerical experiments on both synthetic and real measurements. Our results demonstrate that the proposed RDPS is highly robust to additive Gaussian and Laplacian noise, generalizes well to out-of-distribution shape geometries, and significantly outperforms current state-of-the-art Bayesian and deep learning baselines (including RTO-MH \cite{bardsley2015randomize}, GPnP-BM3D \cite{bouman2023generative}, and DP-SGS \cite{ling2025split}). 
\end{enumerate}

The remainder of this paper is organized as follows. Section \ref{sec:related} reviews related work. Section \ref{sec:diffusion} provides a brief introduction to score-based diffusion models, including a description of unconditional and conditional sampling schemes. Section \ref{sec:proposed} details our proposed regularized conditional sampling method. In Section~\ref{sec:graph_network}, we describe the architecture of the graph-based neural network used to approximate the score function. In Section \ref{sec:eit}, we describe the mathematical model of EIT. Finally, in Section \ref{sec:experiments} we present the numerical experiments with synthetic and real data. Section \ref{sec:concl} will draw limits and conclusions.

\section{Related Work}
\label{sec:related}

\paragraph{Deep Learning for Inverse Problems.}
The integration of deep learning into inverse problems has rapidly evolved from naive end-to-end mapping to structurally informed architectures. Methods such as unrolled networks \cite{monga2021algorithm} and Plug-and-Play (PnP) priors \cite{venkatakrishnan2013plug} explicitly separate the forward physics from the learned regularization, typically using a pretrained denoiser to impose structural priors. We refer to \cite{hertrich2025learning} and to the references therein for a comparison of learned regularization approaches for inverse problems. Recent advances have extended PnP frameworks to generative models, such as the GPnP-BM3D algorithm \cite{bouman2023generative}, which leverages generative prior sampling but generally assumes grid-structured data.

\paragraph{Diffusion Models for Inverse Problems.}
Diffusion models \cite{ho2020denoising} have recently become the standard for high-fidelity generative modeling. Their application to inverse problems typically involves conditional sampling techniques. Score-based generative models \cite{song2020score} enable Bayesian posterior sampling by combining the learned score of the prior with the gradient of the log-likelihood. The diffusion posterior sampling (DPS) method \cite{DPS24} provides a robust approximation for this conditional sampling, even for nonlinear forward operators. However, standard DPS applications are overwhelmingly image-centric and rely on regular 2D/3D grids. For a review on other generative models approaches for inverse problems, including generative adversarial networks and variational autoencoders, the reader is referred to \cite{duff-etal-2024}.

\paragraph{Machine Learning for Electrical Impedance Tomography.}
EIT is notoriously challenging due to its severe nonlinearity and ill-posedness. Classical approaches rely on variational methods like regularized Gauss-Newton (RGN-TV) \cite{borsic2009vivo}, and on Bayesian methods like the Randomize-Then-Optimize (RTO-MH) proposal in \cite{bardsley2015randomize}. Recent data-driven approaches have sought to improve resolution and speed. While end-to-end networks, acting as black boxes without knowledge of the underlying physics, fail to achieve physically plausible solutions, hybrid methods for stabilizing the ill-conditioned EIT inverse problem have been proposed as variational networks in \cite{RGN2022} and PnP in \cite{Col2023DeepplugandplayPG}. However, these approaches typically require supervised training. 
Notably, the split Gibbs diffusion posterior sampling (DP-SGS) \cite{ling2025split} approach applies diffusion priors to EIT, but requires complex splitting schemes. Our work differs by directly embedding the unsupervised diffusion process into the native finite-element mesh structure of the EIT domain providing a numerically stable solution to the challenging EIT inverse problem. For reviews on machine learning approaches to EIT, see \cite{TanyuNingHauptmannJinMaass+2025+437+470,denker2025deep}.

\paragraph{Graph Neural Networks for PDEs.}
Graph neural networks (GNNs) have become the key tool to respect the continuous, often irregular geometries of physical domains governed by PDEs \cite{pfaff2021learning,eliasof-etal-2021}. GNNs operate naturally on unstructured meshes, making them ideal for physics-based learning \cite{Lazzaro2025APG}. For example,  \cite{Lino22,LINO2025} proposed multi-scale, U-Net-inspired GNNs for fluid dynamics. Our network architecture adapts this graph-based paradigm, allowing the diffusion prior to be learned accurately on the same triangular finite-element mesh used by the forward PDE solver, thereby eliminating grid-interpolation artifacts. Other works focusing on using GNNs for inverse problems in PDEs, but not in the context of diffusion models, include \cite{qzhao2022graphpde,lauga2025graph}. We explicitly note that GNNs have been used to solve graph inverse problems, beyond the PDE context \cite{eliasof2025learning}.

\section{Background: Score-Based Diffusion Models}\label{sec:diffusion}

Diffusion models decompose the generative modeling task into two continuous Markov processes. The \textit{forward process} gradually corrupts the data with Gaussian noise until it reaches a tractable prior distribution, while the \textit{reverse process} learns to iteratively invert this corruption to generate new samples. This section briefly reviews the standard unconditional diffusion framework, which forms the foundation for the conditional sampling approach, discussed at the end of the section.

\subsection{Forward Diffusion Process}
\label{sec:fd}

Let $x_0 \in \mathbb{R}^N$ denote a clean data sample drawn from an unknown data distribution $p_{\mathrm{data}}$. The forward diffusion process gradually corrupts $x_0$ by adding Gaussian noise over a sequence of discrete time steps $t = 1,\dots,T$. Specifically, the transition kernel is defined recursively as:
\begin{equation}\notag
    q(x_t \mid x_{t-1}) = \mathcal{N}\!\left(x_t; \sqrt{1-\beta_t}\,x_{t-1}, \beta_t I\right),
\end{equation}
where $\{\beta_t\}_{t=1}^T$ is a predefined variance schedule with $\beta_t \in (0,1)$, typically chosen to be monotonically increasing. 
By the reproductive property of Gaussian distributions, the forward process admits a closed-form marginal distribution given $x_0$:
\begin{equation} \label{eq:forward_diffusion_closed}
    x_t = \sqrt{\bar\alpha_t}\,x_0 + \sqrt{1-\bar\alpha_t}\,\varepsilon, \qquad \varepsilon \sim \mathcal{N}(0,I),
\end{equation}
where $\alpha_t = 1 - \beta_t$ and $\bar\alpha_t = \prod_{s=1}^t \alpha_s$. Equation \eqref{eq:forward_diffusion_closed} shows that $x_t$ can be interpreted as a noisy observation of the scaled clean data $\sqrt{\bar{\alpha}_t} x_0$, corrupted by isotropic Gaussian noise with variance $(1-\bar{\alpha}_t)I$. This structure allows us to leverage Tweedie's formula to estimate the original signal. Recall that for a latent variable $\mu$ and an observation $z \sim \mathcal{N}(\mu, \Sigma)$, Tweedie's formula provides the posterior mean $\mathbb{E}[\mu \mid z] = z + \Sigma \nabla_z \log p(z)$, with $p(z)$ the marginal density of $z$. Applying this to the diffusion forward process yields:
$$\mathbb{E}[\sqrt{\bar{\alpha}_t} x_0 \mid x_t] = x_t + (1-\bar{\alpha}_t)\nabla_{x_t}\log p_t(x_t).$$
Dividing by $\sqrt{\bar{\alpha}_t}$, we obtain the Bayes optimal (minimum mean squared error) estimate of the original clean data $x_0$ given the noisy observation $x_t$:
\begin{equation}
    \label{eq:twed_diff}
    \mathbb{E}[x_0 \mid x_t] = \frac{1}{\sqrt{\bar{\alpha}_t}} \Big( x_t + (1-\bar{\alpha}_t)\nabla_{x_t}\log p_t(x_t) \Big),
\end{equation}
where $\nabla_{x_t} \log p_t(x_t)$ is the score function of the distribution $p_t(x_t)$, representing the direction of steepest ascent toward higher data density.

\subsection{Reverse  Process}
\label{sec:rp}

Rather than estimating the score function $\nabla_{x_t}\log p_t(x_t)$ directly, diffusion models typically adopt a noise-prediction parameterization. A neural network $\varepsilon_\theta(x_t,t)$ is trained to predict the forward noise $\varepsilon$ injected at step $t$. The network parameters $\theta$ are optimized by minimizing a mean squared error objective:
\begin{equation}
    \label{eq:ddpm_loss}
    \mathcal{L}(\theta) = \mathbb{E}_{x_0,t,\varepsilon} \Big[ \big\| \varepsilon - \varepsilon_\theta(x_t,t) \big\|_2^2 \Big].
\end{equation}
This loss encourages $\varepsilon_\theta(x_t,t)$ to approximate the conditional expectation $\mathbb{E}[\varepsilon \mid x_t,t]$. Consequently, minimizing the objective \eqref{eq:ddpm_loss} yields an implicit estimator for the score function:
\begin{equation}
    \label{eq:scoreest}
    \nabla_{x_t}\log p_t(x_t) \approx -\frac{1}{\sqrt{1-\bar\alpha_t}}\, \varepsilon_\theta(x_t,t).
\end{equation}

Using Tweedie's formula, one can estimate the clean signal $x_0$ directly from the noisy state $x_t$ via the conditional expectation $x_0^{\,\star}(x_t) := 
\mathbb{E}[x_0 \mid x_t]$, defined in \eqref{eq:twed_diff}. By substituting the network's score approximation in \eqref{eq:twed_diff}, we obtain the computable Tweedie estimator:
\begin{equation}
\hat x_0(x_t)
=
\frac{1}{\sqrt{\bar\alpha_t}}
\left(
x_t - \sqrt{1-\bar\alpha_t}\,\varepsilon_\theta(x_t,t)
\right).
\label{eq:tweedie_estimator}
\end{equation}
Once trained, the neural network $\varepsilon_\theta$ predicts the noise added at step $t$, and $\hat{x}_0$ is the corresponding estimate of the clean image.

 As detailed in the following sections, this learned score function is the core component of the sampling process, enabling both unconditional data generation and posterior inference.

\subsection{Unconditional Sampling}
The parameterized score function is employed to solve the reverse stochastic differential equation or its discrete counterparts, such as denoising diffusion probabilistic models
 (DDPM) \cite{ho2020denoising} and denoising diffusion implicit models (DDIM) \cite{song2021denoising}. This mechanism transforms Gaussian noise $x_T \sim \mathcal{N}(0,I)$ into high-fidelity samples $x_0 \sim p_{\mathrm{data}}$.

In DDPM the reverse step is stochastic. Starting from the forward transition $q(x_t \mid x_{t-1})$, one can show that the mean of the reverse distribution $q(x_{t-1} \mid x_t, x_0)$ is:
\begin{equation}
    \tilde{\mu}_t(x_t, x_0)
    = \frac{\sqrt{\bar{\alpha}_{t-1}}\,\beta_t}{1-\bar{\alpha}_t}\,x_0
    + \frac{\sqrt{\alpha_t}(1-\bar{\alpha}_{t-1})}{1-\bar{\alpha}_t}\,x_t,
    \label{eq:ddpm_mean}
\end{equation}
where $\beta_t = 1 - \alpha_t$ and $\alpha_t = \bar{\alpha}_t / \bar{\alpha}_{t-1}$.
The full reverse step is:
\begin{equation}
    x_{t-1} = \tilde{\mu}_t\!\big(x_t,\,\hat{x}_0(x_t)\big) + \sqrt{\tilde{\beta}_t}\,z,
    \qquad z \sim \mathcal{N}(0,I),
    \label{eq:ddpm_step}
\end{equation}
with posterior variance $\tilde{\beta}_t = \dfrac{1-\bar{\alpha}_{t-1}}{1-\bar{\alpha}_t}\beta_t$.
By combining \eqref{eq:ddpm_mean} and \eqref{eq:ddpm_step} we obtain
\begin{equation}
    \boxed{x_{t-1}^{\mathrm{DDPM}}
    = \frac{\sqrt{\bar{\alpha}_{t-1}}\,\beta_t}{1-\bar{\alpha}_t}\,\hat{x}_0(x_t)
    + \frac{\sqrt{\alpha_t}(1-\bar{\alpha}_{t-1})}{1-\bar{\alpha}_t}\,x_t
    + \sqrt{\tilde{\beta}_t}\,z.}
    \label{eq:ddpm_final}
\end{equation}
This is a linear combination of $\hat{x}_0$ and $x_t$, plus stochastic noise.
The linearity in $\hat{x}_0$ is precisely what allows for the additive propagation of the guidance correction.

DDIM removes stochasticity and constructs a deterministic reverse process consistent with the same forward dynamics. The step reads:
\begin{equation}
    \boxed{x_{t-1}^{\mathrm{DDIM}}
    = \sqrt{\bar{\alpha}_{t-1}}\,\hat{x}_0(x_t)
    + \sqrt{1-\bar{\alpha}_{t-1}}\,
      \frac{x_t - \sqrt{\bar{\alpha}_t}\,\hat{x}_0(x_t)}{\sqrt{1-\bar{\alpha}_t}}.}
    \label{eq:ddim_final}
\end{equation}
In both cases the reverse step can be written compactly as:
\begin{equation}
    x_{t-1} = A_t \cdot \hat{x}_0(x_t) + B_t \cdot x_t + C_t,
    \label{eq:linear}
\end{equation}
where the coefficients $A_t$, $B_t$, and $C_t$ depend only on the noise schedule $\bar{\alpha}_t$
(and on $z$ in the DDPM case).
The explicit values are:

\medskip
\begin{center}
\begin{tabular}{lccc}
\toprule
\textbf{Sampler} & $A_t$ & $B_t$ & $C_t$ \\
\midrule
DDPM &
$\dfrac{\sqrt{\bar{\alpha}_{t-1}}\,\beta_t}{1-\bar{\alpha}_t}$ &
$\dfrac{\sqrt{\alpha_t}(1-\bar{\alpha}_{t-1})}{1-\bar{\alpha}_t}$ &
$\sqrt{\tilde{\beta}_t}\,z$ \\[10pt]
DDIM &
$\sqrt{\bar{\alpha}_{t-1}} - \sqrt{1-\bar{\alpha}_{t-1}}\,\dfrac{\sqrt{\bar{\alpha}_t}}{\sqrt{1-\bar{\alpha}_t}}$ &
$\dfrac{\sqrt{1-\bar{\alpha}_{t-1}}}{\sqrt{1-\bar{\alpha}_t}}$ &
$0$ \\
\bottomrule
\end{tabular}
\end{center}

\subsection{Diffusion Posterior Sampling}
\label{sec:dps}

While standard diffusion models generate samples from the unconditional prior $p_{\mathrm{data}}$, solving inverse problems requires sampling from the posterior distribution $p(x_0 \mid y)$, where 
\begin{equation}
y = \mathcal{F}(x_0) + \eta
\label{eq:degmod}
\end{equation}
represents the observed measurements corrupted by noise $\eta$. Here, $\mathcal{F}$ denotes the forward map, modeling the measurement process. Diffusion posterior sampling (DPS) adapts the generative process to perform this conditional sampling directly within the diffusion space, \cite{DPS24}.

A key observation underlying DPS is that the posterior distribution at each diffusion
time admits a score decomposition into a prior term and a likelihood term.
The prior contribution is naturally provided by the diffusion model through the
learned score, while the likelihood contribution can be incorporated via a
gradient-based correction.
This results in a principled and flexible posterior sampling method that retains the
generative capabilities of diffusion models while enforcing consistency with the
observed data.

Rather than sampling in the space of the clean signal $x_0$, DPS considers the posterior distribution of the noisy variable $x_t$ conditioned on the observation $y$, denoted as $p_t(x_t \mid y)$. Applying Bayes' theorem, this diffusion-space posterior admits the factorization:
\begin{equation}\label{eq:bayes-pt}
p_t(x_t \mid y) \propto p_t(x_t)\,p_t(y \mid x_t),
\end{equation}
where $p_t(x_t)$ is the prior distribution at time $t$ induced by the diffusion model, and $p_t(y \mid x_t)$ represents the likelihood expressed in the diffusion space.

Since in score-based diffusion models sampling is driven by the gradients of the log densities, taking the logarithm and differentiating with respect to $x_t$ yields the posterior score decomposition:
\begin{equation}
\label{eq:posterior_score_dps}
\nabla_{x_t} \log p_t(x_t \mid y)
=
\nabla_{x_t} \log p_t(x_t)
+
\nabla_{x_t} \log p_t(y \mid x_t).
\end{equation}
In this decomposition, the first term on the right-side is the unconditional prior score, which is estimated by the pre-trained neural network as $-\frac{1}{\sqrt{1-\bar\alpha_t}}\,\varepsilon_\theta(x_t,t)$ and encodes the statistics of the data distribution. 
The central challenge in evaluating \eqref{eq:posterior_score_dps} is the second term: the measurement model \eqref{eq:degmod} is defined on the clean signal $x_0$, making the exact marginal likelihood $p_t(y \mid x_t)$ intractable. DPS elegantly resolves this by leveraging the properties of the forward diffusion process.

The Tweedie estimate \eqref{eq:tweedie_estimator}  allows us to approximate the intractable diffusion-space likelihood in \eqref{eq:posterior_score_dps} with the clean-space likelihood evaluated at $\hat{x}_0(x_t)$:
\begin{equation}\label{eq:plugin}
p_t(y \mid x_t) \approx p\!\left(y \mid \hat{x}_0(x_t)\right).
\end{equation}
Consequently, the required likelihood gradient in \eqref{eq:posterior_score_dps} can be approximated via backpropagation through the forward operator and the neural network:
\begin{equation}\label{eq:plugin1}
\nabla_{x_t}\log p_t(y \mid x_t) \approx \nabla_{x_t}\log p\!\left(y \mid \hat{x}_0(x_t)\right).
\end{equation}
By combining the learned diffusion prior with this tractable likelihood gradient, DPS effectively guides the reverse diffusion trajectory to ensure consistency with the measured data $y$.

\section{Regularized Diffusion Posterior Sampling (RDPS)}\label{sec:proposed}

Given the severe ill-posedness of many inverse problems, such as EIT, 
standard data-fidelity constraints are often insufficient to yield reliable and physically meaningful reconstructions. 
In these regimes, the likelihood provides limited information, and explicit regularization is required to suppress spurious oscillations and artifacts.

Within the diffusion posterior sampling framework discussed in section~\ref{sec:dps}, we propose augmenting the implicit, data-driven diffusion prior with an explicit, physics-informed regularizer $R(x)$. From a Bayesian perspective, this corresponds to defining a ``guided" target posterior $\tilde{p}(x \mid y)$ that is proportional to the standard posterior multiplied by a reference distribution $p_R(x) \propto \exp(-\lambda R(x))$ with $\lambda >0$:
\begin{equation}
	\label{eq:comp}
	\tilde{p}(x \mid y) \propto p(x \mid y) \cdot p_R(x) = \frac{p(y \mid x)\,p(x)\,p_R(x)}{p(y)}.
\end{equation}

Because diffusion models operate over a sequence of noisy states, we must translate this regularized target distribution into the diffusion space. We define the regularized diffusion-space posterior for the noisy variable $x_t$ conditioned on the observation $y$, i.e.
\begin{equation}	
	\label{eq:preg}
	p_t(x_t \mid y) \propto p_t(y \mid x_t)\,p_t(x_t)\,p_R(x_t),
\end{equation}
which  differs from \eqref{eq:bayes-pt} because of the addition of the prior $p_R(x_t)$.
Taking the logarithm and differentiating with respect to $x_t$ yields the regularized posterior score decomposition:
\begin{equation}
	\label{eq:posterior_scoredps1}
	\nabla_{x_t} \log p_t(x_t \mid y)
	=
	\nabla_{x_t} \log p_t(x_t)
	+
	\nabla_{x_t} \log p_t(y \mid x_t)
	+
	\nabla_{x_t} \log p_R(x_t).
\end{equation}
To construct a practical sampling algorithm, we approximate the three terms on the right-hand side of \eqref{eq:posterior_scoredps1} using the pre-trained score network $\varepsilon_\theta(x_t,t)$ and the Tweedie estimator $\hat{x}_0(x_t)$ introduced in \eqref{eq:tweedie_estimator}.

\noindent \textbf{Implicit Diffusion Prior.} The unconditional score is directly approximated by the neural network:
	\begin{equation}\notag
		\nabla_{x_t}\log p_t(x_t) \approx -\frac{1}{\sqrt{1-\bar\alpha_t}}\,\varepsilon_\theta(x_t,t).
	\end{equation}
	
\noindent \textbf{Data Likelihood.} Assuming the Gaussian degradation model $y = \mathcal{F}(x) + \eta$ with $\eta \sim \mathcal{N}(0,\sigma_y^2 I)$, we approximate the intractable diffusion-space likelihood $p_t(y \mid x_t)$ by evaluating the clean-space likelihood  ($ p(y \mid x)  \propto \exp \left(- \frac{1}{2\sigma_y^2}\|y-\mathcal{F}(x)\|_2^2 \right) $ at the Tweedie estimate $\hat{x}_0(x_t)$ given in \eqref{eq:tweedie_estimator}, eq. \eqref{eq:plugin1} reads as:
	 \begin{equation}\label{eq:plugin2}
	 	\nabla_{x_t}\log p_t(y \mid x_t) \approx - \nabla_{x_t} \left( \frac{1}{2\sigma_y^2}\|y-\mathcal{F}(\hat{x}_0(x_t))\|_2^2 \right).
	 \end{equation}
	
\noindent \textbf{Explicit Regularization.} Similarly, we approximate the gradient of the explicit prior evaluated on the noisy state by evaluating it on the denoised estimate:
	\begin{equation}
		\label{eq:reg}
		\nabla_{x_t}\log p_R(x_t) \approx -\lambda \nabla_{x_t} R(\hat{x}_0(x_t)).
	\end{equation}

To integrate the regularized score \eqref{eq:posterior_scoredps1} into a generative diffusion step, we formally extend the standard unconditioned  Tweedie formulation for $\mathbb{E}[x_0 \mid x_t] $ utilized in a generic sampling  to account for the explicit reference prior and the measurement model.

\begin{proposition}[Regularized Conditional Posterior Mean]
\label{pro:cpm}
Assuming the forward diffusion transition $p_t(x_t \mid x_0) = \mathcal{N}(x_t \mid \sqrt{\bar{\alpha}_t} x_0, (1-\bar{\alpha}_t)I)$, the conditional posterior mean for the regularized distribution $p_t(x_t \mid y)$ defined in \eqref{eq:preg} is given by:
\begin{equation}
\label{eq:Ec}	
\mathbb{E}[x_0 \mid x_t, y] = \mathbb{E}[x_0 \mid x_t] + \frac{1-\bar{\alpha}_t}{\sqrt{\bar{\alpha}_t}} \Big(\nabla_{x_t} \log p_t(y \mid x_t) + \nabla_{x_t} \log p_R(x_t)\Big).
\end{equation}	
\end{proposition}
\textit{Proof.} See Appendix~\ref{appendix}.\smallskip

We now describe how Proposition~\ref{pro:cpm} can be used to develop a new conditional sampling scheme.

For notational convenience, define the \emph{guidance correction}:
\begin{equation}
    \Delta_t
    \;:=\;
    \frac{1-\bar{\alpha}_t}{\sqrt{\bar{\alpha}_t}}
    \Big(
        \nabla_{x_t}\log p_t(y \mid x_t)
        + \nabla_{x_t}\log p_R(x_t)
    \Big),
    \label{eq:delta}
\end{equation}
so that Proposition~\ref{pro:cpm} reads compactly as
\begin{equation}
\label{eq:rcpm}
\mathbb{E}[x_0 \mid x_t, y] = \hat{x}_0(x_t) + \Delta_t.
\end{equation}

\medskip

In the \emph{conditional} case, the sampler uses \eqref{eq:rcpm}
in place of the unconditional $\hat{x}_0(x_t)$.
By the linearity \eqref{eq:linear}:
\begin{align}
    x_{t-1}^{\text{cond}}
    &= A_t \cdot \big(\hat{x}_0(x_t) + \Delta_t\big) + B_t \cdot x_t + C_t \notag\\
    &= \underbrace{A_t \cdot \hat{x}_0(x_t) + B_t \cdot x_t + C_t}_{x_{t-1}^{\mathrm{DD}}}
       + A_t \cdot \Delta_t.
    \label{eq:step1}
\end{align}
Hence:
\begin{equation}
    x_{t-1}^{\text{cond}} = x_{t-1}^{\mathrm{DD}} + A_t \cdot \Delta_t.
    \label{eq:step1b}
\end{equation}

Substituting the definition \eqref{eq:delta} of $\Delta_t$ into \eqref{eq:step1b}:
\begin{equation}
    x_{t-1}^{\text{cond}}
    = x_{t-1}^{\mathrm{DD}}
    + \underbrace{A_t \cdot \frac{1-\bar{\alpha}_t}{\sqrt{\bar{\alpha}_t}}}_{\displaystyle=:\,\eta_t}
      \Big(
          \nabla_{x_t}\log p_t(y \mid x_t)
          + \nabla_{x_t}\log p_R(x_t)
      \Big).
    \label{eq:step2}
\end{equation}
The scalar $\eta_t := A_t \cdot \frac{1-\bar{\alpha}_t}{\sqrt{\bar{\alpha}_t}} > 0$
depends only on the noise schedule and is treated as a (tunable) step size.

Both the likelihood score in \eqref{eq:plugin2} and the regularization score in \eqref{eq:reg} are evaluated at the
Tweedie estimate $\hat{x}_0(x_t)$ rather than at the intractable diffusion-space
distribution.

Inserting \eqref{eq:plugin2} and \eqref{eq:reg} into \eqref{eq:step2},
and noting that both gradient approximations carry a negative sign:
\begin{align}
    \nabla_{x_t}\log p_t(y \mid x_t) + \nabla_{x_t}\log p_R(x_t)
    &\approx
    -\,\nabla_{x_t}\!\left[
        \frac{1}{2\sigma_y^2}\|y-\mathcal{F}(\hat{x}_0(x_t))\|_2^2
        + \lambda\,R(\hat{x}_0(x_t))
    \right].
    \label{eq:combined_grad}
\end{align}
Substituting \eqref{eq:combined_grad} into \eqref{eq:step2} yields the final updated formula.

Let DD denote a general unconditional sampling step (e.g., DDPM or DDIM). The proposed \textbf{regularized diffusion posterior sampling (RDPS)} reads as:
\begin{align}
    x_{t-1}^{\mathrm{DD}} &= \texttt{unconditional\_step}\!\left(x_t,\,\varepsilon_\theta(x_t,t)\right), \label{eq:rdps} \\
    x_{t-1} &= x_{t-1}^{\mathrm{DD}} - \eta_t\,\nabla_{x_t}\!\left[ \dfrac{1}{2\sigma_y^2} \left\|y - \mathcal{F}\big(\hat{x}_0(x_t)\big)\right\|_2^2 + \lambda\,R\big(\hat{x}_0(x_t)\big) \right]. 
    \label{eq:rdps_2}
\end{align}

The conditional sampling process can be decoupled into two distinct phases at each timestep: a standard unconditional generative step, followed by a gradient-based ``guidance" step that nudges the sample toward the measurement manifold and the regularized subspace. 

This formulation essentially embeds a single step of gradient descent on a regularized variational objective directly into the reverse diffusion trajectory.

\subsection{Algorithm DDIM-RDPS}

Algorithm \ref{algorithm_reg} summarizes the complete procedure, incorporating the DDIM deterministic sampler. 

\begin{algorithm}[h]
    \caption{DDIM-RDPS (Regularized Diffusion Posterior Sampling)}
    \begin{algorithmic}[1] 
        \Require Measurements $y$, score network $\varepsilon_\theta$, forward model $\mathcal{F}$, timesteps $T$, step size $\eta \in [0,1]$, regularization parameter $\lambda$, regularizer $R(\cdot)$, positive constant $\epsilon$
        \State $x_T \sim \mathcal{N}(0, \mathbf{I})$
        \For{$t = T$ \textbf{to} $1$}
        \State \textbf{Score Estimation:} $\varepsilon_{\theta} \leftarrow \varepsilon_\theta(x_t, t)$
        \State \textbf{Tweedie's Formula:} $\hat{x}_{0} \leftarrow \frac{x_{t}-\sqrt{1-\bar{\alpha}_t}\varepsilon_{\theta}}{\sqrt{\bar{\alpha}_t}}$
        \State \textbf{Reverse Transition:} $\tilde{x}_{t-1} \leftarrow \sqrt{\bar{\alpha}_{t-1}}\hat{x}_{0} + \sqrt{1-\bar{\alpha}_{t-1}} \varepsilon_{\theta}$
        \State \textbf{Forward Solve:} \quad $y_0 \leftarrow \mathcal{F}(\hat{x}_{0})$
        \State \textbf{Adaptive step size:} $\eta_t \leftarrow \frac{\eta}{\|y - y_0\|_2 + \epsilon}$
        \State \textbf{Gradient Step:} \quad $x_{t-1} \leftarrow \tilde{x}_{t-1} - \eta_t \nabla_{x_t} \big( \|y - y_0\|_2^2 + \lambda R(\hat{x}_0) \big)$
        \State \textbf{Project (Domain Constraint):} $x_{t-1} \leftarrow \max(x_{t-1}, \epsilon)$
        \EndFor
        \State \Return $x_0^*$
    \end{algorithmic}
\label{algorithm_reg}
\end{algorithm}

\paragraph{Adaptive Step Size.} The choice of $\eta_t$ in Algorithm \ref{algorithm_reg} addresses the severe ill-conditioning of the nonlinear forward map $\mathcal{F}$. The gradient magnitude can vary drastically across the diffusion trajectory. By normalizing the step size by the current residual $\|y - y_0\|_2$, we provide a stable guidance that prevents the likelihood term from destabilizing the reverse process, particularly in the high-noise regimes characteristic of early sampling phases.

\paragraph{Positivity Projection.} Line 9 implements a projection operator $x_{t-1} \leftarrow \max(x_{t-1}, \epsilon)$. Standard diffusion models often operate on unconstrained domains (e.g., image pixels scaled to $[-1, 1]$). However, when the target parameter in specific applications  is required to be strictly positive physically, this projection ensures the intermediate samples remain within the physically valid domain required by the forward solver.

\vspace{0.2cm}

Figure \ref{fig:grafico1}(b) illustrates the conditional sampling workflow followed by Algorithm 1 which takes as input the measurement $y$ and starting from a noisy mesh $x_T$ produces after $T$ steps a sample $x_0^*$ conditioned by the given measurements and by the physical model (forward operator $\mathcal{F}$).

\begin{rem}
The reverse DDIM transition in line 5 is equivalent to equation \eqref{eq:ddim_final}. Specifically, by rearranging the Tweedie estimator in \eqref{eq:tweedie_estimator} to isolate $\varepsilon_\theta(x_t,t) = \frac{x_t - \sqrt{\bar\alpha_t}\,\hat x_0(x_t)}{\sqrt{1-\bar\alpha_t}}$, the DDIM update can be expressed directly in terms of the predicted noise $\varepsilon_\theta(x_t,t)$, namely $x_{t-1} = \sqrt{\bar\alpha_{t-1}}\,\hat x_0(x_t) + \sqrt{1-\bar\alpha_{t-1}}\,\varepsilon_\theta(x_t,t).$
 \end{rem}
\subsubsection{Explicit Regularizers}
The proposed framework is highly modular, and the explicit regularization term $R(\hat{x}_0)$ can be tailored to the specific structural priors of the inverse problem. For a given vector $x \in \R^n$, defined on the nodes of a graph 
 $G = (\mathcal{V},\mathcal{E})$  composed of nodes/vertices $\mathcal{V}$ and edges $\mathcal{E}$, we investigate three established choices:

\paragraph{(Tik) Tikhonov Regularization.} A simple zeroth-order $\ell_2$ penalty:
\begin{equation}
	R(x) = \|x\|_2^2,
\label{eq:regT}
\end{equation}
which penalizes large solution values and biases the solution towards small-norm reconstructions. In fact, in order to resemble the $L^2$ norm of the function that $x$ discretizes, this norm should take into account the structure of the graph. However, since in this work all the meshes considered will be approximately uniform, this has a minor impact and will not be considered.

\paragraph{(GTik)  Generalized Tikhonov Regularization.} 
To enforce spatial smoothness over the discretized domain, we penalize the squared $\ell_2$-norm of the discrete gradient of $x$:
\begin{equation}
R(x) = \|D x\|_2^2
\label{eq:regGT}
\end{equation}
where $D \in \mathbb{R}^{|\mathcal{E}| \times |\mathcal{V}|}$ represents the discrete gradient matrix (incidence matrix) defined on the triangular mesh connectivity, with $|\mathcal{V}|$ and $|\mathcal{E}|$ denoting the number of nodes and edges, respectively. This discrete quadratic form effectively suppresses unphysical, high-frequency oscillations across adjacent mesh elements while ensuring a computationally efficient optimization landscape.

\paragraph{(TV) Total Variation Regularization.} To preserve sharp transitions in the solution, we adopt an anisotropic total variation penalty defined over the mesh graph structure. The regularization term is expressed as the sum of directional discrete derivatives along the edges connecting each vertex to its 1-ring neighborhood:
\begin{equation}R(x) = \sum_{i \in \mathcal{V}} \sum_{j \in \mathfrak{N}(i)} |x_j - x_i|,\label{eq:regTV}
\end{equation}
where $\mathcal{V}$ is the set of mesh vertices, $\mathfrak{N}(i)$ denotes the 1-ring neighborhood of node $i$, and $x_j - x_i$ represents the discrete directional derivative along the shared edge $e_{i,j}$. Although this graph-based TV functional is non-smooth at zero, we employ automatic differentiation on a slightly smoothed approximation during optimization. This approach provides robust piecewise-constant regularization that effectively penalizes oscillations while preserving the sharp, blocky boundaries typical of EIT reconstructions.

\vspace{0.5cm}

In Figure \ref{fig:grafico1}(a) we illustrate the overall diffusion  process composed by the forward diffusion, described in Section \ref{sec:fd}, and the reverse process, detailed in Section \ref{sec:rp}.
The forward process progressively corrupts the input data $x_0$ into noise $x_T$ on the mesh, and
the reverse process is trained to predict the noise added at each step $t$ during the forward diffusion using the loss function \eqref{eq:ddpm_loss}. In the following, we describe the graph neural network that is learned to approximate the score function $\varepsilon_{\theta}(x_t, t)$.

\section{Graph-Based Unconditional Score Network}
\label{sec:graph_network}

In this section we describe the neural architecture used to approximate the score
function
\[
\varepsilon_\theta(x_t, t)
\;\approx\;
\nabla_{x_t} \log p_t(x_t),
\]
at each time $t \in [0,T]$.
To respect the continuous geometry of the physical domain without resorting to grid interpolations, we parameterize the score function using a multi-scale denoising graph network, denoted as $\text{DGN}_{\theta}$.

In this work, we focus on the case when the signal $x$ is defined on the nodes of a graph 
 $G = (\mathcal{V},\mathcal{E})$. 
 For example, the graph may correspond to an unstructured mesh, arising from the discretization of a PDE. The graph is augmented with node features $V = \{ v_i \in \mathbb{R}^{F_v} \mid i \in \mathcal{V} \}$ and edge features $E = \{ e_{i,j} \in \mathbb{R}^{F_e} \mid (i,j) \in \mathcal{E} \}$. The initial node feature set is populated by the vector $x \in \mathbb{R}^{N}$ (where $N = |\mathcal{V}|$ and we consider one feature per vertex, $F=1$).

Adapted from the DGN4CFD architecture proposed in \cite{LINO2025}, the network is designed as a hierarchical non-linear operator that predicts the noise $\epsilon_\theta^{t-1}$ (and variance interpolation $s_\theta^{t-1}$) during the reverse diffusion process:
\begin{equation}
\label{eq:dgneq}
[ \epsilon_{\theta}^{t-1}, s_{\theta}^{t-1}] \leftarrow DGN_{\theta}(G, V, E, t),
\end{equation}
for each diffusion time $t=T,\ldots,1.$
We remark that as in DDIM approaches, the variance  $ s_{\theta}$ is not used.

\subsection{Feature Embeddings $\mathcal{E}_{in}$}
Before processing the spatial data, the inputs are mapped into a high-dimensional hidden space by means of the input embedding operator 
\[
\mathcal{E}_{in}\colon \R^{|\mathcal{V}|+1}\to \R^{F_h}\times (\R^{F_h})^{|\mathcal{V}|}\times (\R^{F_h})^{|\mathcal{E}|},\qquad (x_t, t) \mapsto (t_{\mathrm{emb}},V,E),
\]
 which we now define. 
To account for the varying noise levels across the diffusion trajectory, the time step $t$ is transformed into a latent embedding $t_{\mathrm{emb}}$, as follows:
\begin{equation}
t_{\mathrm{emb}} = W_{\mathrm{proj}}^t\,\phi\bigl(W_t\,\mathrm{PE}(t)\bigr),
\label{eq:temb}
\end{equation}
where $\mathrm{PE}(\cdot)$ is the sinusoidal positional encoding, $\phi$ denotes the Scaled Exponential Linear Unit (SELU) \cite{Tallman2020}, $W_t$ projects the encoding into the embedding space, and $W_{\mathrm{proj}}^t$ maps it to the network's hidden dimension $F_h$.

The input embedding operator $\mathcal{E}_{in}$ incorporates this time embedding alongside the physical node and edge data to obtain the feature sets $V,E$:
\begin{equation}
\begin{aligned}
[v_i] &\leftarrow W_{\mathrm{proj}}^n \Bigl( \phi\bigl([W_{n}x_t^i \mid {t}_{\mathrm{emb}}]\bigr) \Bigr), \qquad &&i \in \mathcal{V}, \\
[e_{i,j}] &\leftarrow W_{\mathrm{proj}}^e \, ( |\mbox{edge}_{i,j} |\,), \qquad &&(i,j)\in\mathcal{E},
\end{aligned}
\end{equation}
where $W_{n} \in \mathbb{R}^{ F_h \times F } $ is the node encoder, $F_h$ is the hidden feature dimension, and $W_{\mathrm{proj}}^n \in \mathbb{R}^{ F_h\times 2F_h}$ is a learnable projection matrix. The notation $[\,\cdot \mid \cdot\,]$ denotes feature concatenation, $|\mbox{edge}_{i,j} |$ represents the physical edge length, and $W_{\mathrm{proj}}^e \in \mathbb{R}^{F \times F_h}$ maps the edge lengths to the hidden space.

\subsection{Multi-Scale GU-Net Architecture}
Once embedded, the features are processed through a multi-scale graph U-Net configuration, illustrated in Fig.~\ref{fig:grafico1}(a), inside the orange box. The graph network is defined as a composition of functions:
\begin{equation}
DGN_\theta = \mathcal{E}_{out} \circ \underbrace{(T_{2L-1} \circ \dots \circ T_L \circ \dots \circ T_1)}_{\text{multi-scale GU-Net}} \circ \; \mathcal{E}_{in}
\label{eq:dgn}
\end{equation}
where the output embedding operator $\mathcal{E}_{out}$ simply maps the processed hidden features back to the required output dimension (e.g., $\epsilon_{\theta}^{t-1} \in \mathbb{R}^{N}$).  Note that within the multi-scale GU-Net block, the operators $T_k$ are interconnected via skip connections that transfer features from the contracting (encoder) layers to the expanding (decoder) layers.

\begin{figure}[h]  
    \centering
    \begin{tabular}{c}
    \includegraphics[width=0.8\textwidth]{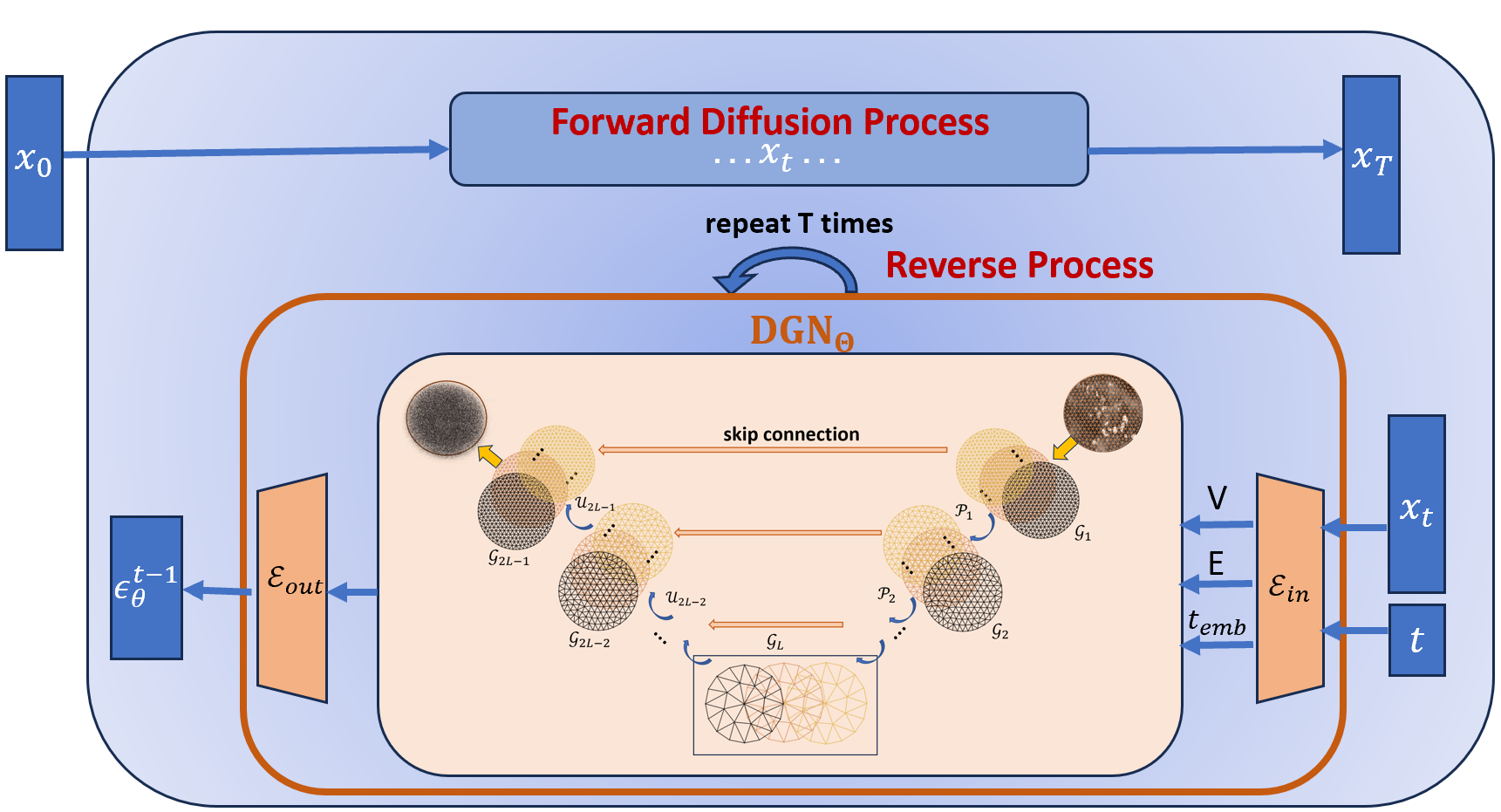}\\
    (a)\\
    \includegraphics[width=0.8\textwidth]{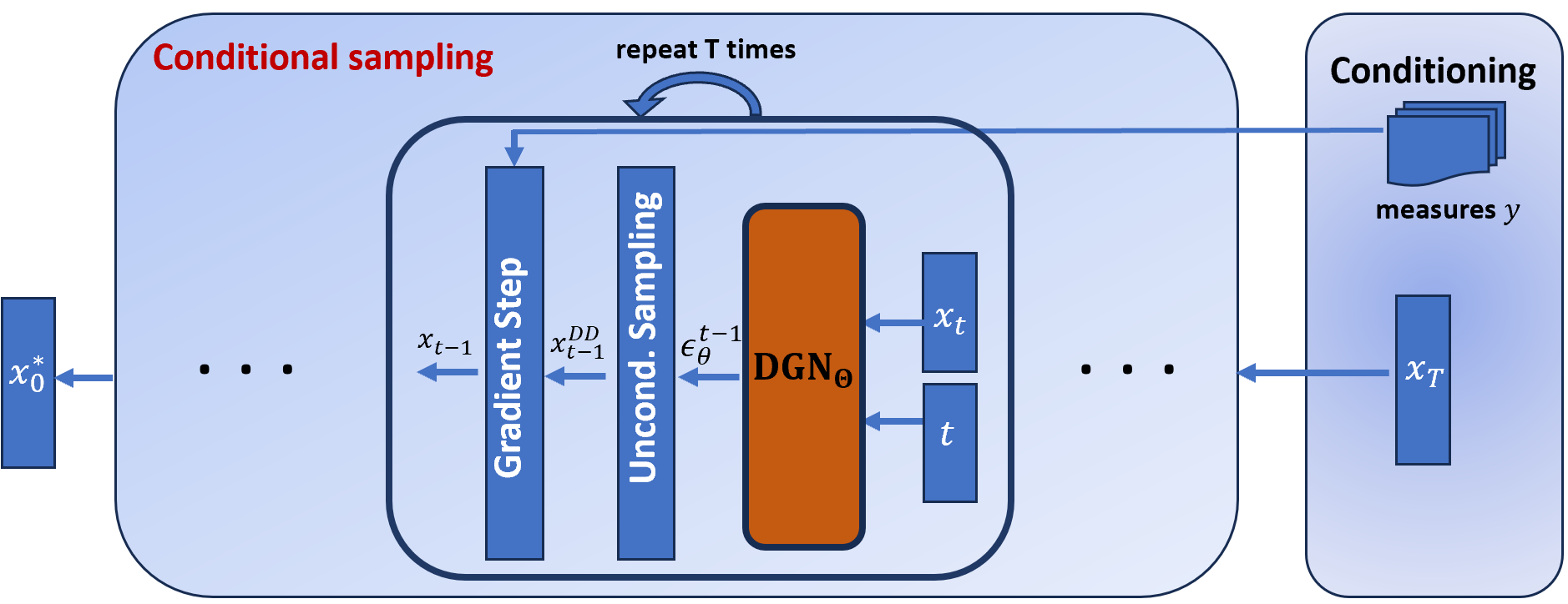}\\
    (b)
    \end{tabular}
    \caption{(a) Forward and reverse diffusion process, in the orange box we show  the graph-based score network with skip connections (red lines); (b) Proposed RDPS.}
    \label{fig:grafico1}
\end{figure}

The $M = 2L-1$ operators $T_k$ act upon a hierarchy of $L$ graphs $\{G_1, \dots, G_L\}$, where $G_1=G$,  with decreasing spatial resolutions $|\mathcal{V}_1| > |\mathcal{V}_2| > \dots > |\mathcal{V}_L|$. Each graph $G_\ell$ is characterized by an adjacency matrix $A_\ell$, encoding spatial connectivity via a $k$-nearest-neighbor (KNN) graph at scale $\ell$. The operators process features through resolution transitions and spatial convolutions as follows:
\begin{itemize}
    \item \textbf{Encoding Operators} ($k < L$): $T_k := \mathcal{P}_k \circ \mathcal{G}_k$. Here, $\mathcal{G}_k$ performs graph convolutions at scale $\ell_k=k$ on the graph $G_k$, as detailed below. The pooling operator $\mathcal{P}_k$ maps the graph from scale $k$ to $k+1$ using Guillard’s coarsening algorithm \cite{guillard:inria-00074773}, which decimates nodes while preserving physical graph topology.
    \item \textbf{Coarsest-Level Operator} ($k = L$): Performs global graph convolutions $\mathcal{G}_L$ on the coarsest graph $G_L$ to maximize the receptive field and capture long-range physical dependencies.
    \item \textbf{Decoding Operators} ($k > L$): $T_k := \mathcal{G}_k \circ \mathcal{U}_k (\cdot, \text{skip})$. The unpooling operator $\mathcal{U}_k$ maps features from scale $2L-k+1$ back to $2L-k$ by propagating them from parent to child nodes, concatenating them with encoder skip connections to restore high-frequency spatial details. This is followed by a  convolution $\mathcal{G}_k$ at scale $\ell_k=2L-k$, as detailed below.
\end{itemize}

Within each operator $T_k$, the spatial feature extraction is performed through a composition of two sequential graph convolutional layers, $\mathcal{G}_{k}= \mathcal{C}_{k} \circ \mathcal{C}_{k}$, with the same architecture but different learned weights. 

A single layer $\mathcal{C}_{k}$ at scale $\ell_k$ updates the node and edge features by aggregating information from neighborhood vertices $\mathfrak{N}_j=\{i:(i,j)\in \mathcal{E}_{\ell_k}\}$ (induced by $A_{\ell}$) utilizing non-linear kernels $\sigma_{edge}$ and $\sigma_{node}$ with learnable weights $\Theta_e$ and $\Theta_v$. Specifically, $\mathcal{C}_{k}$ operates on edges $E$ and nodes $V$ as follows:
\begin{align}
e_{ij} &= W_e e_{ij} + \sigma_{edge}([e_{ij} \mid v_i \mid v_j] ; \Theta_e), \quad &&\text{for each } (i,j) \in \mathcal{E}_{\ell} \label{eq:edge_update} \\
v_j &= W_v v_j + \sigma_{node}([\bar{e}_j \mid v_j] ; \Theta_v), \quad &&\text{for each } j \in \mathcal{V}_{\ell} \label{eq:node_update}
\end{align}
where $\bar e_j = \sum_{i \in \mathfrak{N}_j} e_{ij} \in \mathbb{R}^{F_h}$ denotes the aggregated incoming edge features at node $j$, and $W_e, W_v$ are learnable linear weight matrices for residual connections.

The graph convolutional kernels $\sigma_{edge}$ and $\sigma_{node}$ are parameterized by $\Theta_e$ and $\Theta_v$, respectively, and incorporate layer normalization ($\text{LN}$) and the SELU activation ($\phi$). The edge convolutional kernel is defined as:
\begin{equation}
\sigma_{edge}(y_e ; \Theta_e) = \phi \left(W_{e}^{(2)} \phi \left( W_{e}^{(1)} \text{LN}(y_e) + b_{e}^{(1)} \right) + b_{e}^{(2)}\right) \, ,
\end{equation}
where $y_e = [e_{ij} \mid v_i \mid v_j] \in \mathbb{R}^{3 F_h}$ and the parameter set $\Theta_e = \{W_{e}^{(1)}, b_{e}^{(1)}, W_{e}^{(2)}, b_{e}^{(2)}\}$ includes
$W_{e}^{(1)} \in \mathbb{R}^{F_h \times (3 F_h)}$,   $W_{e}^{(2)} \in \mathbb{R}^{F_h \times F_h}$, and bias vectors $b_{e}^{(1)}, b_{e}^{(2)} \in \mathbb{R}^{F_h}$.

Similarly, the node convolutional kernel is defined as:
\begin{equation}
\sigma_{node}(y_v ; \Theta_v) = \phi \left(W_{v}^{(2)} \phi \left( W_{v}^{(1)} \text{LN}(y_v) + b_{v}^{(1)} \right) + b_{v}^{(2)}\right) \, , 
\end{equation}
where the input is $y_v = [\bar{e}_j \mid v_j]\in \mathbb{R}^{F_h + F_h}$. The parameters $\Theta_v = \{W_{v}^{(1)}, b_{v}^{(1)}, W_{v}^{(2)}, b_{v}^{(2)}\}$ includes $W_{v}^{(1)} \in \mathbb{R}^{F_h \times 2F_h}$, $W_{v}^{(2)} \in \mathbb{R}^{F_h \times F_h}$, and bias vectors $b_{v}^{(1)}, b_{v}^{(2)} \in \mathbb{R}^{F_h}$.

\section{Electrical Impedence Tomography: a Case Study}\label{sec:eit}
Electrical impedance tomography (EIT) is an imaging modality aimed at recovering the conductivity of a body through electrical measurements taken at the boundary \cite{cheney-isaacson-newell-1999,borcea-2002}. From the mathematical modelling perspective, EIT is a nonlinear inverse problem that is known
to be severely ill-posed \cite{alessandrini,mandache}.
In other words, small perturbations in boundary measurements can lead to large variations in the
reconstructed conductivity.
As a consequence, EIT constitutes a challenging test case for Bayesian inverse
methods, particularly in settings where measurement noise and limited data are
present.

In this paper, EIT is used as a case study to evaluate the behavior of the proposed regularized diffusion
posterior sampling method in a severely ill-posed regime.
Rather than focusing on recovery guarantees, the goal is to assess how diffusion
priors and explicit regularization can be combined to obtain stable and physically
plausible posterior samples.

Let us now describe the mathematical model for EIT. Let 
 $\sigma\in L_+^\infty(\Omega)$ denote the spatially varying
conductivity distribution inside a bounded domain
$\Omega \subset \mathbb{R}^d$. We assume that $\sigma$ is bounded by below and above by positive constants.
Given an applied mean-zero current pattern $g\in L^2(\partial\Omega)$ on $\partial\Omega$, the electric potential $u\in H^1(\Omega)$ satisfies the
following Neumann boundary value problem
\begin{equation*}
\nabla \cdot \big( \sigma\,\nabla u \big) = 0
\quad \text{in } \Omega,\qquad \sigma\partial_\nu u=g\quad \text{on } \partial\Omega.
\end{equation*}
The data consist of electrical measurements of the boundary voltage $u|_{\partial\Omega}$ for different choices of current patterns $g$. More formally, the data can be represented by the Neumann-to-Dirichlet map $\Lambda_\sigma\colon L_0^2(\partial\Omega)\to L_0^2(\partial\Omega)$, defined by $g\mapsto u|_{\partial\Omega}$, where $L_0^2(\partial\Omega)$ denotes the space of mean-zero square-integrable functions on $\partial\Omega$. The corresponding inverse problem consists of recovering $\sigma$ from the knowledge of $\Lambda_\sigma$, and is known as the Calderón problem in the mathematical literature \cite{calderon1980inverse,calderon-book}. 

In practice, the applied currents and the corresponding measurements on the boundary are finite, and relative to a fixed configuration of electrodes. In the numerical simulations of this paper, we model the electrodes and the corresponding boundary value problem with the complete electrode model (we refer to \cite{Somersalo} for the details).
In this context, the forward operator maps the conductivity $\sigma$ to boundary voltage
measurements and can be written abstractly as
\begin{equation*}
y = \mathcal{F}(\sigma) + \eta,
\end{equation*}
where $\mathcal{F}(\cdot)$ denotes the nonlinear parameter-to-observable map and $\eta$ models
measurement noise.
From an inverse problem perspective, reconstructing $\sigma$ from $y$ is severely
ill-posed.
Thus, the likelihood $p(y \mid \sigma)$ is often weakly informative, and meaningful
posterior inference requires the incorporation of strong prior information.

\section{Experimental Results}\label{sec:experiments}

This section validates the proposed regularized diffusion posterior sampling by applying it to the non-linear, highly ill-conditioned EIT inverse problem. We demonstrate that the regularization integrated into the diffusion model is a critical component for achieving a physically plausible approximation of the solution. All our experiments are in $d=2$ dimensions, and $\Omega$ is the unit ball.

Before detailing the experimental results, we describe the datasets in Section \ref{sec:dataset} and define the performance metrics used for evaluation in Section \ref{sec:fm}.

In the first experiment (section~\ref{sub:ex1})  the performance of the proposed DDIM-RDPS diffusion model is investigated. The first phase focuses on a baseline comparison between DDPM and DDIM without any regularization, and the second phase takes the superior model, DDIM, and applies it across four distinct configurations: one without regularization and three employing different regularization techniques.

The second experiment (section~\ref{sub:ex2}) investigates the robustness of the solver against perturbations of EIT voltage measurements by additive white Gaussian and Laplacian noise. 
The experiment simulates a realistic scenario where noise levels scale proportionally with signal amplitude.

In the third case study (section~\ref{sub:ex3}), we assess how well DDIM-RDPS generalizes to out-of-distribution structural data and demonstrate its robustness when handling empirical datasets. 

The fourth experiment (section~\ref{sub:ex4}) demonstrates the competitive performance of the proposed DDIM-RDPS model against leading state-of-the-art methods in EIT reconstruction. This benchmark serves to validate the model's competitive edge, comparing its reconstruction accuracy, and resolution against established variational, Bayesian and deep learning-based solvers currently used in inverse EIT reconstruction.

\subsection{Datasets for the Training of the Graph-based Score Network}
\label{sec:dataset}

The synthetic EIT measurements were generated on a fine mesh $\Tau_h$ composed of $N = 3766$ vertices.
To avoid the inverse crime, the inversion and diffusion-based reconstruction process 
were performed on a coarser mesh ($N = 1602$).
This ensures that the forward operator used for data generation
differs from the one employed during reconstruction,
preventing artificial over-optimistic results.

All examples simulate a circular tank slice of
unitary radius. In the circular boundary ring, $p=32$ equally spaced electrodes are located. 
The conductivity of the background liquid is set to be $x_0 = 1.0$ $ S m^{-1}.$

The study employs four datasets consisting of circular tanks with one to three inclusions of random locations, random sizes 
and random constant conductivities. DATASET1 and DATASET2 contain circular and triangular inclusions, respectively (conductivity values are sampled with respect to the uniform distribution $\mathcal{U} [0.5, 1.5]$). DATASET3 features blob-shaped inclusions with a wider conductivity range ($\mathcal{U}[0.3, 1.7]$). These datasets consist of $5000$ examples on the coarser mesh with $N = 1602$ vertices. Training has been performed with ADAM  optimizer \cite{Adam}, using a learning rate equal to $2.5 \times 10^{-3}$ through $1000$ epochs, with a mini-batch size of $10$ instances. Finally, an out-of-distribution set, DATASET4, is introduced using horseshoe-shaped inclusions to test the model's robustness.

The acquisition of $m$ measurements is simulated through opposite injection - adjacent measurement protocol. In all the examples, the setup
is considered blind, that is no a priori information about the sizes or locations of the inclusions is known.
In the forward calculations of $\mathcal{F}$, we applied the KTCFwd forward solver, a two-dimensional version of the FEM described in \cite{Vetal1999}, kindly provided by the authors\footnote{\url{https://github.com/CUQI-DTU/KTC2023-CUQI4}}, which is based on a FEM implementation of the CEM model on triangle elements. The electric potential is discretized using second-order polynomial basis functions, while the conductivity is discretized on the nodes using linear basis functions on triangle elements.

\subsection{Performance Metrics}
\label{sec:fm}

Reconstructions are assessed qualitatively through visual inspection of artifacts and quantitatively by using reconstruction error metrics.

Let $x^{GT}\in\mathbb{R}^{N}$ denote the ground-truth conductivity
distribution and $x^{*}\in\mathbb{R}^{N}$ be the reconstructed conductivity,
both represented on the same mesh with $N$ vertices.
The mean squared error (MSE) is defined as
\begin{equation*}
\mathrm{MSE}(x^{GT},x^{*})
:=\frac{1}{N}\left\lVert x^{GT}-x^{*}\right\rVert_2^2,
\end{equation*}
and the corresponding $\mathrm{RMSE}
:=\sqrt{\mathrm{MSE}}$, which
 measures the average absolute deviation of the reconstructed conductivity
from the ground truth (in the same physical unit as the conductivity).
Further,  where appropriate, we report the relative $L^2$ error to quantify the discrepancy
normalized by the magnitude of the ground truth:
\begin{equation}
\mathrm{RelErr}(x^{GT},x^{*})
:=\frac{\left\lVert x^{GT}-x^{*}\right\rVert_2}
{\left\lVert x^{GT}\right\rVert_2}.
\label{eq:relerr}
\end{equation}

In addition, we report the structural similarity index measure (SSIM) to assess the similarity between the reconstructed and ground-truth conductivity distributions beyond pointwise error metrics.
For a local window $\Omega_k$ (one-ring neighborhood) along the ordered conductivity vector, the SSIM is defined as
\begin{equation}
\mathrm{SSIM}_k(x^{GT},x^{*})
:=
\frac{(2\mu^{GT}_k\mu^{*}_k + C_1)(2\sigma^{GT,*}_k + C_2)}
{\left((\mu^{GT}_k)^2 + (\mu^{*}_k)^2 + C_1\right)
\left((\sigma^{GT}_k)^2 + (\sigma^{*}_k)^2 + C_2\right)},
\end{equation}
where $\mu^{GT}_k$ and $\mu^{*}_k$ denote the local means, $(\sigma^{GT}_k)^2$ and $(\sigma^{*}_k)^2$ the local variances, and $\sigma^{GT,*}_k$ the local covariance within the window $\Omega_k$. The constants $C_1$ and $C_2$ are included for numerical stability.
The final SSIM score is obtained by averaging over all local windows:
\begin{equation}
\mathrm{SSIM}(x^{GT},x^{*})
:=
\frac{1}{K}\sum_{k=1}^{K}\mathrm{SSIM}_k(x^{GT},x^{*}),
\end{equation}
where $K$ is the number of evaluated windows. A value closer to $1$ indicates higher similarity between the reconstructed and ground-truth conductivity vectors.

\setlength{\tabcolsep}{3pt}
\renewcommand{\arraystretch}{1.0}
\newcolumntype{C}[1]{>{\centering\arraybackslash}m{#1}}

\newcommand{\imgmetric}[3][2.2cm]{%
  \makecell[c]{\includegraphics[width=#1]{#2}\\[-1mm]{\scriptsize #3}}%
}

\newcommand{\imgonly}[2][2.6cm]{%
  \includegraphics[width=#1]{#2}%
}

\subsection{Example 1: Diffusion Posterior Sampling Validation}\label{sub:ex1}

\subsubsection{DDPM-DPS vs DDIM-DPS}
In this experiment, we compare the performance of two sampling strategies, DDPM and DDIM, on DATASET1.
To isolate the effect of the sampler itself, we disable the regularization term $R$.
For a fair comparison, we ensure that for each test sample both DDPM-DPS and DDIM-DPS start from the \emph{same} initial noise realization $x_T$.

We evaluate both samplers on $50$ test samples. For each sample, reconstructions are generated using DDPM and DDIM with $T=1000$ denoising steps, a value determined to be optimal through preliminary testing with smaller step sizes.

Table~\ref{tab:ddpm_ddim_50} reports the reconstruction errors for DDPM-DPS and DDIM-DPS as RMSE (mean$\pm$std) and RelErr (mean).
The sampler with the lower mean RMSE is selected as the default unconditional sampling strategy \eqref{eq:rdps} for subsequent experiments.

\begin{table}[h]
	\centering
	\caption{Example 1: DDPM-DPS vs.\ DDIM-DPS on 50 test samples (no regularization). }
	\label{tab:ddpm_ddim_50}
	\begin{tabular}{lccc}
		\hline
		Method & RMSE (mean) & RMSE (std) & RelErr (mean) \\
		\hline
		DDPM-DPS & \texttt{0.076400} & \texttt{0.018883} & \texttt{7.582\%}  \\
		DDIM-DPS & \texttt{0.072636} & \texttt{0.020798} & \texttt{7.216\%}  \\
		\hline
	\end{tabular}
\end{table}

\subsubsection{Comparison of  Different Regularization Terms}

We further investigate how different regularization terms $R$ affect the reconstruction quality for DATASET 1. The regularization parameter $\lambda$ has been tuned to obtain the optimal reconstruction results for each regularizer.
We consider the following variants of diffusion posterior sampling:

\paragraph{[No-Reg] Unregularized DDIM--DPS.}
In the baseline configuration, posterior sampling is performed using the diffusion
prior and the data-fidelity term only.
No additional explicit regularization is introduced.
This setting highlights the stabilizing effect of the learned diffusion prior
alone.

\paragraph{[Tik] Tikhonov-regularized DDIM--RDPS.}
We incorporate a classical Tikhonov regularization term of the form \eqref{eq:regT} which penalizes large conductivity values and promotes globally smooth
solutions.

\paragraph{[GTik] Generalized Tikhonov-Regularized DDIM--RDPS.}
We also consider a generalized Tikhonov regularization, as in \eqref{eq:regGT}, which penalizes spatial variations and enforces smoothness while respecting
the mesh structure.

\paragraph{[TV] Total Variation regularized DDIM--RDPS.}
Finally, we study total variation regularization, defined in \eqref{eq:regTV},
which promotes piecewise-constant conductivity profiles and preserves sharp interfaces.
TV regularization is particularly well suited for EIT settings involving
inclusions with distinct conductivity contrasts.

\begin{figure}
\begin{table}[H]

\centering
\small

\begin{adjustbox}{width=\linewidth}
\begin{tabular}{C{2.2cm} C{2.8cm} C{2.8cm} C{2.8cm} C{2.8cm} C{2.8cm}}

{GT} &
\imgonly{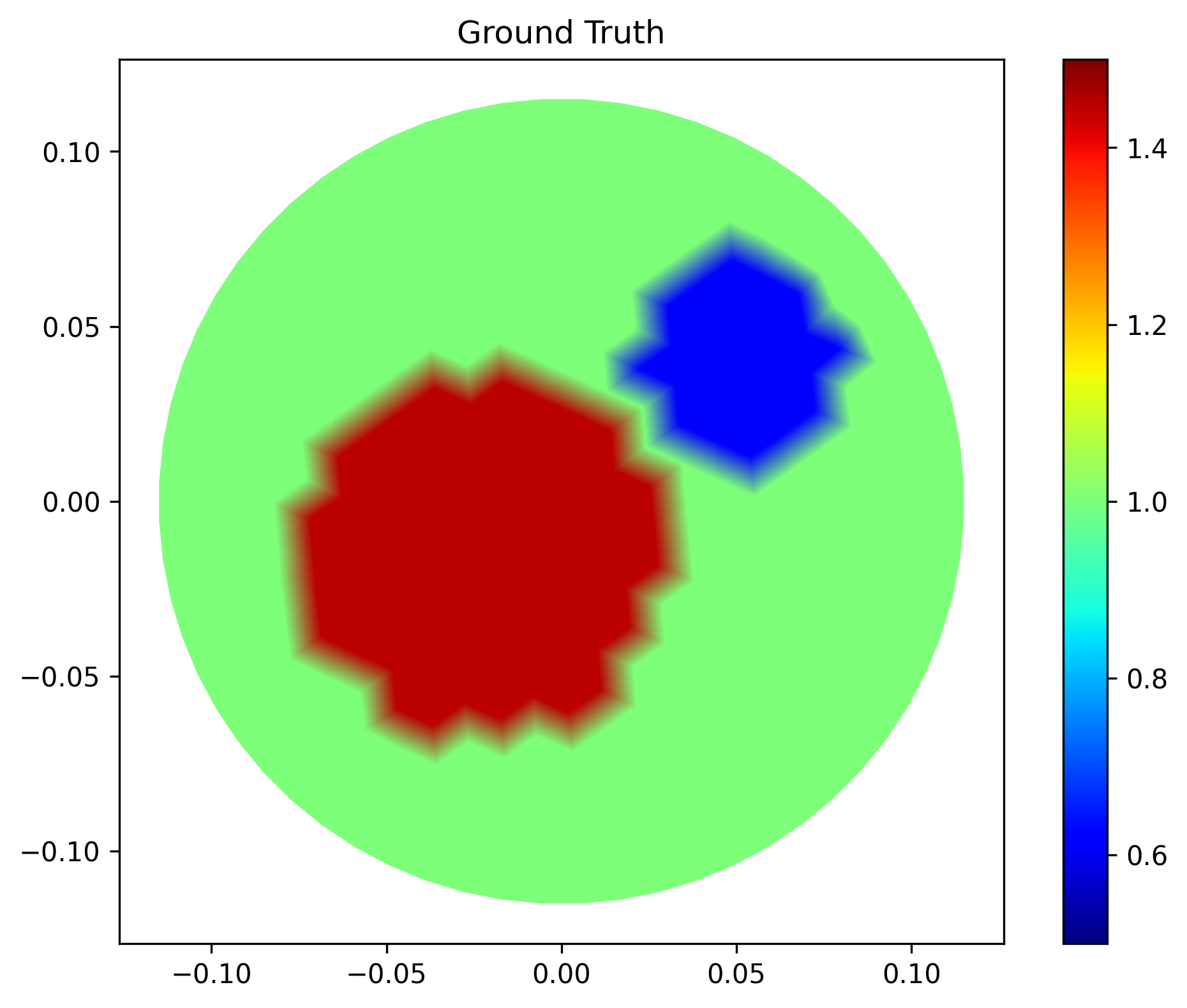} &
\imgonly{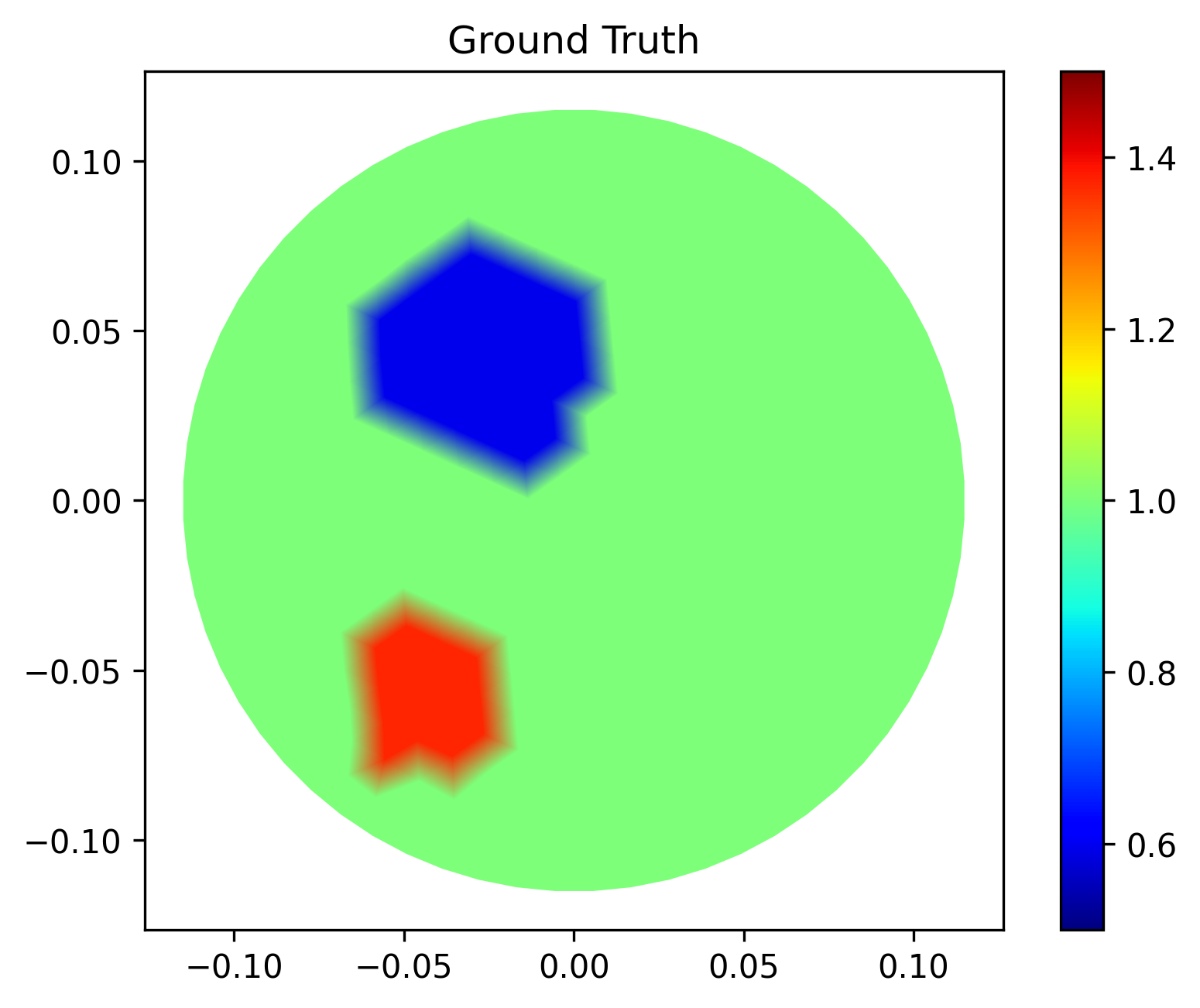} &
\imgonly{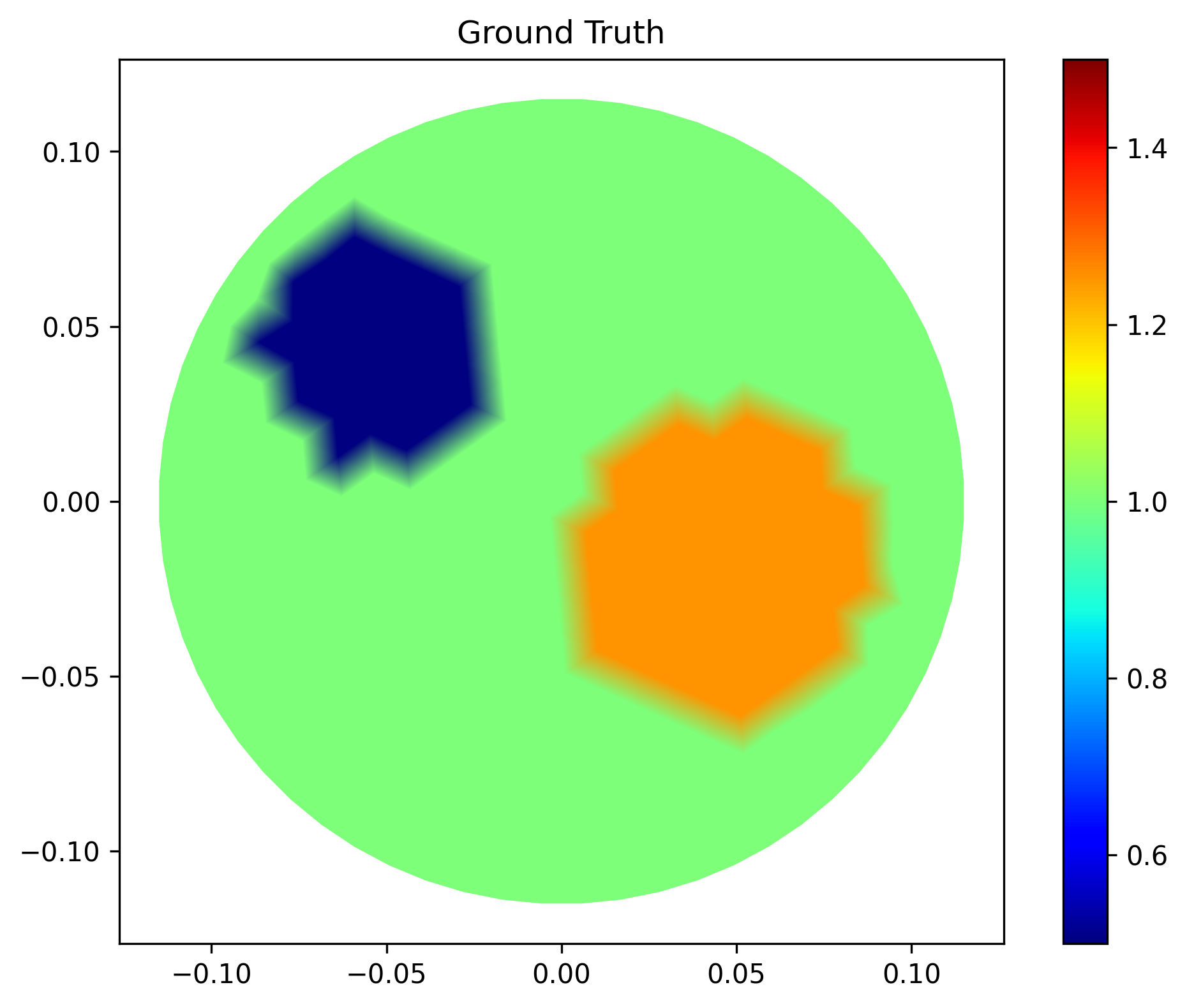} &
\imgonly{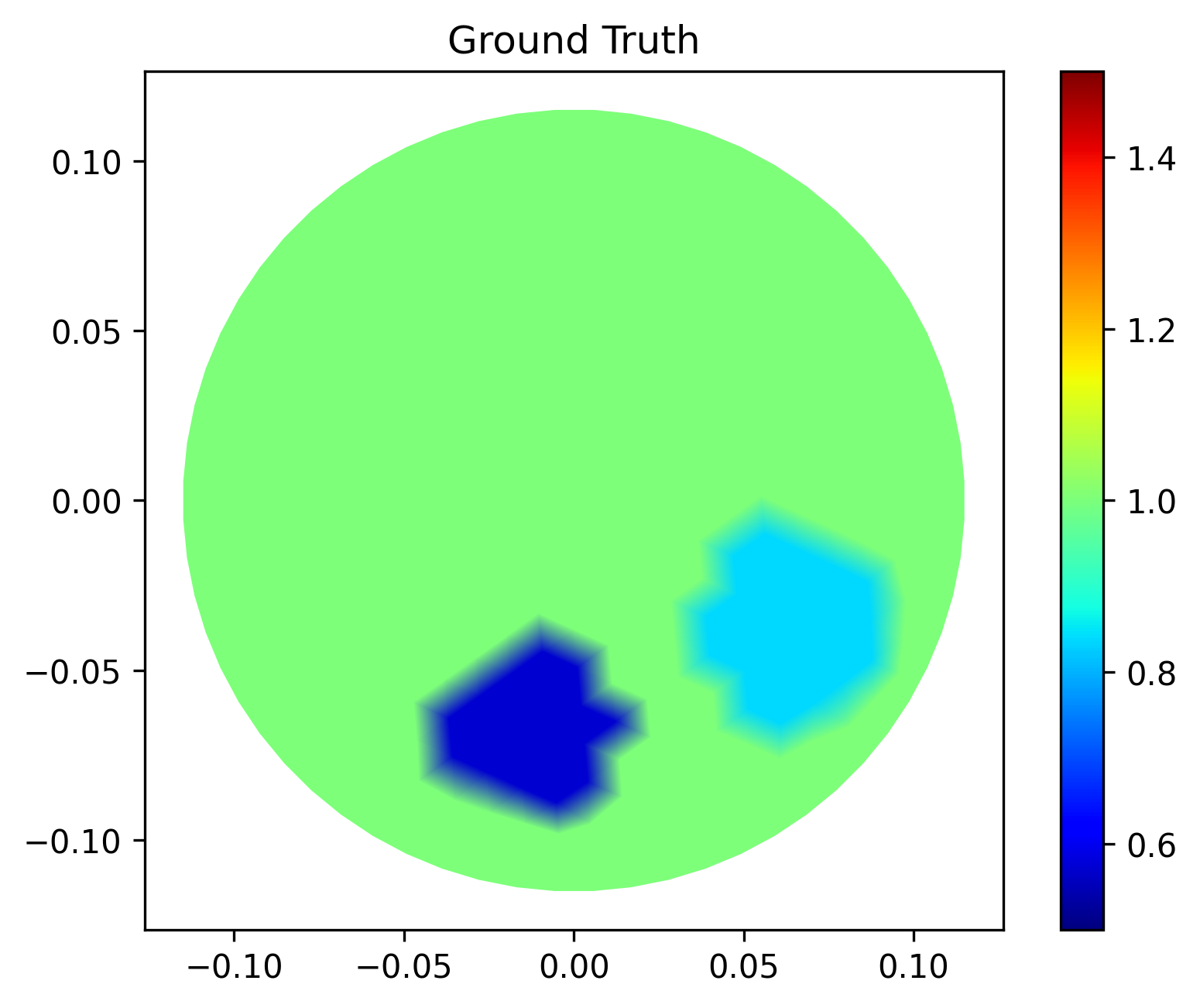} &
\imgonly{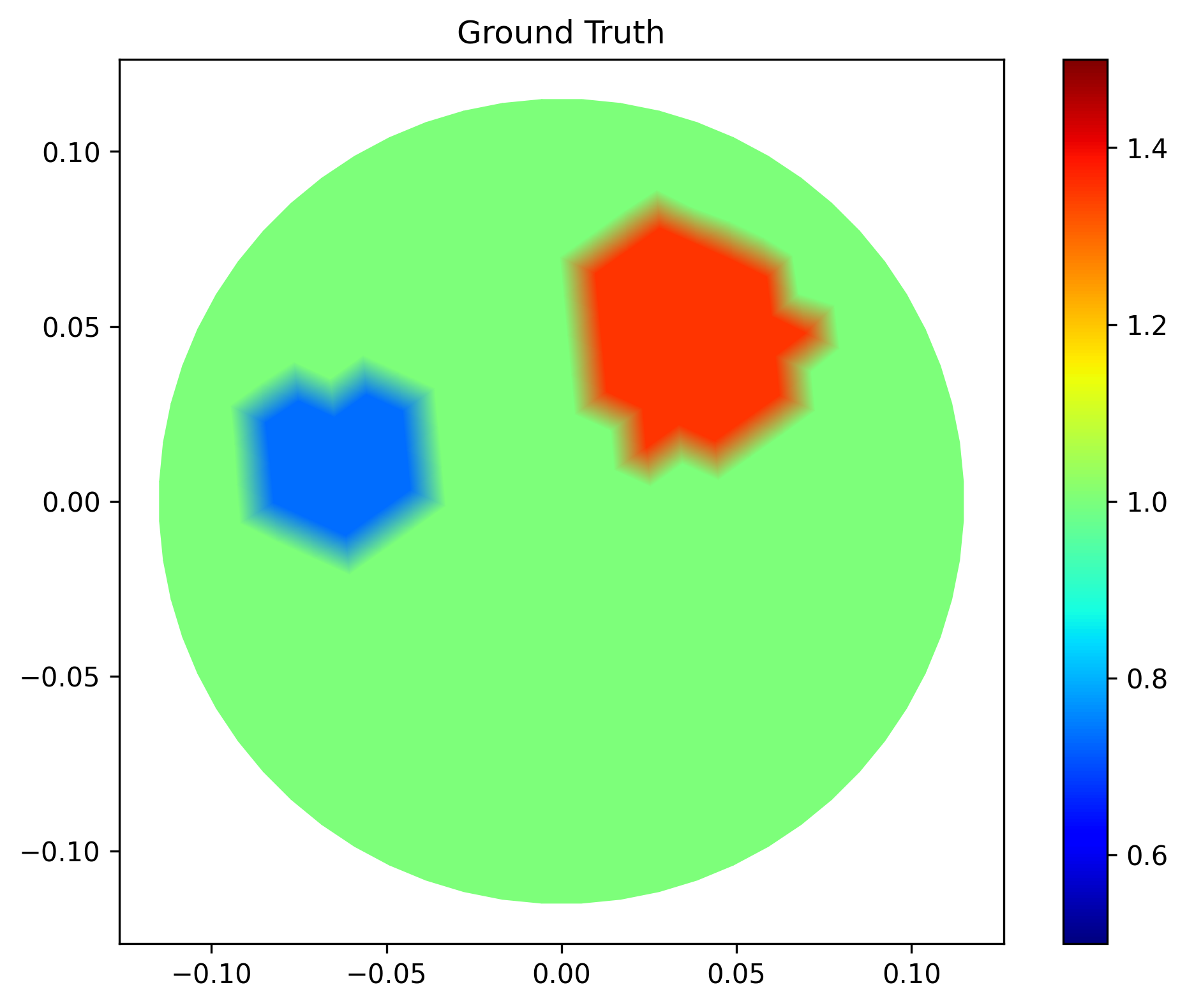} \\
\hline

No-Reg  &
\imgmetric{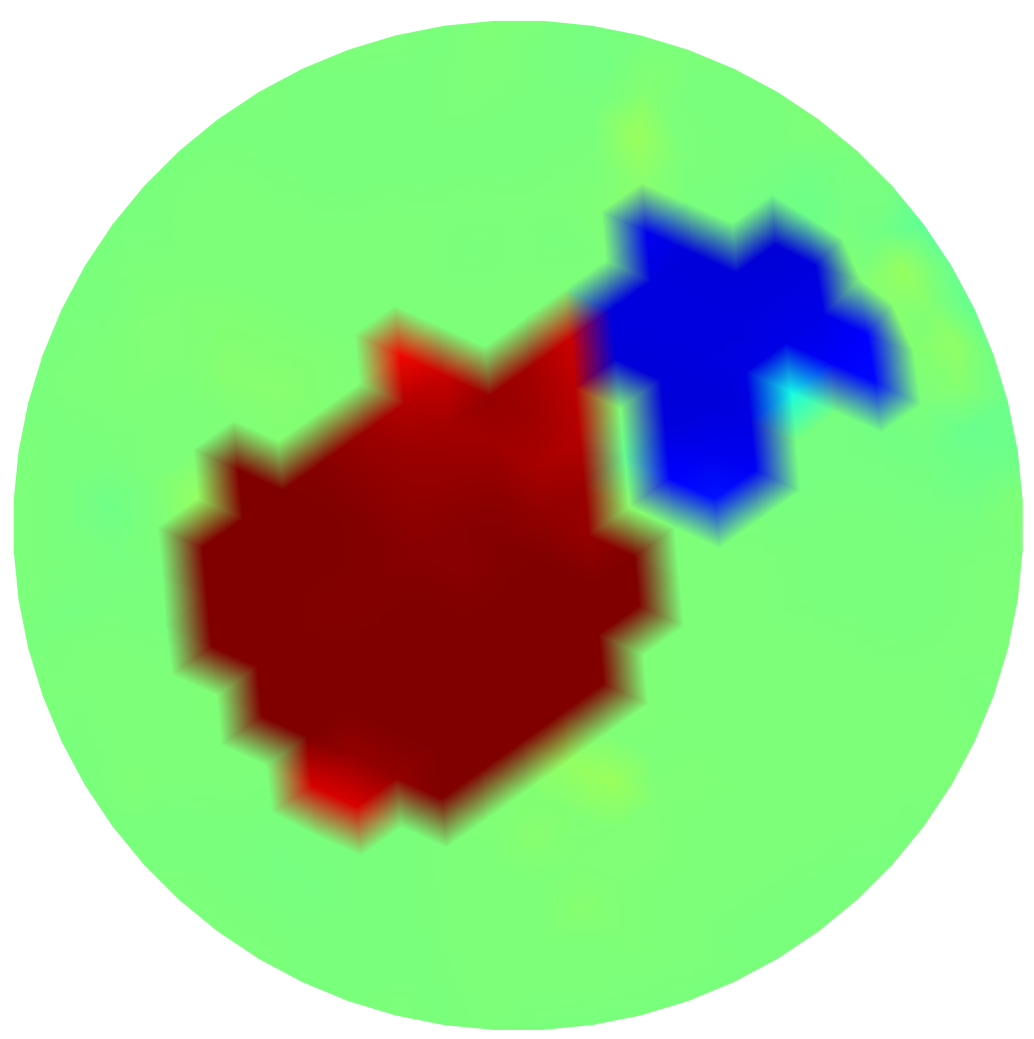}{(0.091, 8.41\%)} &
\imgmetric{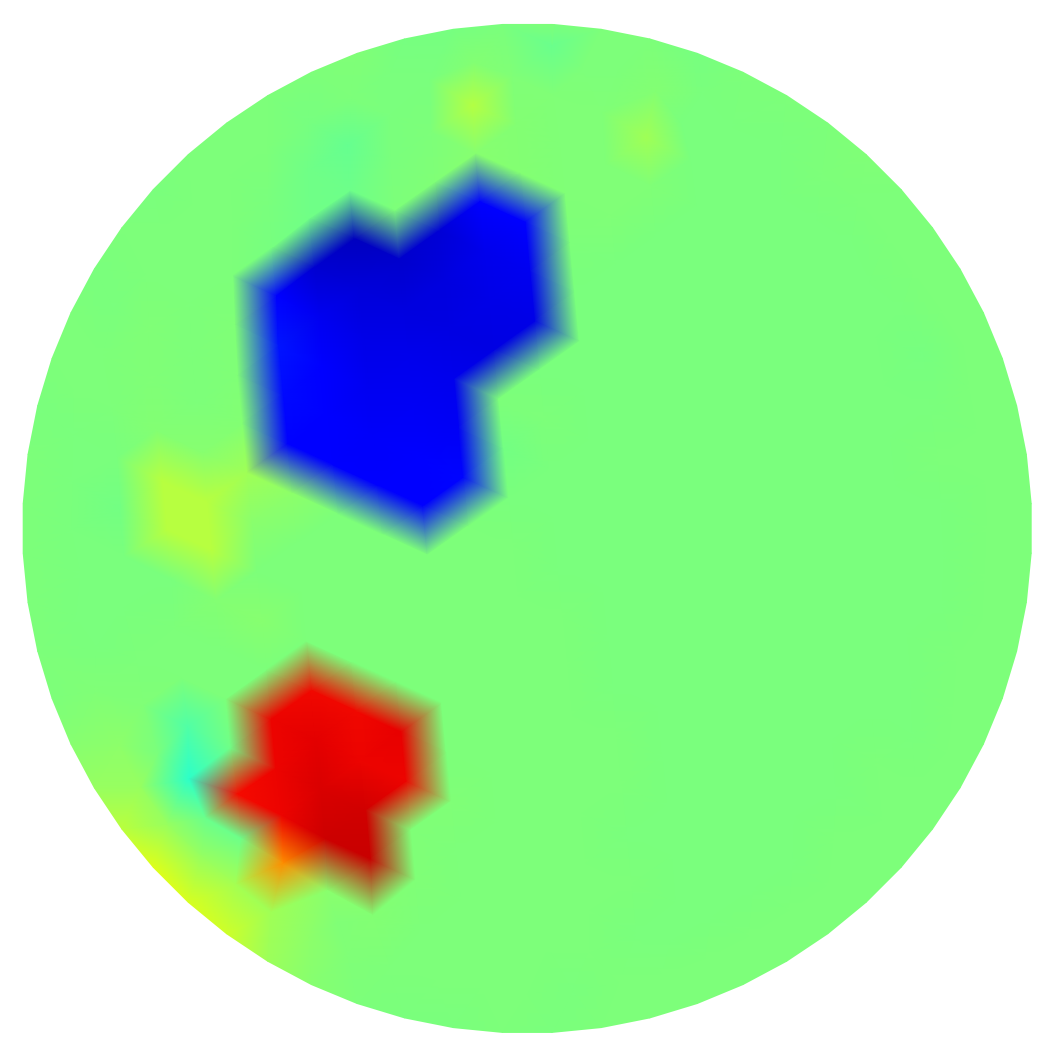} {(0.068, 6.86\%)}&
\imgmetric{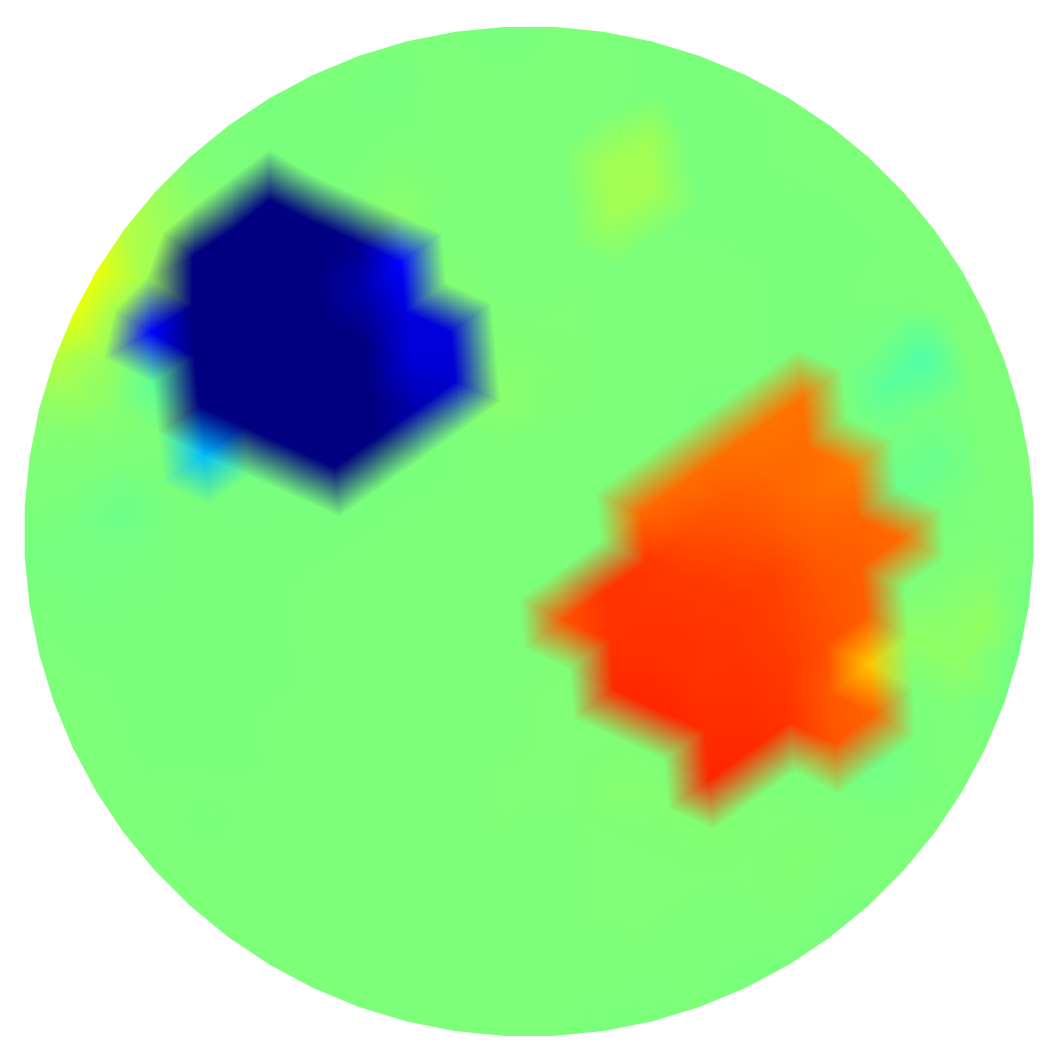} {(0.067, 6.67\%)}&
\imgmetric{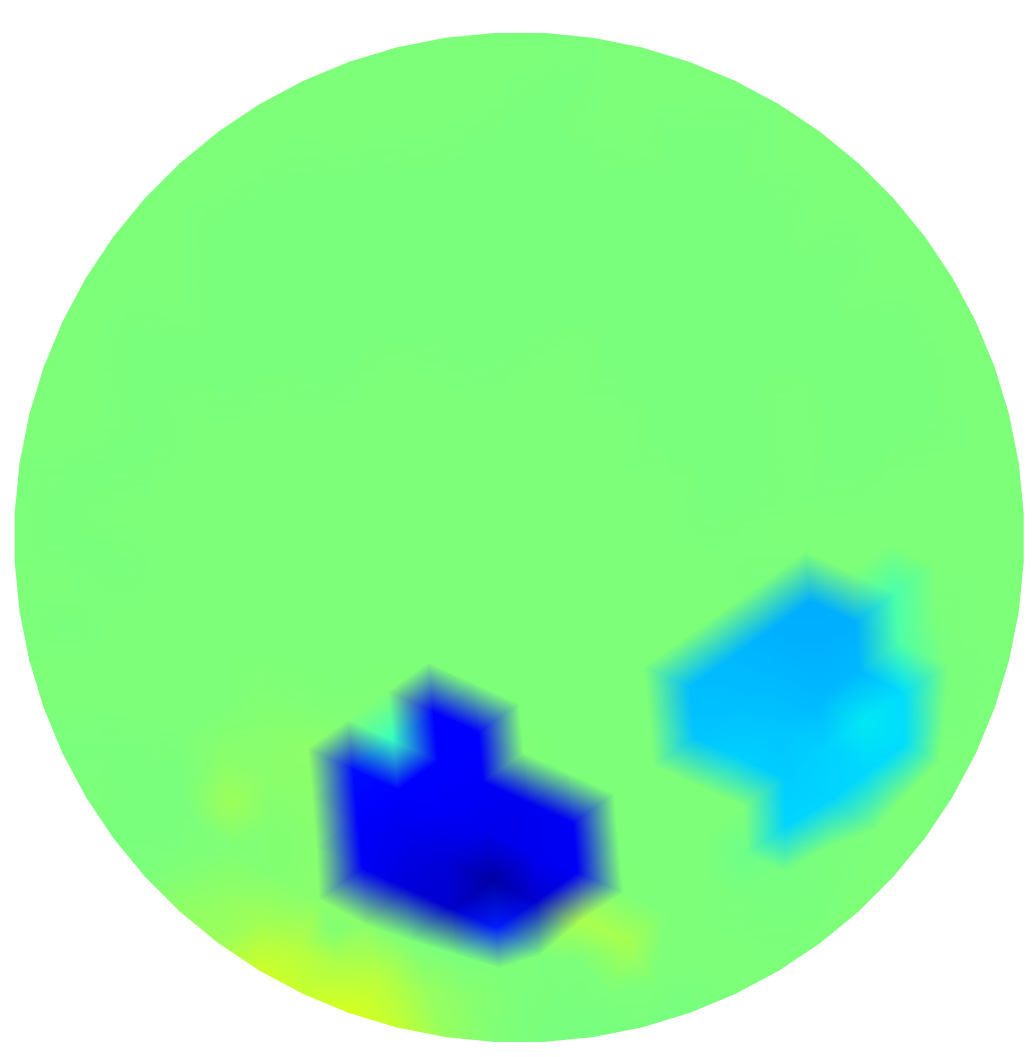} {(0.044, 4.50\%)}&
\imgmetric{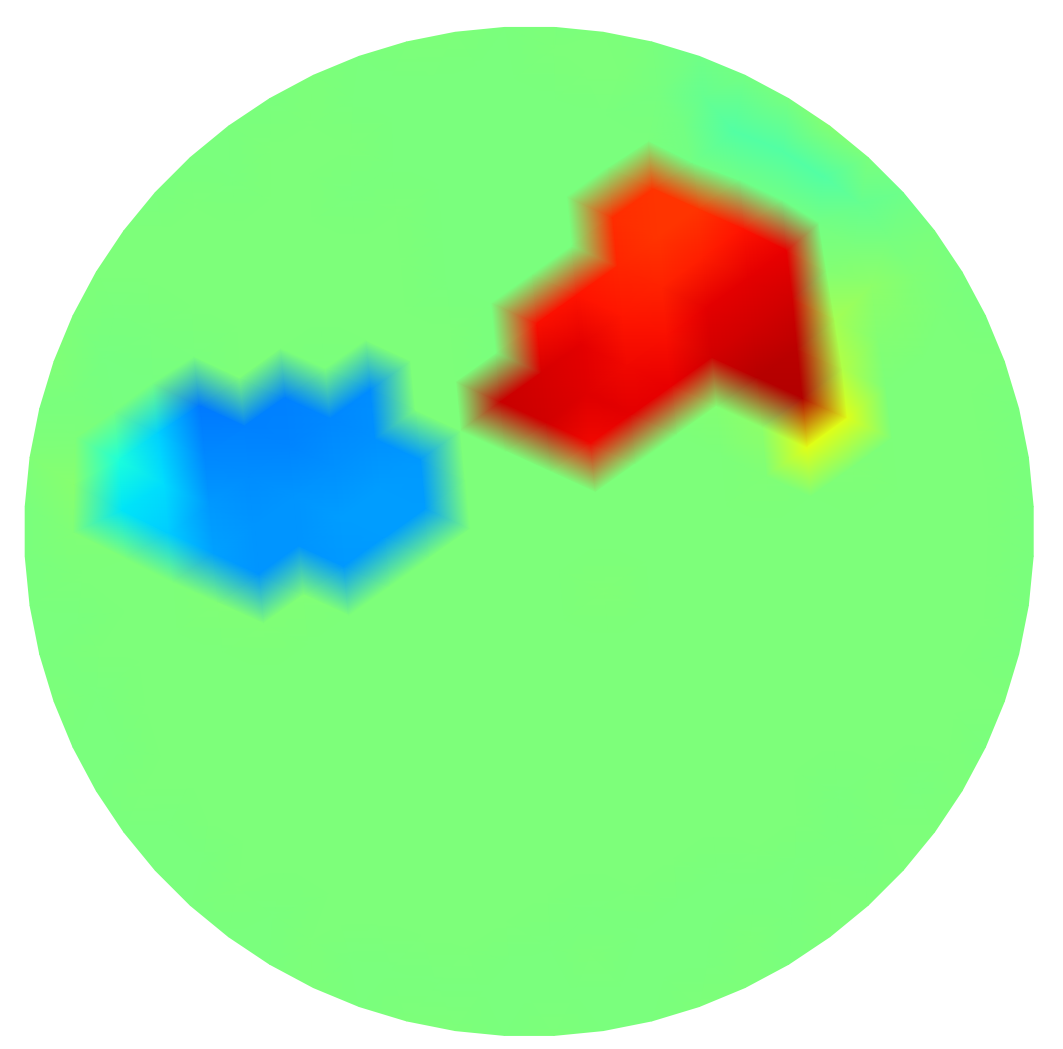} {(0.073, 7.10\%)}\\

Tik  &
\imgmetric{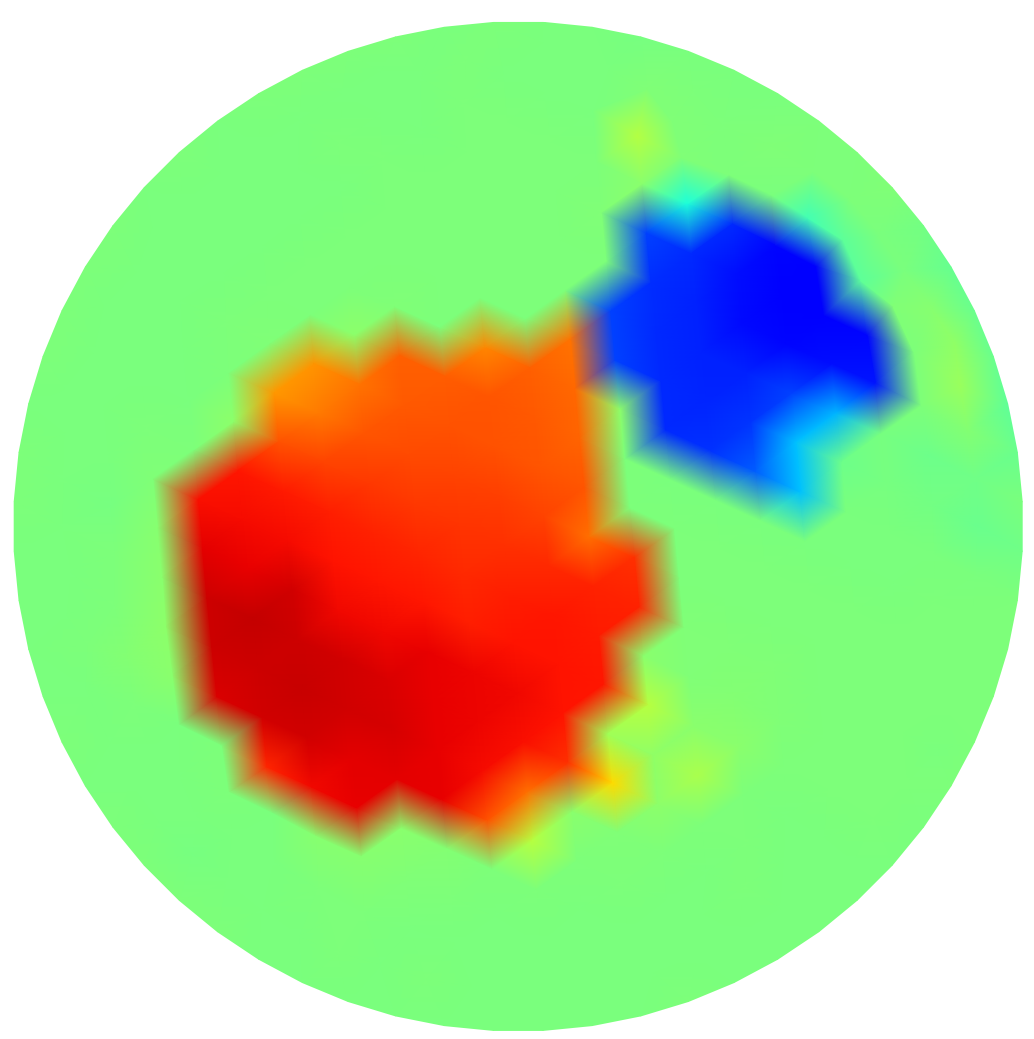}{(0.073, 6.75\%)} &
\imgmetric{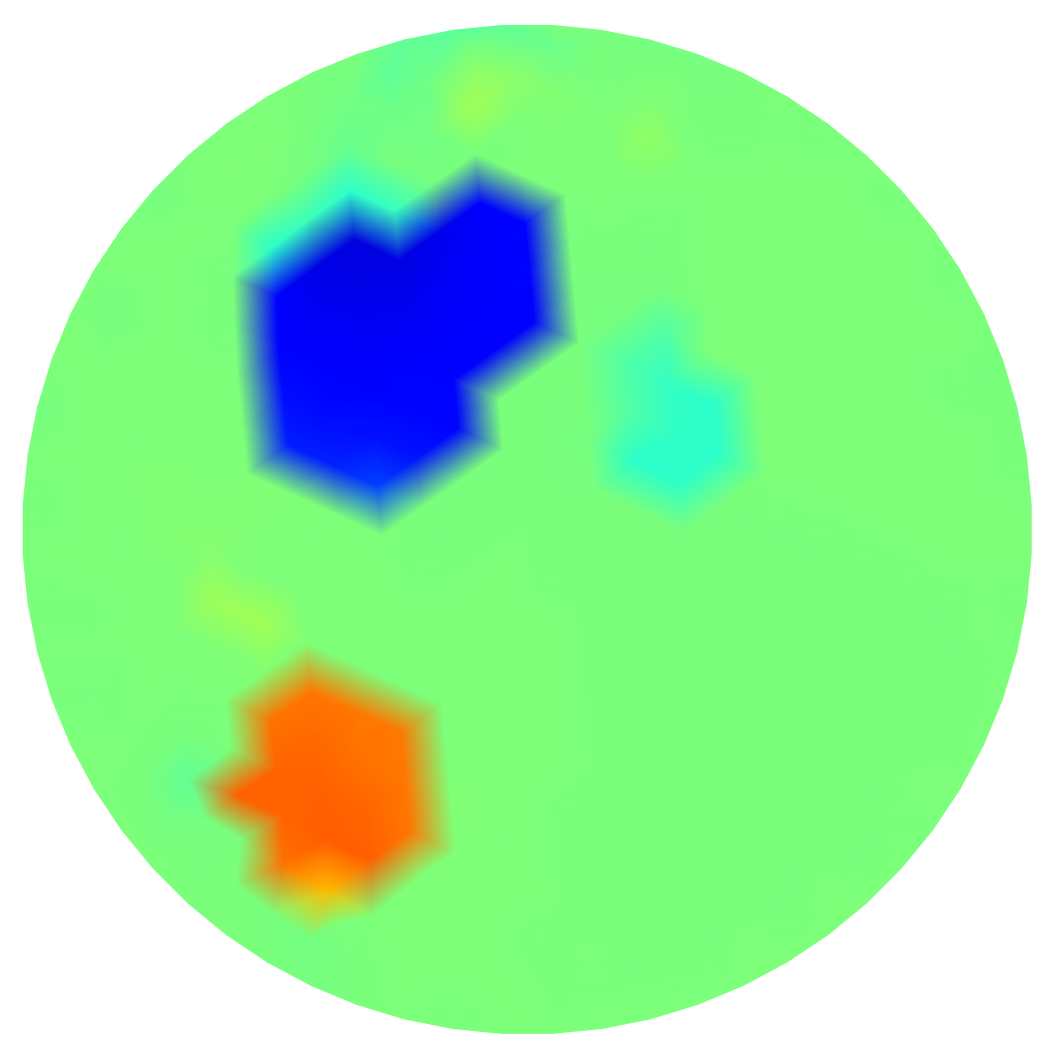} {(\textbf{0.063}, 6.39\%)}&
\imgmetric{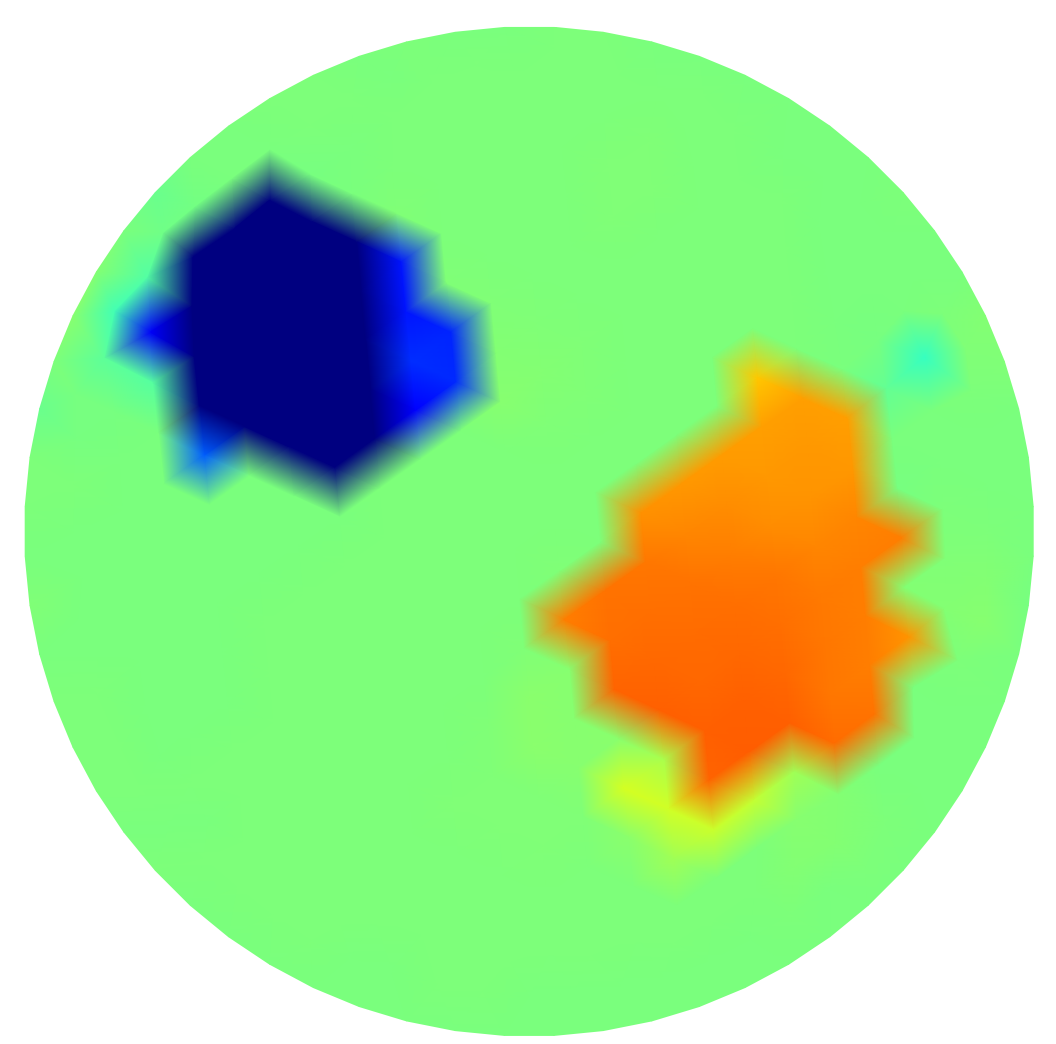} {(0.060, 5.98\%)}&
\imgmetric{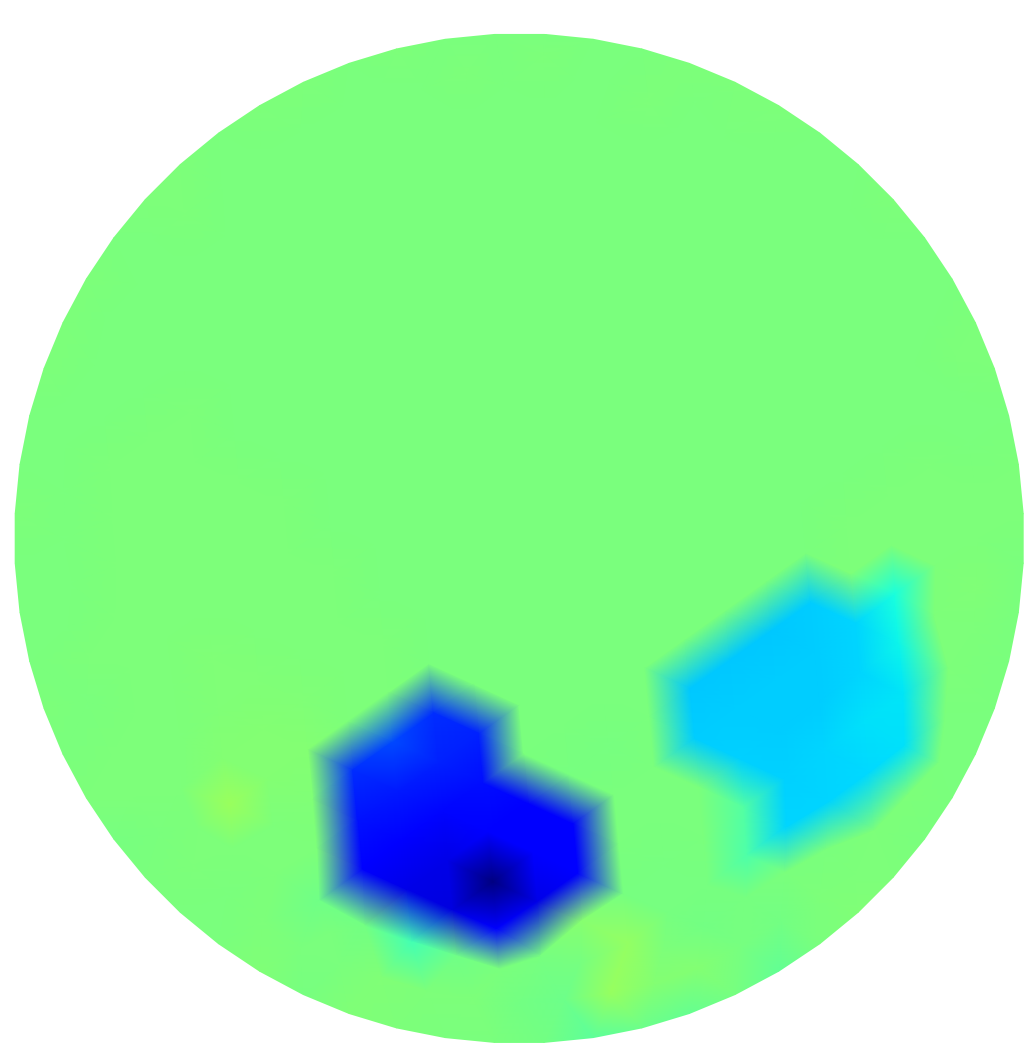} {(0.043, 4.38\%)}&
\imgmetric{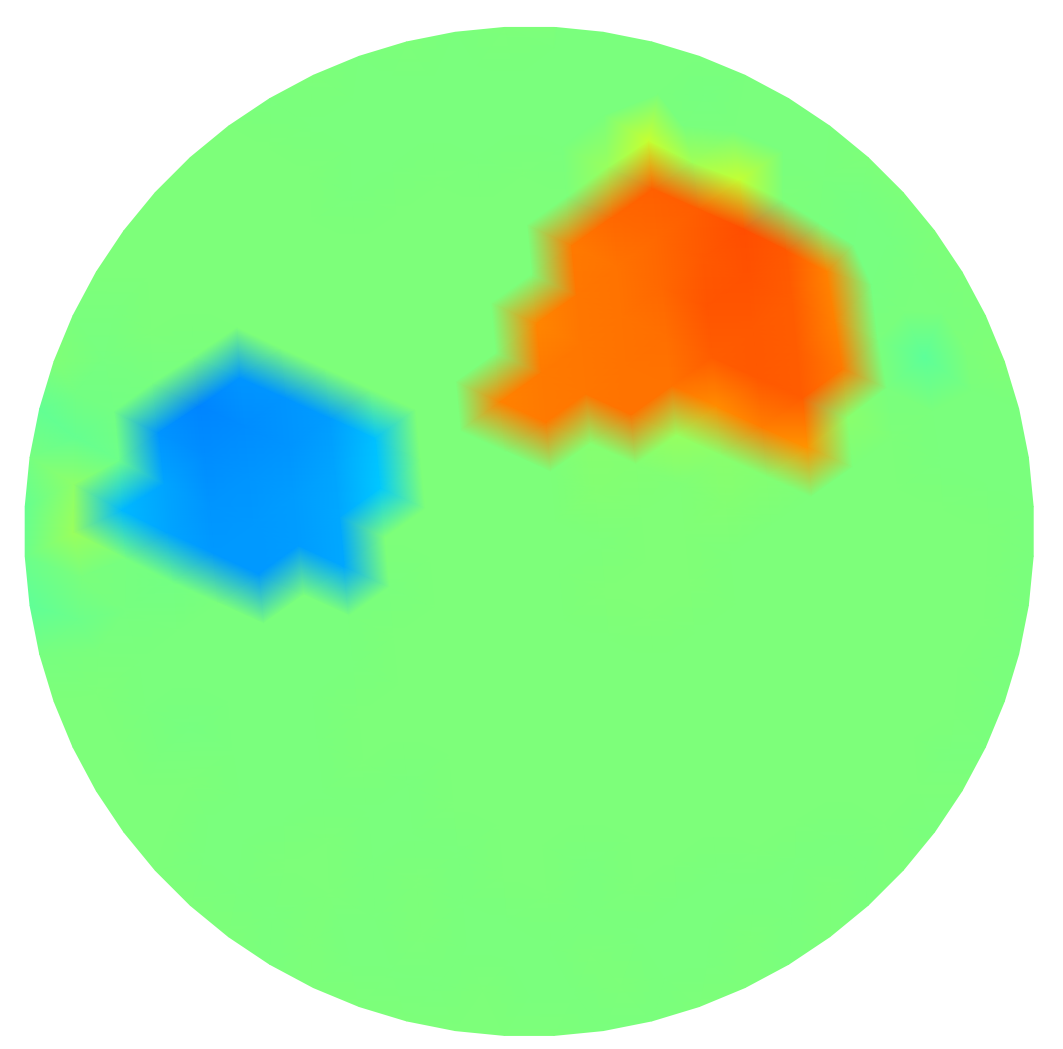} {(0.053, 5.22\%)}\\

GTik &
\imgmetric{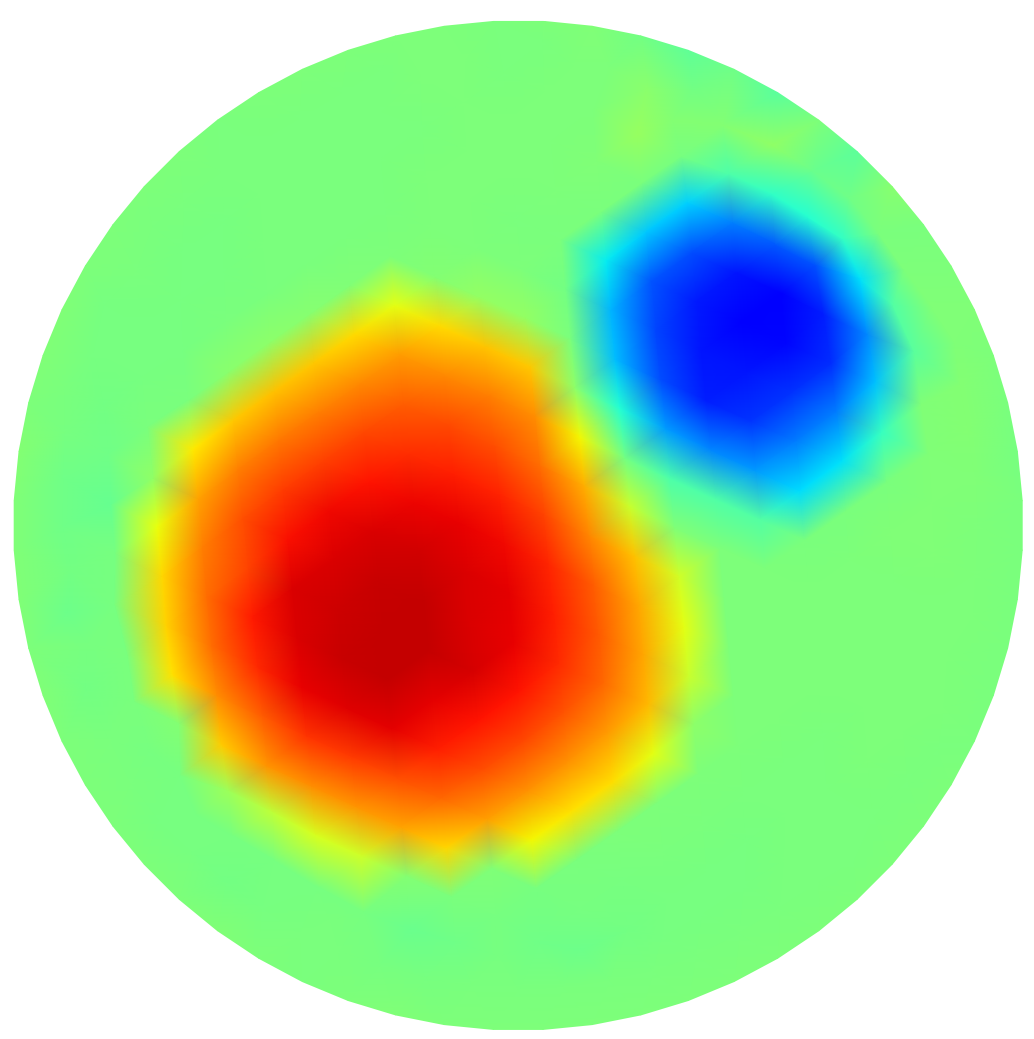} {(0.083, 7.71\%)}&
\imgmetric{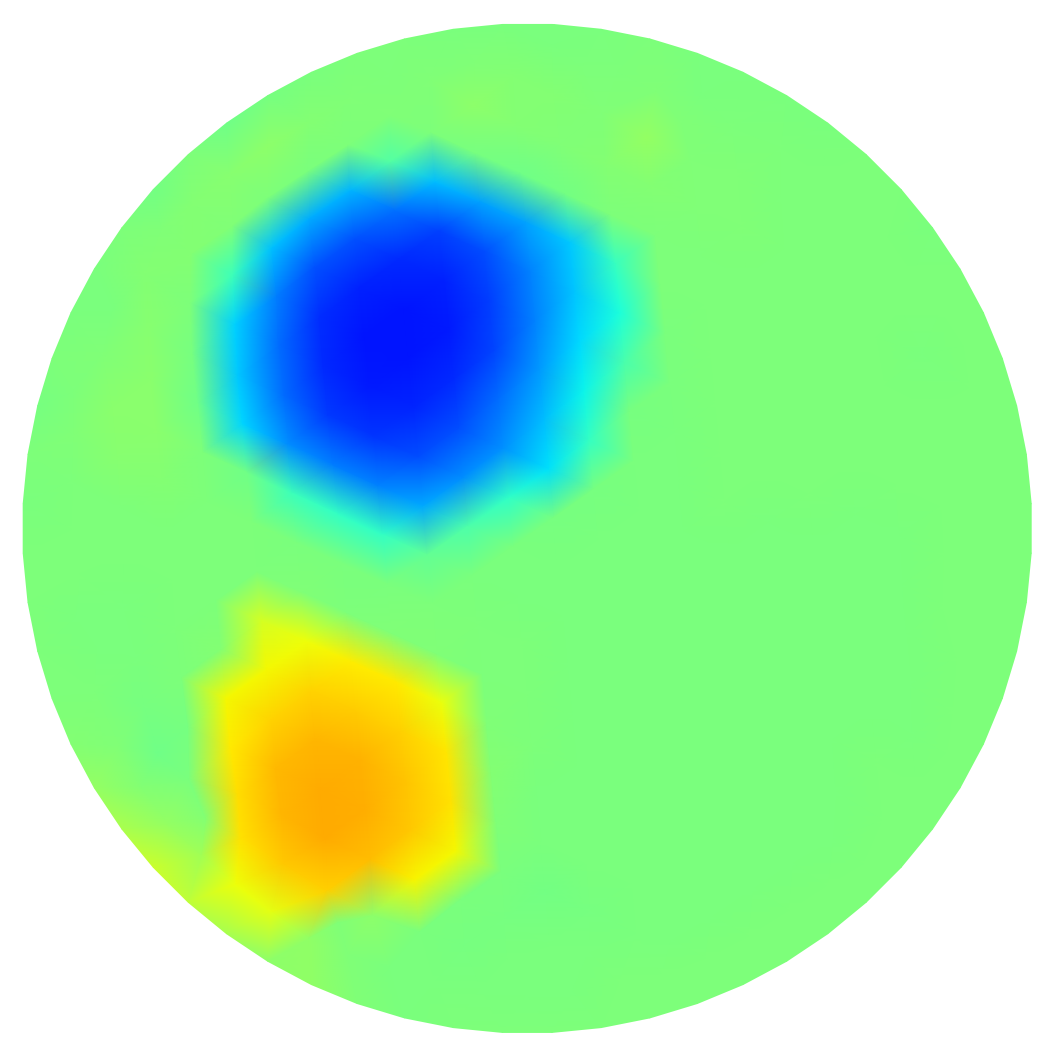} {(0.069, 6.95\%)}&
\imgmetric{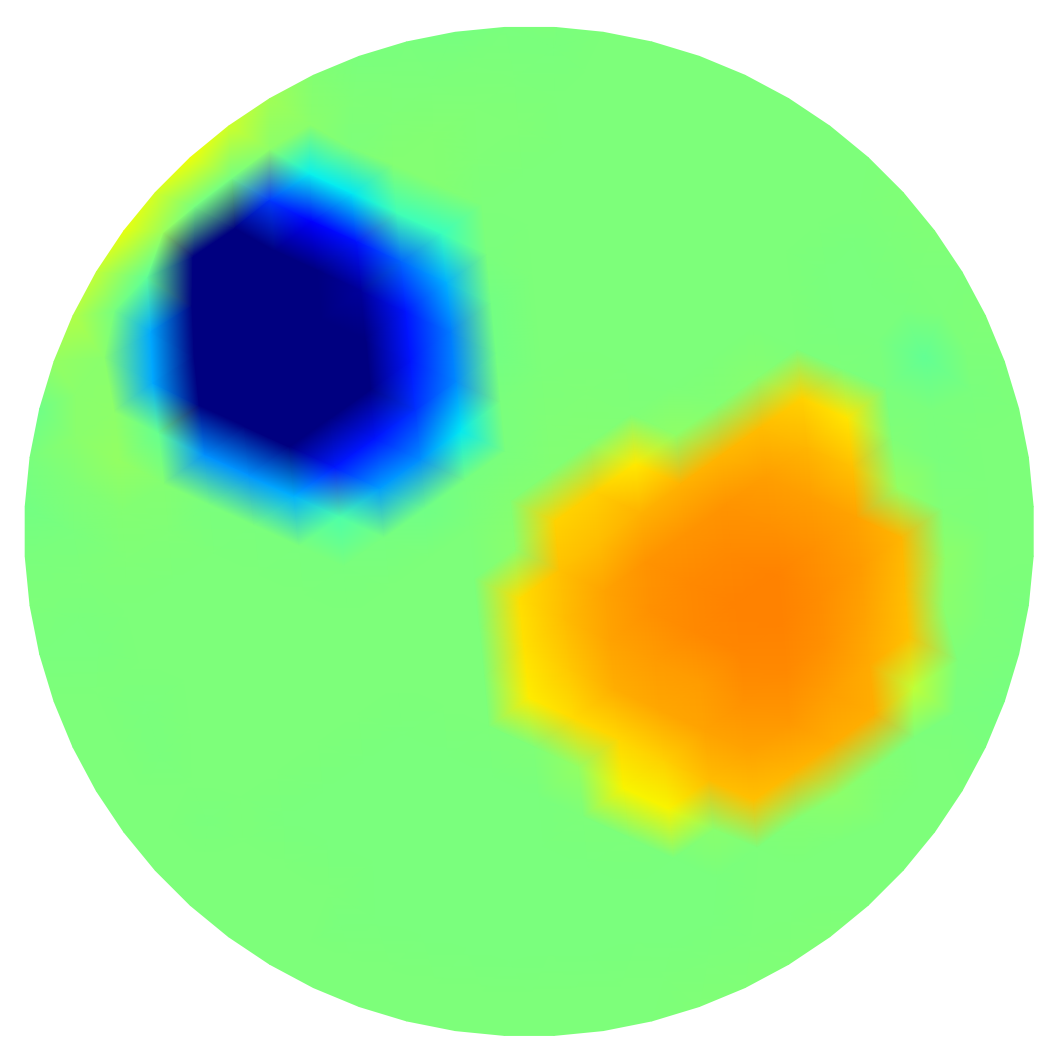} {(0.055, 5.47\%)}&
\imgmetric{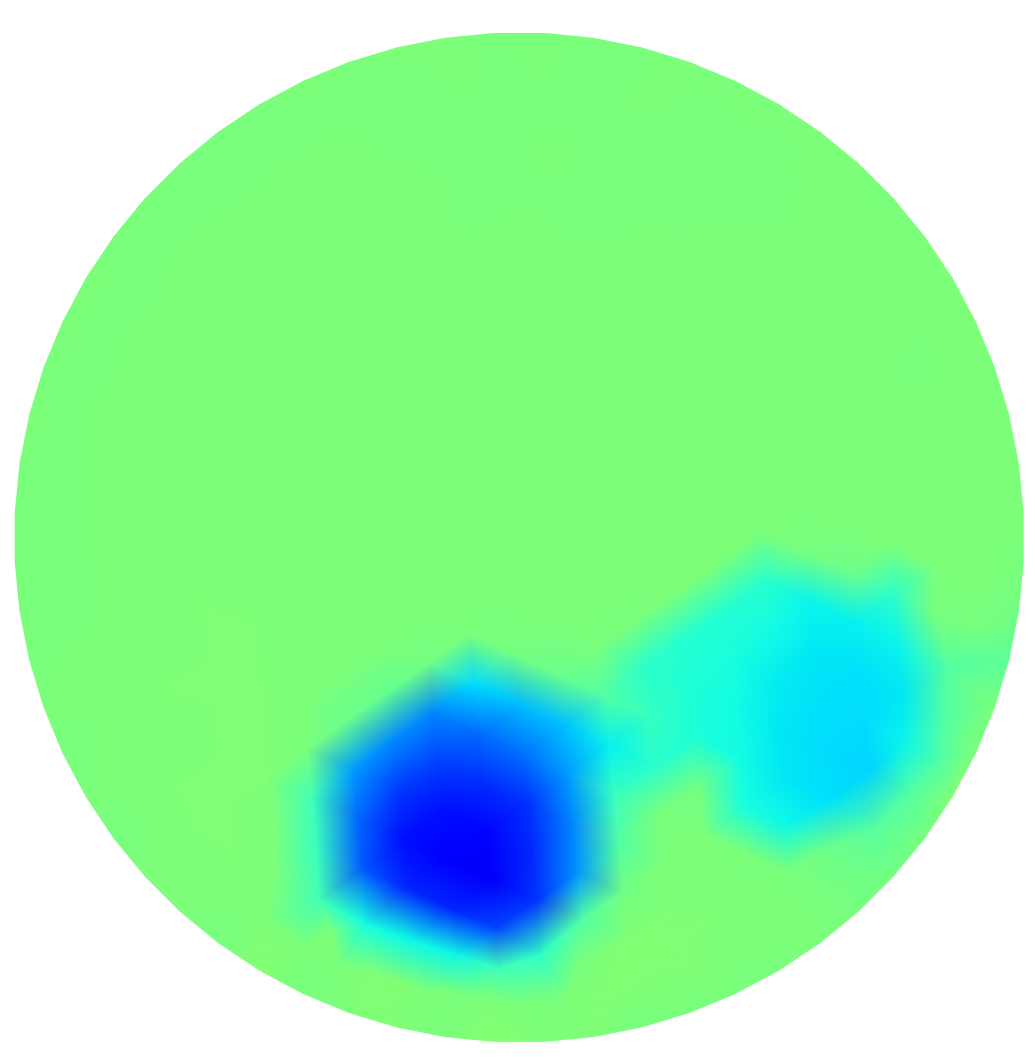} {(0.043, 4.43\%)}&
\imgmetric{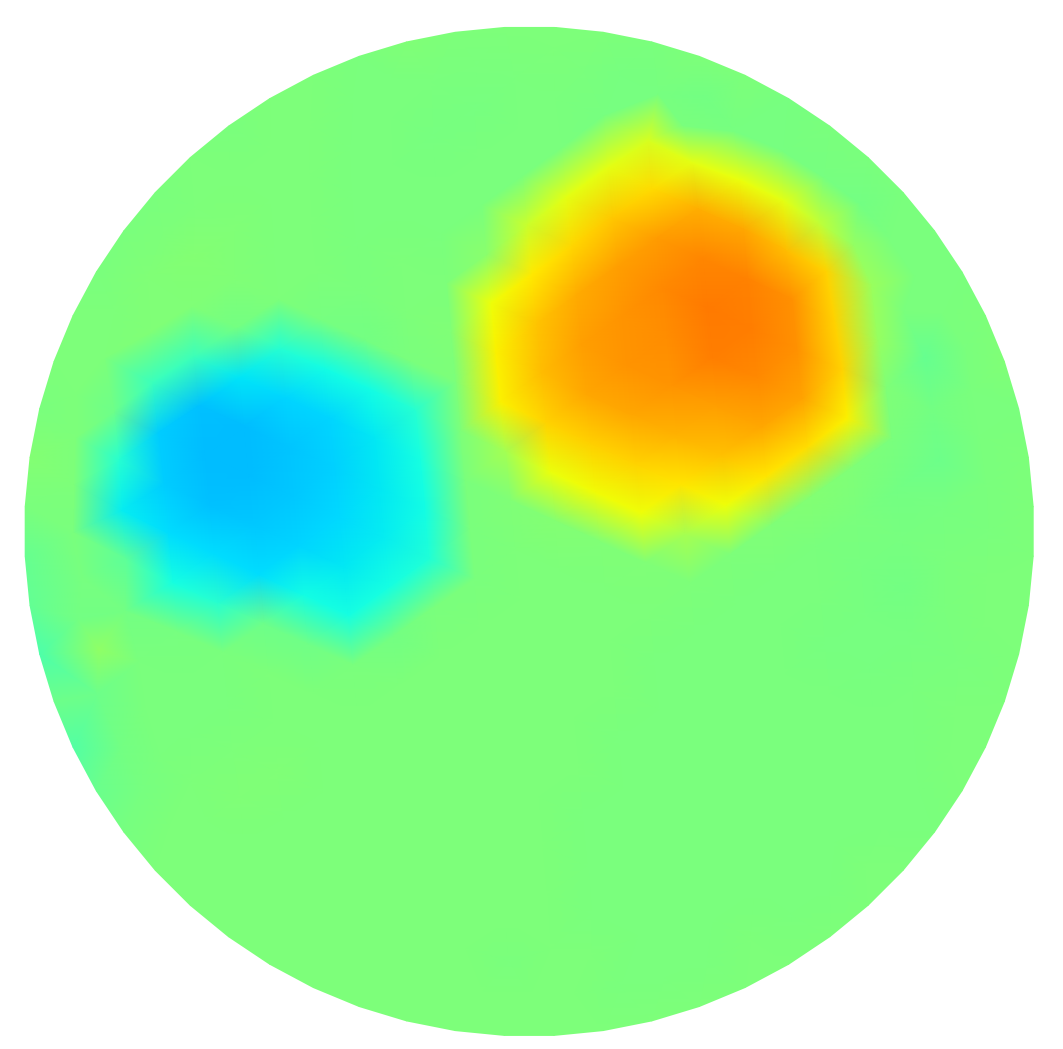} {(0.061, 5.96\%)}\\

TV &
\imgmetric{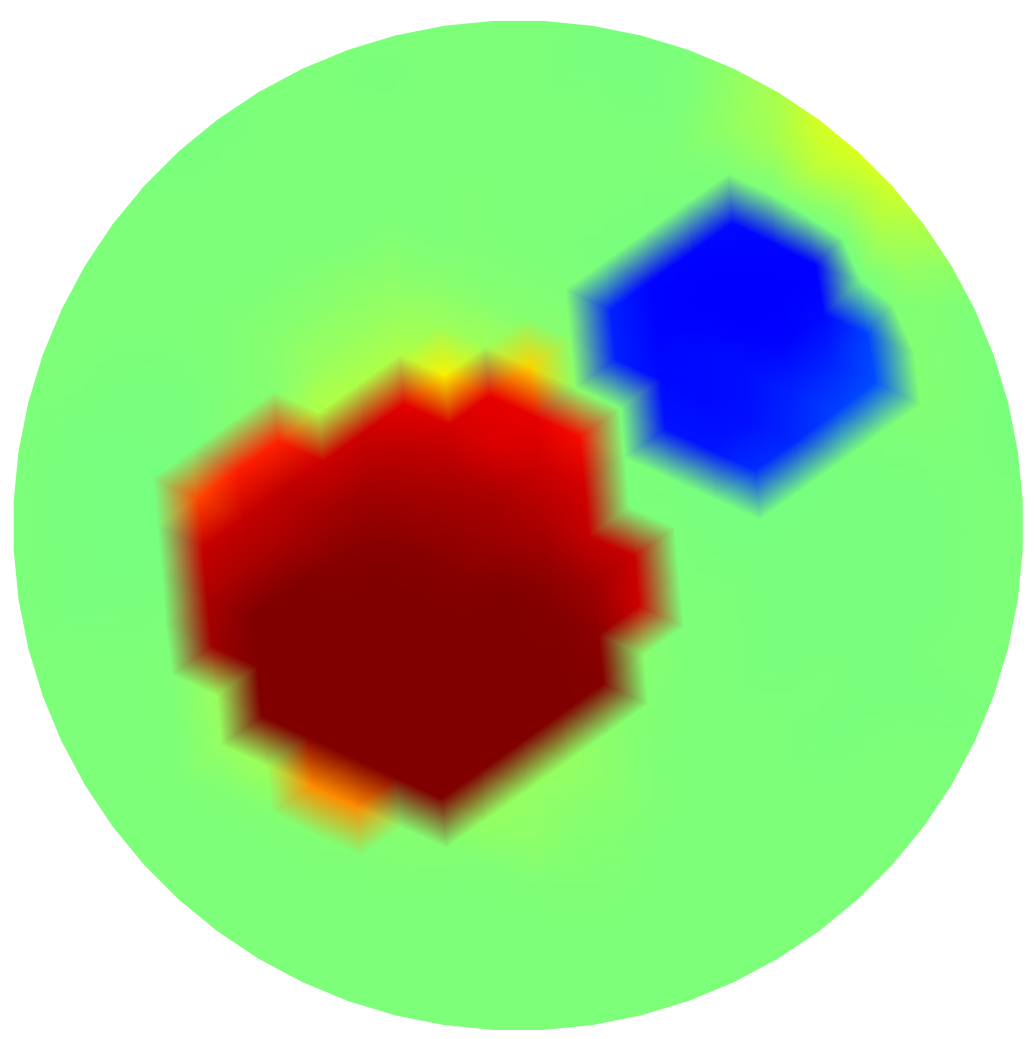} {(\textbf{0.064}, 5.95\%)}&
\imgmetric{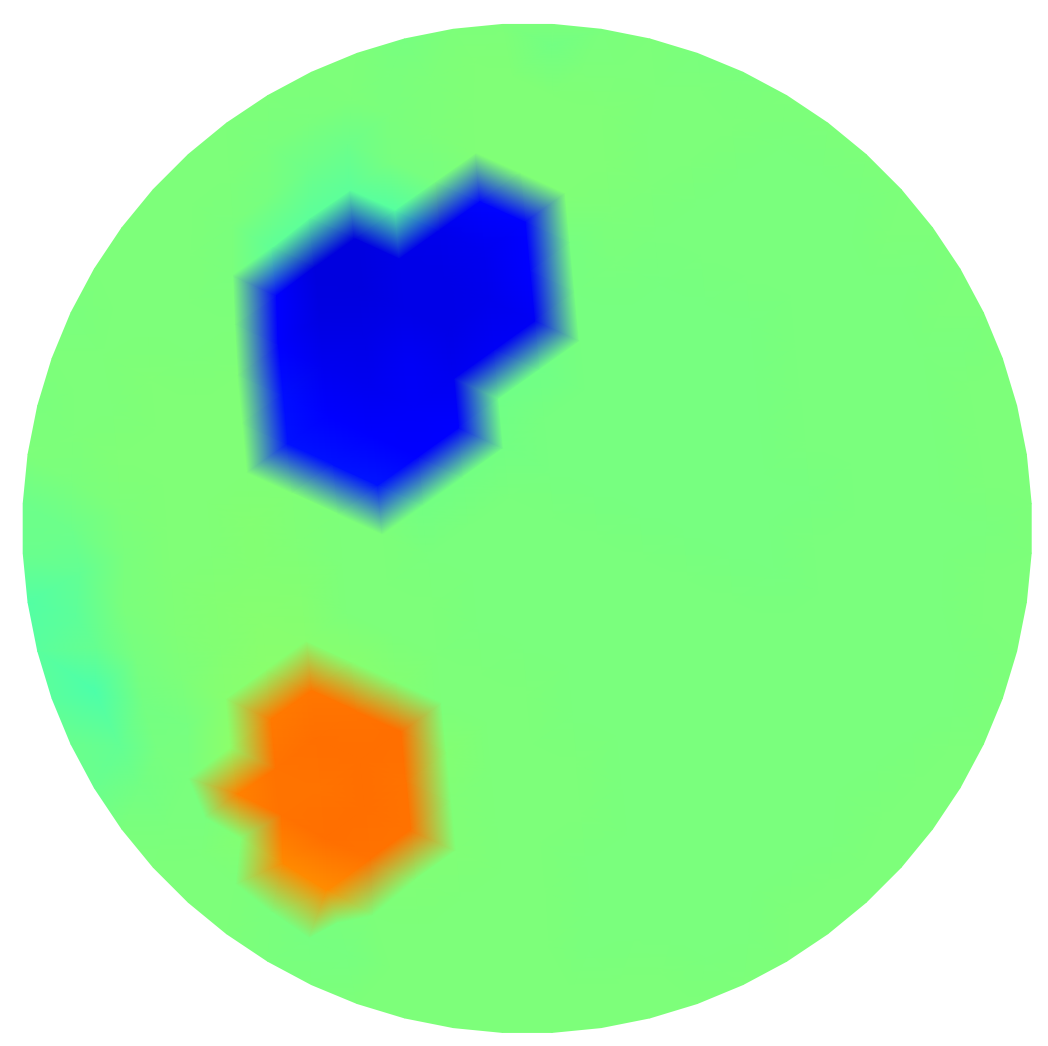} {(\textbf{0.063}, 6.34\%)}&
\imgmetric{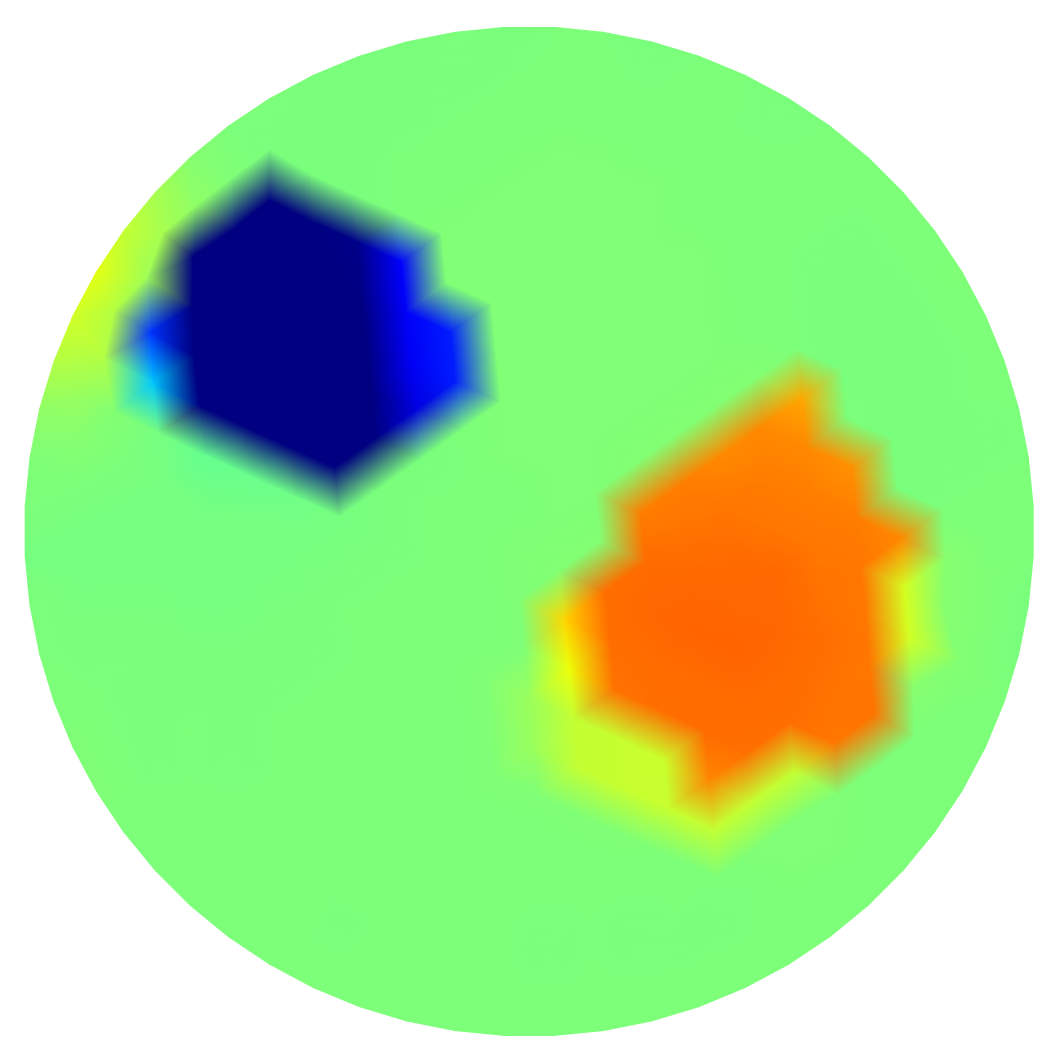}{(\textbf{0.054}, 5.40\%)} &
\imgmetric{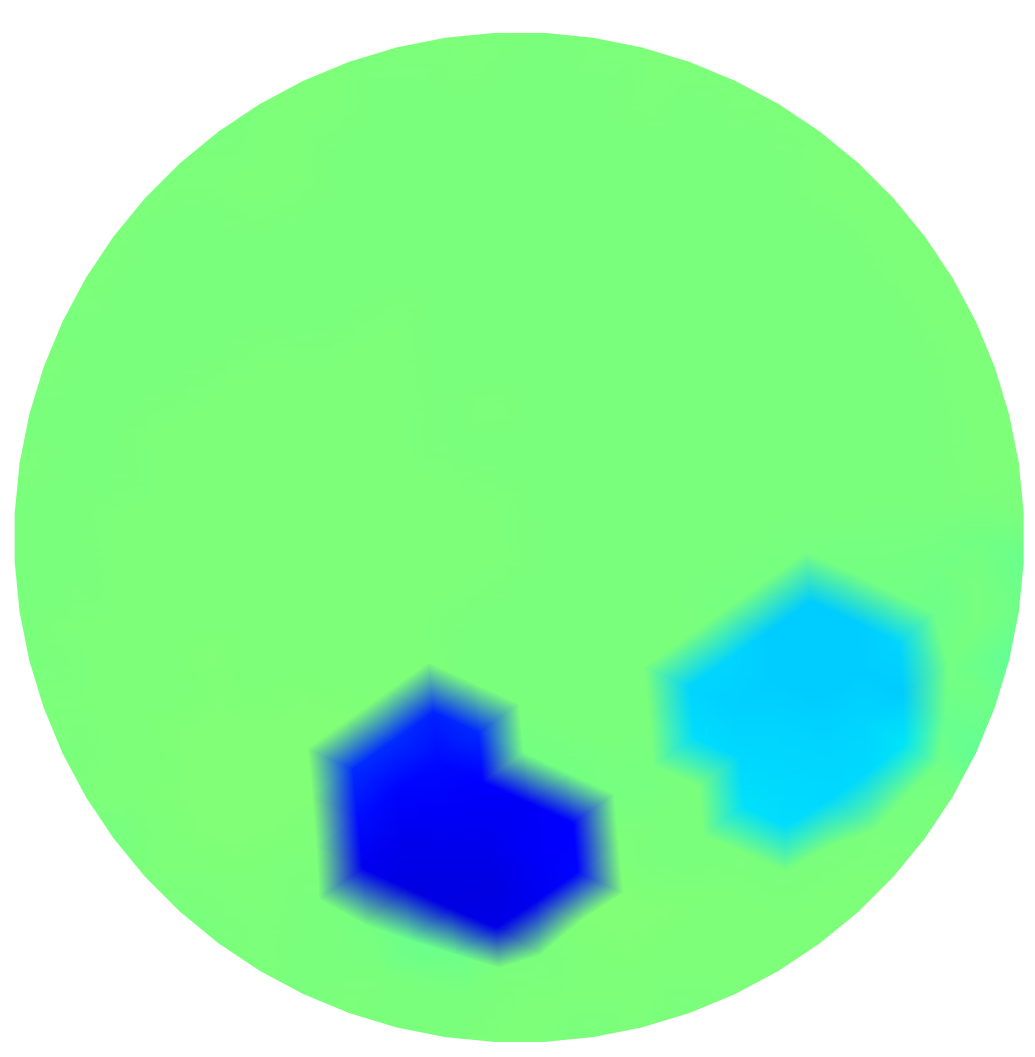}{(\textbf{0.041}, 4.23\%)} &
\imgmetric{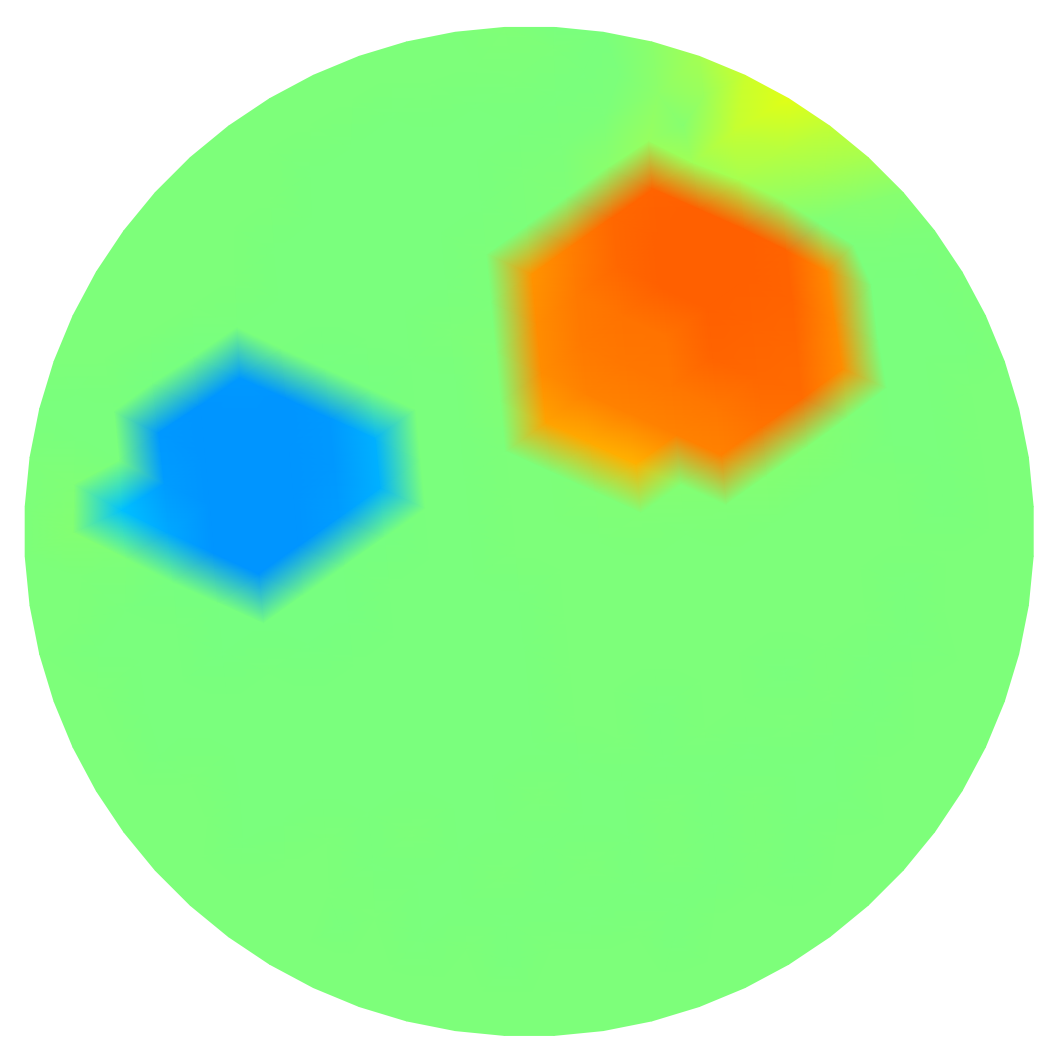} {(\textbf{0.045}, 4.37\%)}
\end{tabular}
\end{adjustbox}
\end{table}
\caption{Example 1: Comparison of different regularizers; (RMSE, RelErr) values are reported for each test case (in bold the lowest errors).}
\end{figure}

Table~\ref{tab:reg_10samples} reports the results on $10$ test samples and summarizes the RMSE values for each regularization choice.  Fig.~\ref{more_results_reg} provides qualitative comparisons for 5 test samples.
As expected, the TV regularizer yields the best results. Its effectiveness is directly linked to the homogeneous structure of the inclusions, which the TV penalty is mathematically designed to preserve.

However, the proposed framework is highly flexible; the regularizer can be effectively tailored to the specific morphology of the data being reconstructed.
\begin{table}[t]
	\centering
	\caption{RMSE on 10 test samples under different regularization terms.}
	\label{tab:reg_10samples}
	\begin{tabular}{ccccc}
		\hline
		Sample ID & No-reg  & Tik & GTik & TV\\
		\hline
		1  & \texttt{0.0980} & \texttt{0.0730} & \texttt{0.0833} & \texttt{0.0643} \\
		2  & \texttt{0.0680} & \texttt{0.0633} & \texttt{0.0689} & \texttt{0.0628} \\
		3  & \texttt{0.0673} & \texttt{0.0604} & \texttt{0.0533} & \texttt{0.0545} \\
		4  & \texttt{0.0439} & \texttt{0.0428} & \texttt{0.0432} & \texttt{0.0423} \\
		5  & \texttt{0.0725} & \texttt{0.0533} & \texttt{0.0609} & \texttt{0.0446} \\
		6  & \texttt{0.0598} & \texttt{0.0443} & \texttt{0.0543} & \texttt{0.0449} \\
		7  & \texttt{0.0459} & \texttt{0.0480} & \texttt{0.0615} & \texttt{0.0407} \\
		8  & \texttt{0.0736} & \texttt{0.0655} & \texttt{0.0481} & \texttt{0.0567} \\
		9  & \texttt{0.0418} & \texttt{0.0340} & \texttt{0.0204} & \texttt{0.0319} \\
		10 & \texttt{0.0515} & \texttt{0.0532} & \texttt{0.0553} & \texttt{0.0455} \\
		\hline
		Mean & \texttt{0.0622} & \texttt{0.0538} & \texttt{0.0549} & \texttt{0.0488} \\
		\hline
	\end{tabular}
\label{more_results_reg}
\end{table}

\subsection{Example 2: Noise Sensitivity Analysis}\label{sub:ex2}

To evaluate the robustness of our method under noisy measurements, we conduct experiments on DATASET1 using both noise-free and noise-corrupted data. In particular, we consider two noise models: additive white Gaussian noise (AWGN) and additive white Laplace noise (AWLN), each at two different noise levels. While AWGN accounts for standard measurement noise such as thermal fluctuations, AWLN is employed to simulate impulsive disturbances and outliers typical of real-world EIT environments. This experiment is designed to assess the stability and robustness of the proposed graph-based diffusion prior under different noise distributions and noise levels.

For each test sample, the corrupted measurement vector $y$ is generated from the noise-free measurement $y_{\mathrm{GT}}$ according to
\[
y = y_{\mathrm{GT}} + \eta.
\]
The noise vector $\eta$ is modeled as a zero-mean random process with standard deviation $\sigma_y$. We consider two distinct noise regimes:
\begin{eqnarray}
\text{[AWGN]} \quad & \eta \sim \mathcal{N}\!\left(0,\,\sigma_y^2 I\right), \\
\text{[AWLN]} \quad & \eta \sim \mathrm{Laplace}(0,\, b I), \quad \text{with } b = \frac{\sigma_y}{\sqrt{2}},
\end{eqnarray}
where $b$ represents the scale parameter of the Laplace distribution, specifically chosen so that both distributions share the same variance.
To assess performance under different noise intensities, we consider two noise levels for each noise type:
\begin{equation}
\label{eq:nl}
\sigma_y \in \left\{\sigma_1 := 2\times10^{-3}\cdot \max\bigl(|y_{\mathrm{GT}}|\bigr),\; \sigma_2 := 1\times10^{-2}\cdot \max\bigl(|y_{\mathrm{GT}}|\bigr)\right\}.
\end{equation}

To address the sensitivity of the reconstruction quality to the regularization parameter $\lambda$, we employ a time-dependent adaptive regularization strategy. This approach helps balancing data fidelity and prior constraints throughout the diffusion process.
According to \eqref{eq:rdps}, the theoretical weight $\lambda_t^{\mathrm{theory}}$ is proportional to the ratio between measurement noise ($\sigma_y^2$) and the time-varying diffusion noise
\[
\lambda_t^{\mathrm{theory}}
\propto \frac{\sigma_y^2}{(1-\bar{\alpha_t})/ \sqrt{\bar{\alpha_t}}}.
\] This ensures that the reconstruction is primarily guided by data consistency in the early sampling stages, while the generative prior is gradually strengthened as the process converges. To maintain stability, the adaptive parameter $\lambda_t$ is obtained by clamping the theoretical weight within a prescribed interval $[\lambda_{\min}, \lambda_{\max}]$, namely
\[
\lambda_t = \mathrm{clip}\!\left(\lambda_t^{\mathrm{theory}},\, \lambda_{\min},\, \lambda_{\max}\right).
\]
The lower bound prevents the prior from becoming negligible at high noise levels, while the upper bound ensures the data-fidelity term is not suppressed during the final refinement steps.  By adjusting the clamping range based on the specific noise model (AWGN vs. AWLN), the framework maintains robustness across different physical measurement disturbances.

In EIT research, noise levels typically range from $0.1\%$ to $2\%$. To demonstrate the high robustness we test the performance for measurements corrupted by a small noise level ($\sigma_1$) as well as a higher noise level ($\sigma_2$), see \eqref{eq:nl}.

Table~\ref{tab:noisy_w/o_noisy_20} summarizes the quantitative reconstruction errors under noisy and noise-free settings. The performance degradation introduced by noise remains moderate, suggesting that our approach is robust to measurement perturbations.

\begin{table}[H]
	\centering
	\caption{Example 2: Noisy vs.\ noiseless on 20 test samples (DDIM + TV regularization).}
	\label{tab:noisy_w/o_noisy_20}
	\begin{tabular}{lcc}
		\hline
		Setting & RMSE (mean)  & RelErr (mean) \\
		\hline
		Noisy (AWGN, $\sigma_1$) & \texttt{0.07083} & \texttt{7.07\%}  \\
        Noisy (AWGN, $\sigma_2$) & \texttt{0.08498} & \texttt{8.51\%}  \\
        Noisy (AWLN, $\sigma_1$) & \texttt{0.07055} & \texttt{7.05\%}  \\
        Noisy (AWLN, $\sigma_2$) & \texttt{0.08139} & \texttt{8.12\%}  \\
		Noiseless & \texttt{0.06220} &  \texttt{6.22\%}  \\
		\hline
	\end{tabular}
\end{table}

 Figure~\ref{fig:noise_robustness} provides a visual comparison of reconstructions for four representative samples. The conductivity reconstructions in Figure~\ref{fig:noise_robustness} are organized row-wise as follows: with noise-free measurements (second row); with perturbed measurements and using an optimal hand-tuned regularization parameter $\lambda$ (third row); and with perturbed measurements but using our proposed automatic parameter selection $\lambda_t$ under AWGN (fourth and fifth rows) and AWLN (sixth and seventh rows) for the two noise levels $\sigma_1$ and $\sigma_2$. 

Comparing the third and fourth rows for each sample, it is evident that our automatic parameter selection strategy, $\lambda_t$, (fourth row) is highly robust. Notably, it eliminates the need for the computationally expensive grid-search typically required to find an optimal hand-tuned regularization parameter $\lambda$ (third row).

Adopting an automatic $\lambda_t$, as shown in Figure~\ref{fig:noise_robustness} (fourth and fifth blocks), the proposed method exhibits stable reconstruction performance under different noise models and noise levels. For both AWGN and AWLN at the lower noise level $\sigma_1$, the reconstructed conductivities remain visually close to the noise-free results, and the corresponding RMSE and relative error are only mildly increased. At the higher noise level $\sigma_2$, the degradation becomes more noticeable, with larger deviations from the ground truth and increased reconstruction errors across all representative samples. Nevertheless, the adaptive parameter-selection strategy still produces visually reasonable reconstructions under both Gaussian and Laplace noise. In particular, the results under AWGN and AWLN are overall comparable at the same noise level. The sample in the third column appears to be consistently more challenging than the other examples, even in the noise-free setting, probably because of the presence of an inclusion far from the boundary, where stability further deteriorates in EIT \cite{Noser}. These observations indicate that the proposed method is robust to moderate measurement noise and degrades gradually as the noise intensity increases, also in view of the severe ill-posedness of this inverse problem.

\begin{figure}[H]
\centering
\setlength{\tabcolsep}{4pt}
\renewcommand{\arraystretch}{1.0}

\begin{tabular}{@{}c c c c c@{}}

GT
&
\begin{minipage}[t]{0.17\textwidth}\centering\vspace{0pt}
  \includegraphics[width=\linewidth]{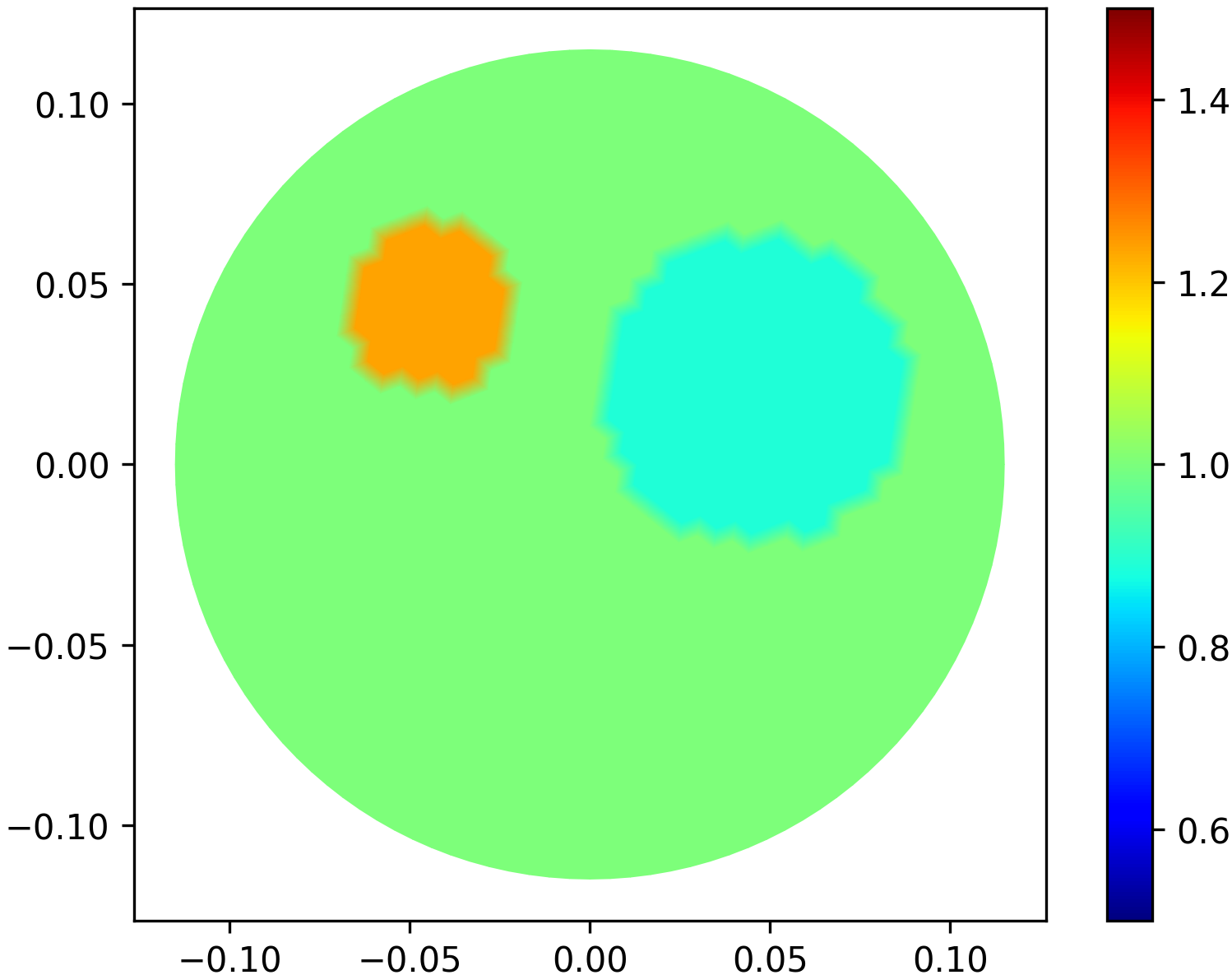}
\end{minipage}
&
\begin{minipage}[t]{0.17\textwidth}\centering\vspace{0pt}
  \includegraphics[width=\linewidth]{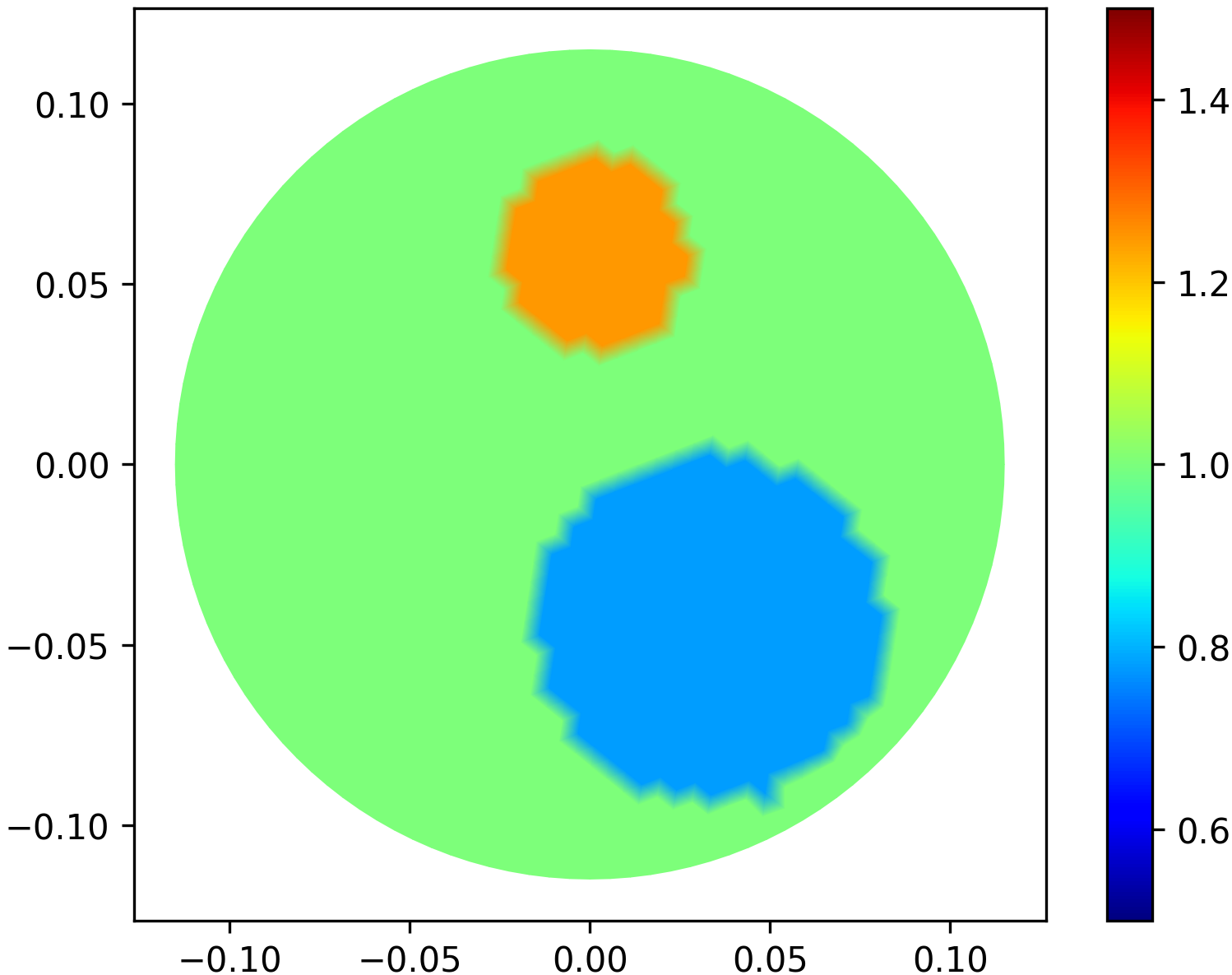}
\end{minipage}
&
\begin{minipage}[t]{0.17\textwidth}\centering\vspace{0pt}
  \includegraphics[width=\linewidth]{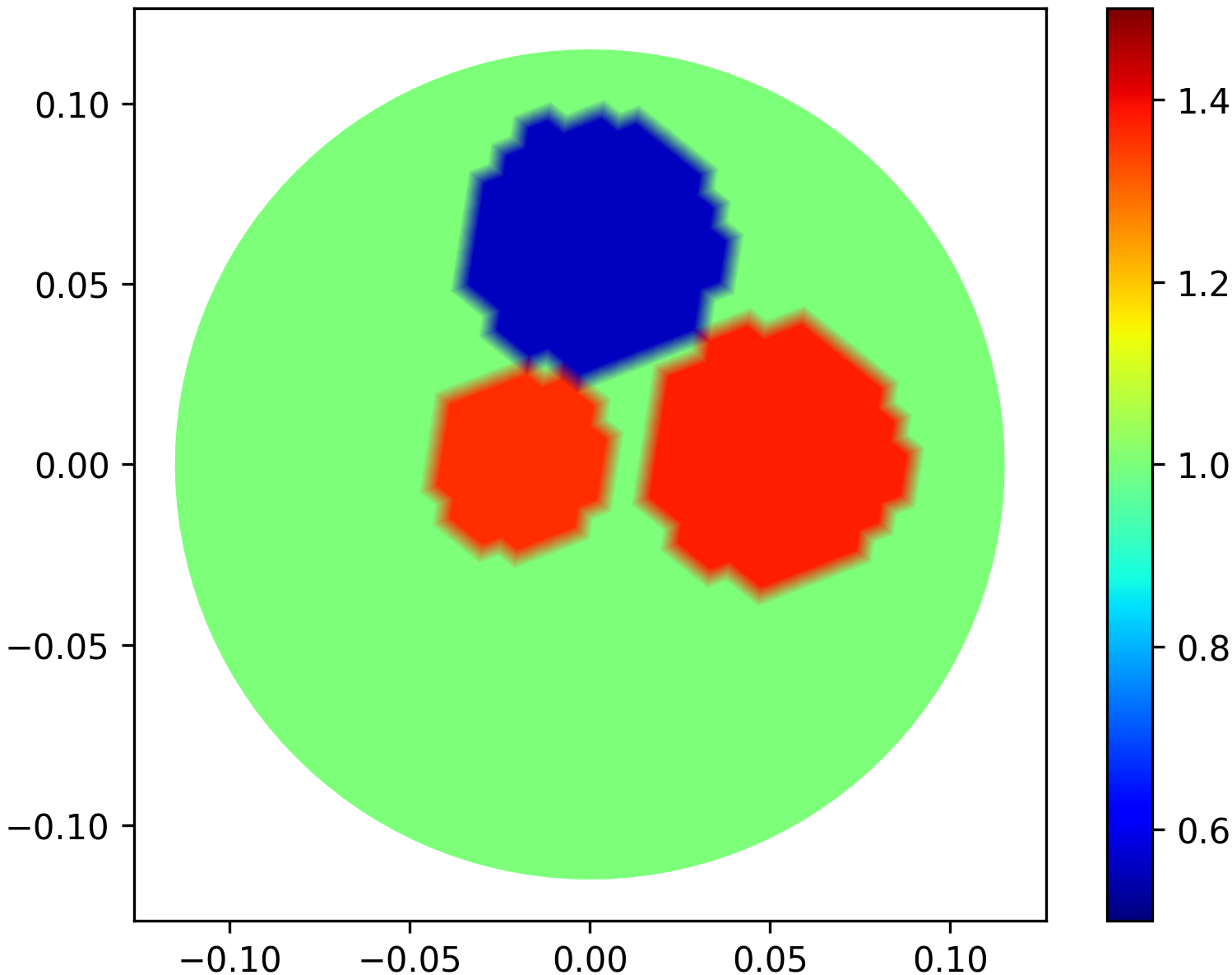}
\end{minipage}
&
\begin{minipage}[t]{0.17\textwidth}\centering\vspace{0pt}
  \includegraphics[width=\linewidth]{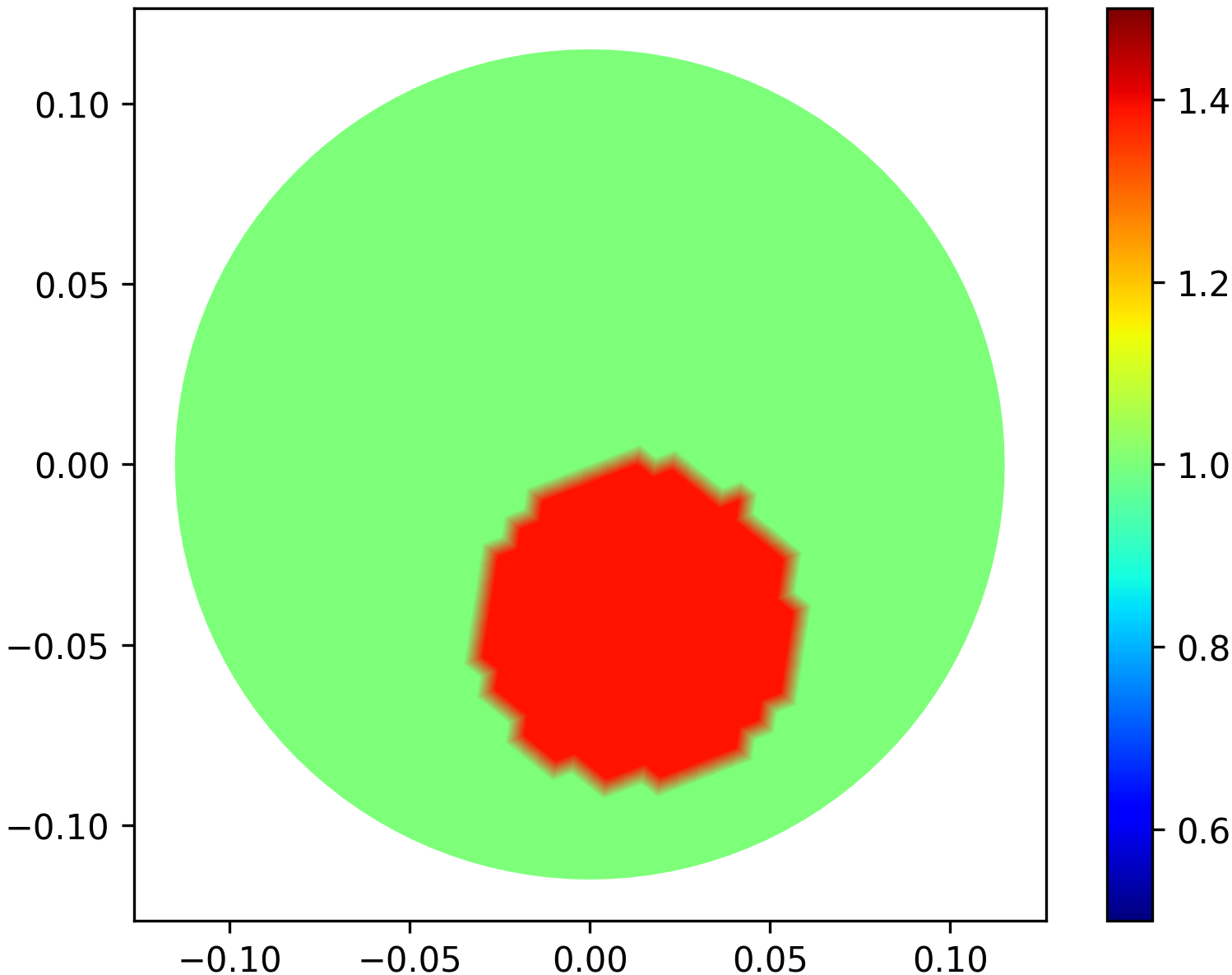}
\end{minipage}
\\ & & & & \\
\hline 

Noise-free
&
\begin{minipage}[t]{0.13\textwidth}\centering\vspace{0pt}
  \includegraphics[width=\linewidth]{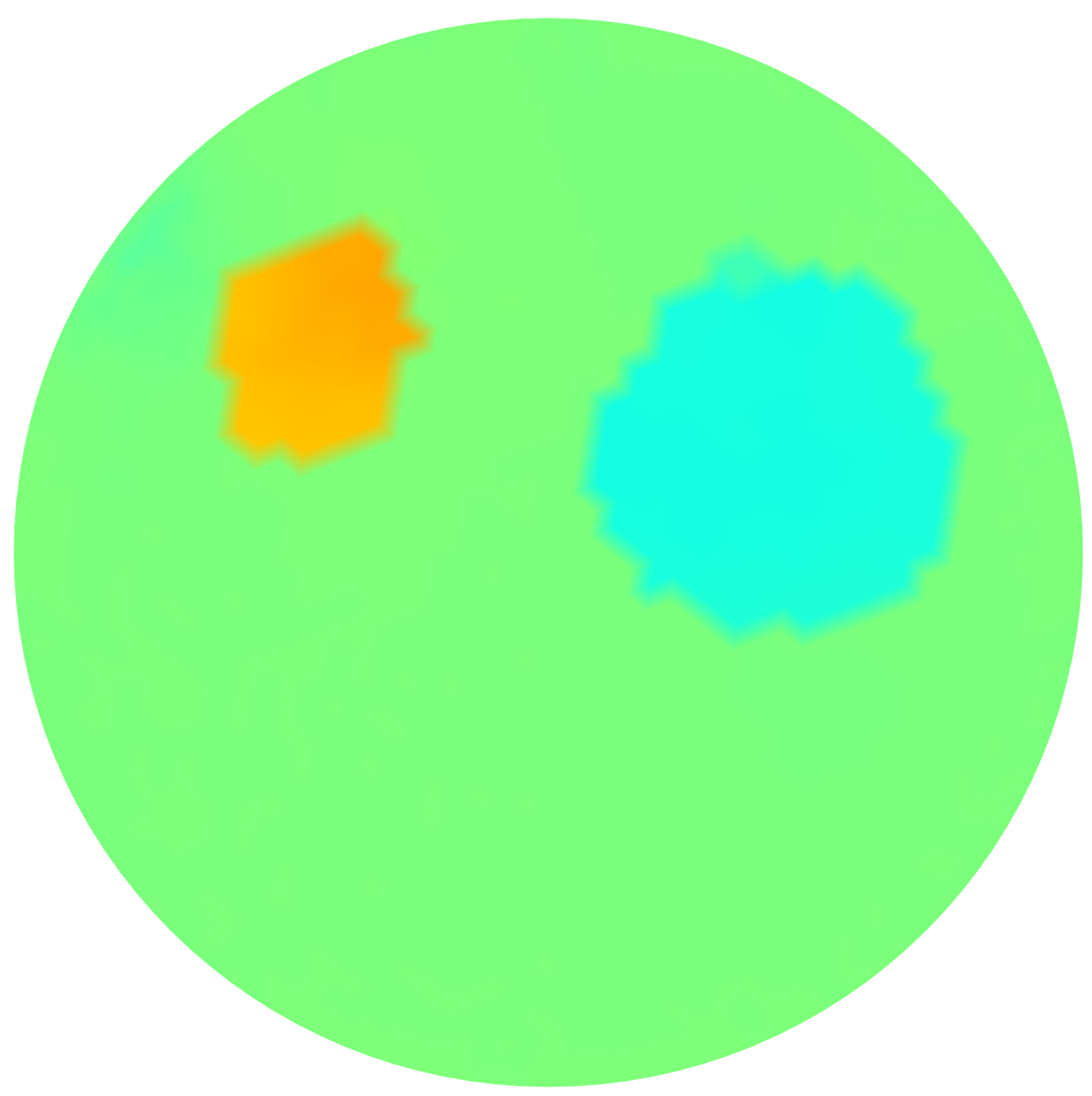}\\[-2pt]
  {\footnotesize (0.031, 3.11\%)}
\end{minipage}
&
\begin{minipage}[t]{0.13\textwidth}\centering\vspace{0pt}
  \includegraphics[width=\linewidth]{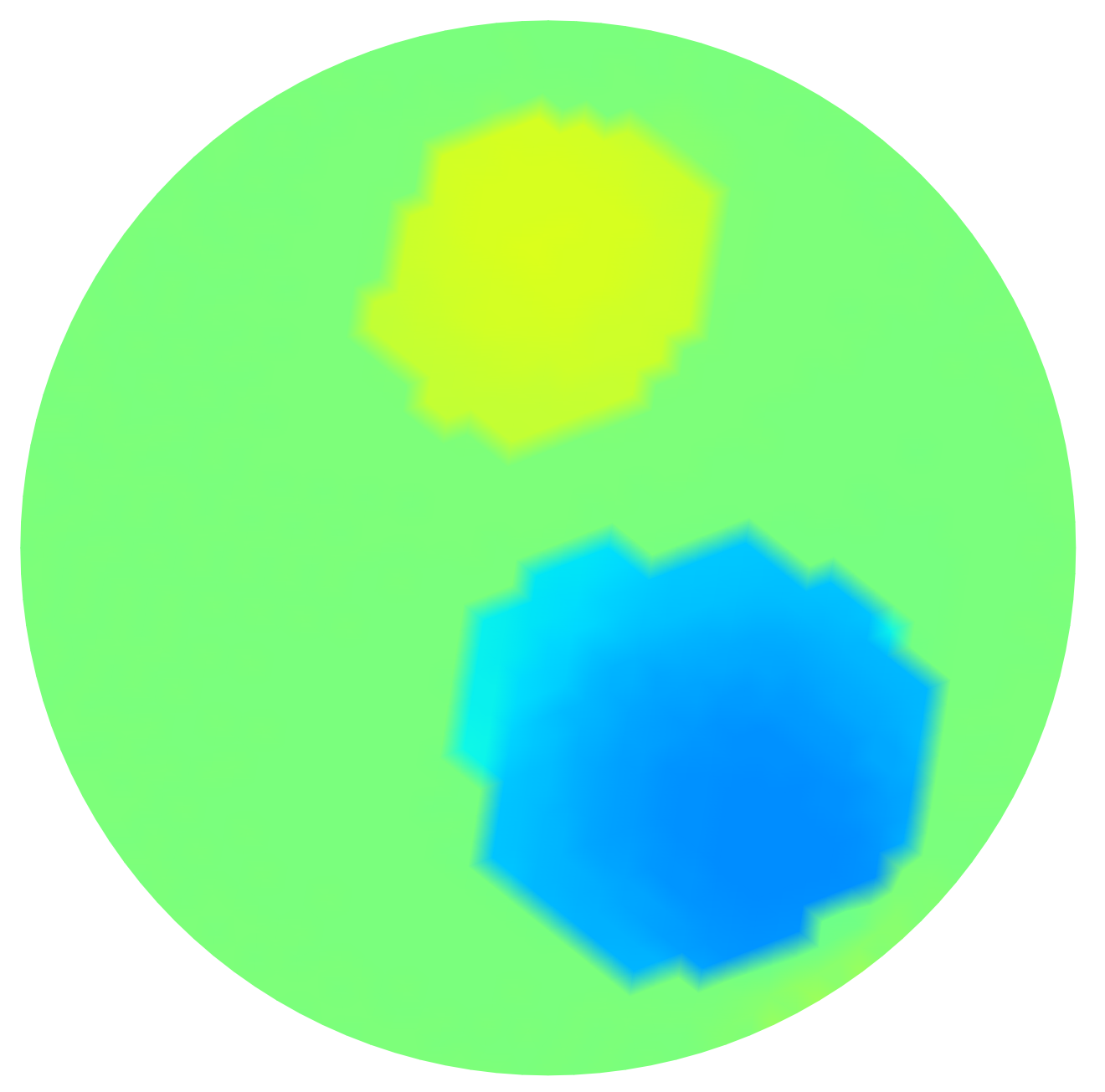}\\[-2pt]
  {\footnotesize (0.045, 4.58\%)}
\end{minipage}
&
\begin{minipage}[t]{0.13\textwidth}\centering\vspace{0pt}
  \includegraphics[width=\linewidth]{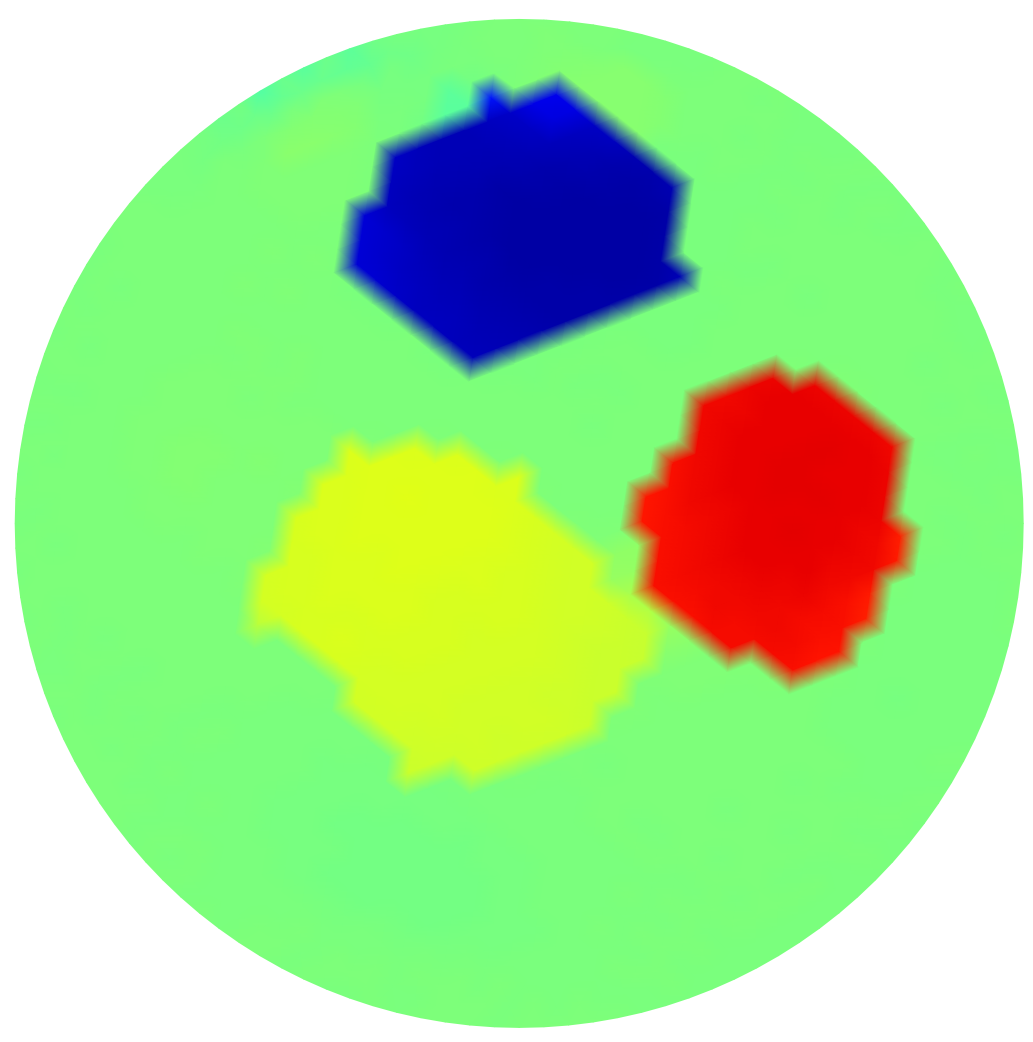}\\[-2pt]
  {\footnotesize (0.115, 11.16\%)}
\end{minipage}
&
\begin{minipage}[t]{0.13\textwidth}\centering\vspace{0pt}
  \includegraphics[width=\linewidth]{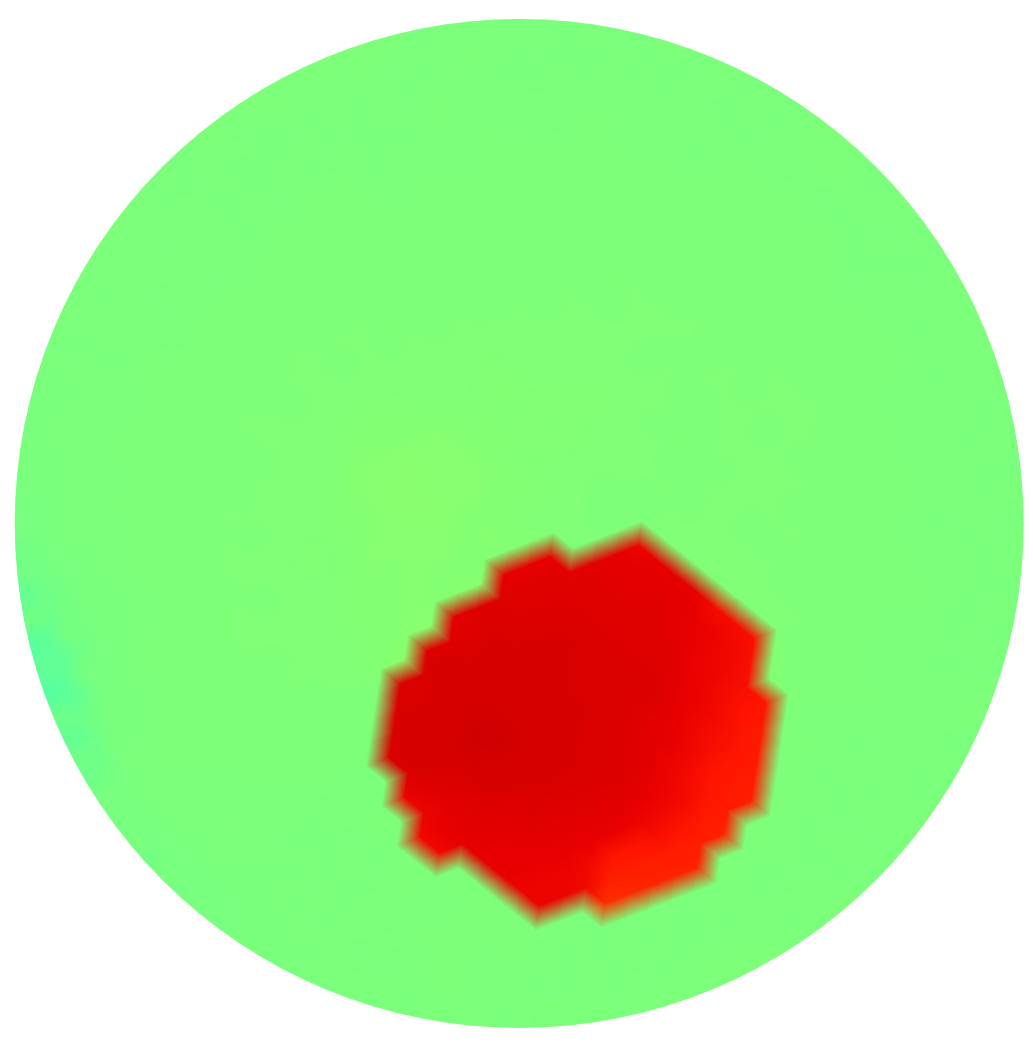}\\[-2pt]
  {\footnotesize (0.043, 4.01\%)}
\end{minipage}
\\
\hline 
$\lambda$, (AWGN, $\sigma_1$)
&
\begin{minipage}[t]{0.13\textwidth}\centering\vspace{0pt}
  \includegraphics[width=\linewidth]{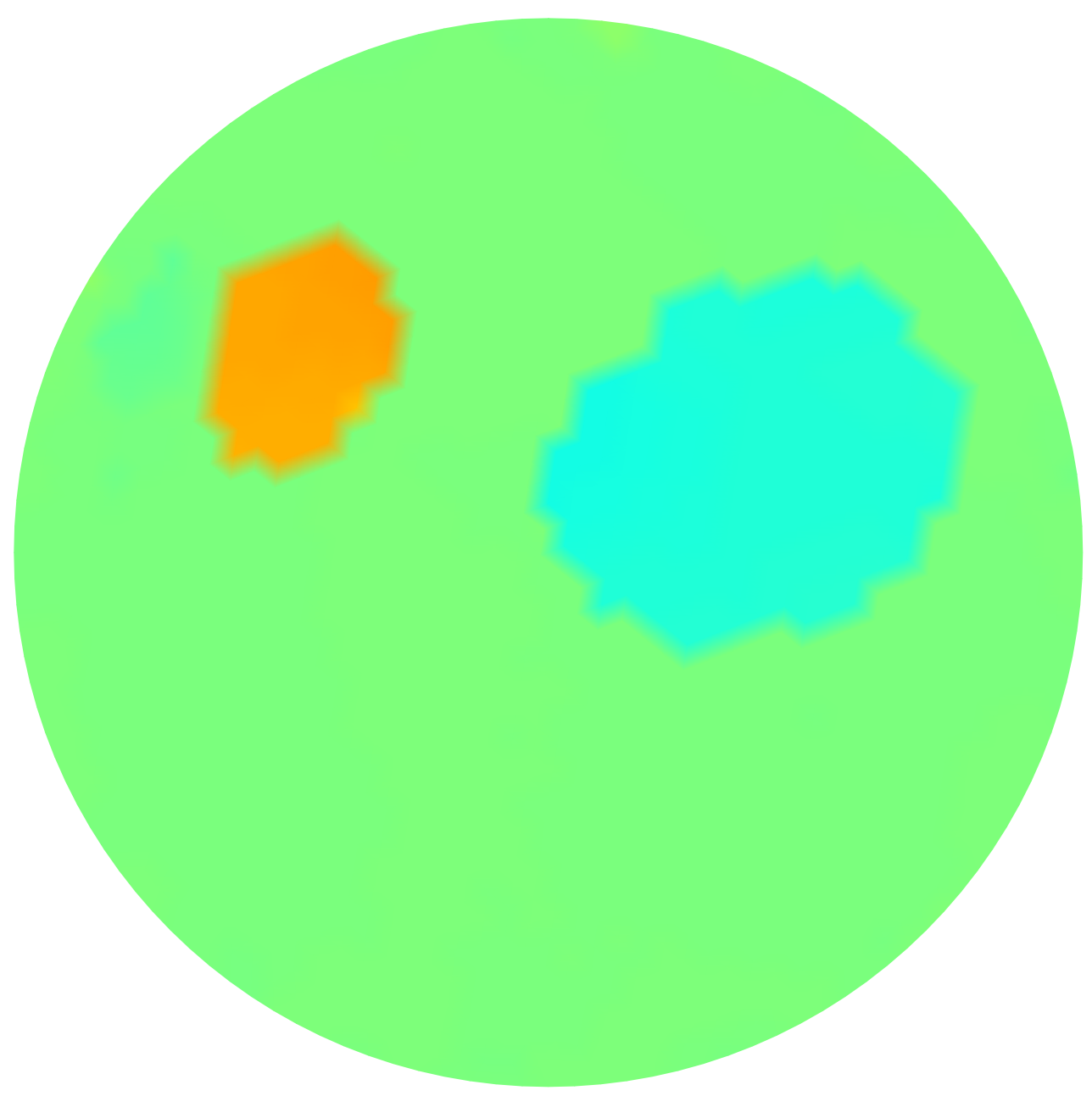}\\[-2pt]
  {\footnotesize (0.036, 3.61\%)}
\end{minipage}
&
\begin{minipage}[t]{0.13\textwidth}\centering\vspace{0pt}
  \includegraphics[width=\linewidth]{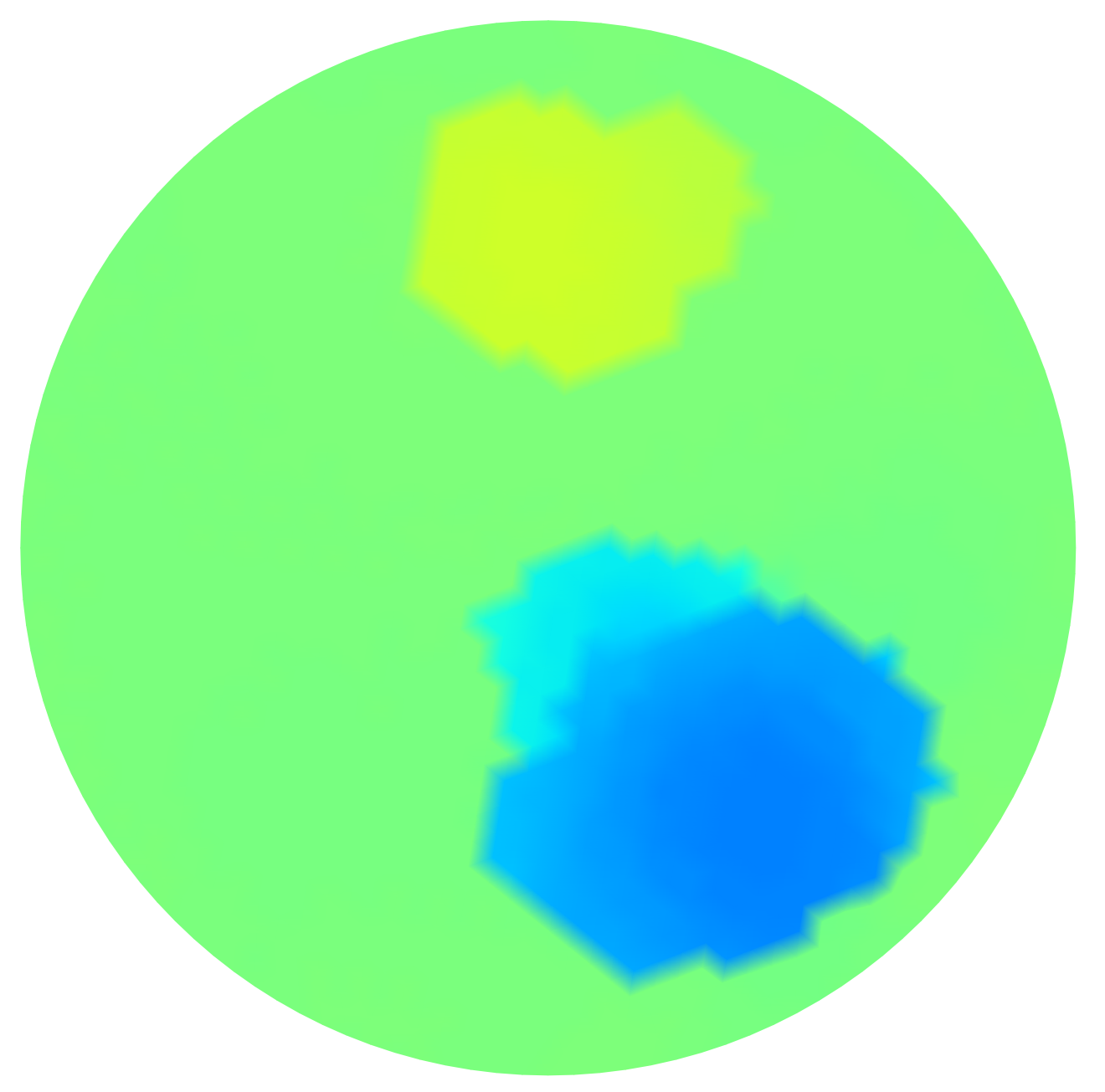}\\[-2pt]
  {\footnotesize (0.052, 5.33\%)}
\end{minipage}
&
\begin{minipage}[t]{0.13\textwidth}\centering\vspace{0pt}
  \includegraphics[width=\linewidth]{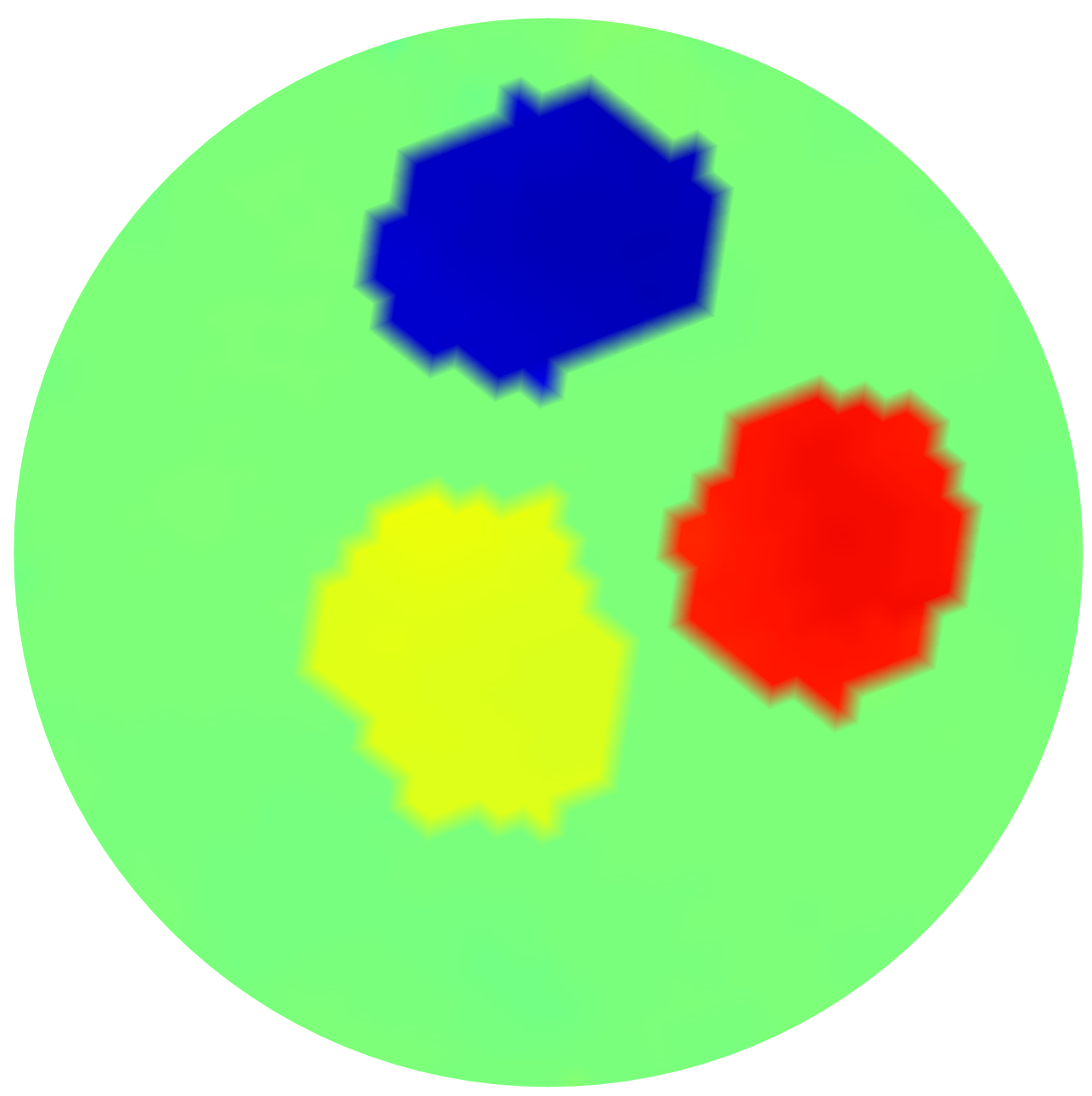}\\[-2pt]
  {\footnotesize (0.116, 11.26\%)}
\end{minipage}
&
\begin{minipage}[t]{0.13\textwidth}\centering\vspace{0pt}
  \includegraphics[width=\linewidth]{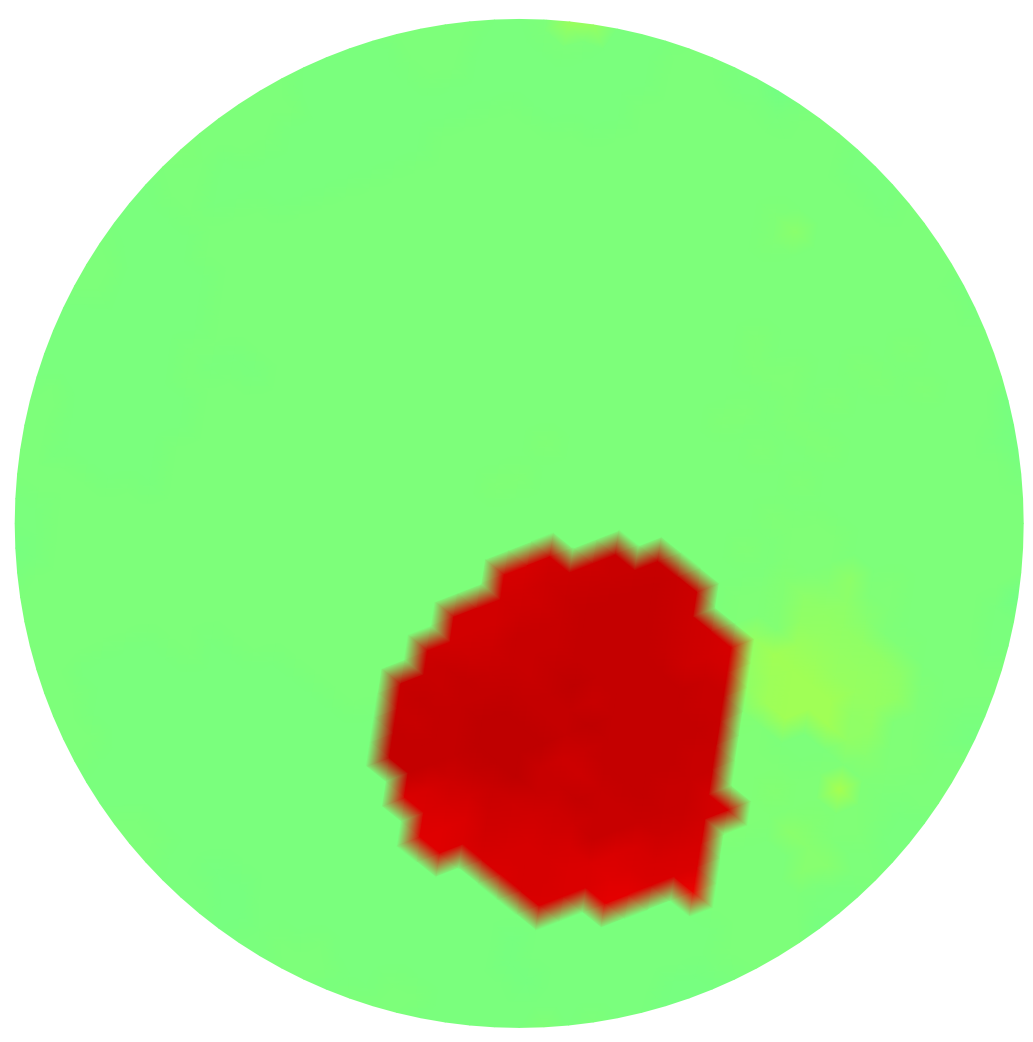}\\[-2pt]
  {\footnotesize (0.060, 5.63\%)}
\end{minipage}
\\[6pt]
\hline

$\lambda_t$, (AWGN, $\sigma_1$)
&
\begin{minipage}[t]{0.13\textwidth}\centering\vspace{0pt}
  \includegraphics[width=\linewidth]{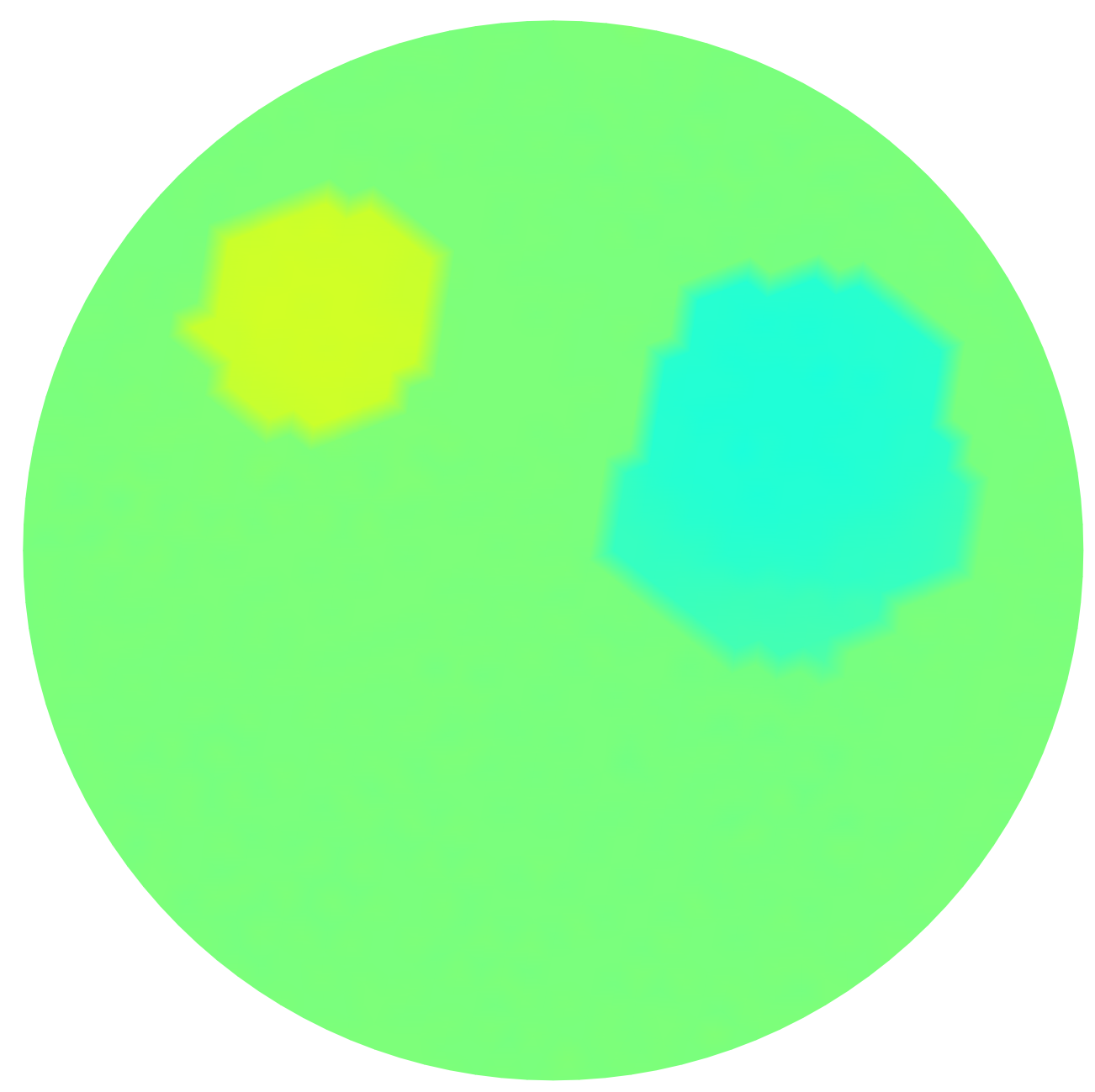}\\[-2pt]
  {\footnotesize (0.039, 3.86\%)}
\end{minipage}
&
\begin{minipage}[t]{0.13\textwidth}\centering\vspace{0pt}
  \includegraphics[width=\linewidth]{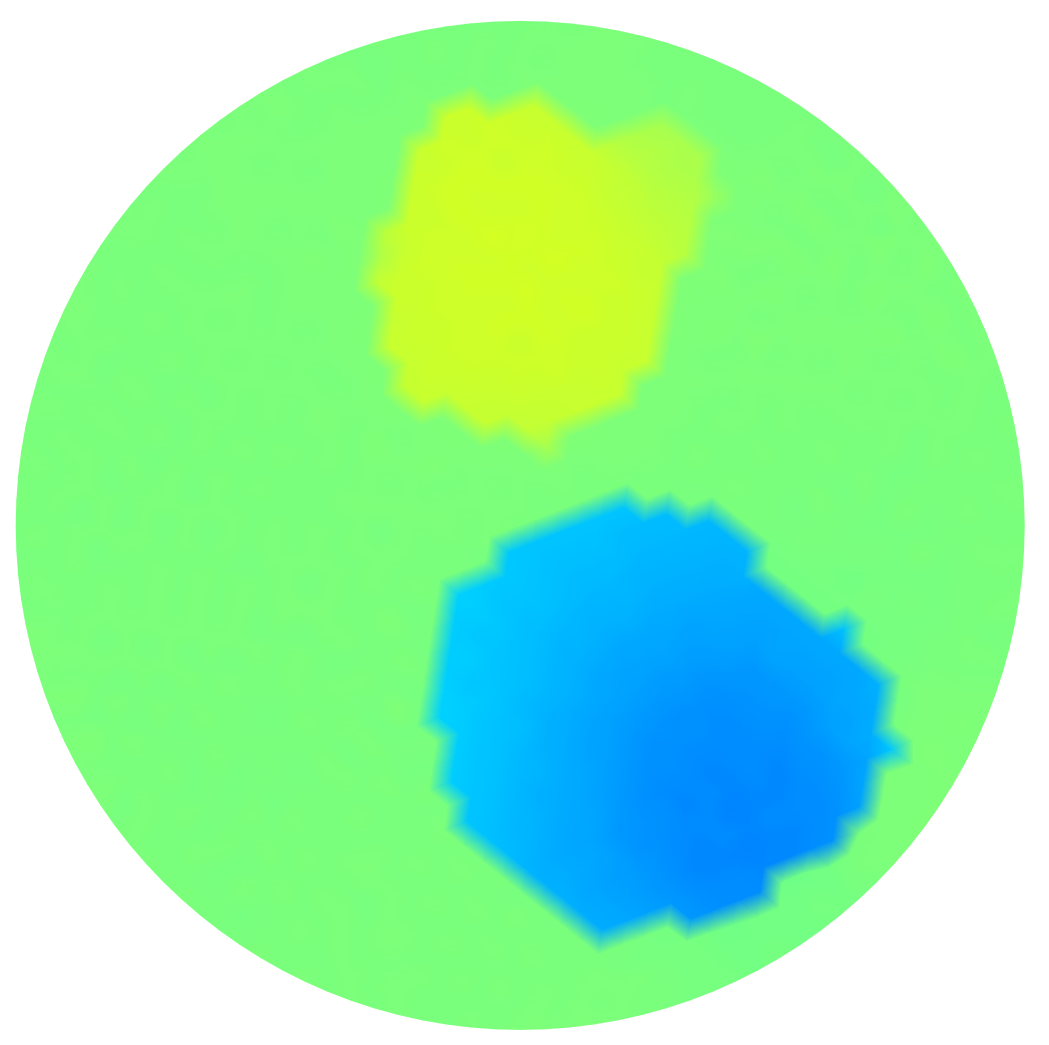}\\[-2pt]
  {\footnotesize (0.049, 4.97\%)}
\end{minipage}
&
\begin{minipage}[t]{0.13\textwidth}\centering\vspace{0pt}
  \includegraphics[width=\linewidth]{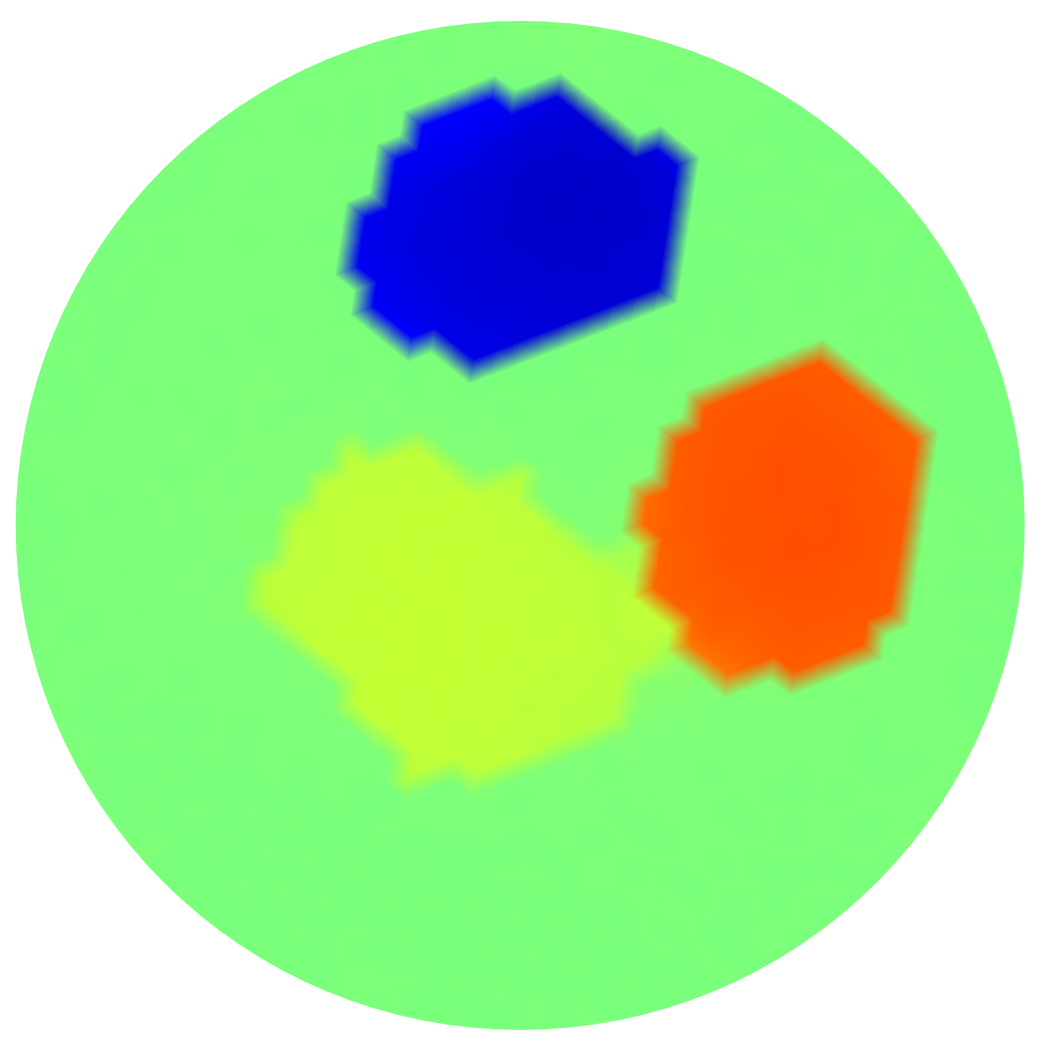}\\[-2pt]
  {\footnotesize (0.116, 11.25\%)}
\end{minipage}
&
\begin{minipage}[t]{0.13\textwidth}\centering\vspace{0pt}
  \includegraphics[width=\linewidth]{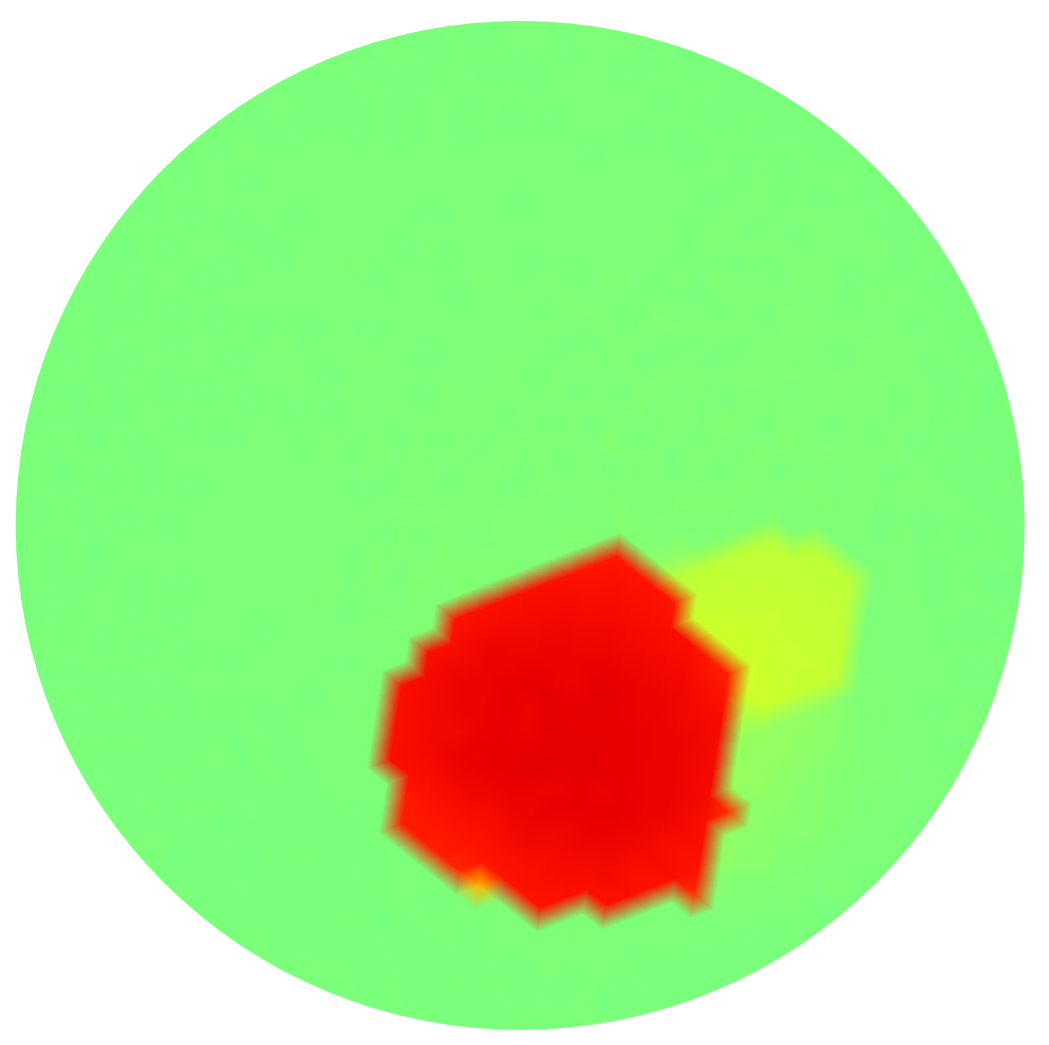}\\[-2pt]
  {\footnotesize (0.060, 5.67\%)}
\end{minipage}
\\[6pt]

$\lambda_t$,  (AWGN, $\sigma_2$)
&
\begin{minipage}[t]{0.13\textwidth}\centering\vspace{0pt}
  \includegraphics[width=\linewidth]{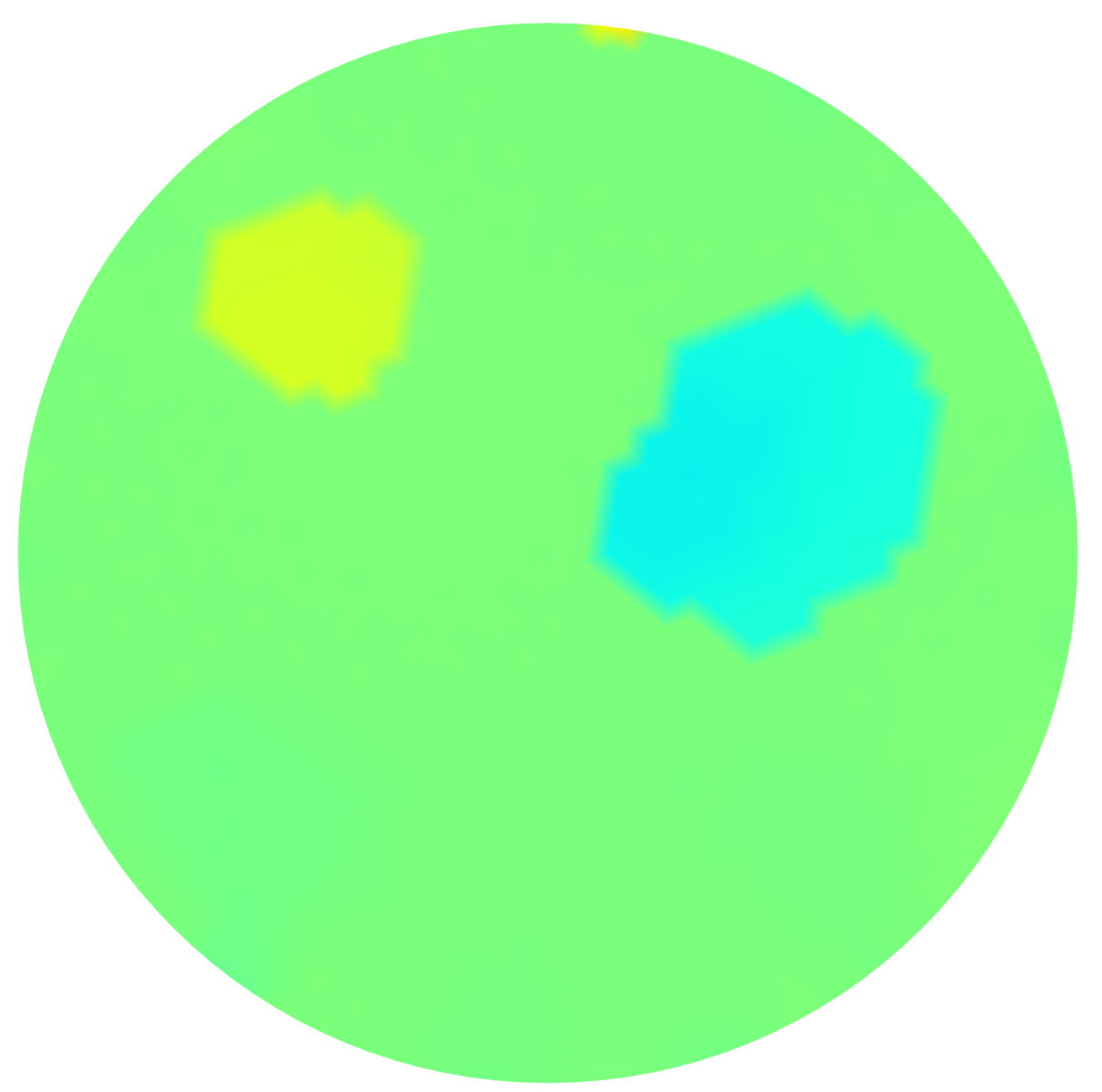}\\[-2pt]
  {\footnotesize (0.045, 4.48\%)}
\end{minipage}
&
\begin{minipage}[t]{0.13\textwidth}\centering\vspace{0pt}
  \includegraphics[width=\linewidth]{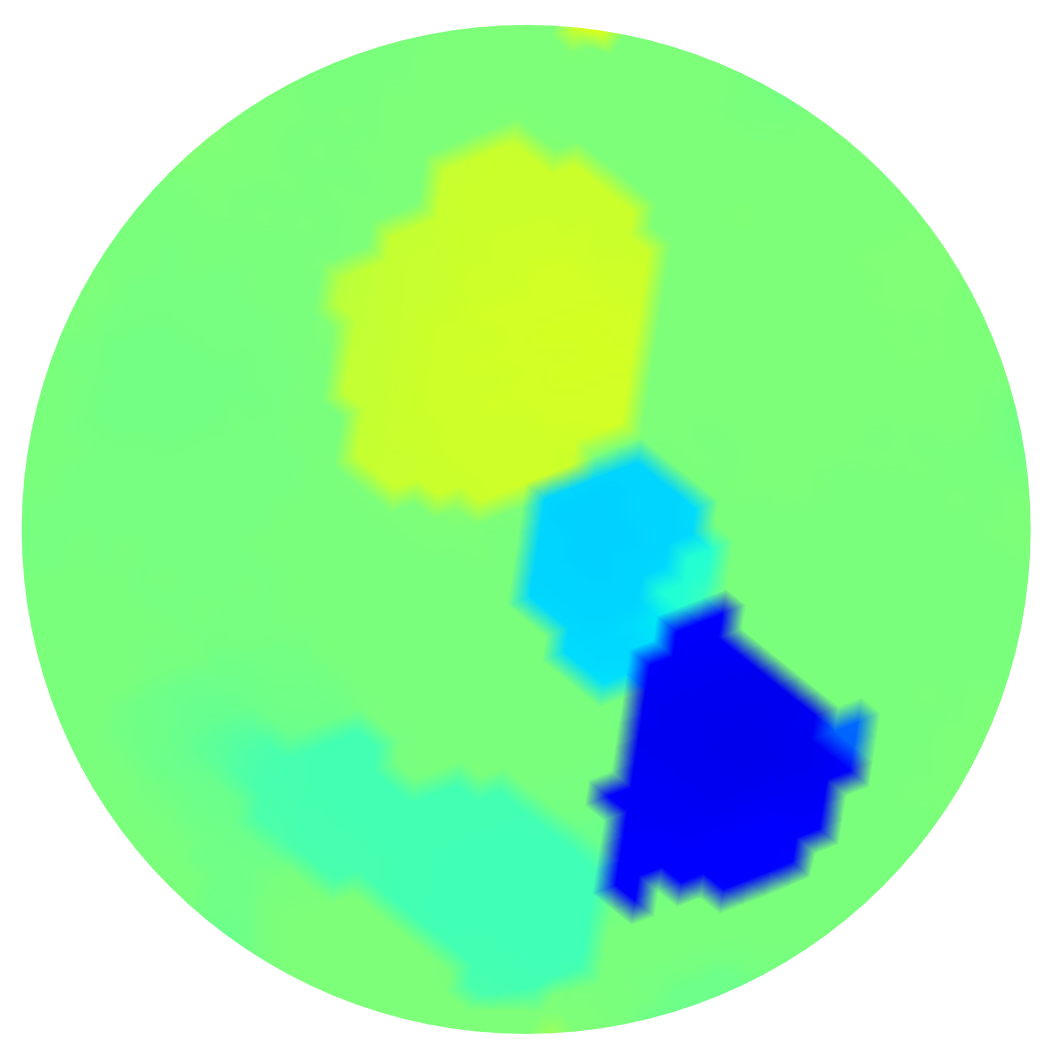}\\[-2pt]
  {\footnotesize (0.087, 8.90\%)}
\end{minipage}
&
\begin{minipage}[t]{0.13\textwidth}\centering\vspace{0pt}
  \includegraphics[width=\linewidth]{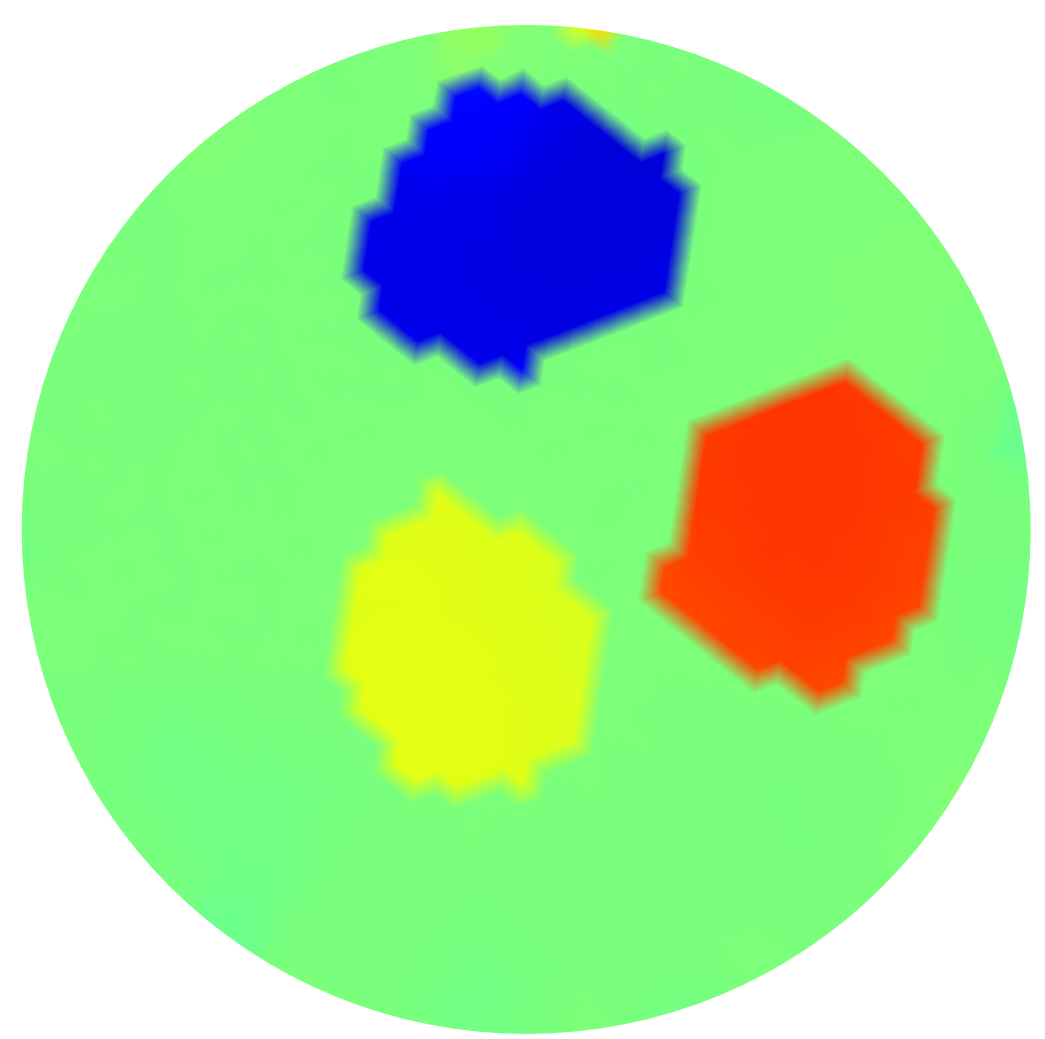}\\[-2pt]
  {\footnotesize (0.127, 12.38\%)}
\end{minipage}
&
\begin{minipage}[t]{0.13\textwidth}\centering\vspace{0pt}
  \includegraphics[width=\linewidth]{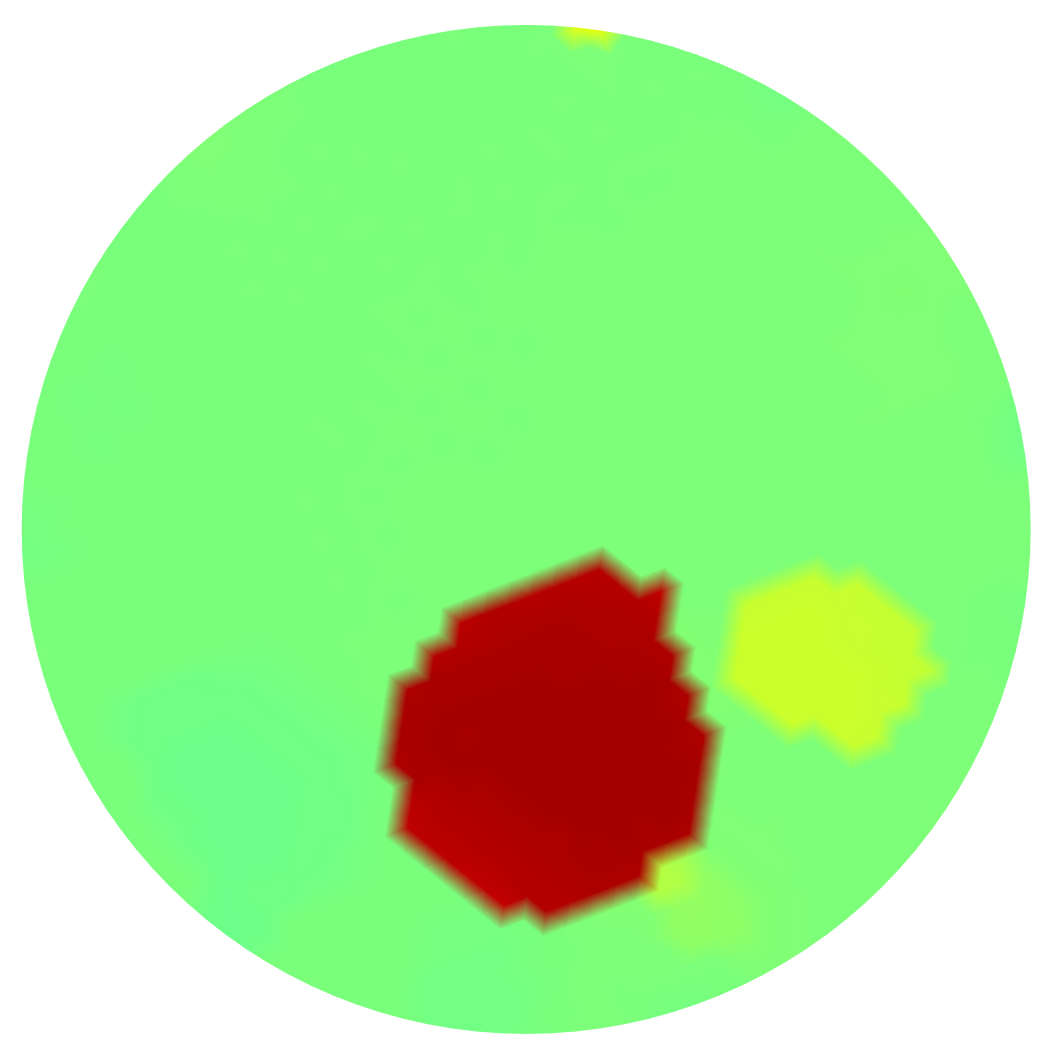}\\[-2pt]
  {\footnotesize (0.083, 7.81\%)}
\end{minipage}
\\[6pt]
\hline

$\lambda_t$,  (AWLN, $\sigma_1$)
&
\begin{minipage}[t]{0.13\textwidth}\centering\vspace{0pt}
  \includegraphics[width=\linewidth]{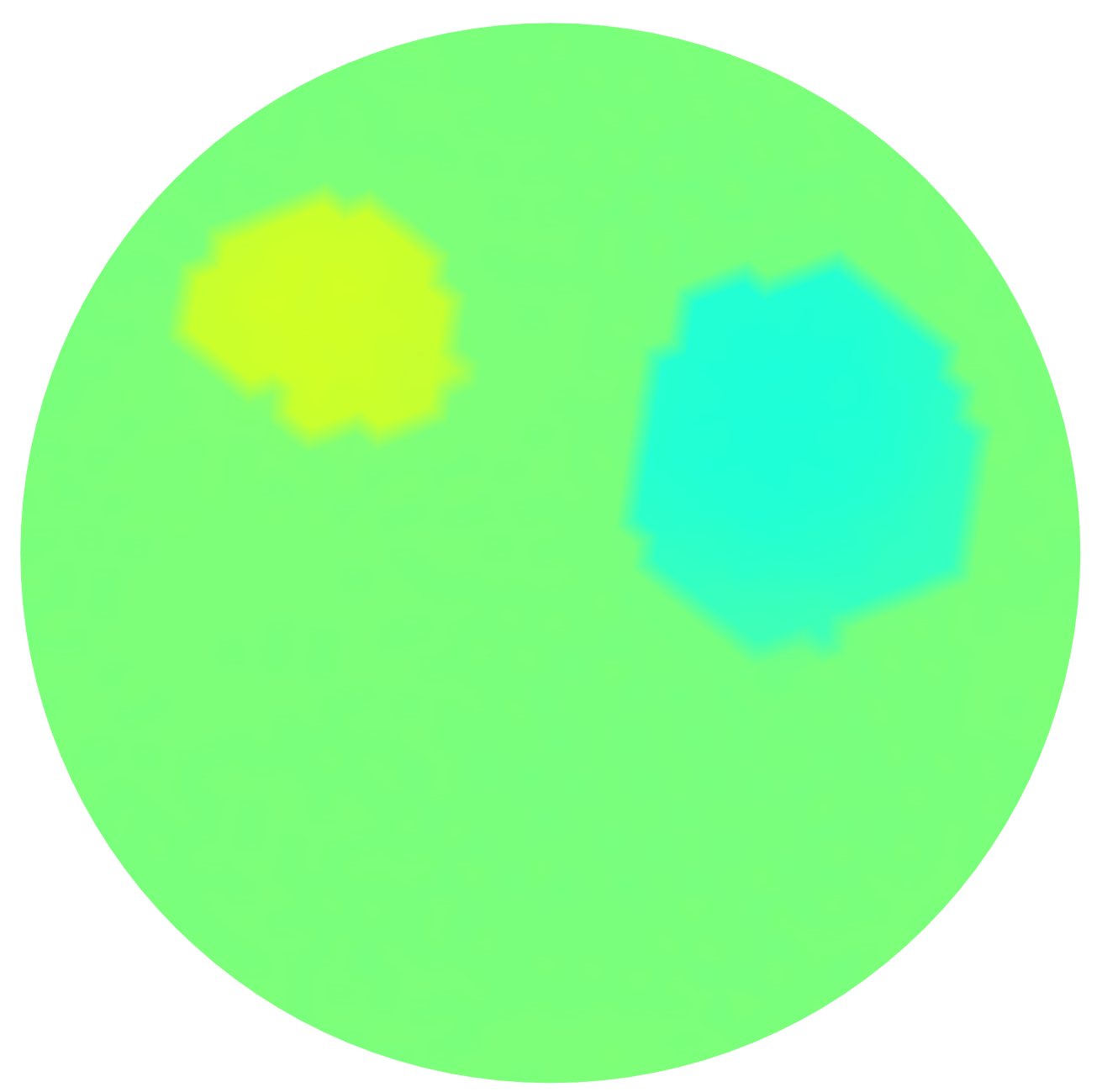}\\[-2pt]
  {\footnotesize (0.038, 3.84\%)}
\end{minipage}
&
\begin{minipage}[t]{0.13\textwidth}\centering\vspace{0pt}
  \includegraphics[width=\linewidth]{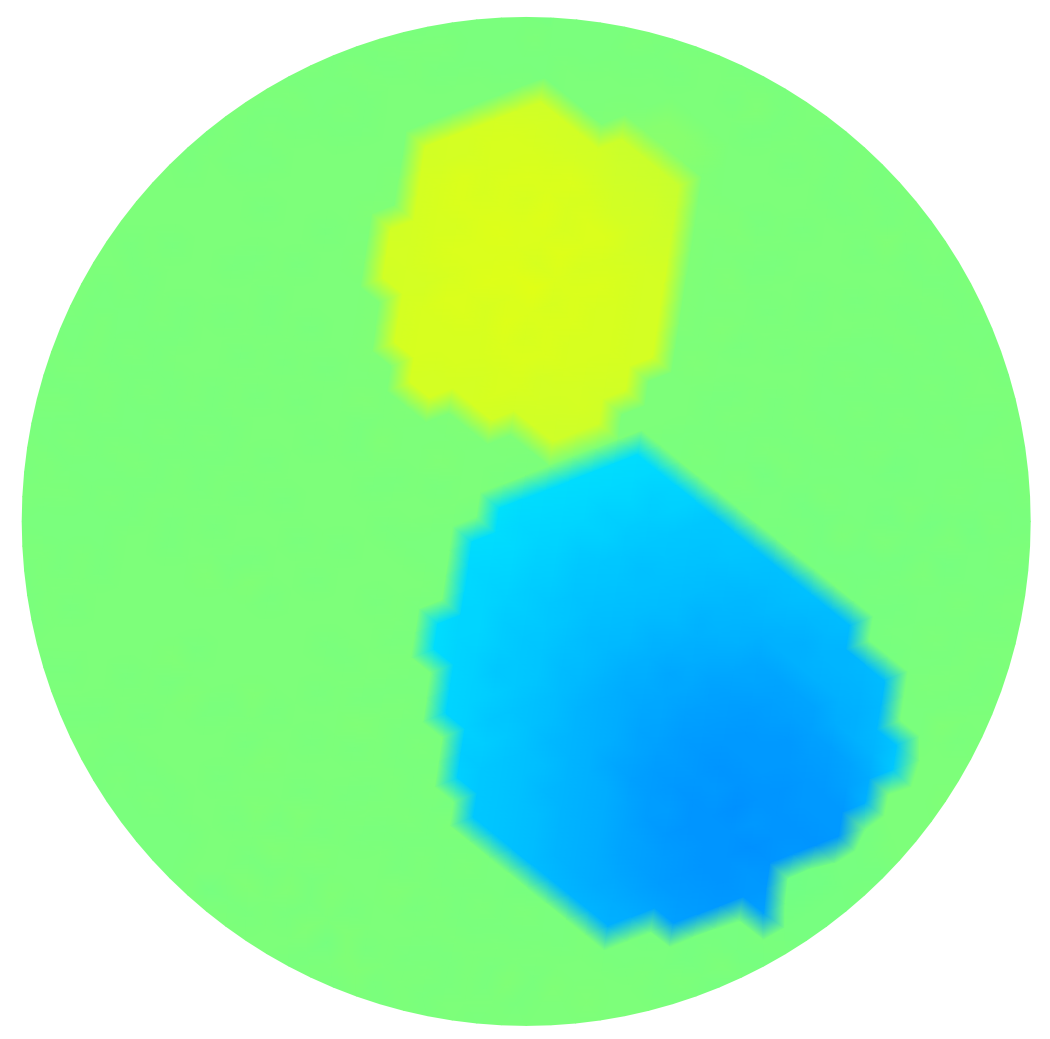}\\[-2pt]
  {\footnotesize (0.051, 5.18\%)}
\end{minipage}
&
\begin{minipage}[t]{0.13\textwidth}\centering\vspace{0pt}
  \includegraphics[width=\linewidth]{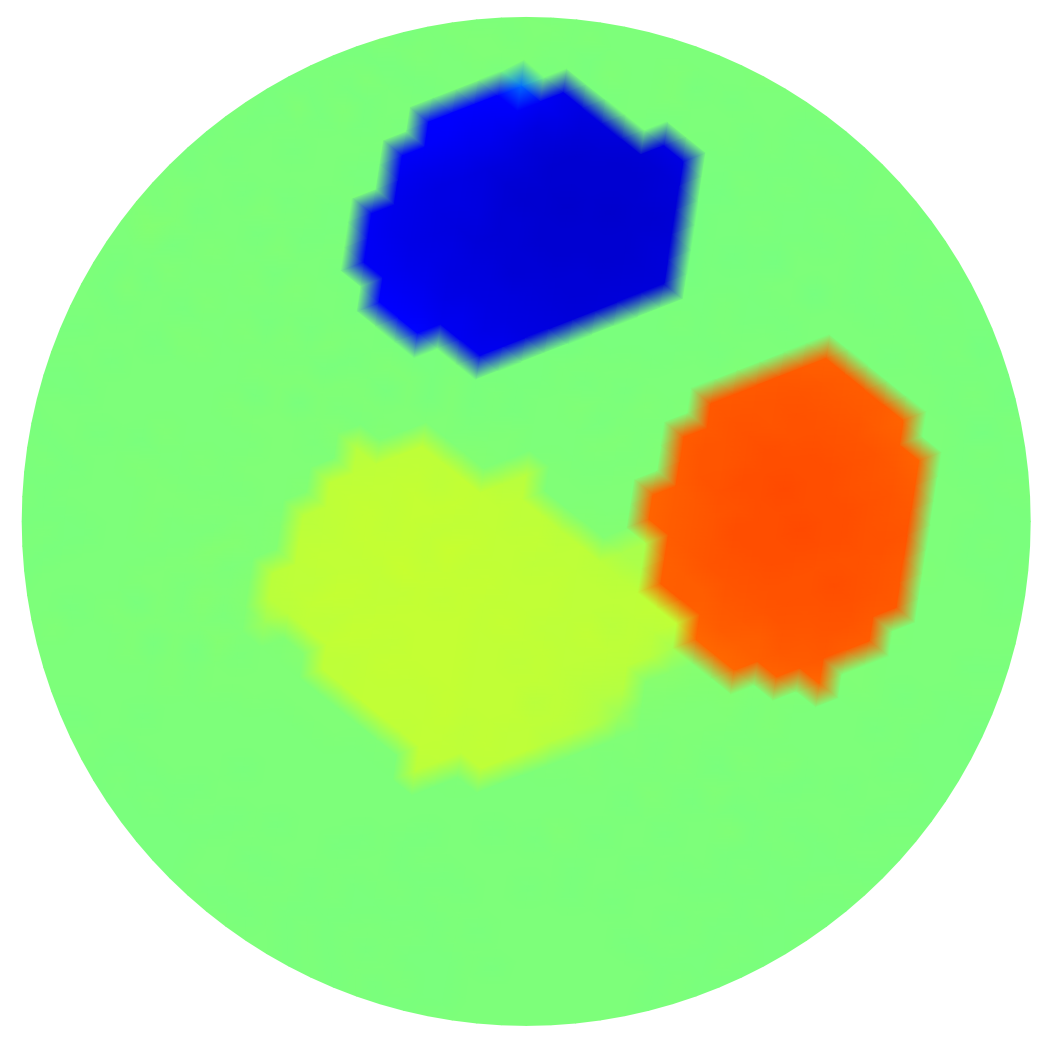}\\[-2pt]
  {\footnotesize (0.116, 11.28\%)}
\end{minipage}
&
\begin{minipage}[t]{0.13\textwidth}\centering\vspace{0pt}
  \includegraphics[width=\linewidth]{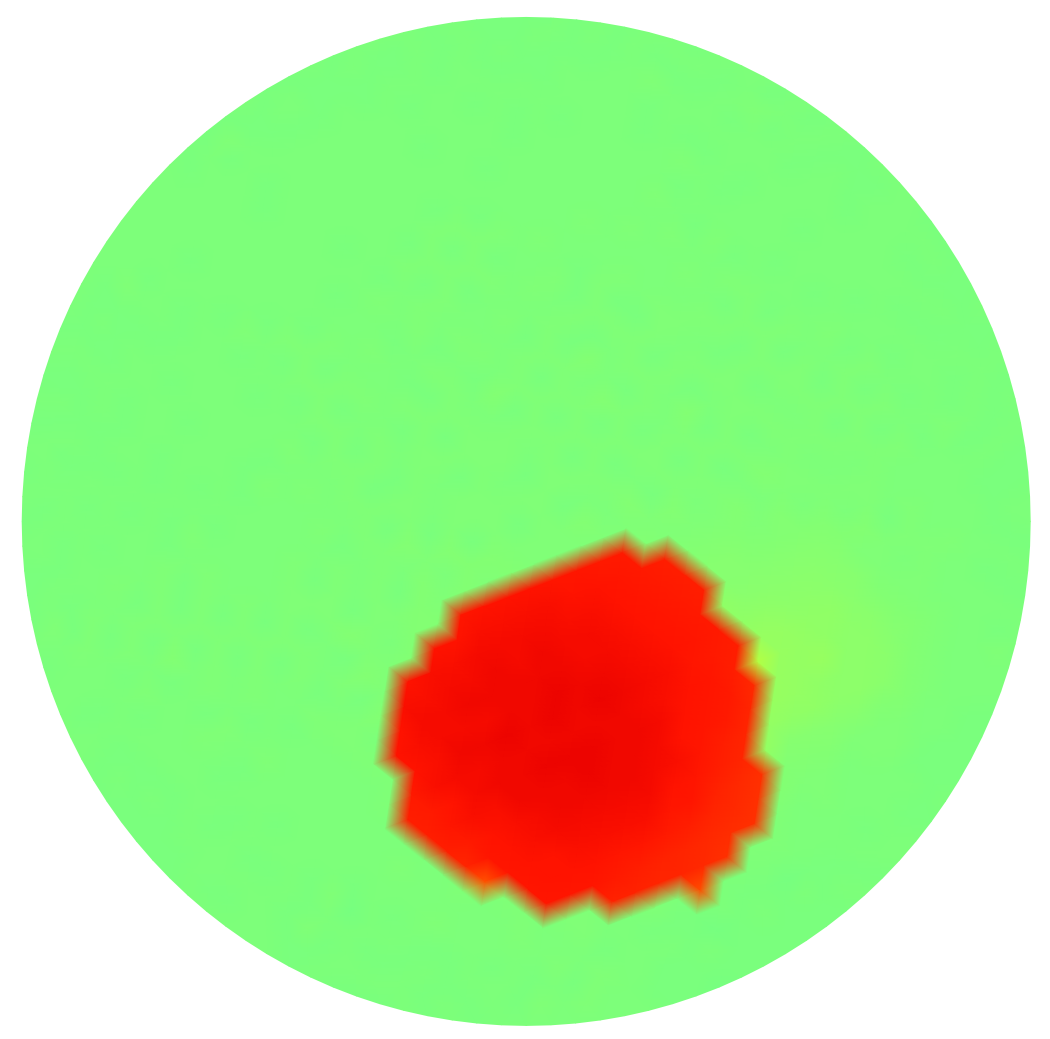}\\[-2pt]
  {\footnotesize (0.053, 4.98\%)}
\end{minipage}
\\[6pt]
$\lambda_t$,  (AWLN, $\sigma_2$)
&
\begin{minipage}[t]{0.13\textwidth}\centering\vspace{0pt}
  \includegraphics[width=\linewidth]{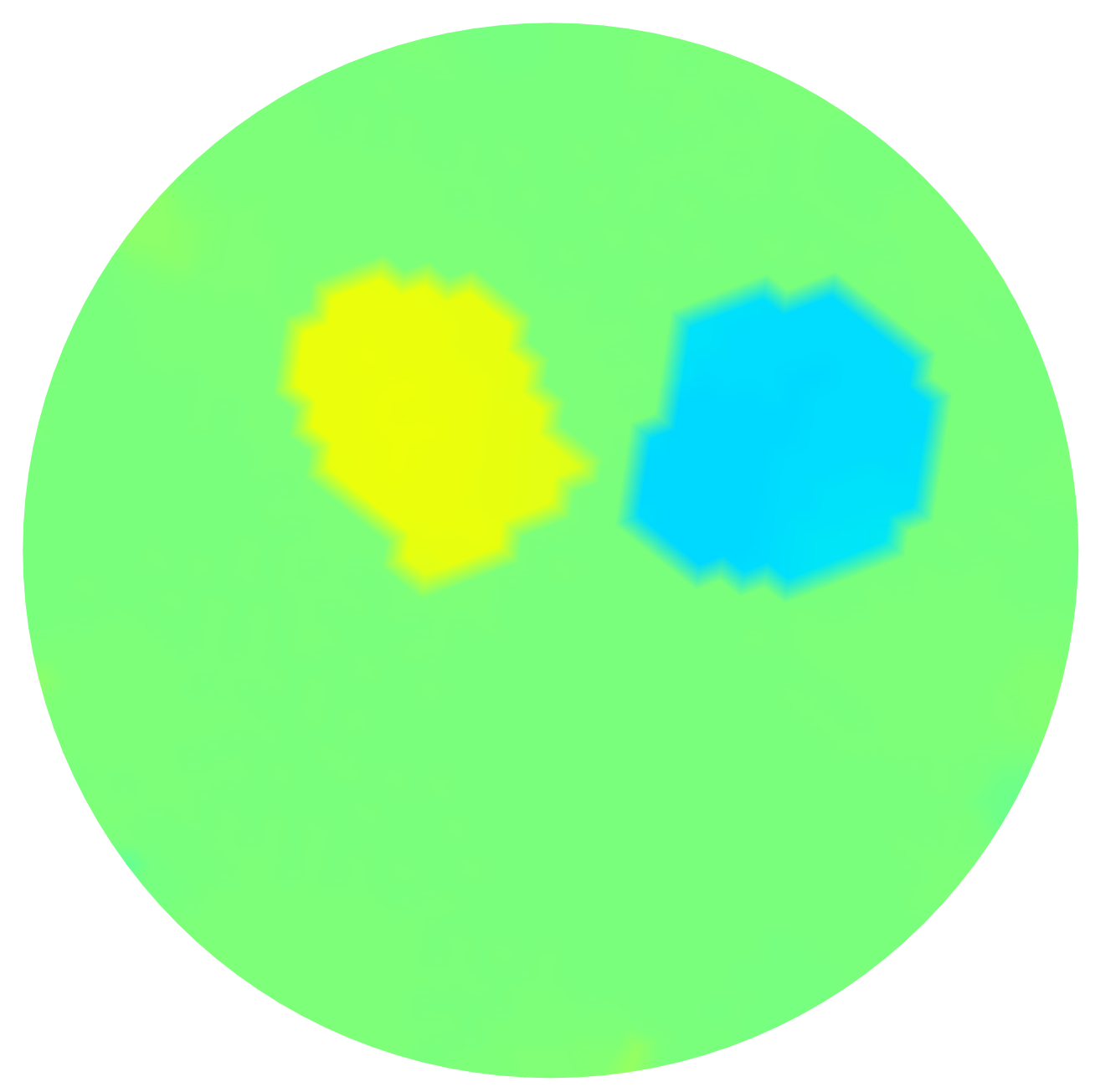}\\[-2pt]
  {\footnotesize (0.050, 5.05\%)}
\end{minipage}
&
\begin{minipage}[t]{0.13\textwidth}\centering\vspace{0pt}
  \includegraphics[width=\linewidth]{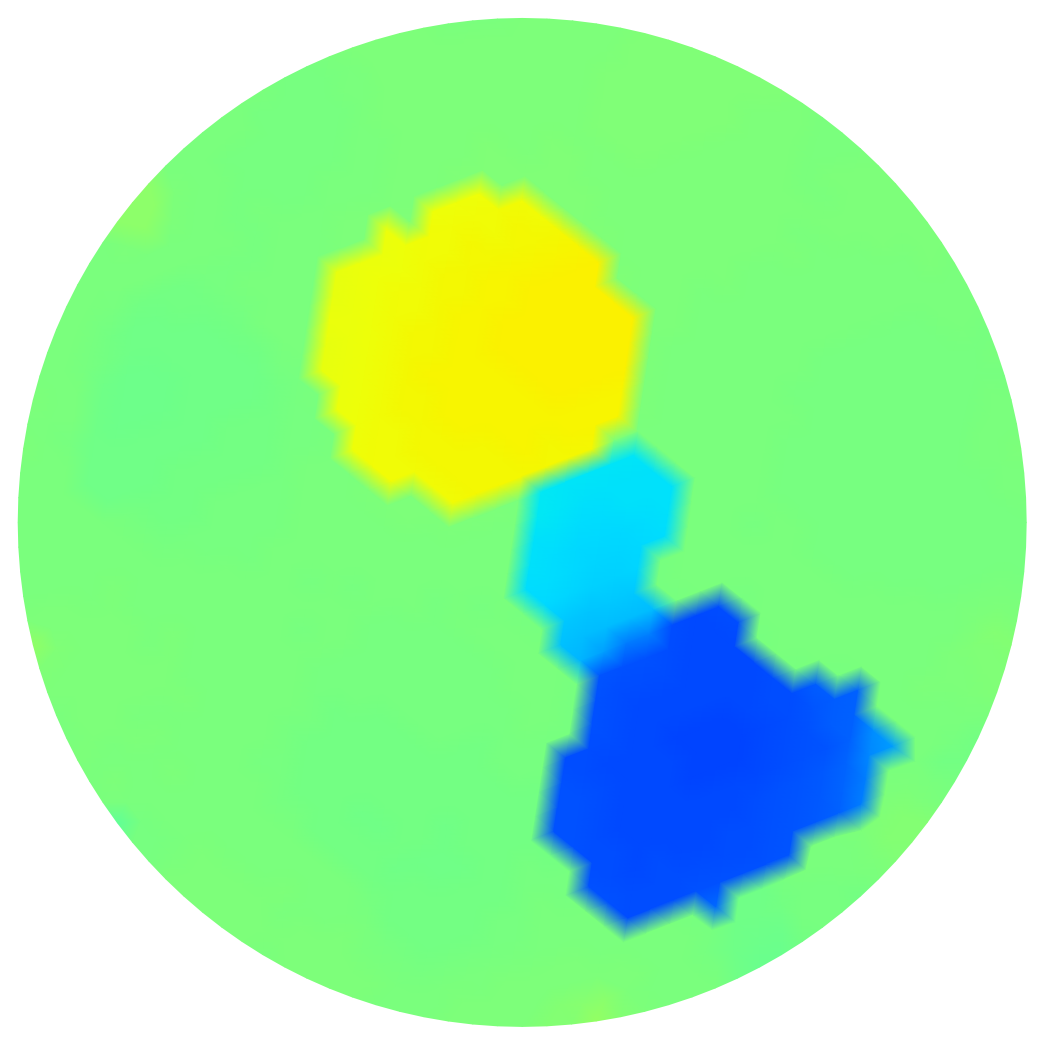}\\[-2pt]
  {\footnotesize (0.078, 7.95\%)}
\end{minipage}
&
\begin{minipage}[t]{0.13\textwidth}\centering\vspace{0pt}
  \includegraphics[width=\linewidth]{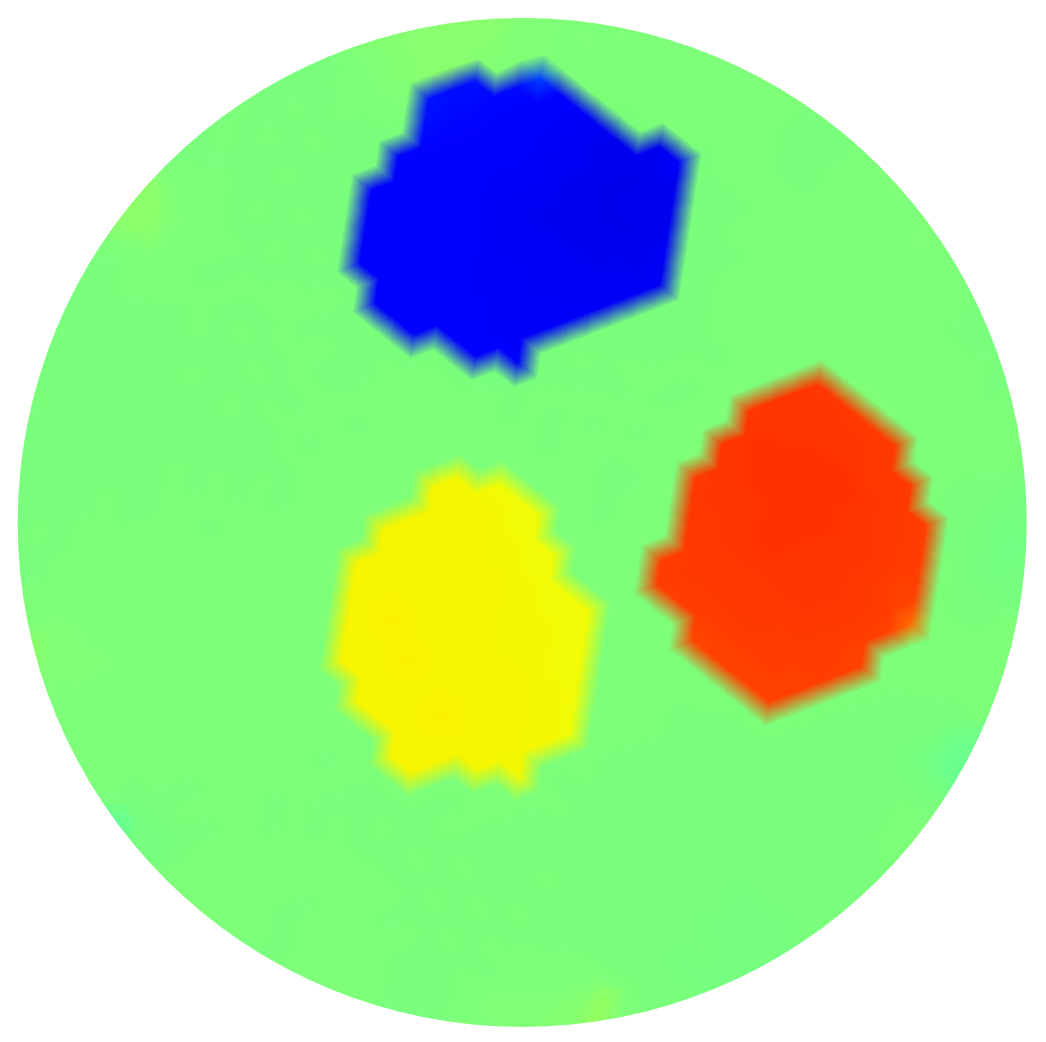}\\[-2pt]
  {\footnotesize (0.127, 12.36\%)}
\end{minipage}
&
\begin{minipage}[t]{0.13\textwidth}\centering\vspace{0pt}
  \includegraphics[width=\linewidth]{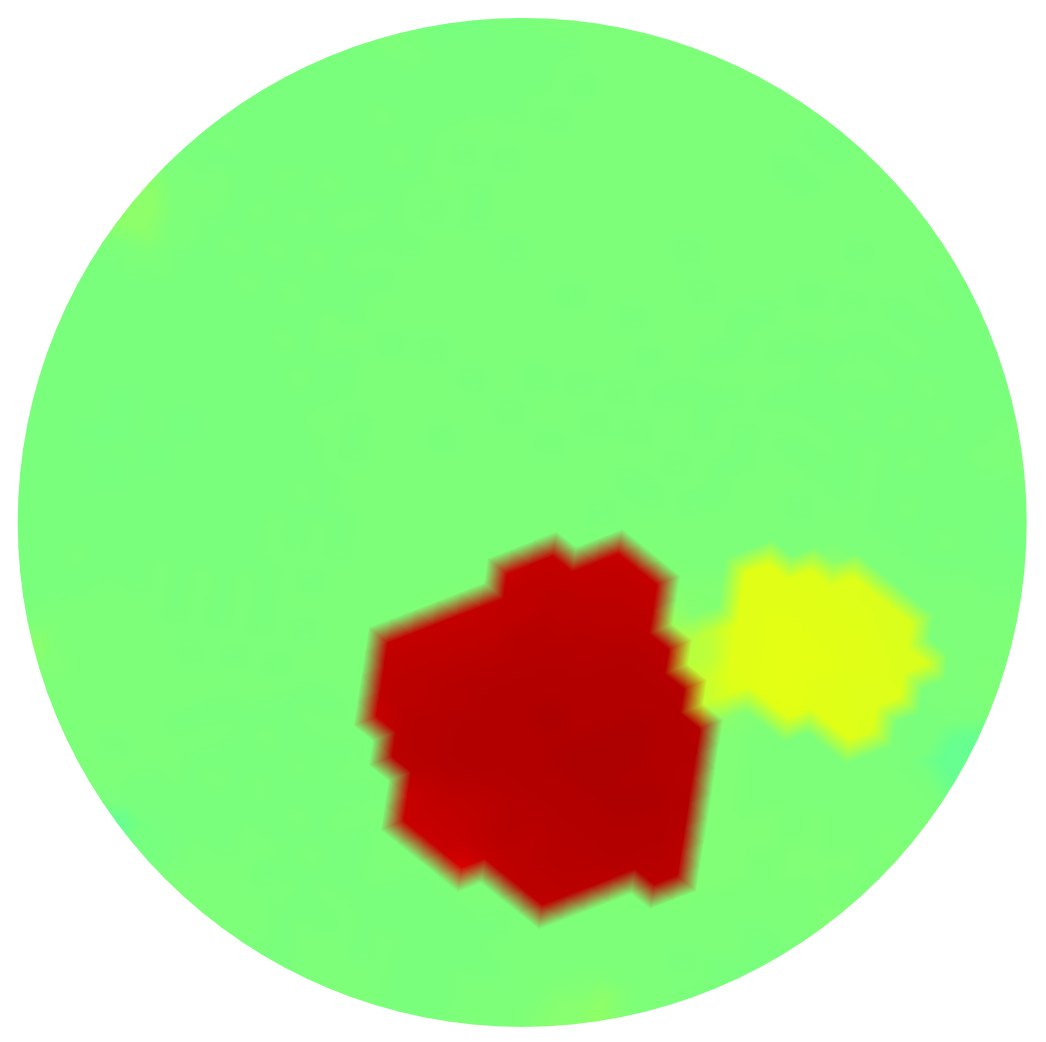}\\[-2pt]
  {\footnotesize (0.078, 7.30\%)}
\end{minipage}
\\[6pt]
\\
\end{tabular}

\caption{Example 2: Robustness to measurement noise.  
Reconstructions obtained with adaptive parameters $\lambda_t$ under AWGN and AWLN  at noise levels $\sigma_1$ and $\sigma_2$. The numbers below each reconstruction denote (RMSE, Rel Err). The clamping ranges are $[0.15,\,0.5]$ for noise level $\sigma_1$ and $[0.2,\,0.8]$ for noise level $\sigma_2$.}

\label{fig:noise_robustness}
\end{figure}

\subsection{Example 3: Out-of-Distribution Generalization and Real Dataset}\label{sub:ex3}
To evaluate the out-of-distribution generalization capability of DDIM-RDPS, we perform experiments on test samples with inclusion geometries and inclusion numbers that differ significantly from those encountered during training. This experiment is designed to examine whether the learned diffusion prior remains effective in recovering irregular and non-circular structures that are absent from the training distribution.
Figure~\ref{ood} summarizes three representative out-of-distribution settings: columns (a)--(b) show models trained on circle-shaped inclusions (DATASET1) and tested on blob-shaped inclusions, columns (c)--(d) show models trained on blob-shaped inclusions (DATASET3) and tested on horseshoe-shaped inclusions, and columns (e)--(f) show models trained on samples from DATASET4 with a random number of 1--3 blobs and tested on samples containing 4 blobs. In each column, the first row gives the ground-truth conductivity, and the second row shows the reconstruction generated by our method, together with the corresponding quantitative error measures. These results indicate that, despite the distribution mismatch, the proposed method remains capable of capturing the main geometric and contrast features of the targets.

We further evaluated the proposed method on the open 2D EIT dataset~\cite{hauptmann2017open}. The experimental setup consists of a circular tank filled with a conductive solution and equipped with 16 uniformly spaced electrodes along its boundary. Since the electrodes extend over the full depth of the tank, the measurements can be modeled using a two-dimensional EIT configuration.
To assess the ability of the method to generalize from synthetic to real data, we directly apply the model trained on DATASET1, which consists of the classical circular inclusions, to this real dataset without retraining. The data are acquired using an adjacent current-injection and adjacent voltage-measurement protocol. Specifically, for each sample, 16 adjacent current injection patterns are applied, and for each injection, 13 voltage measurements are recorded on adjacent electrode pairs, excluding the current-injecting electrodes. This yields a total of 208 voltage measurements per sample. These measurements are then used to reconstruct the conductivity distribution using both DDIM-RDPS, and the regularized Gauss-Newton (RGN-TV) method included as a baseline for comparison \cite{RGN2022,borsic2009vivo}.

Representative reconstructions are presented in Fig.~\ref{fig:recon_2d_EIT}. Overall, the proposed method yields reconstructions with higher fidelity and fewer artifacts than RGN-TV, demonstrating its superior performance on real experimental data.

\begin{figure}[H]
\centering
\small

\begin{adjustbox}{width=\linewidth}
\begin{tabular}{C{2.8cm} C{2.8cm}| C{2.8cm} C{2.8cm}| C{2.8cm} C{2.8cm}}

\imgonly{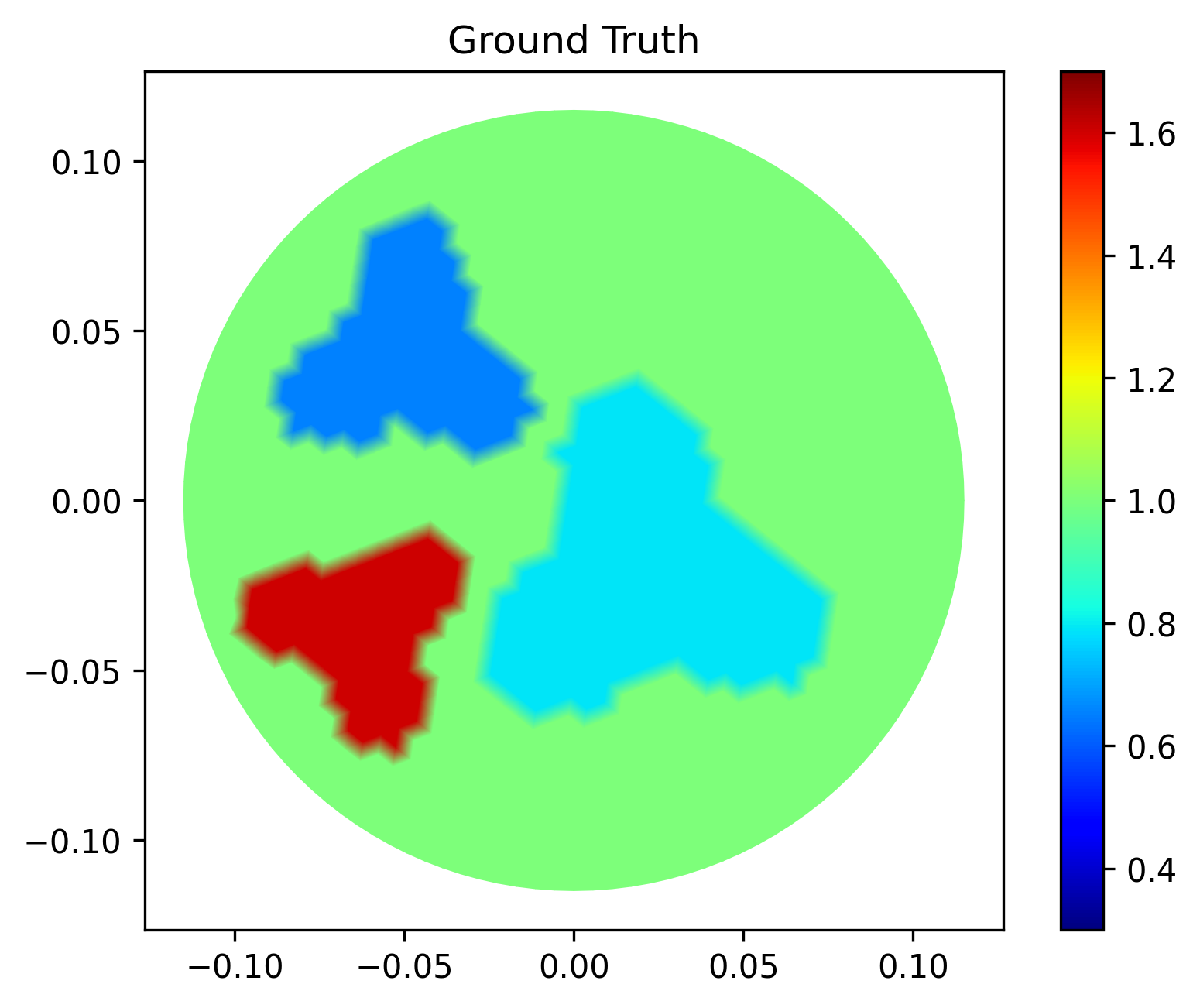} &
\imgonly{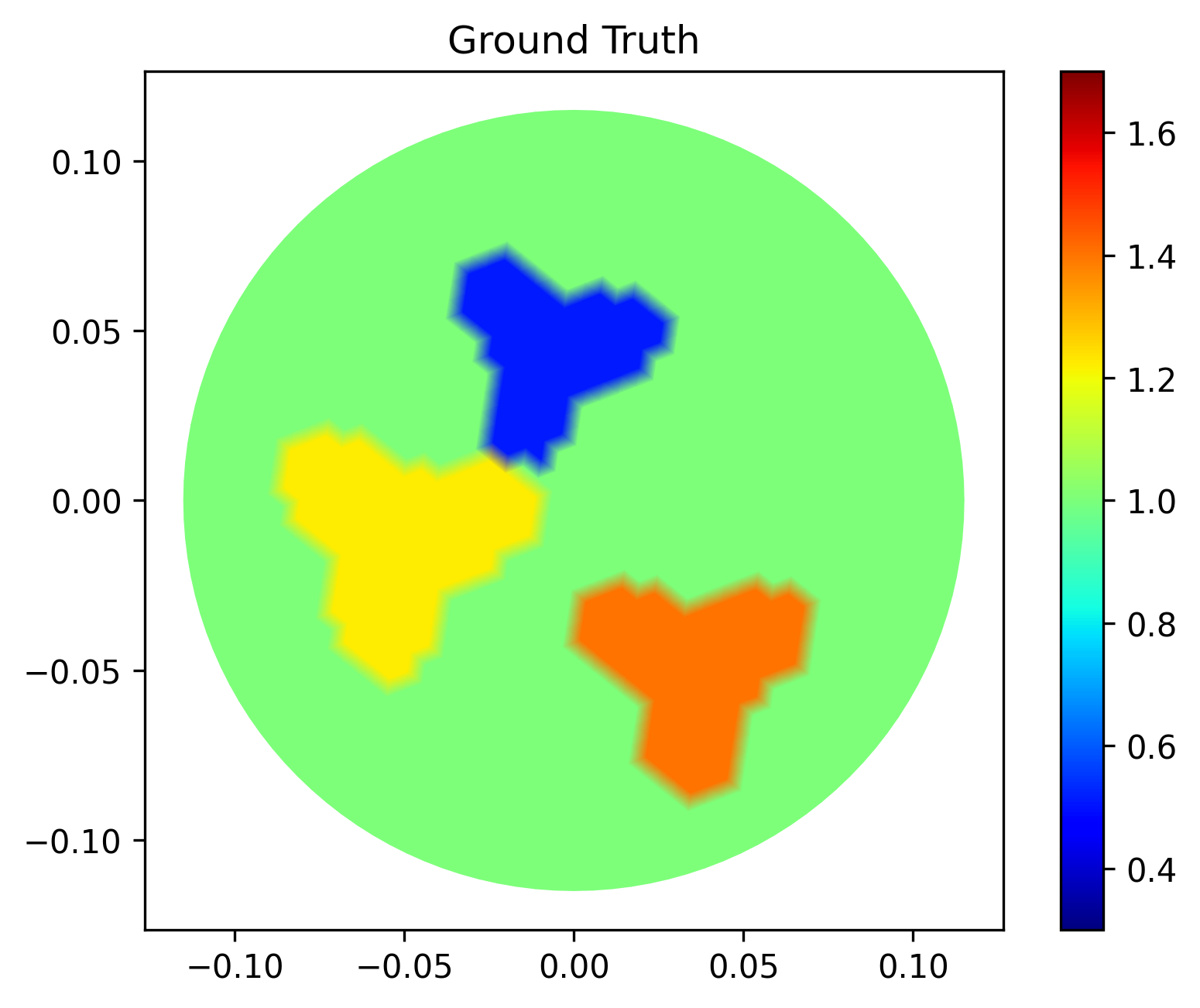} &
\imgonly{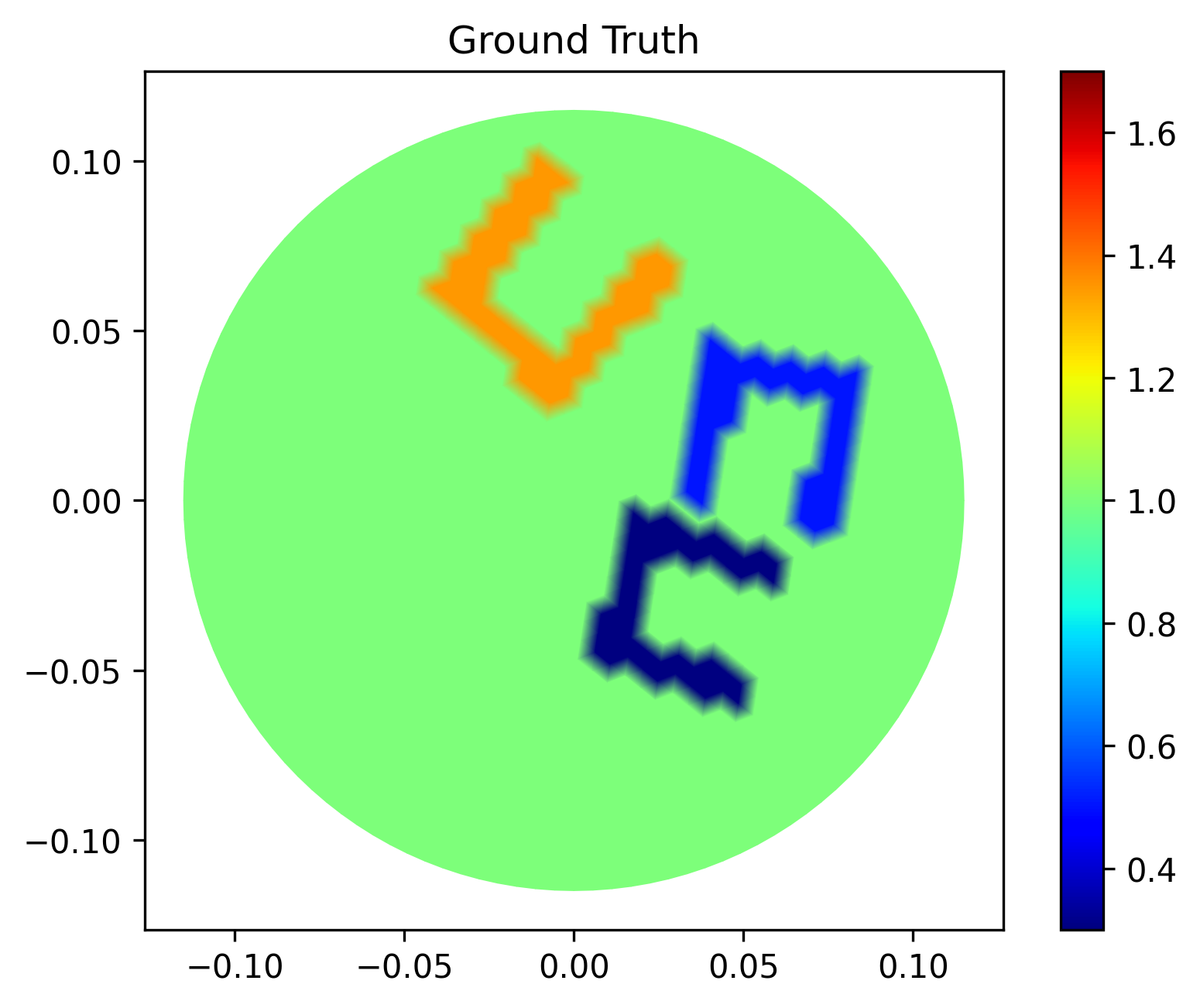} &
\imgonly{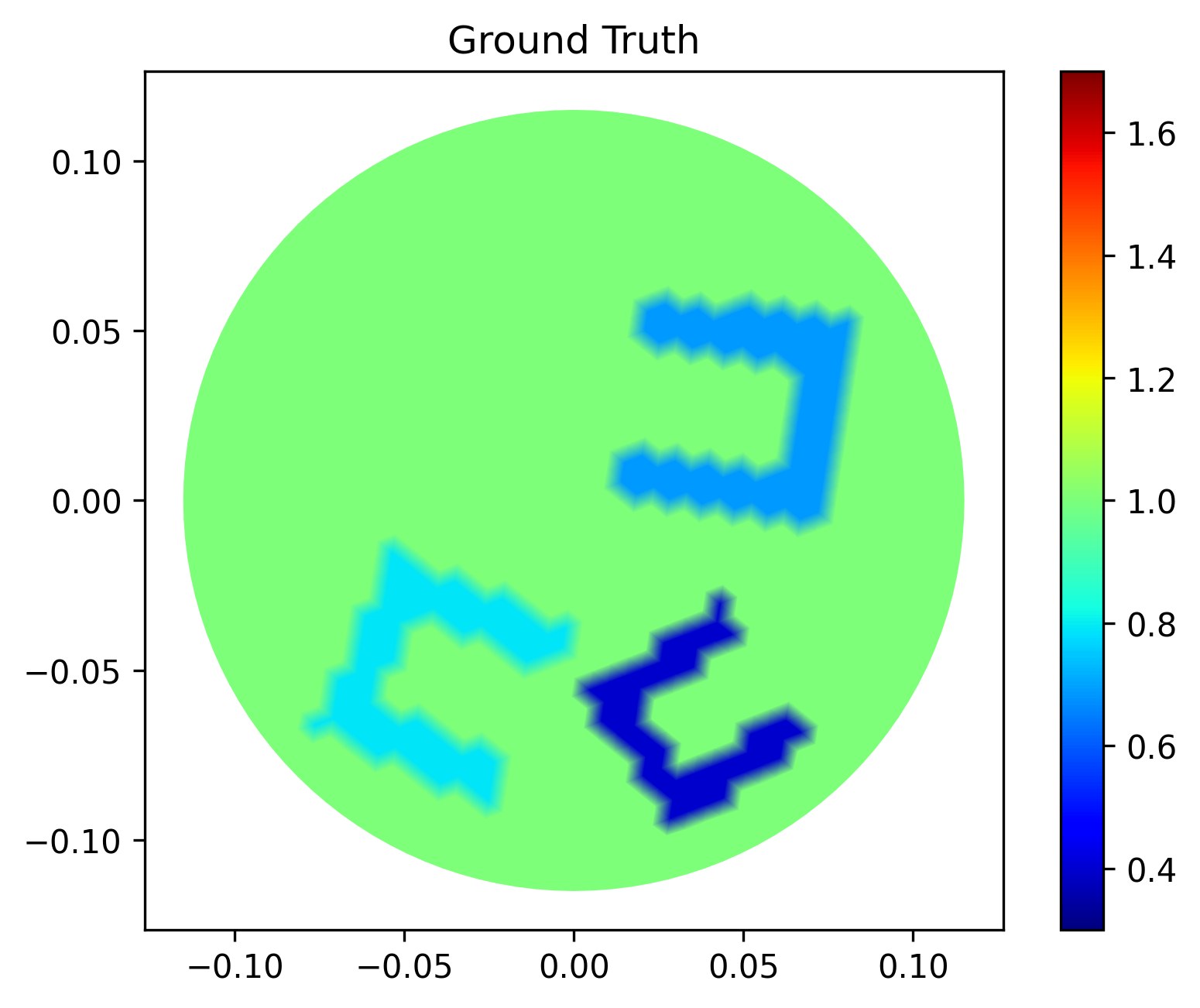} &
\imgonly{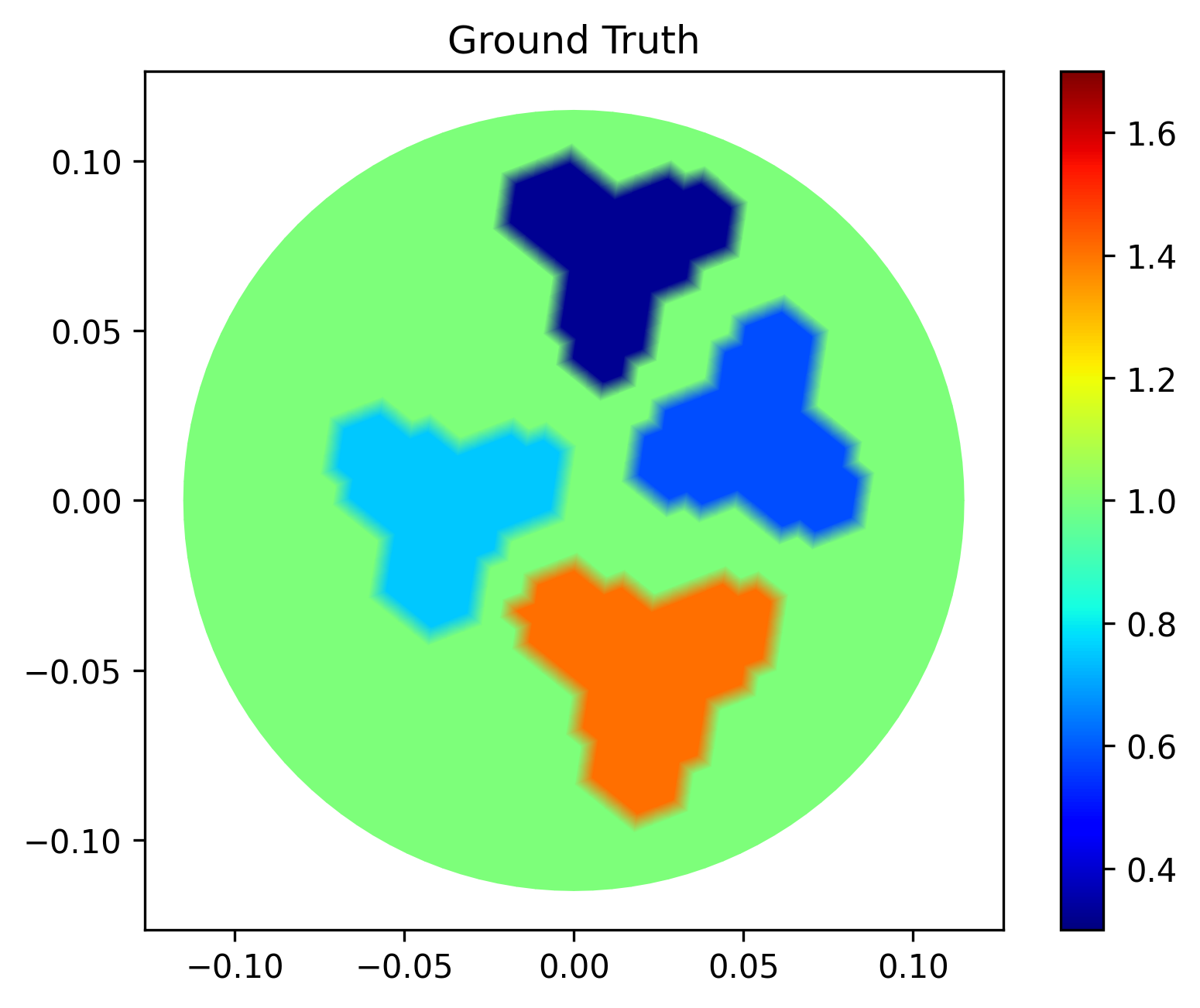} &
\imgonly{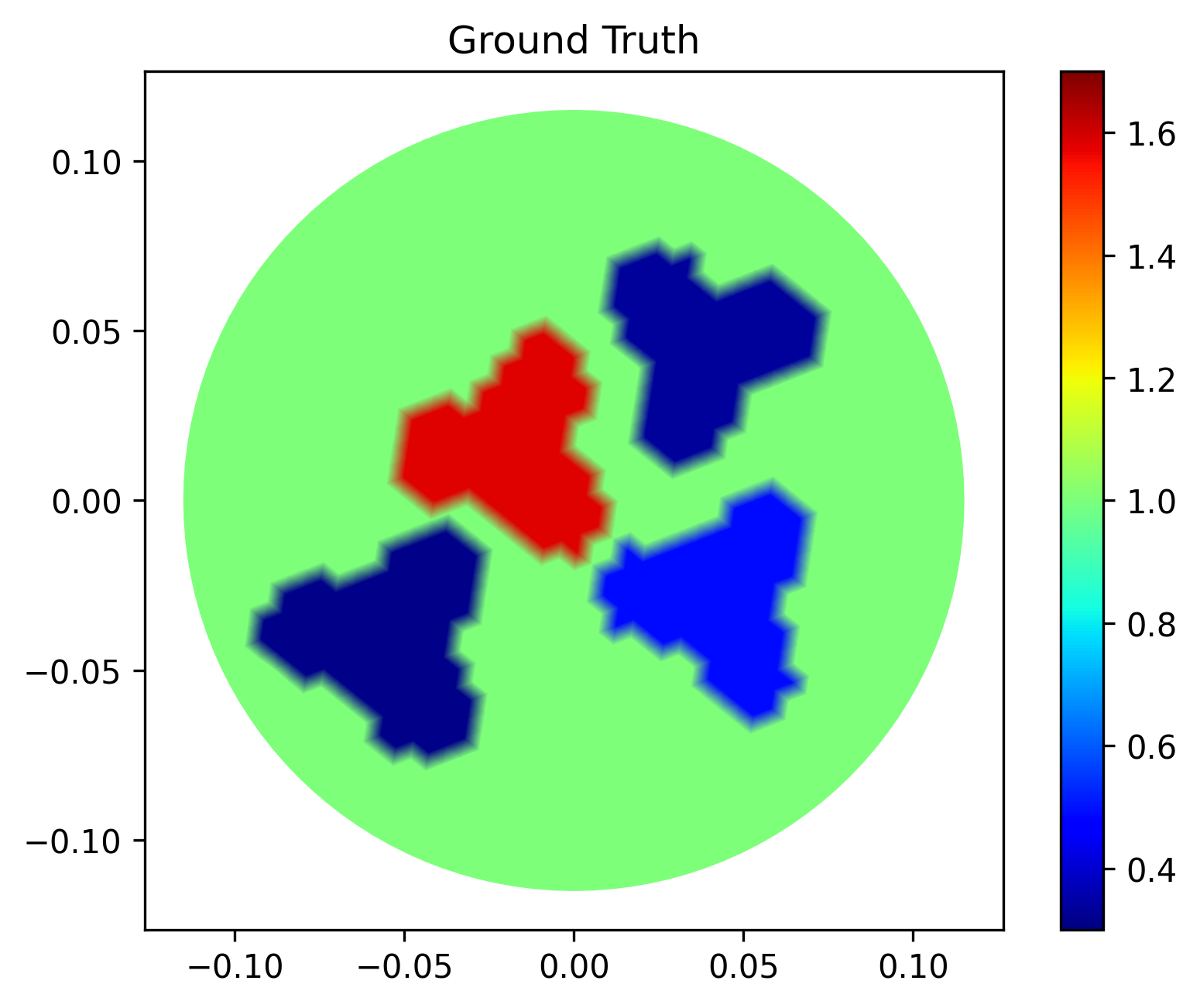} \\[10mm]

\imgmetric{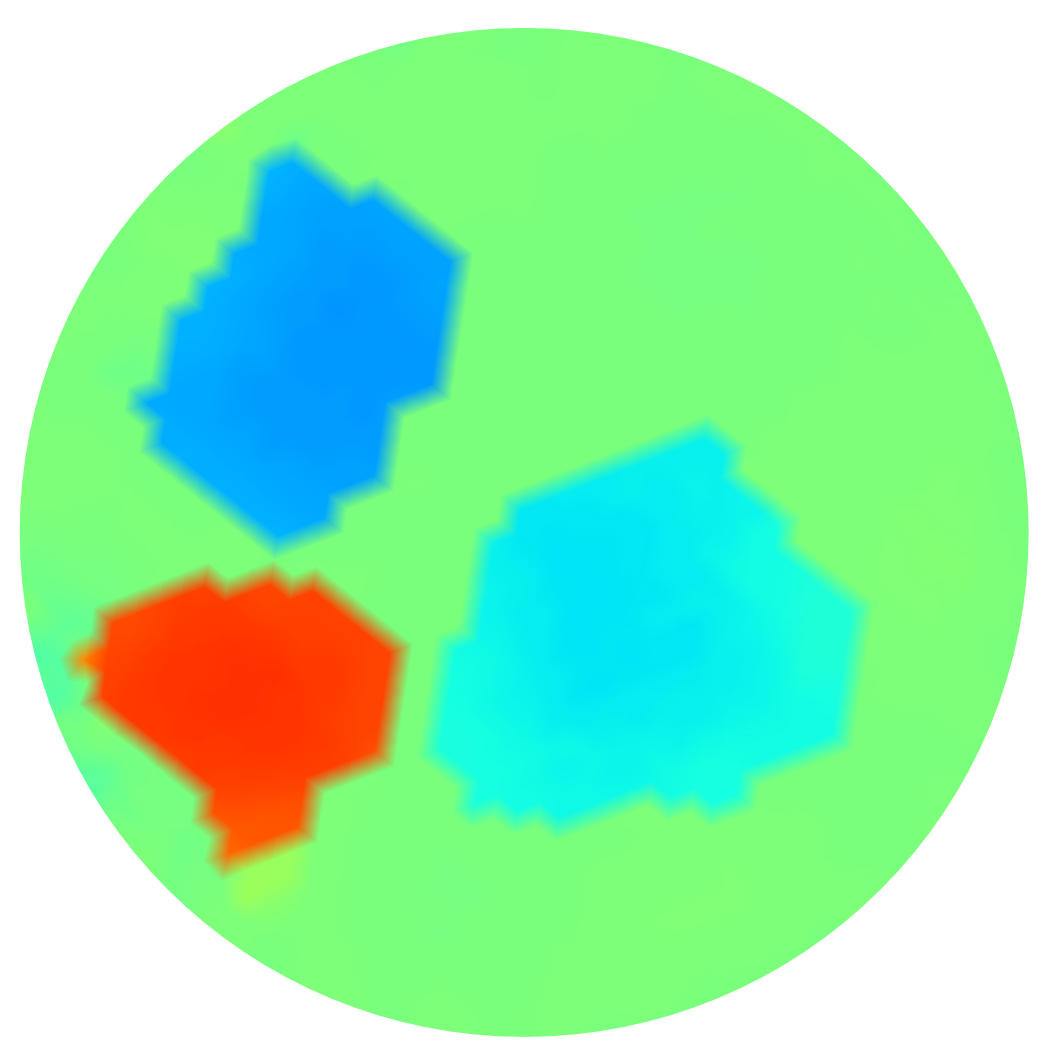}{(0.106, 10.60\%)} &
\imgmetric{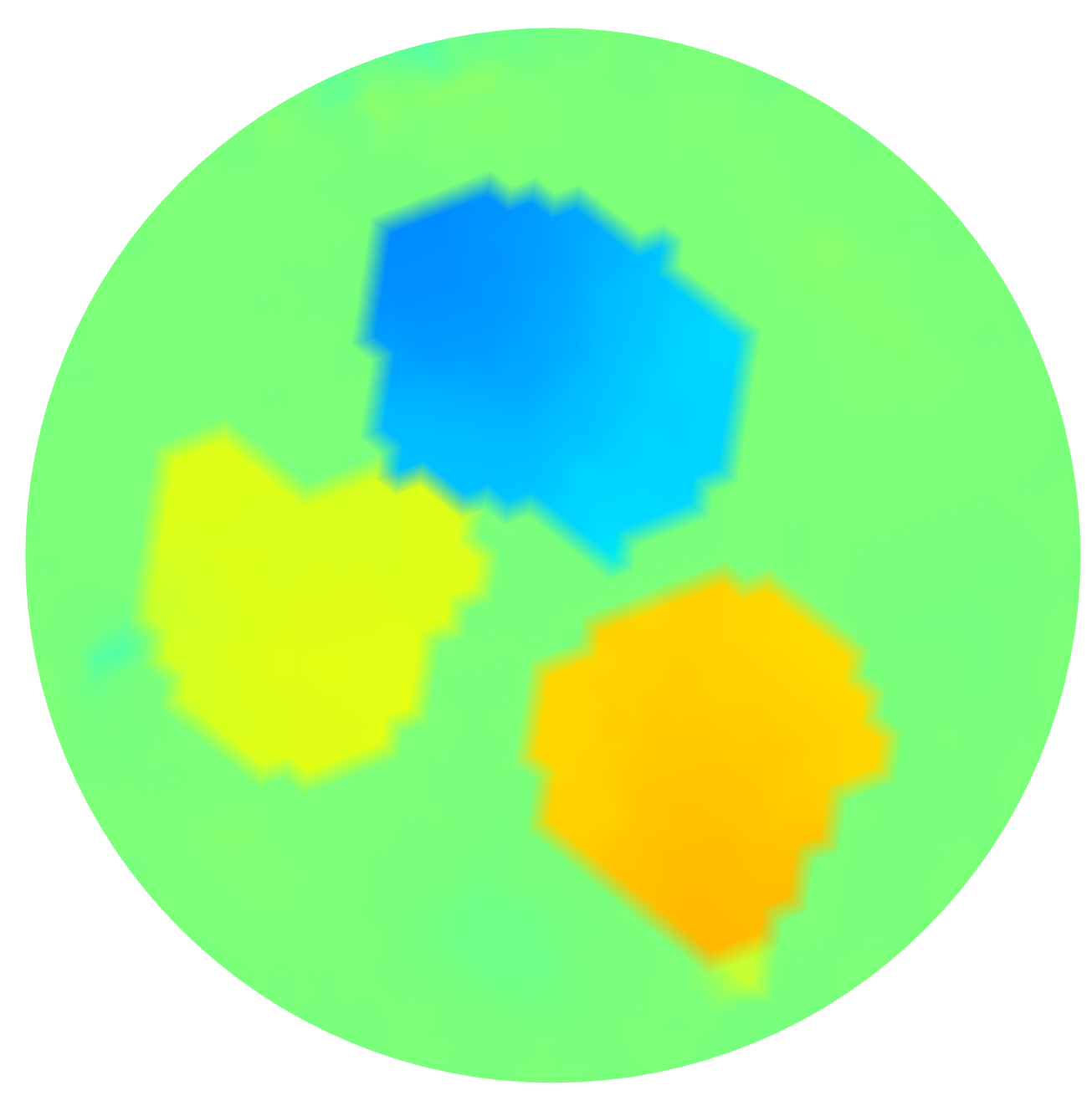}{(0.105, 10.20\%)} &
\imgmetric{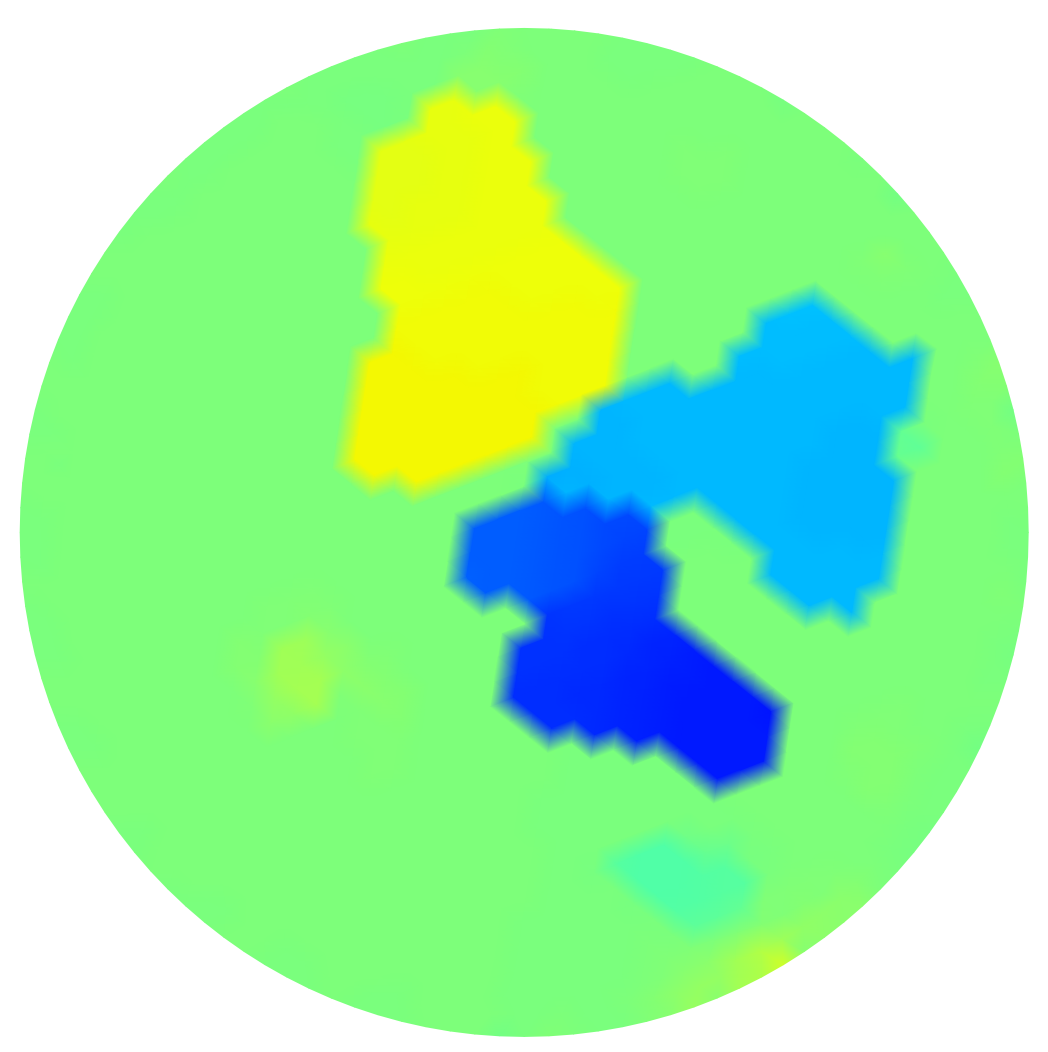}{(0.160, 16.22\%)} &
\imgmetric{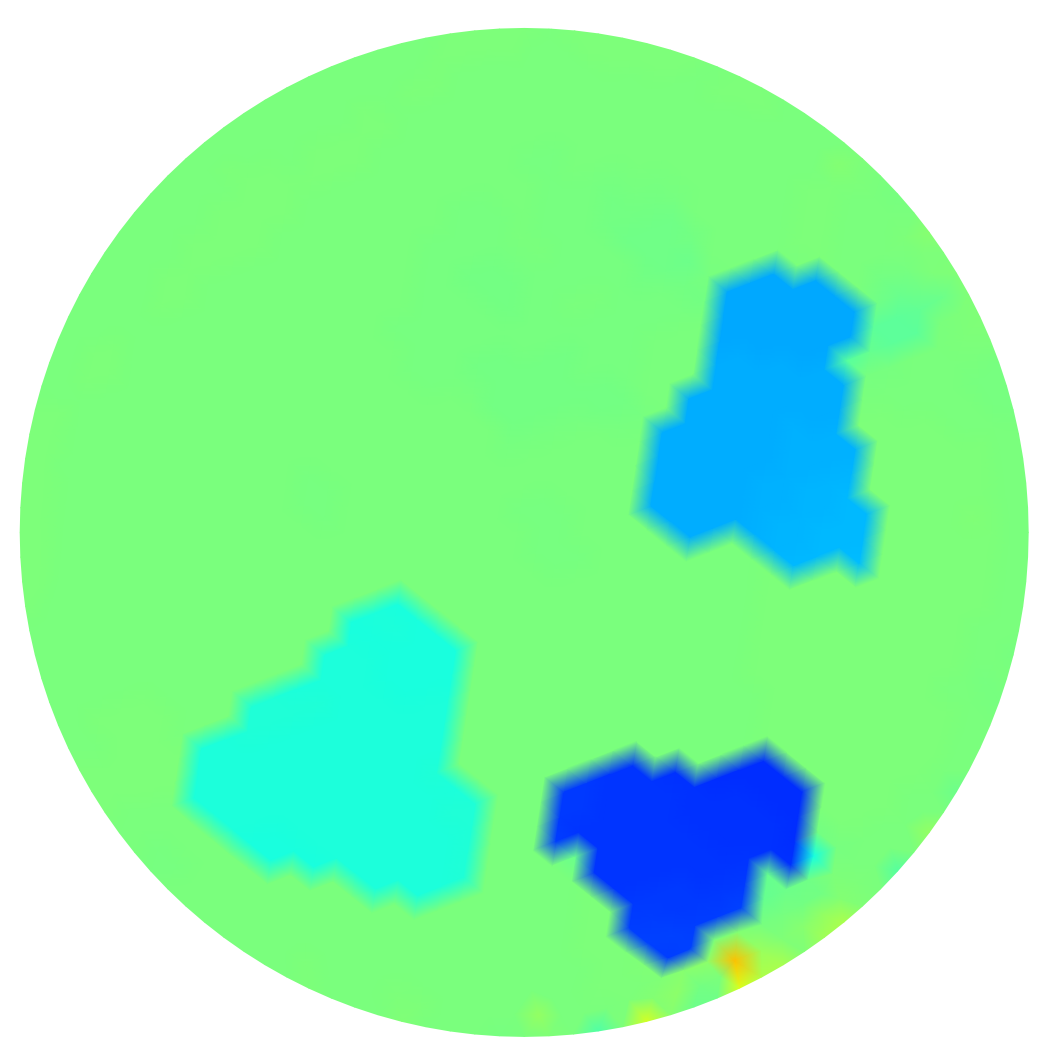}{(0.120, 12.47\%)} &
\imgmetric{results_1602/Out-of-Distribution/sampling_with_metrics_5_77_0.02.png}{(0.123, 12.62\%)} &
\imgmetric{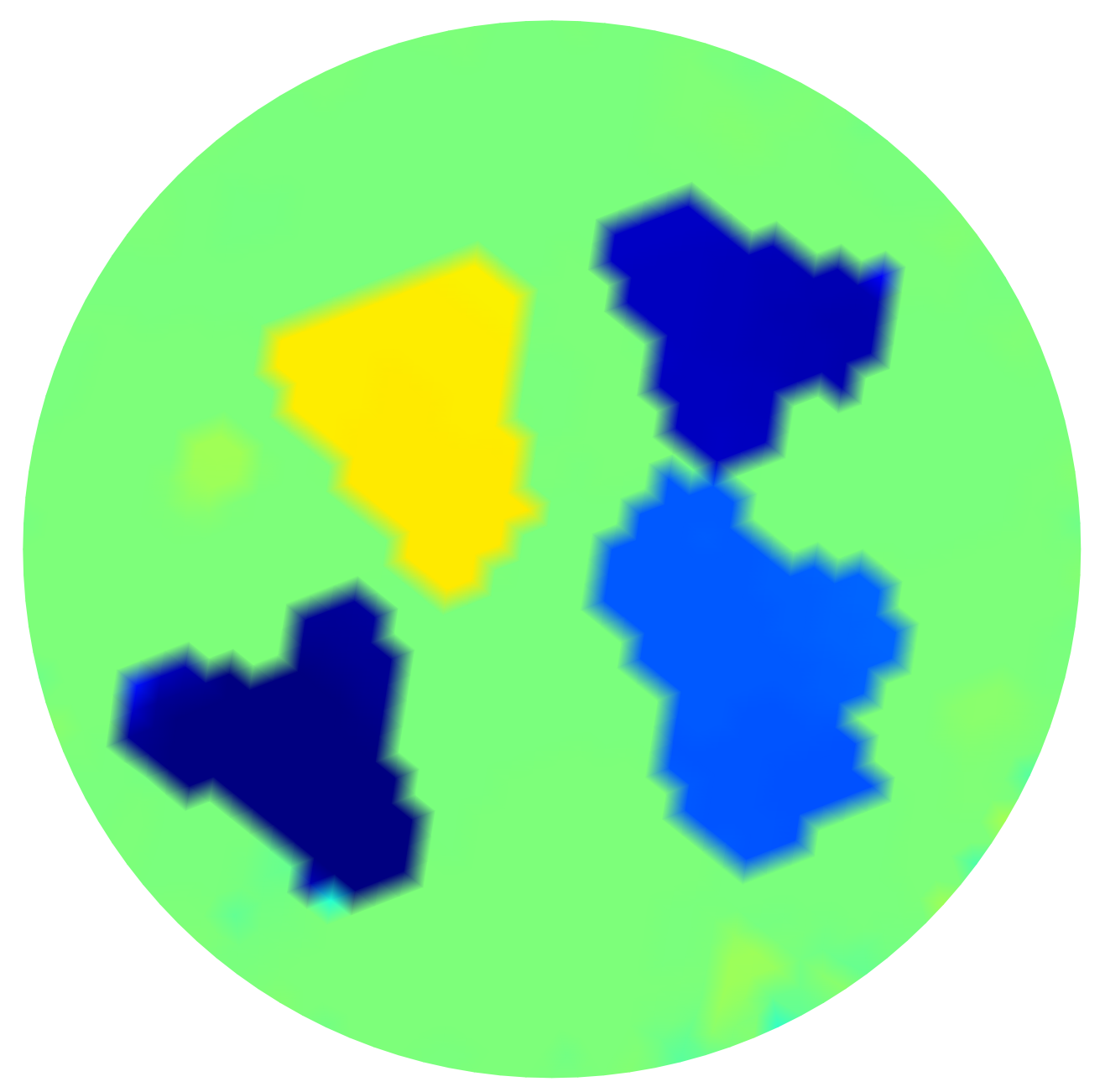}{(0.171, 17.74\%)} \\[2mm]

(a) & (b) & (c) & (d) & (e) & (f) \\
\end{tabular}
\end{adjustbox}

\caption{Example 3: Out-of-distribution reconstruction results. Columns (a)--(b) correspond to models trained on circular inclusions and tested on blob-shaped inclusions; columns (c)--(d) correspond to models trained on blob-shaped inclusions and tested on horseshoe-shaped inclusions; columns (e)--(f) correspond to models trained on samples with a random number of 1--3 blobs and tested on samples with 4 blobs. The first row shows the ground-truth conductivity, and the second row shows the reconstructions. The corresponding (RMSE, Rel Err) values are reported for each test case.}
\label{ood}
\end{figure}

\begin{figure}[H]
\centering
\renewcommand{\arraystretch}{1.15}

\begin{tabular}{>{\centering\arraybackslash}m{0.12\textwidth}
                @{\hspace{0.8em}}
                >{\centering\arraybackslash}m{0.18\textwidth}
                @{\hspace{0.8em}}
                >{\centering\arraybackslash}m{0.18\textwidth}
                @{\hspace{0.8em}}
                >{\centering\arraybackslash}m{0.18\textwidth}
                @{\hspace{0.8em}}
                >{\centering\arraybackslash}m{0.18\textwidth}}

& {Case 2.1} & {Case 3.1} & {Case 4.3} & {Case 4.4} \\[2pt]
\hline \\[-6pt]

&
\includegraphics[width=0.13\textwidth]{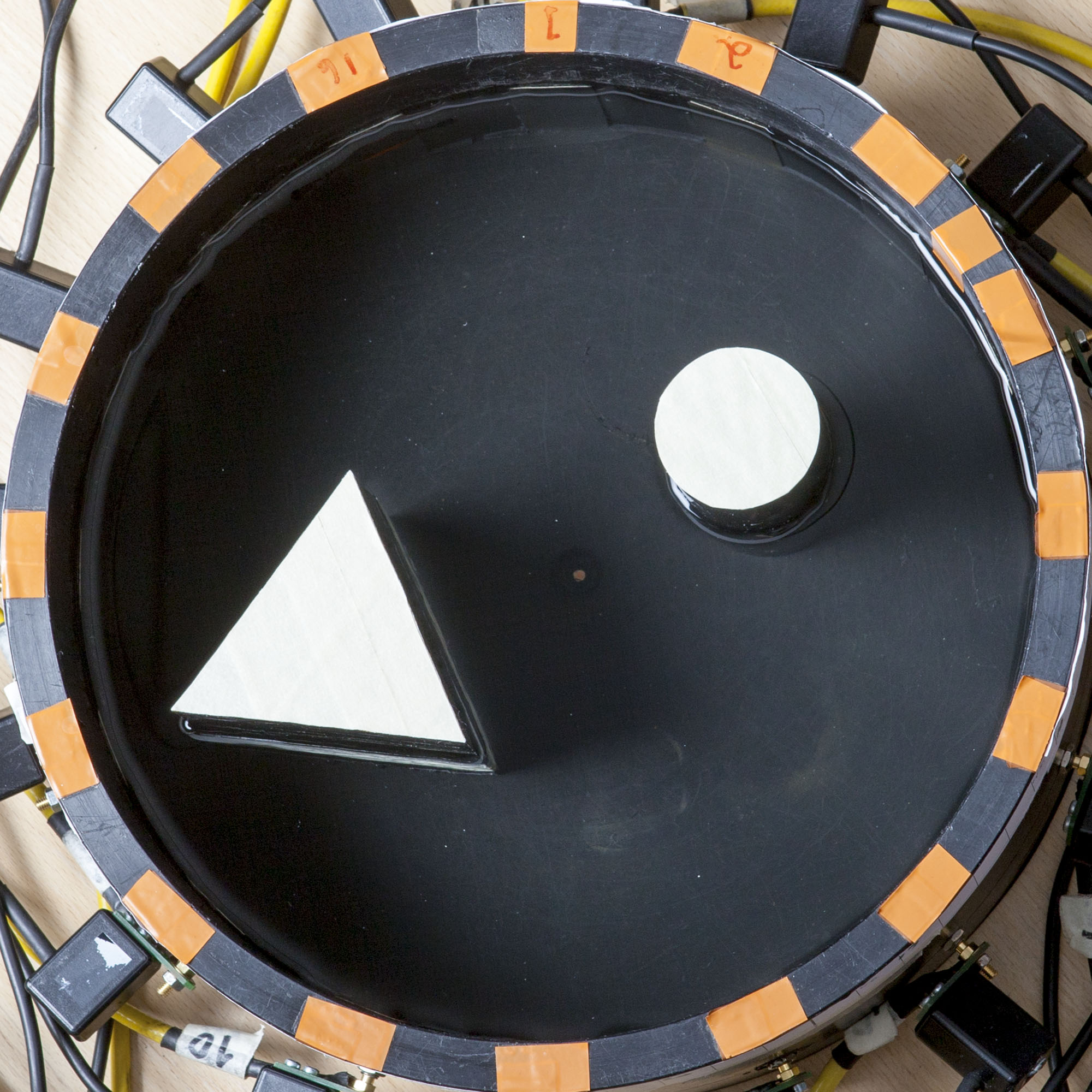} &
\includegraphics[width=0.13\textwidth]{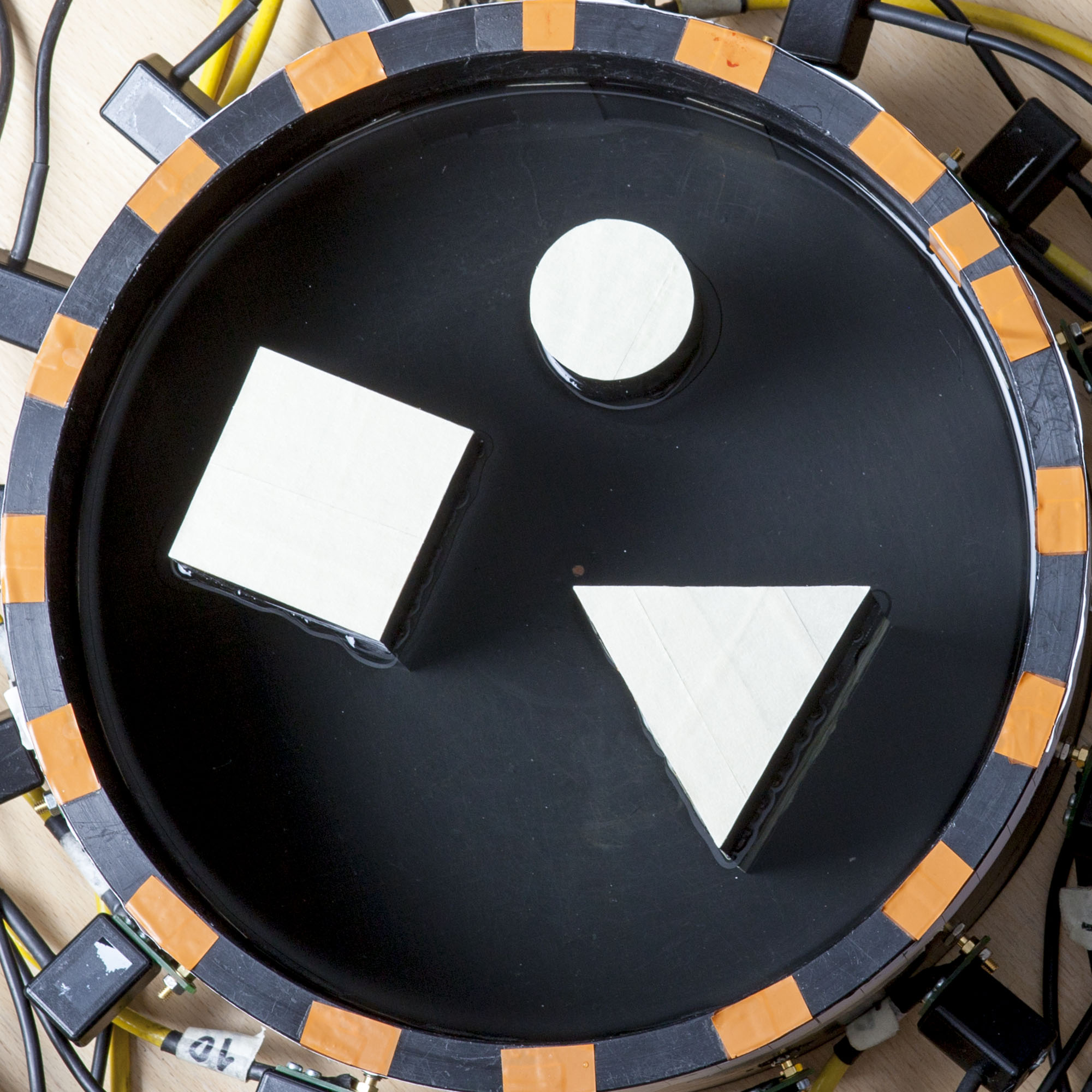} &
\includegraphics[width=0.13\textwidth]{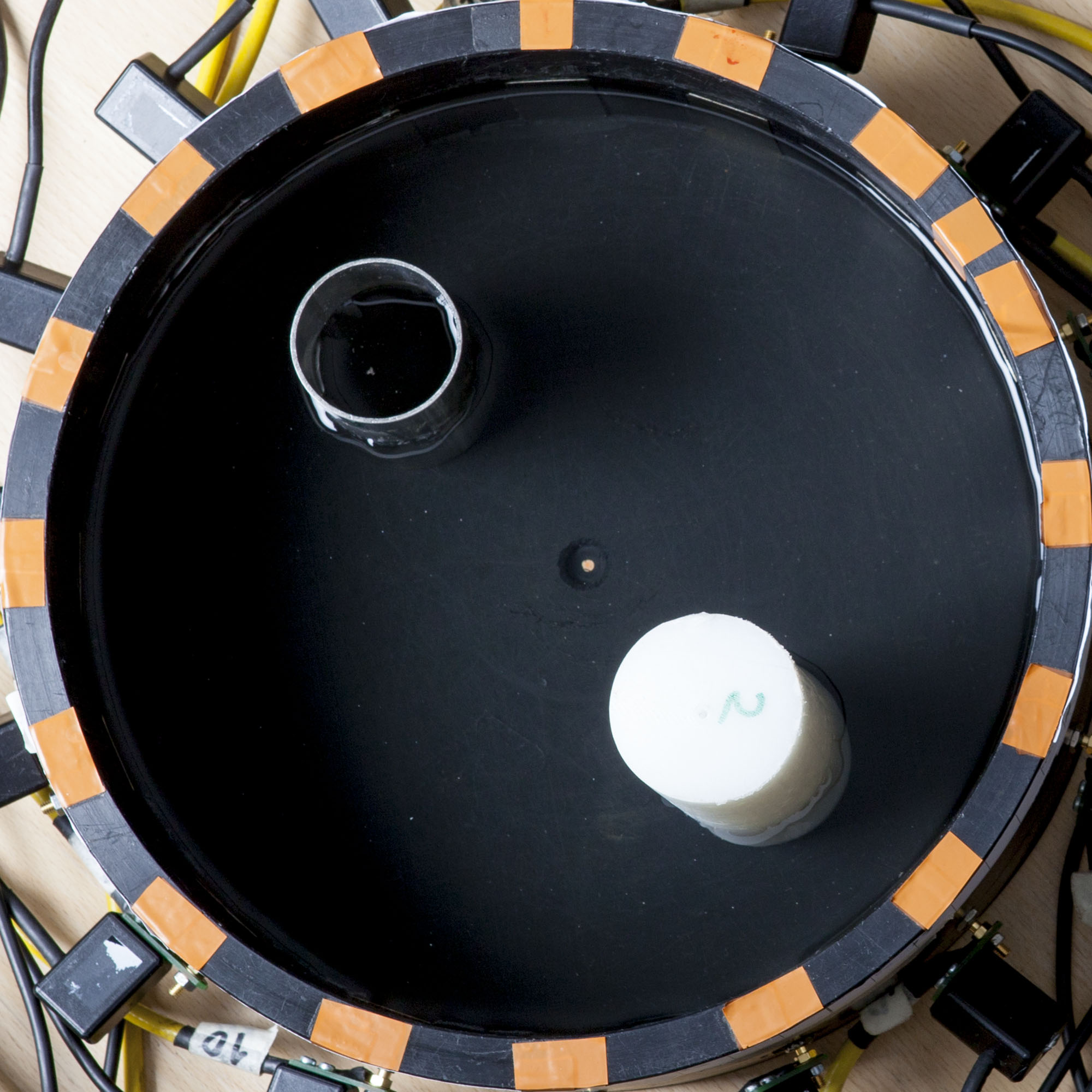} &
\includegraphics[width=0.13\textwidth]{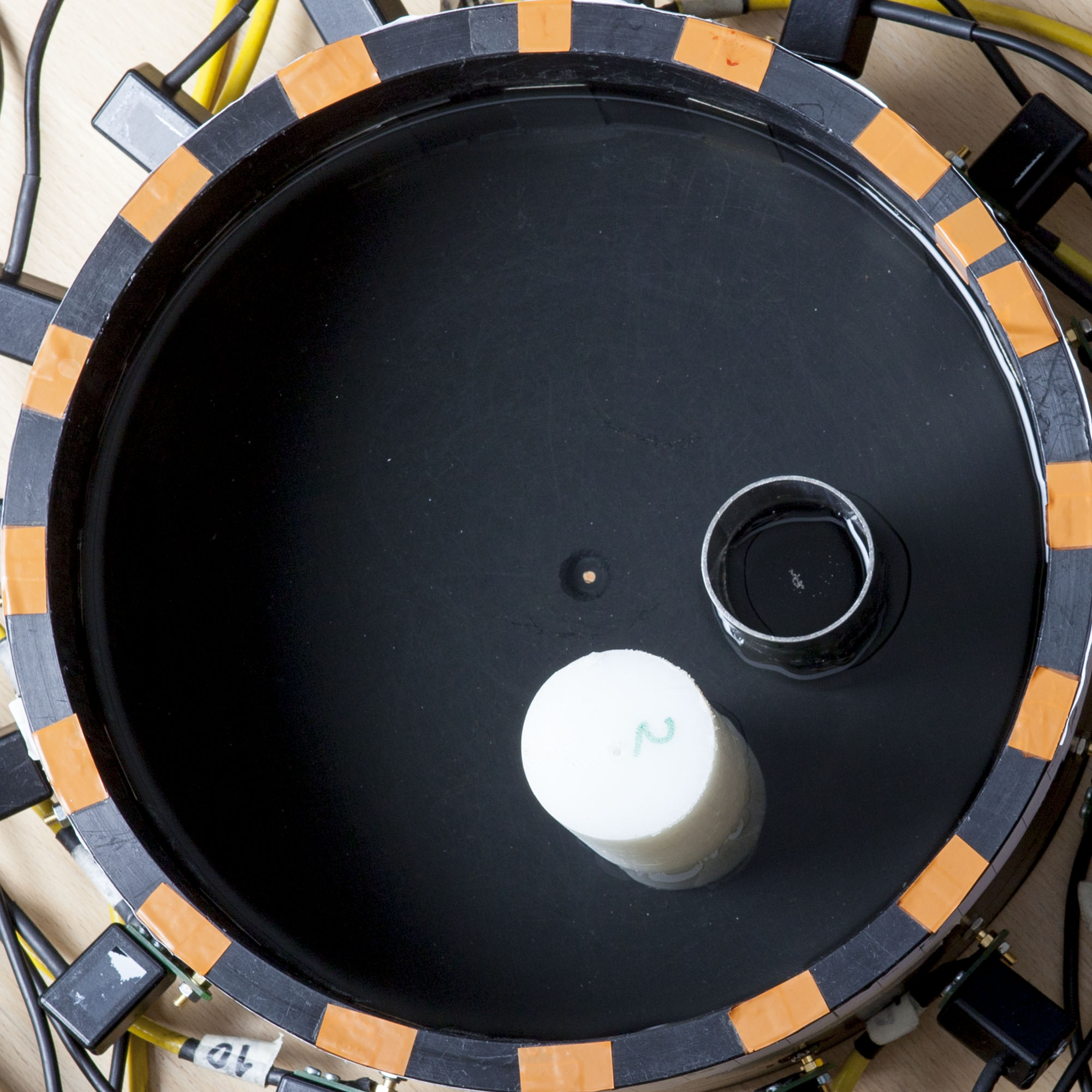} \\[12pt]

{RGN-TV} &
\includegraphics[width=0.13\textwidth]{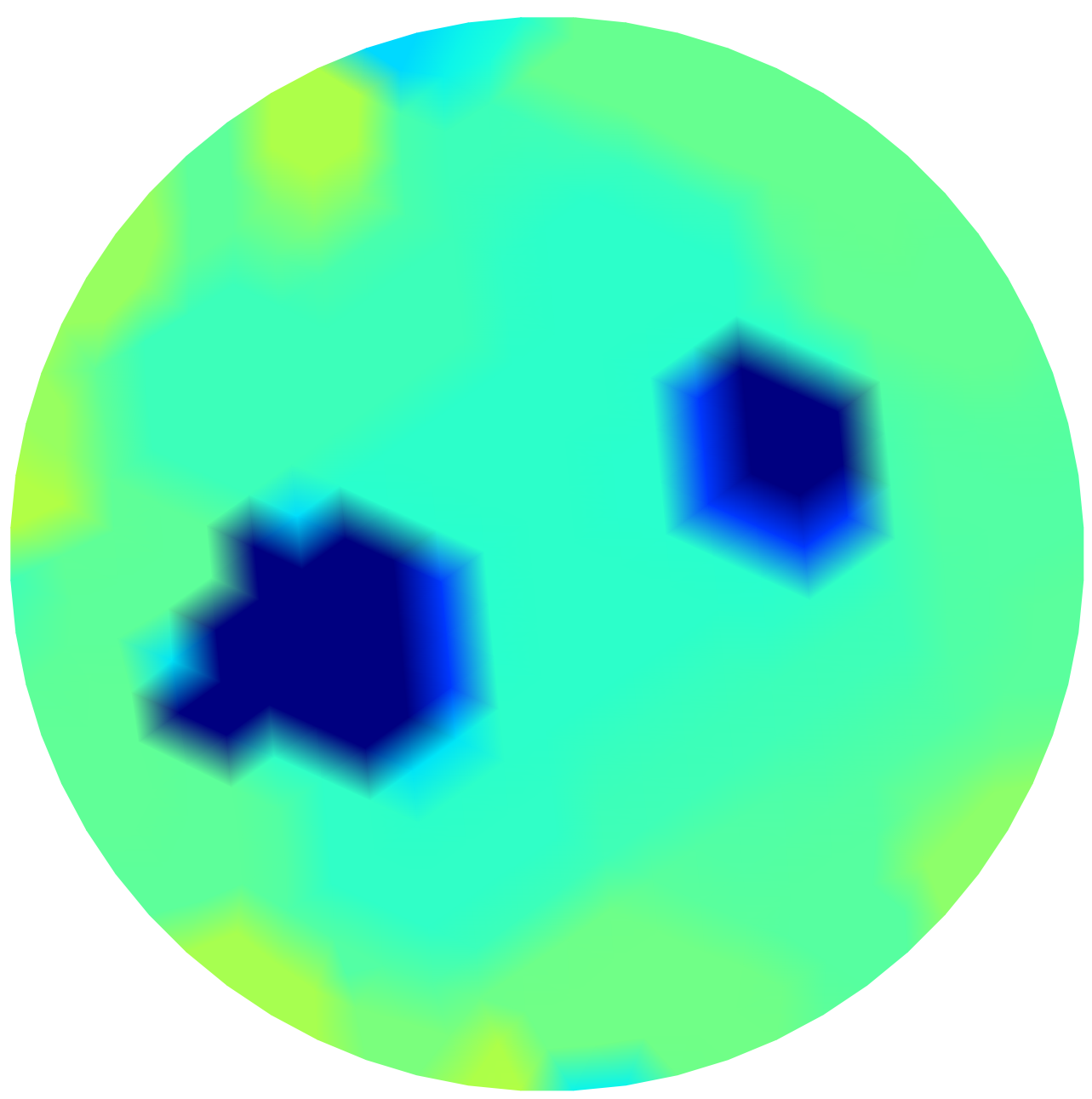} &
\includegraphics[width=0.13\textwidth]{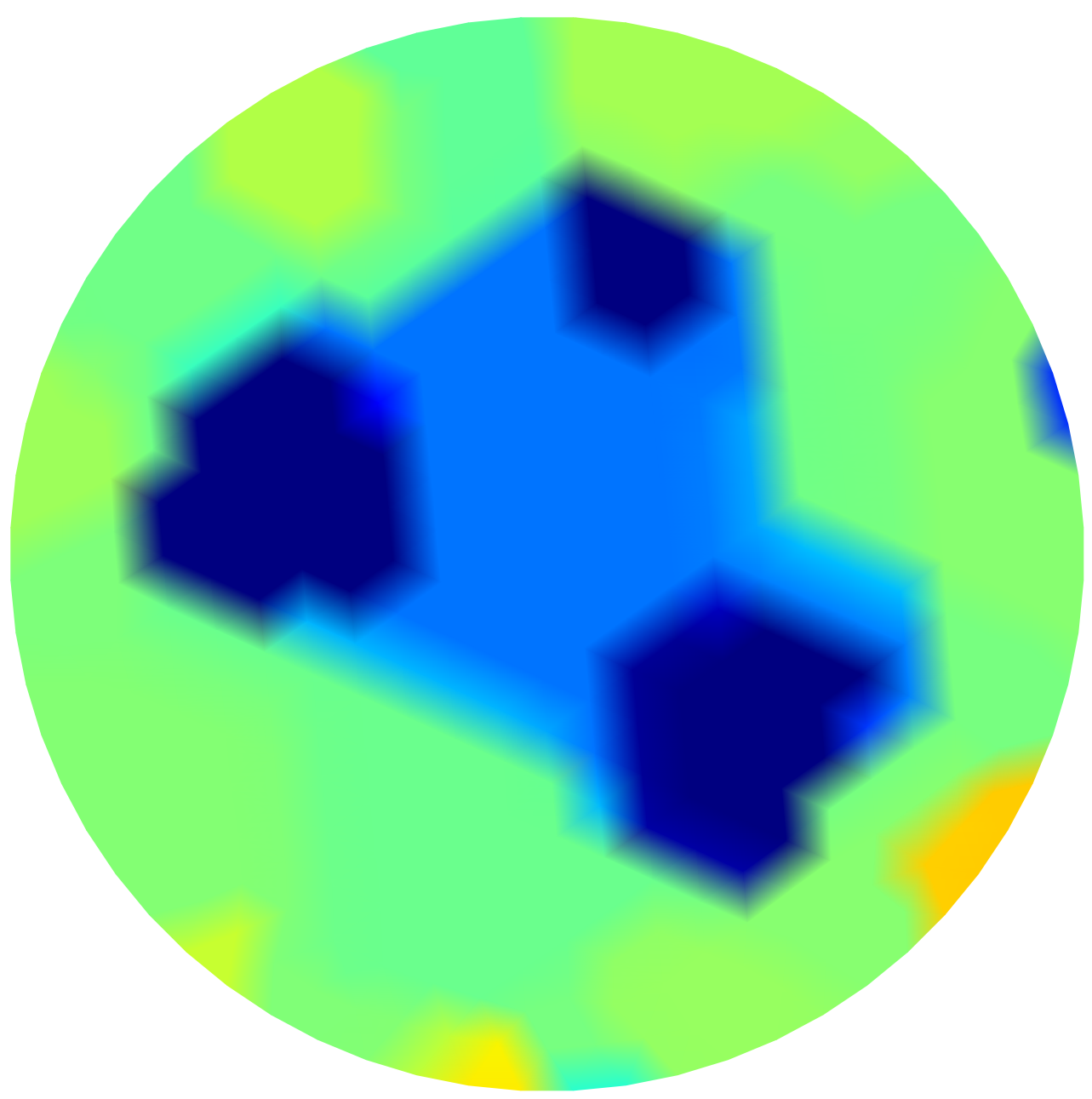} &
\includegraphics[width=0.13\textwidth]{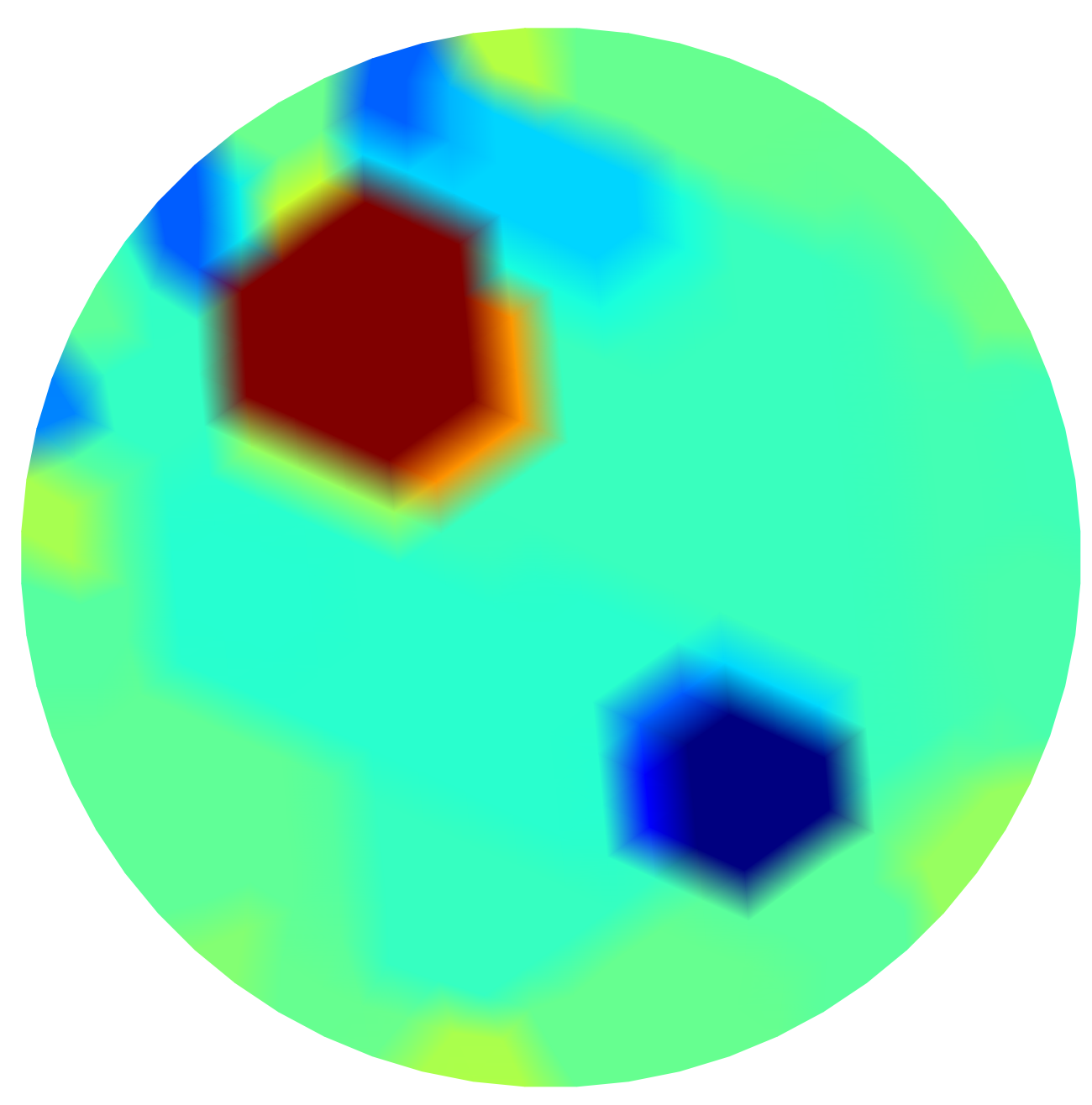} &
\includegraphics[width=0.13\textwidth]{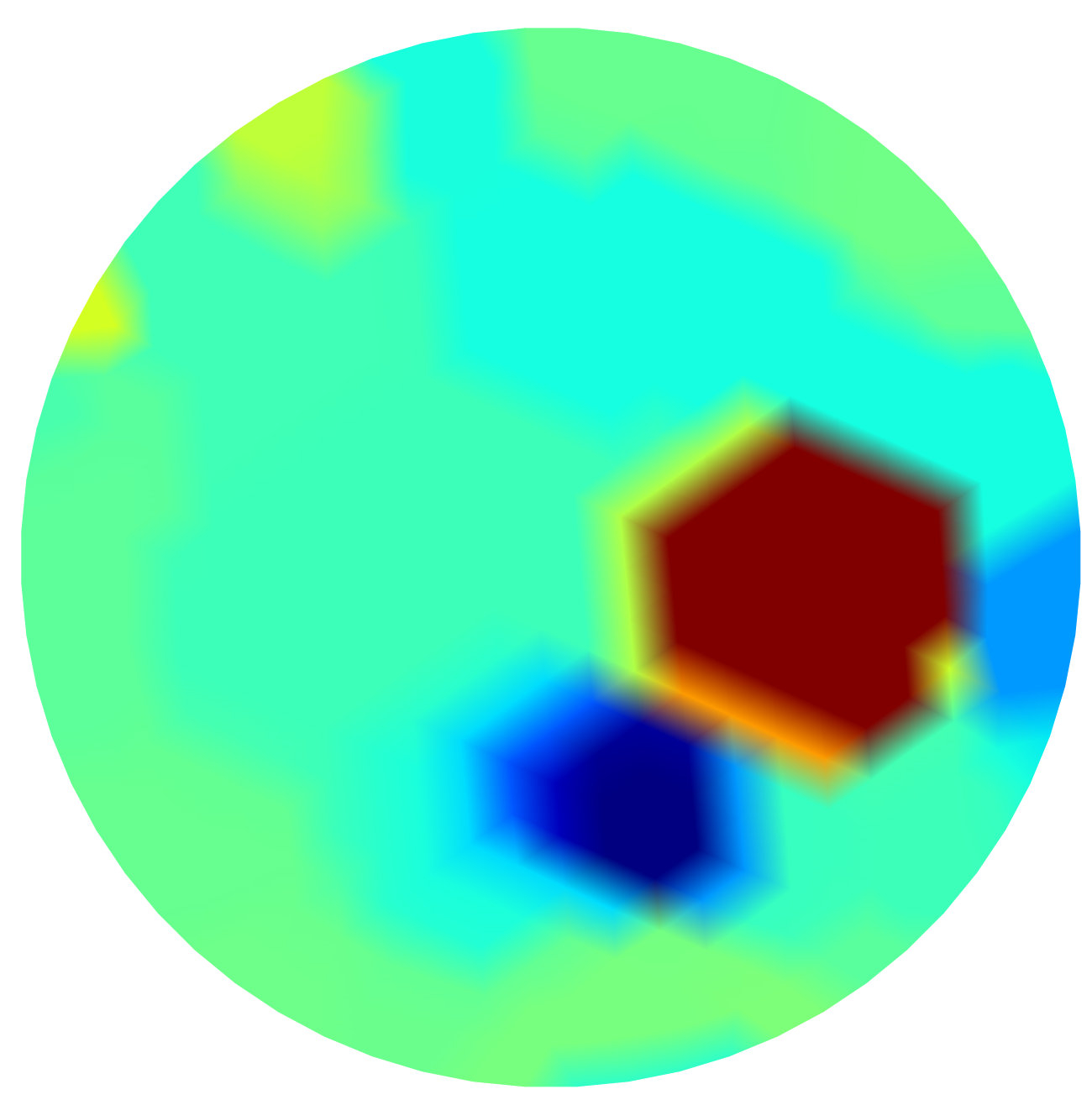} \\[12pt]

{Ours} &
\includegraphics[width=0.13\textwidth]{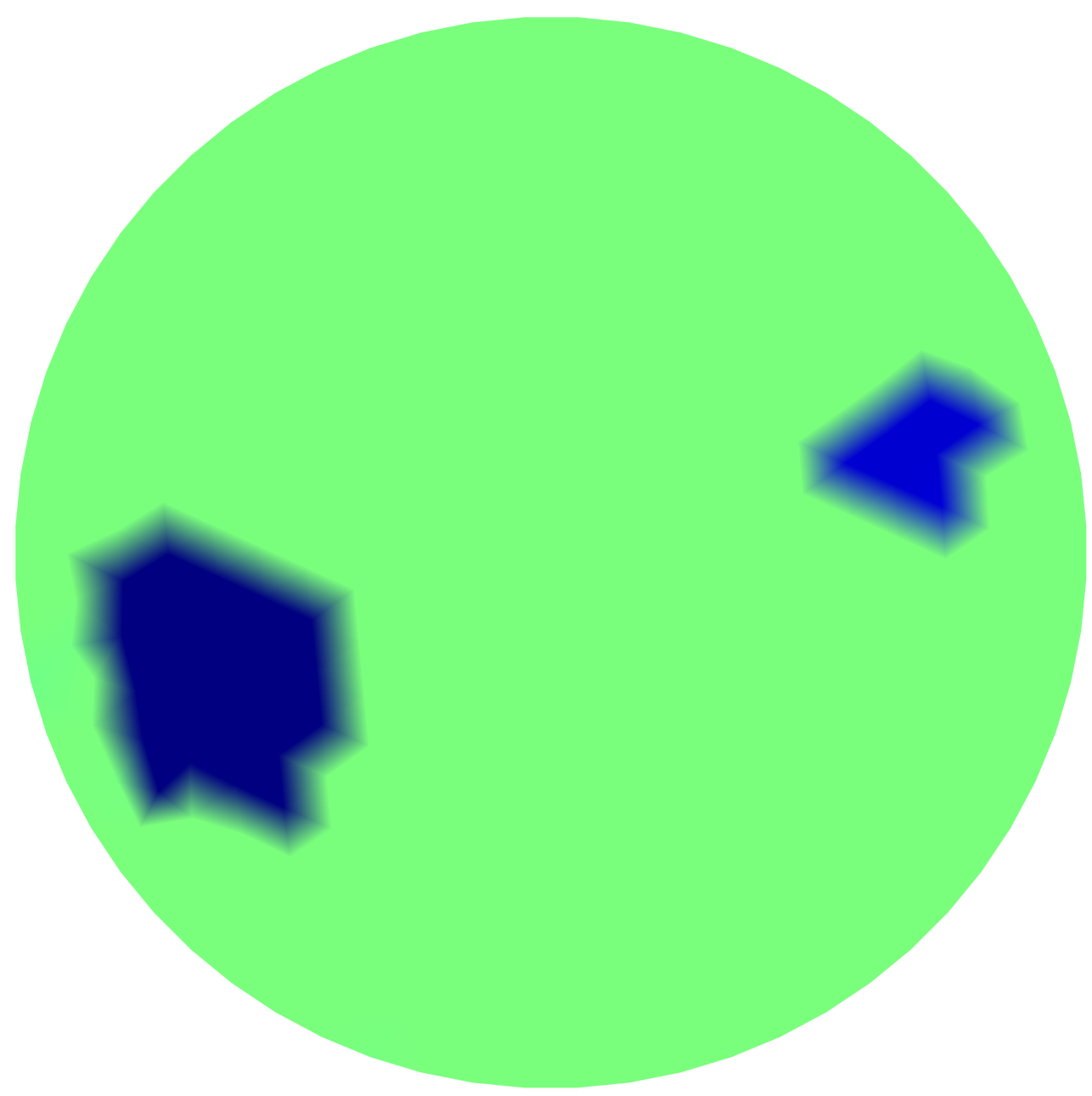} &
\includegraphics[width=0.13\textwidth]{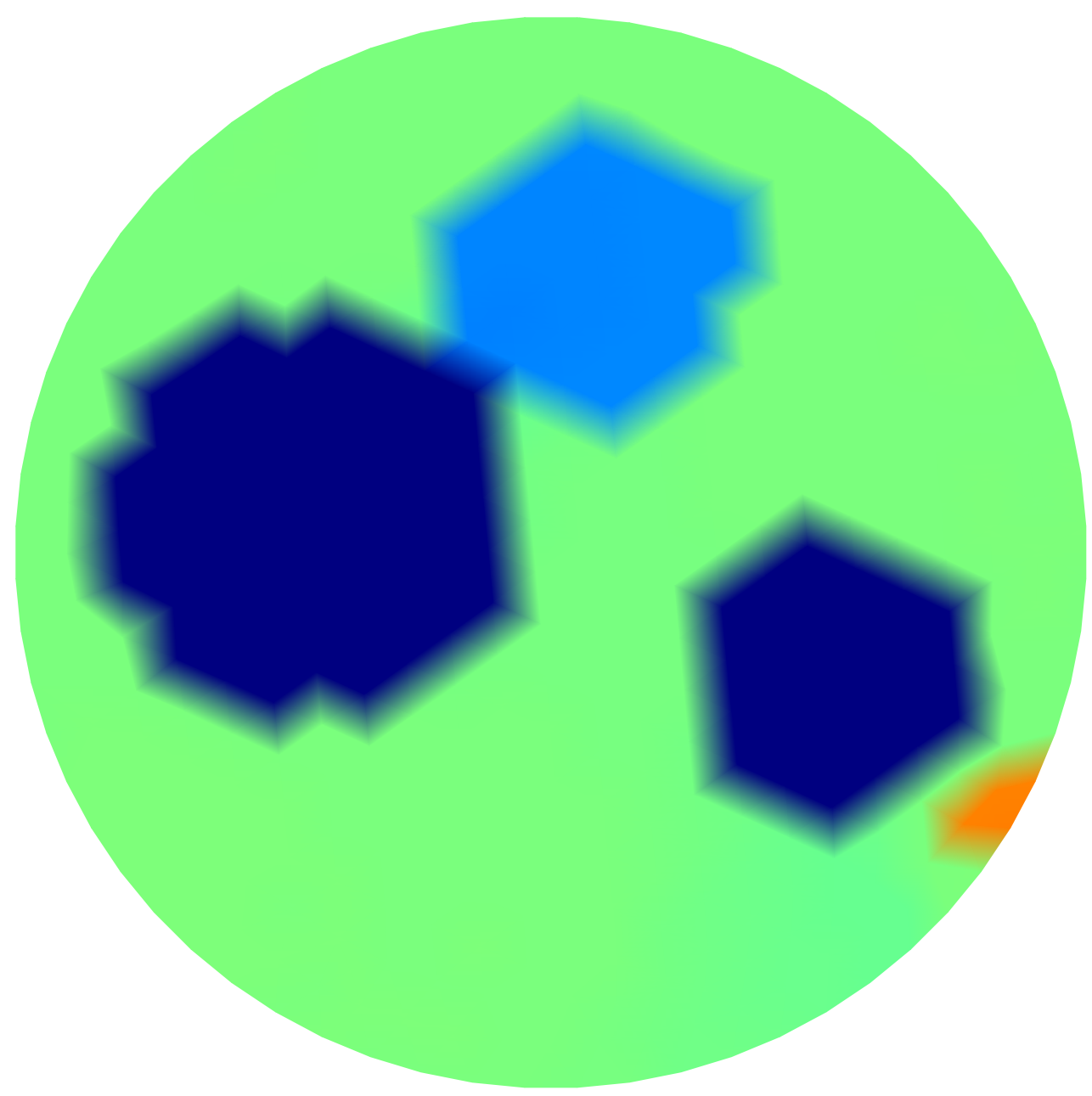} &
\includegraphics[width=0.13\textwidth]{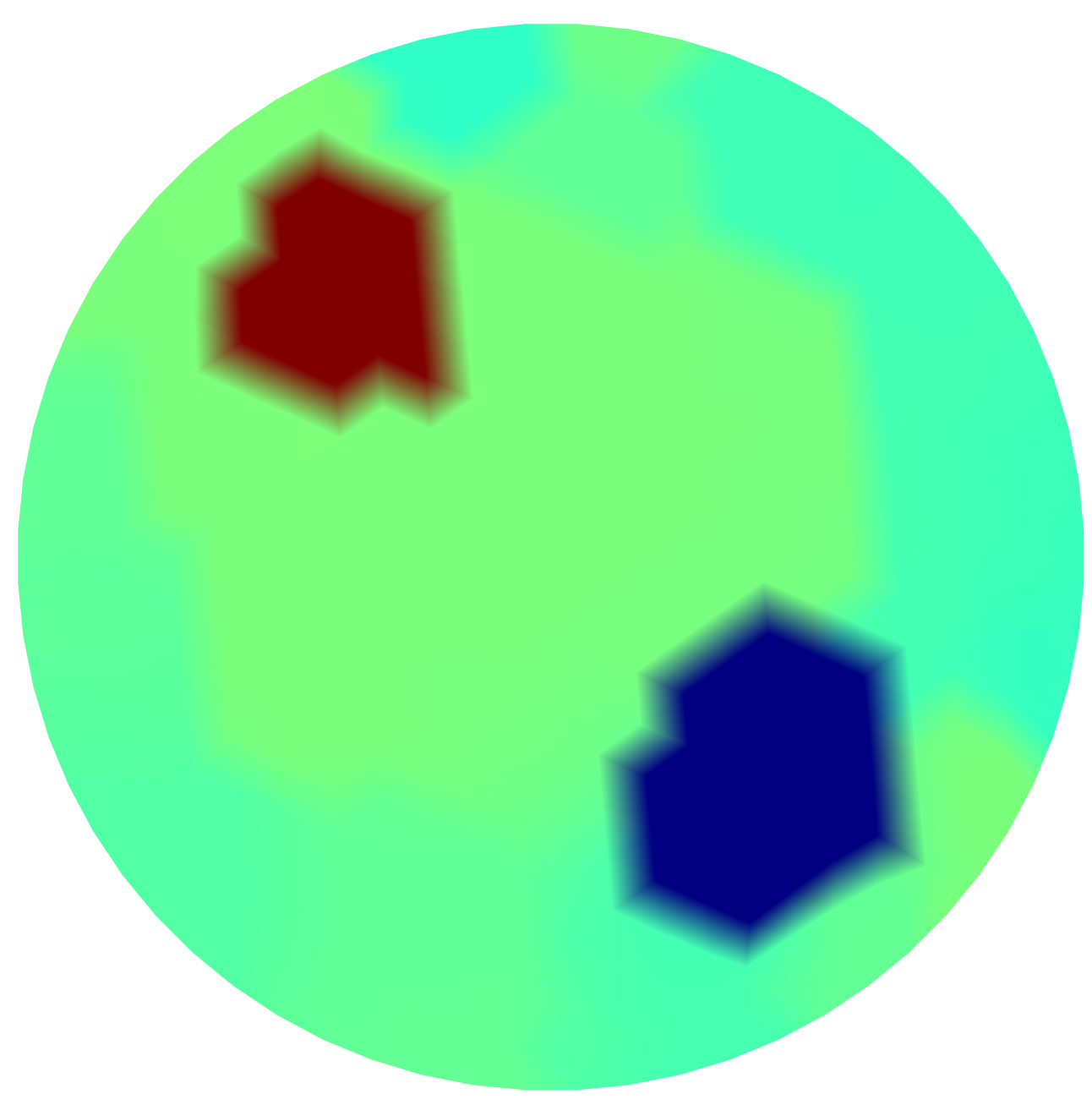} &
\includegraphics[width=0.13\textwidth]{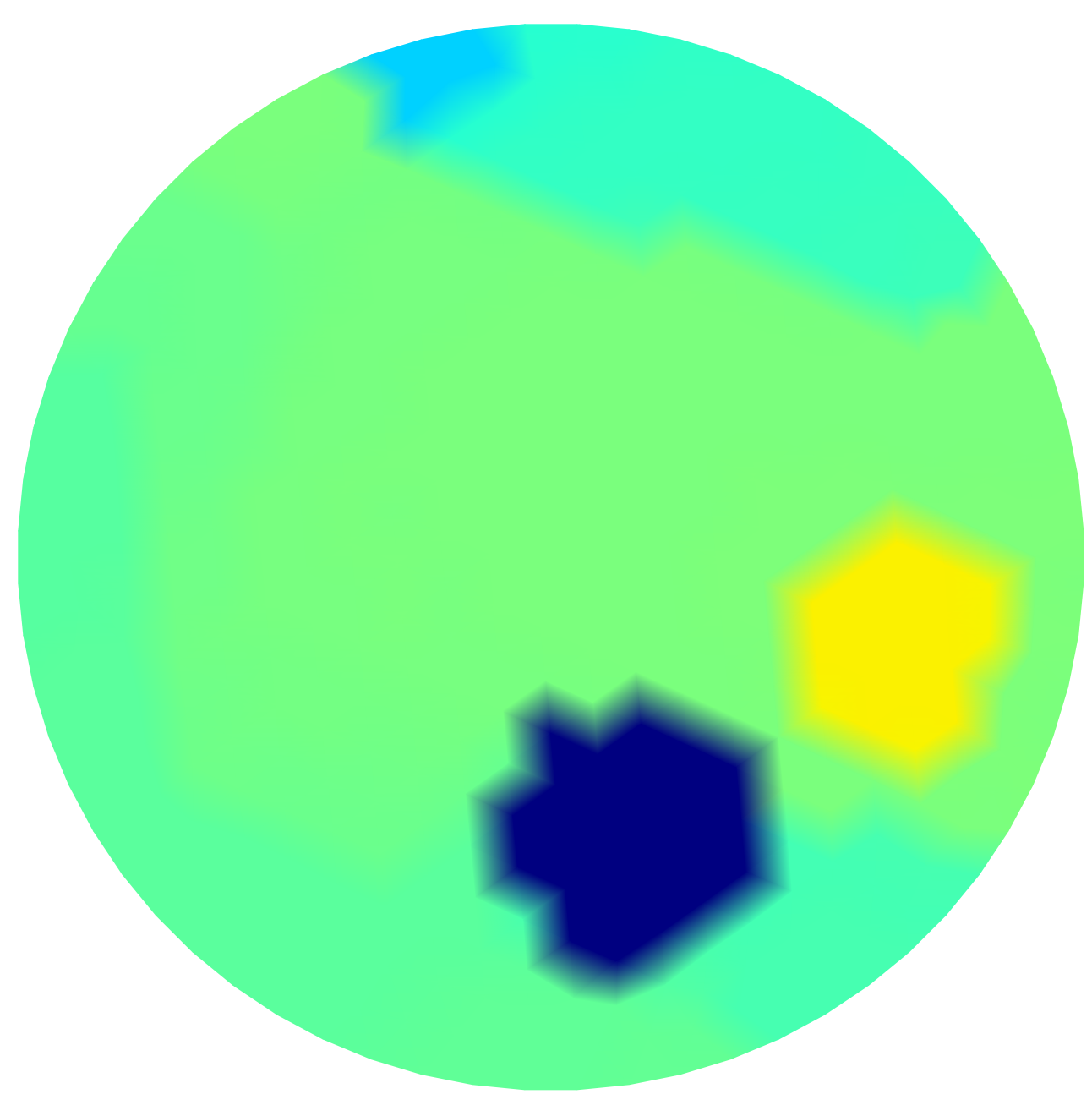} \\

\end{tabular}

\caption{Example 3: Reconstructions on the 2D EIT dataset~\cite{hauptmann2017open}. Each column corresponds to one example case. The first row shows the experimental setup, the second row shows reconstructions obtained with RGN-TV, and the third row shows reconstructions with our method.}
\label{fig:recon_2d_EIT}
\end{figure}

\subsection{Example 4: Comparative Benchmarking}\label{sub:ex4}
We compare our method against four approaches representing distinct paradigms in Bayesian inverse problems and classical regularization. Specifically, we evaluate the regularized Gauss-Newton (RGN-TV) \cite{RGN2022,borsic2009vivo}, the randomize-then-optimize Metropolis-Hastings (RTO-MH) algorithm \cite{bardsley2015randomize}, the Generative Plug-and-Play (GPnP-BM3D) method \cite{bouman2023generative}, and the Deep Posterior Sampling (DP-SGS) framework \cite{ling2025split}.
The codes for RTO-MH, GPnP-BM3D, and DP-SGS were downloaded from GitHub (\url{https://github.com/Ling991028/DP-SGS}).
Since our EIT data are originally defined on a triangular mesh, whereas these methods require regular image-grid inputs (e.g., 128×128 images), we converted the data from the triangular mesh representation to a regular image grid for all comparisons. The RGN-TV code was provided by the authors. To ensure optimal performance, all methods were rigorously trained and tuned on the datasets described in Section \ref{sec:dataset}, using grid interpolation where necessary.
Reconstruction results for the conductivity distributions are presented in Fig. \ref{tab:comp}. We report the Relative Error (RelErr) and the Structural Similarity Index (SSIM) below each case. While the former assesses the accuracy of the inclusion intensity, the latter is employed to evaluate how well the inclusion's shape and structure are preserved.

Regardless of inclusion morphology, ranging from circles and triangles to non-convex blobs, our results appear more sharp, consistently preserving both the geometric shape and the target conductivity values. In contrast, the most closely related method DP-SGS produces significantly noisier reconstructions. Meanwhile, RGN-TV oversimplifies the inclusion geometry, and RTO-MH tends to blur the edges; overall, all baseline methods struggle to accurately recover triangular shapes.

\begin{figure}
\begin{table}[H]
\centering
\small

\begin{adjustbox}{width=\linewidth}
\begin{tabular}{C{2.8cm} C{2.8cm} C{2.8cm} C{2.8cm} C{2.8cm} C{2.8cm}}

{GT} & & & & &\\
\imgonly{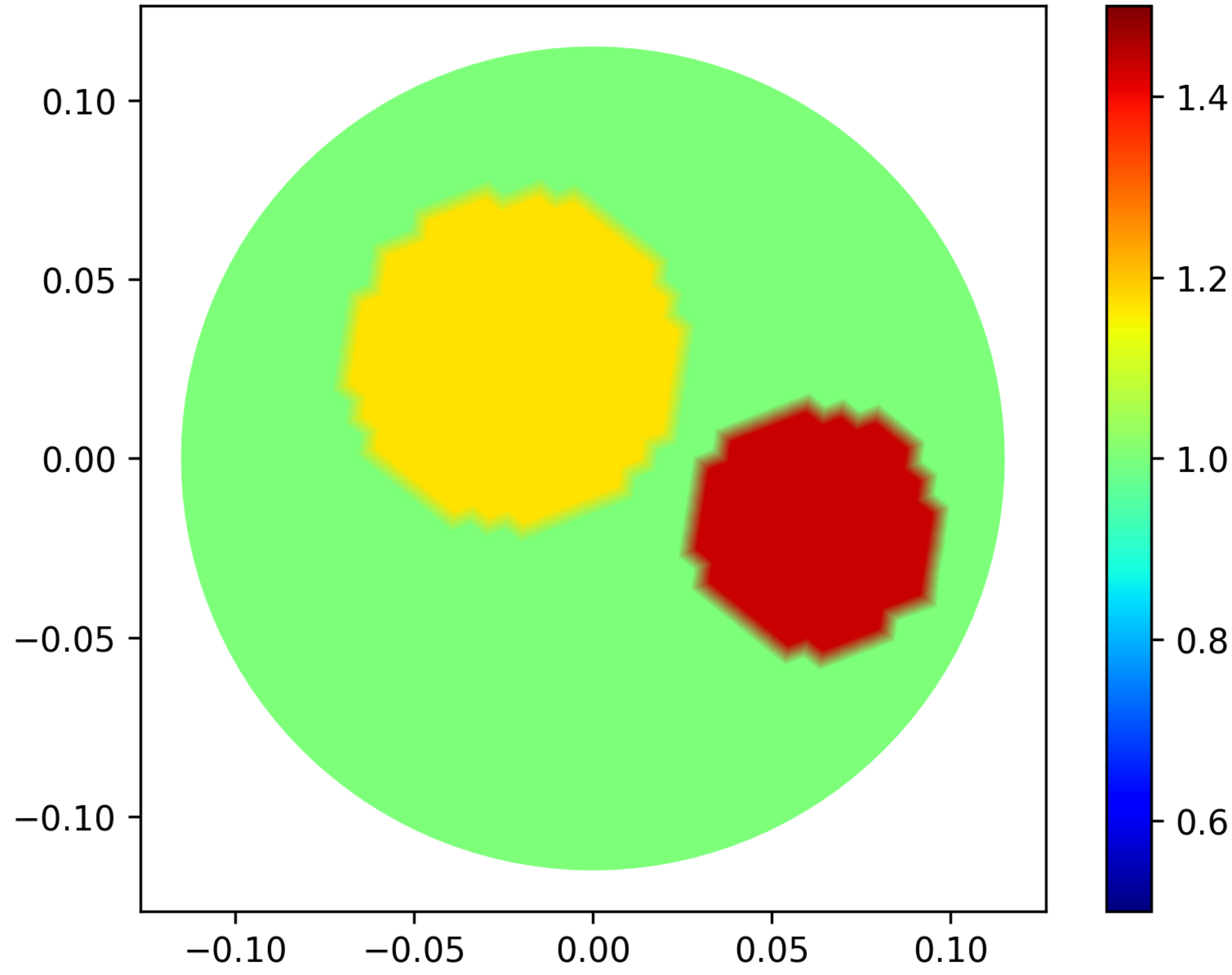} &
\imgonly{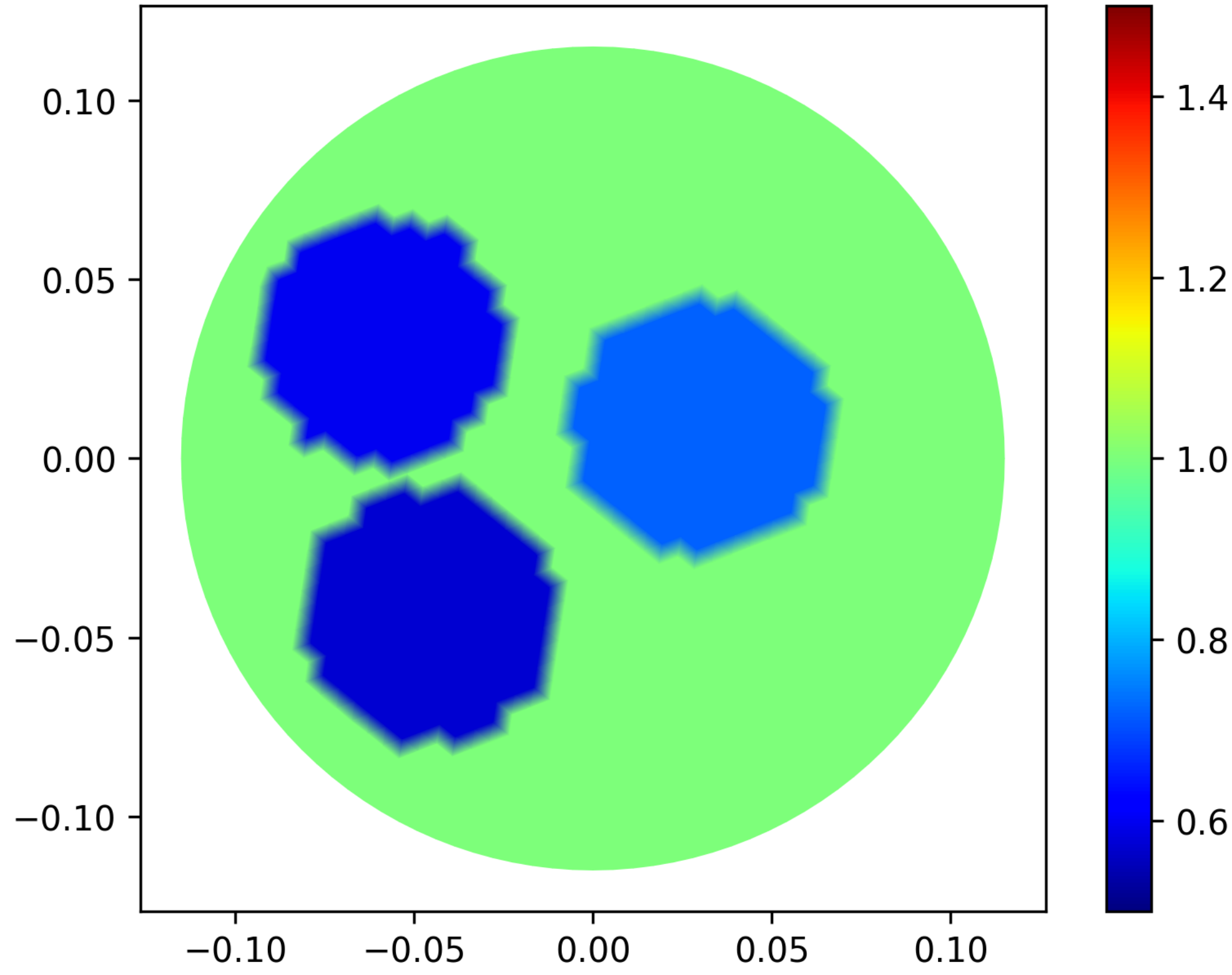} &
\imgonly{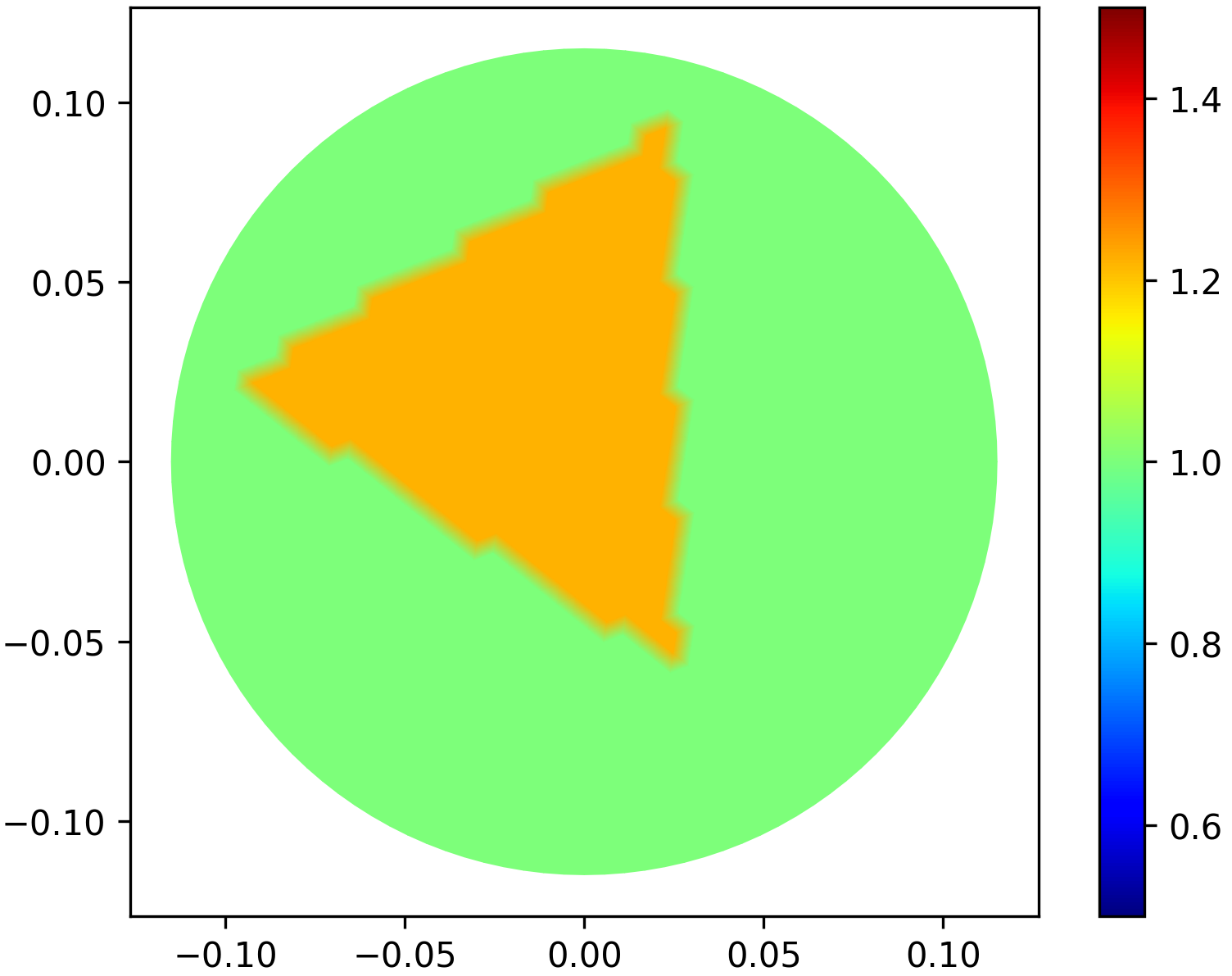} &
\imgonly{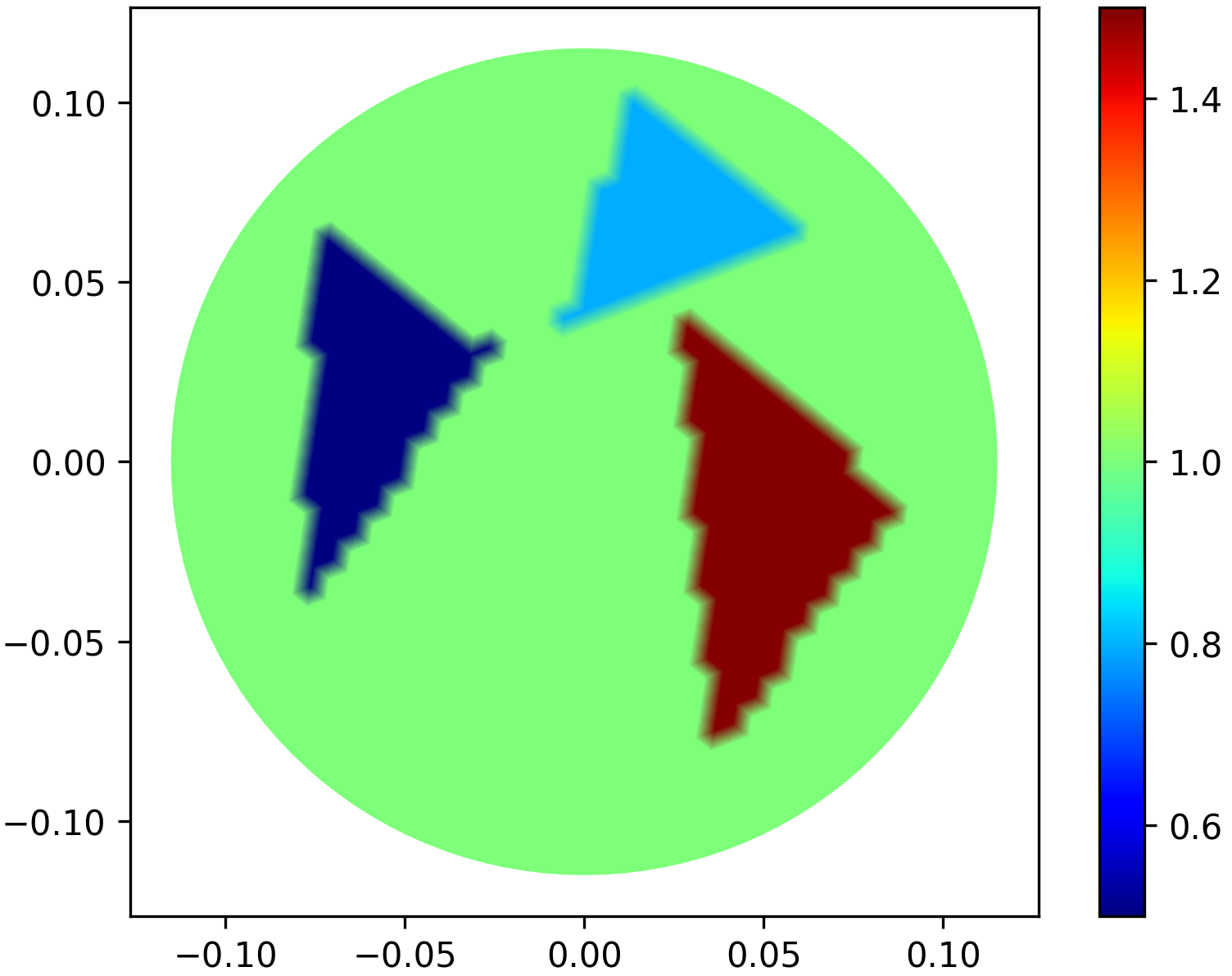} &
\imgonly{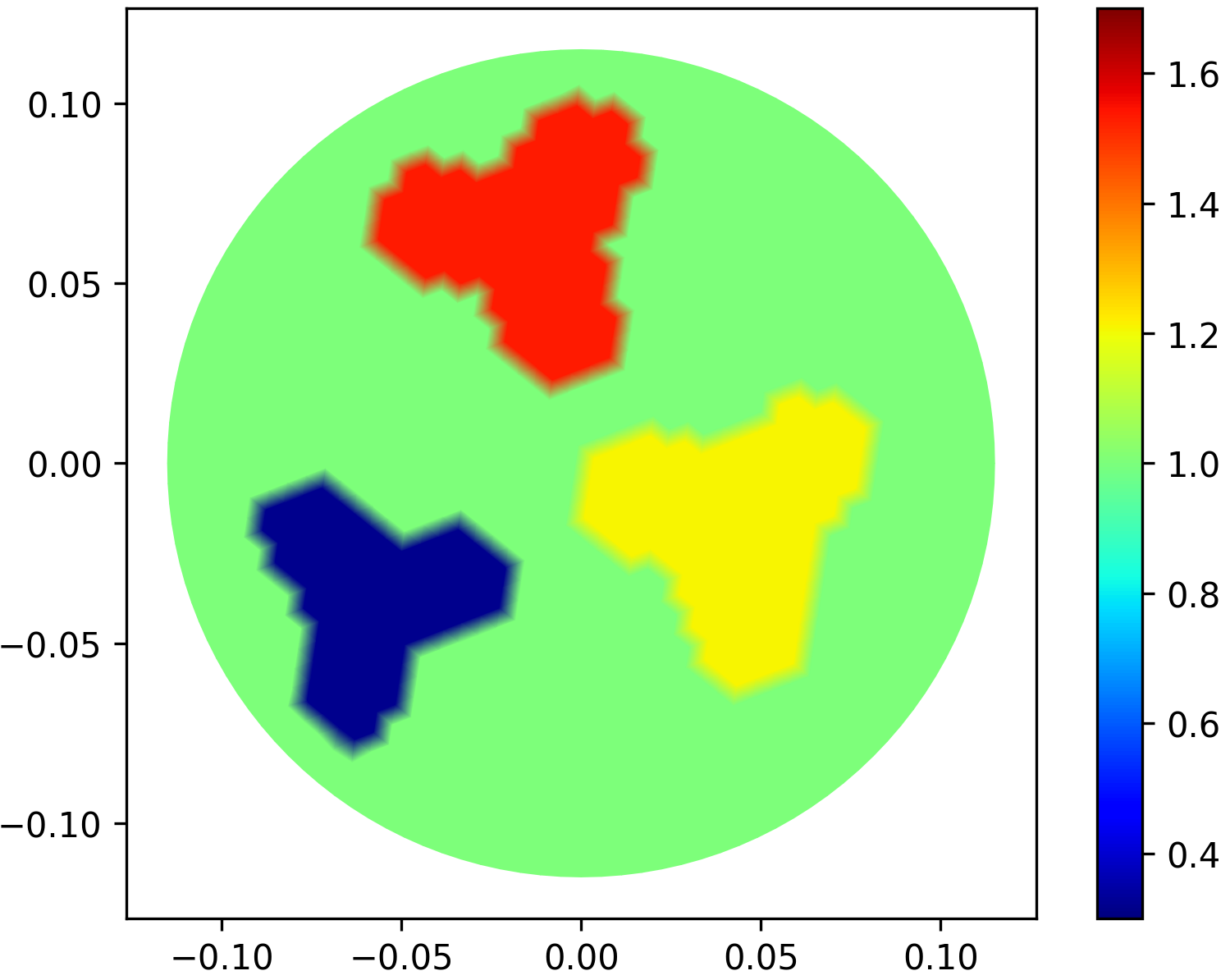} &
\imgonly{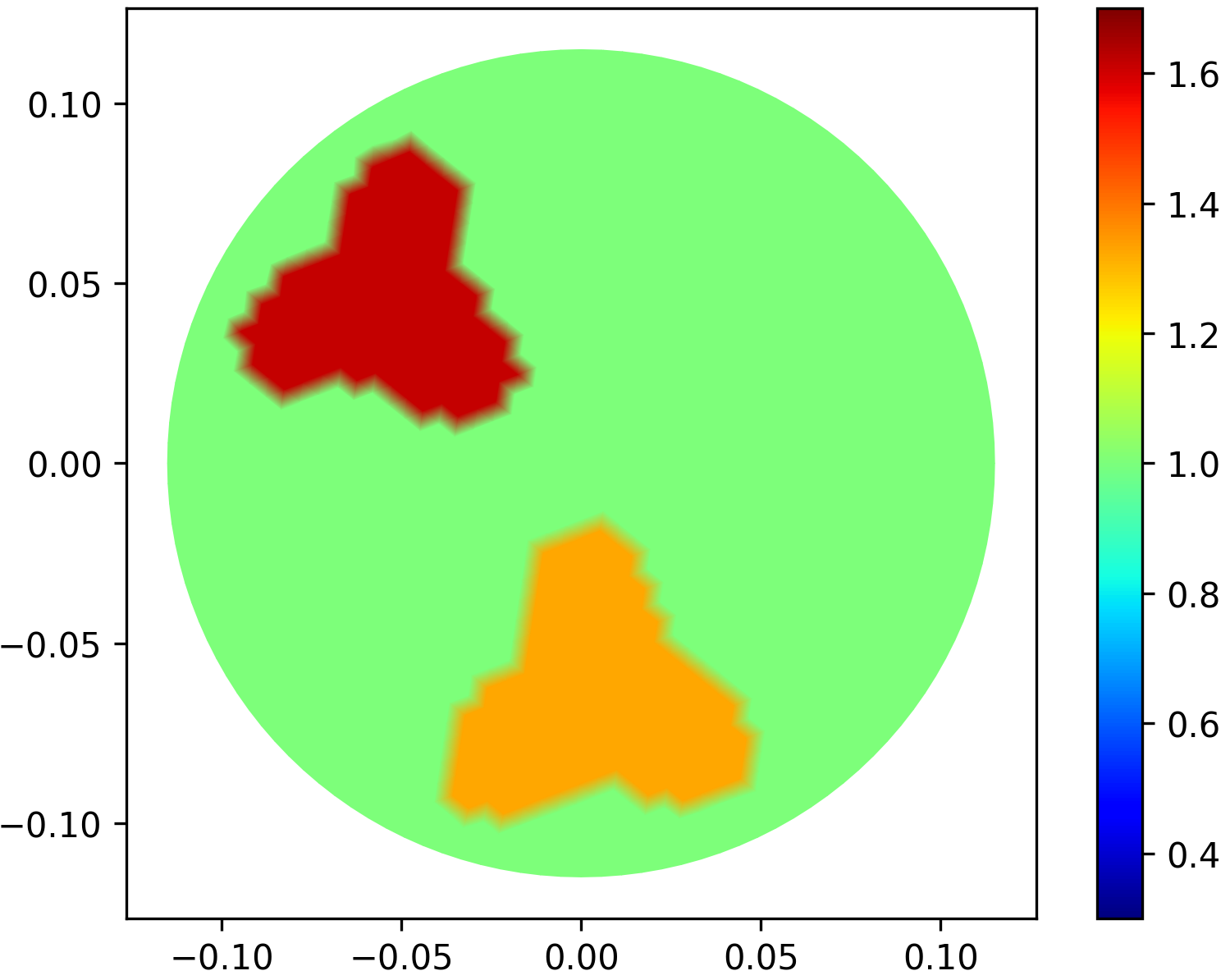} \\[10mm]

\hline

RTO-MH~\cite{bardsley2015randomize} & & & & &\\
\imgmetric{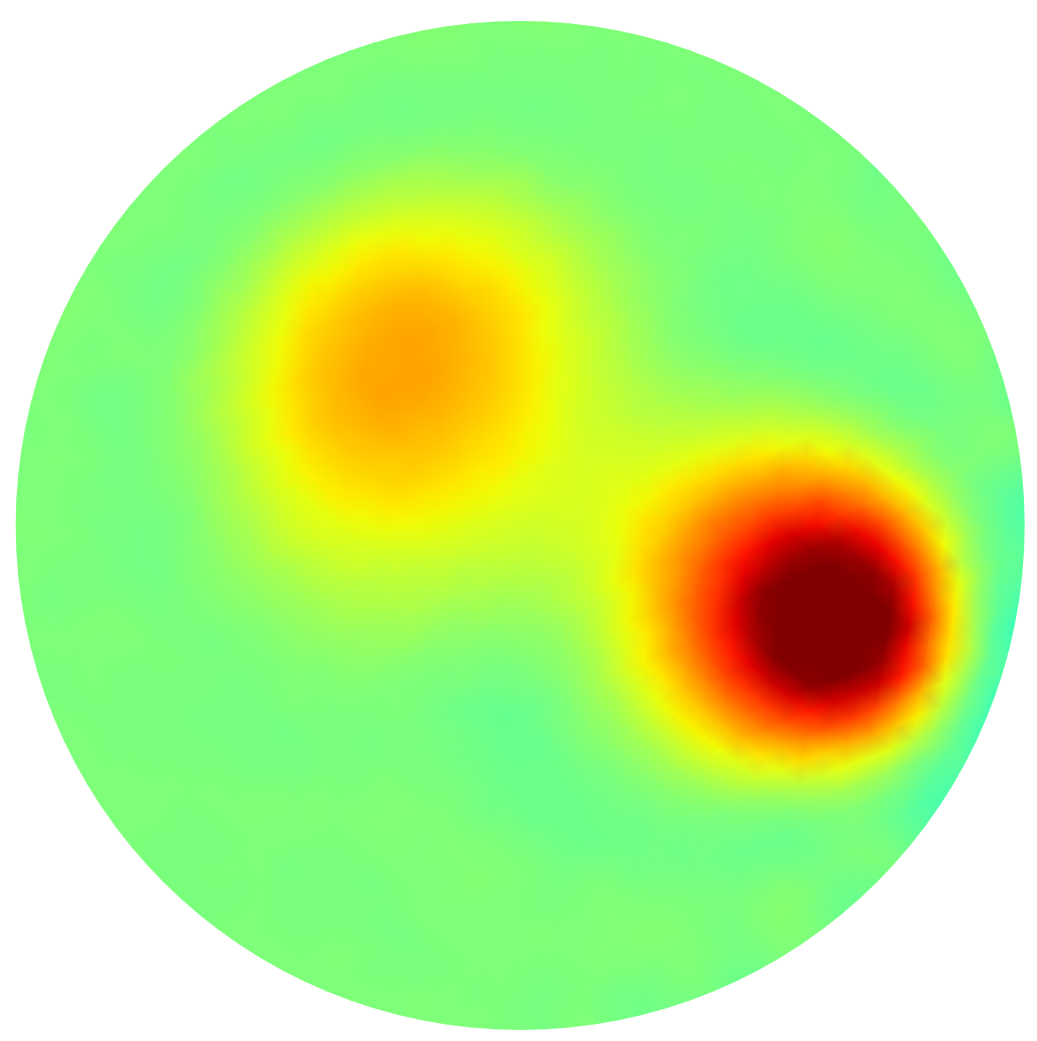}{(0.057, 0.7633)} &
\imgmetric{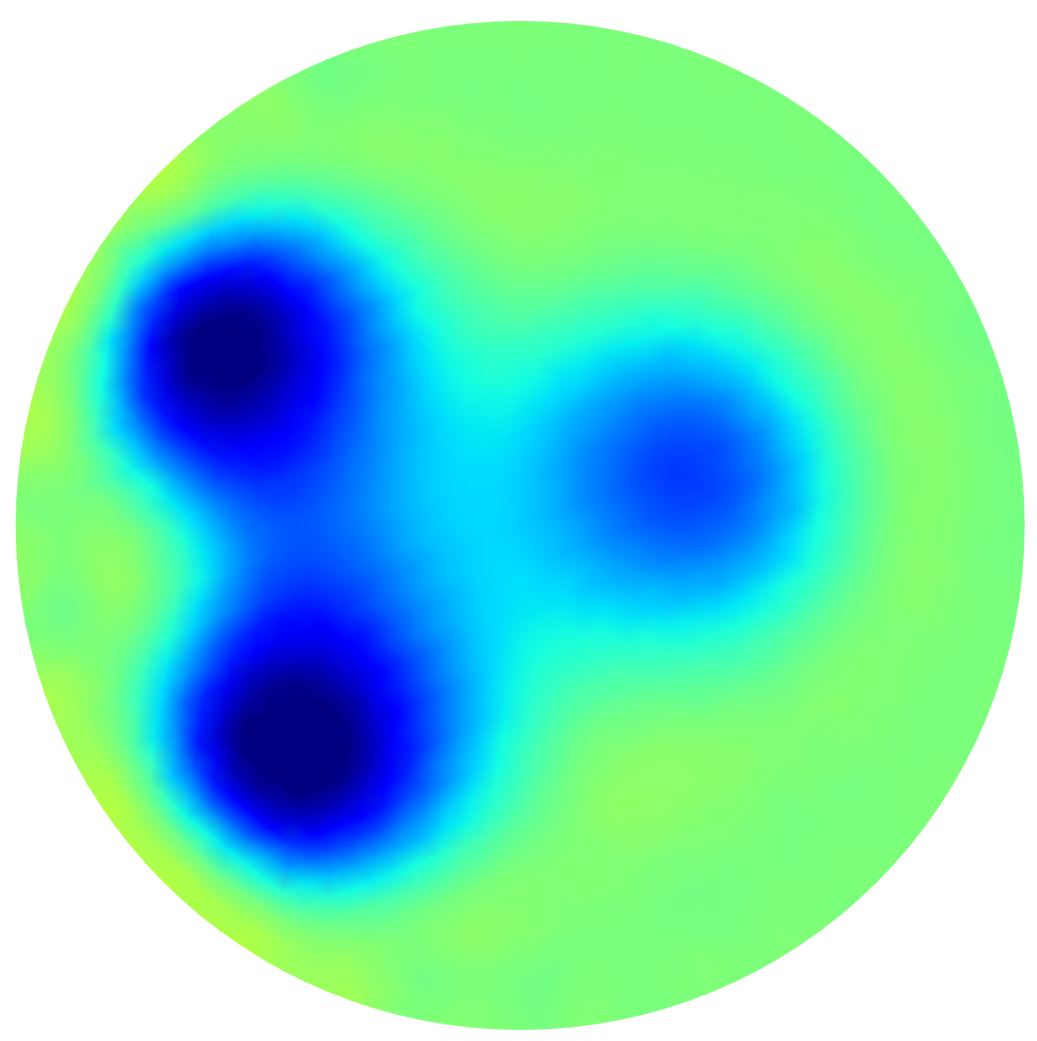}{(0.081, 0.6573)} &
\imgmetric{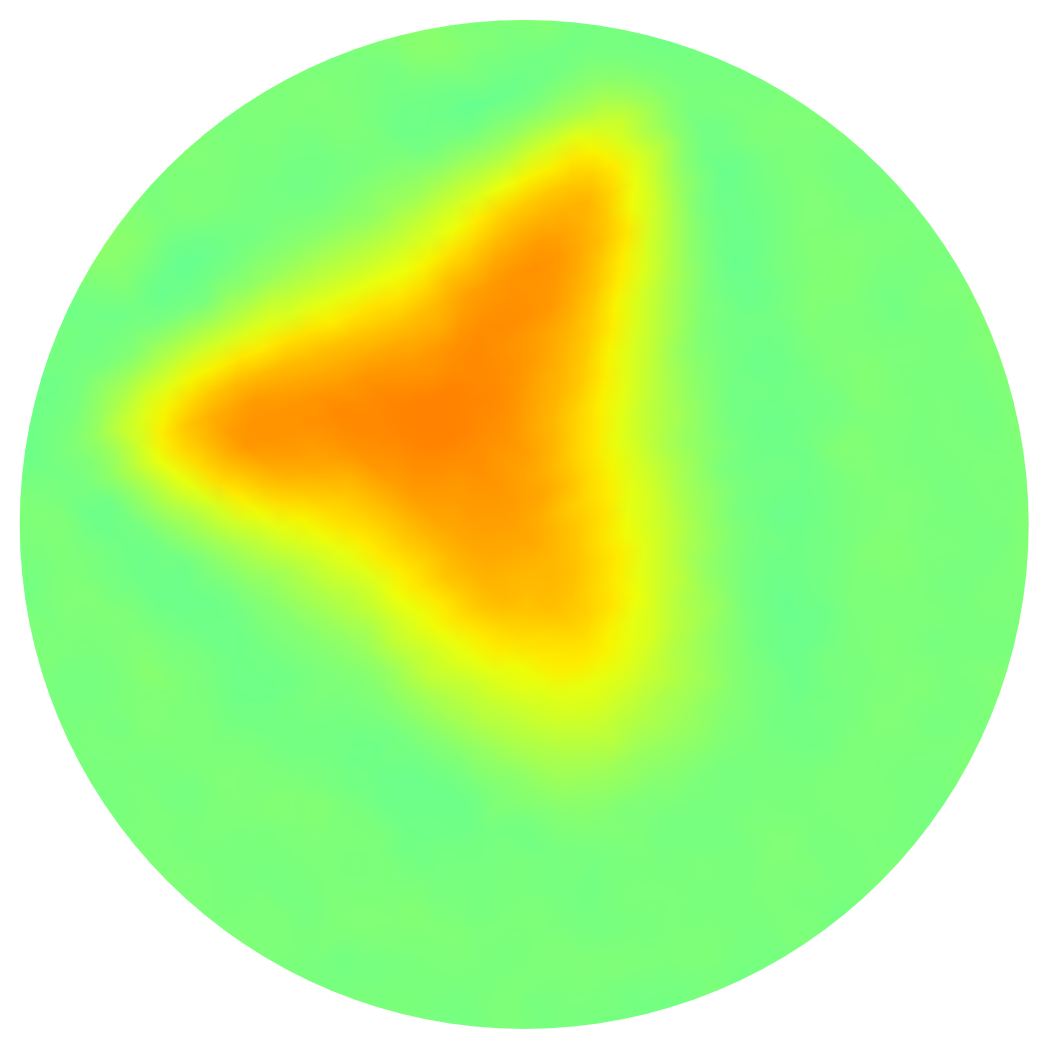}{(\textbf{0.037}, 0.7191)} &
\imgmetric{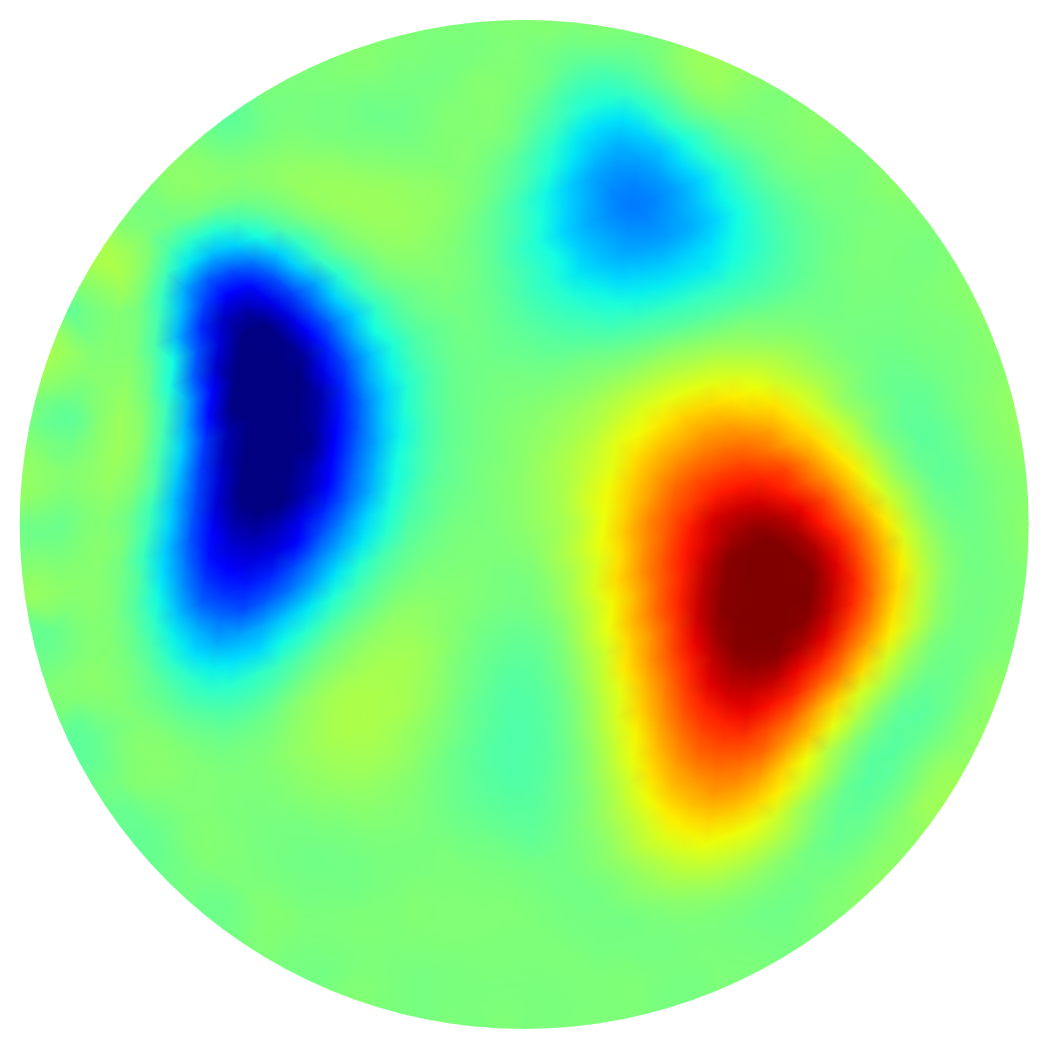}{(\textbf{0.093}, 0.7019)} & \imgmetric{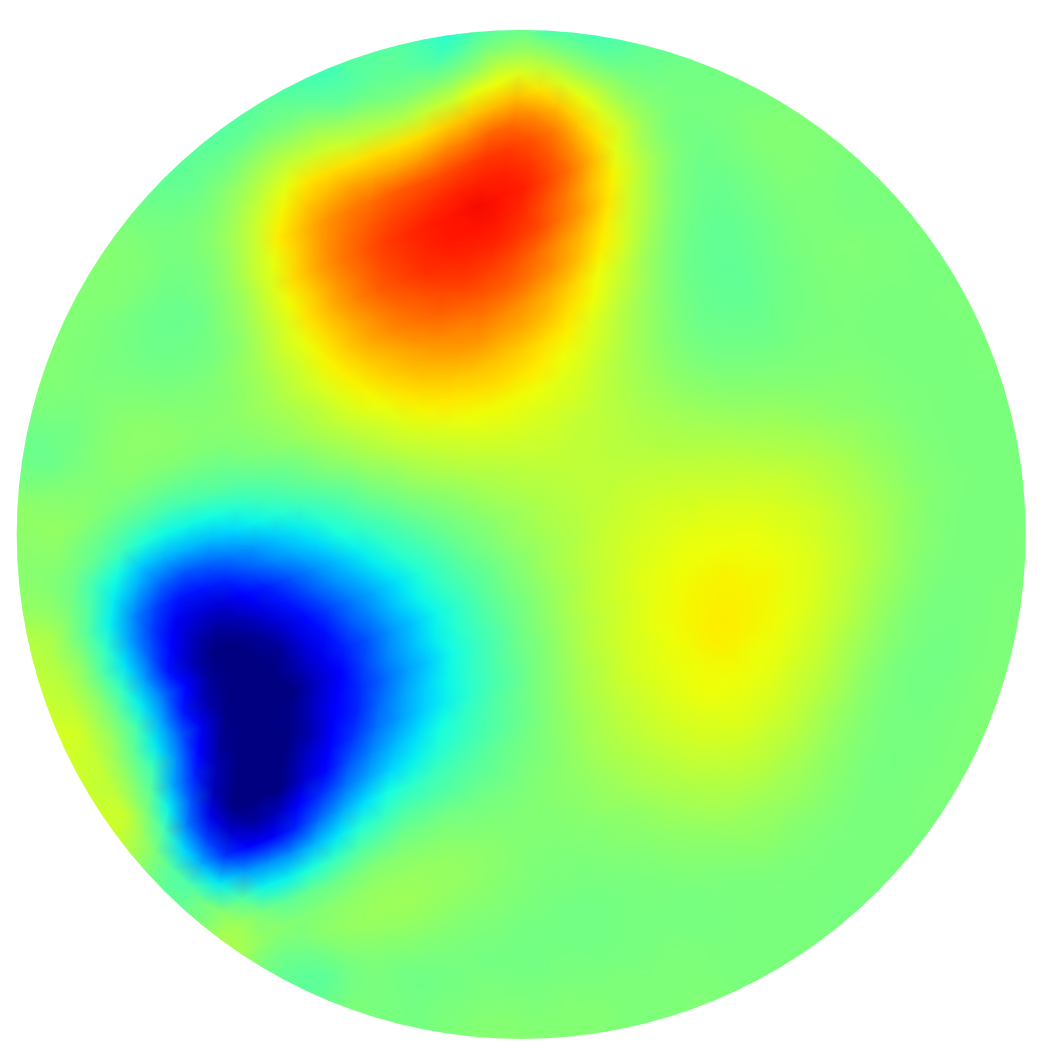}{(0.120, 0.6653)} &
\imgmetric{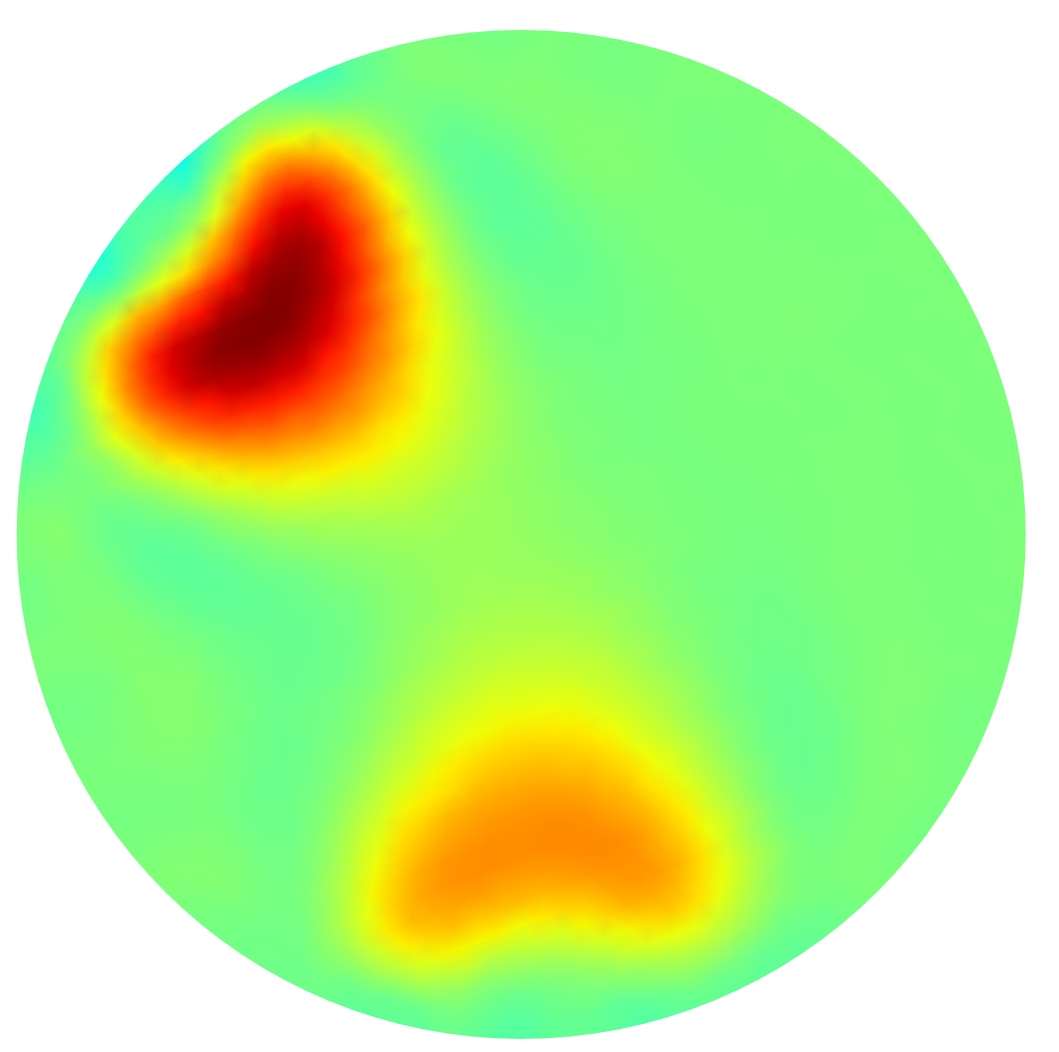}{(0.091, 0.6814)} \\[20mm]

GPnP-BM3D~\cite{bouman2023generative} & & & & &\\
\imgmetric{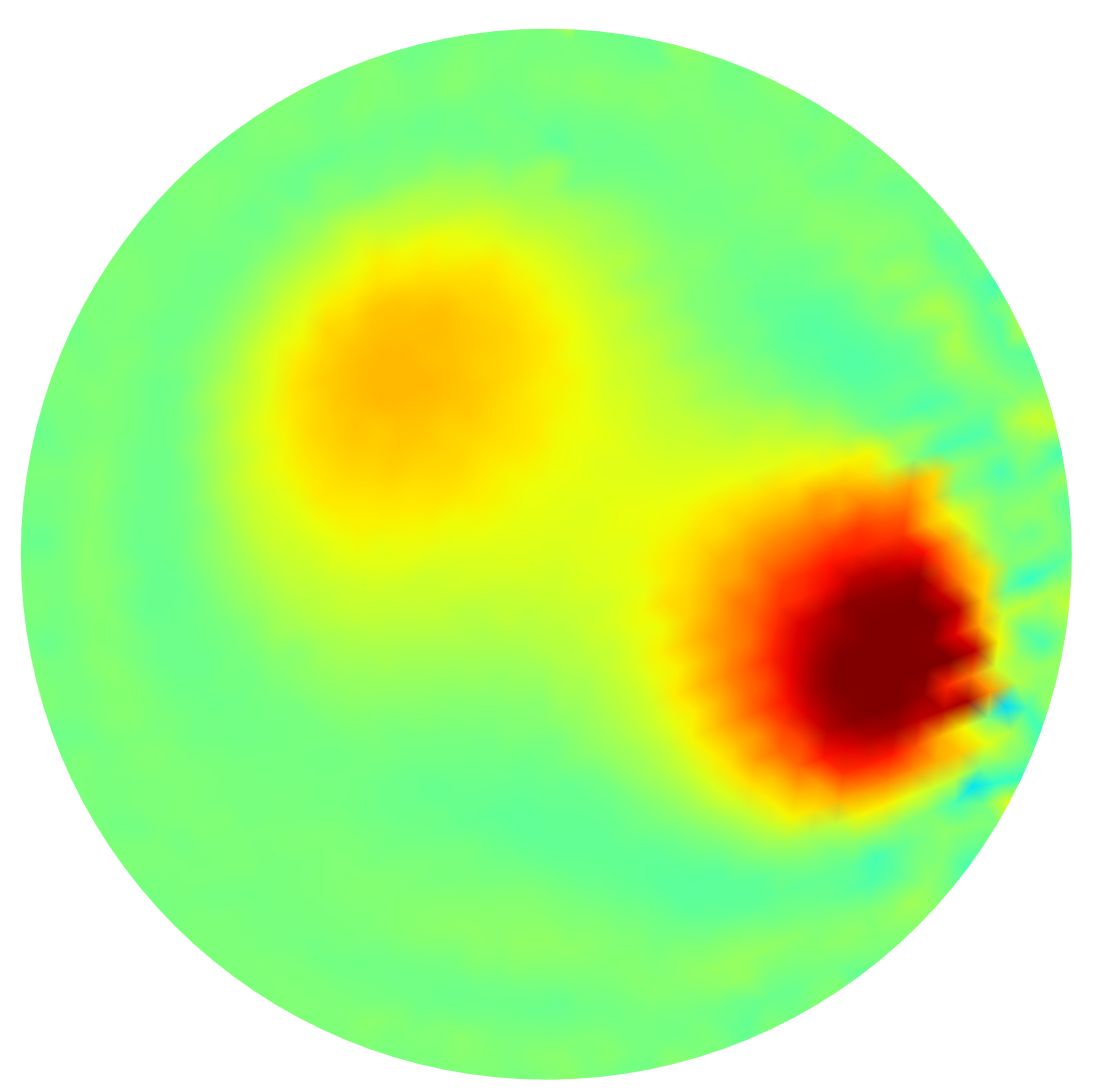}{(0.062, 0.6879)} &
\imgmetric{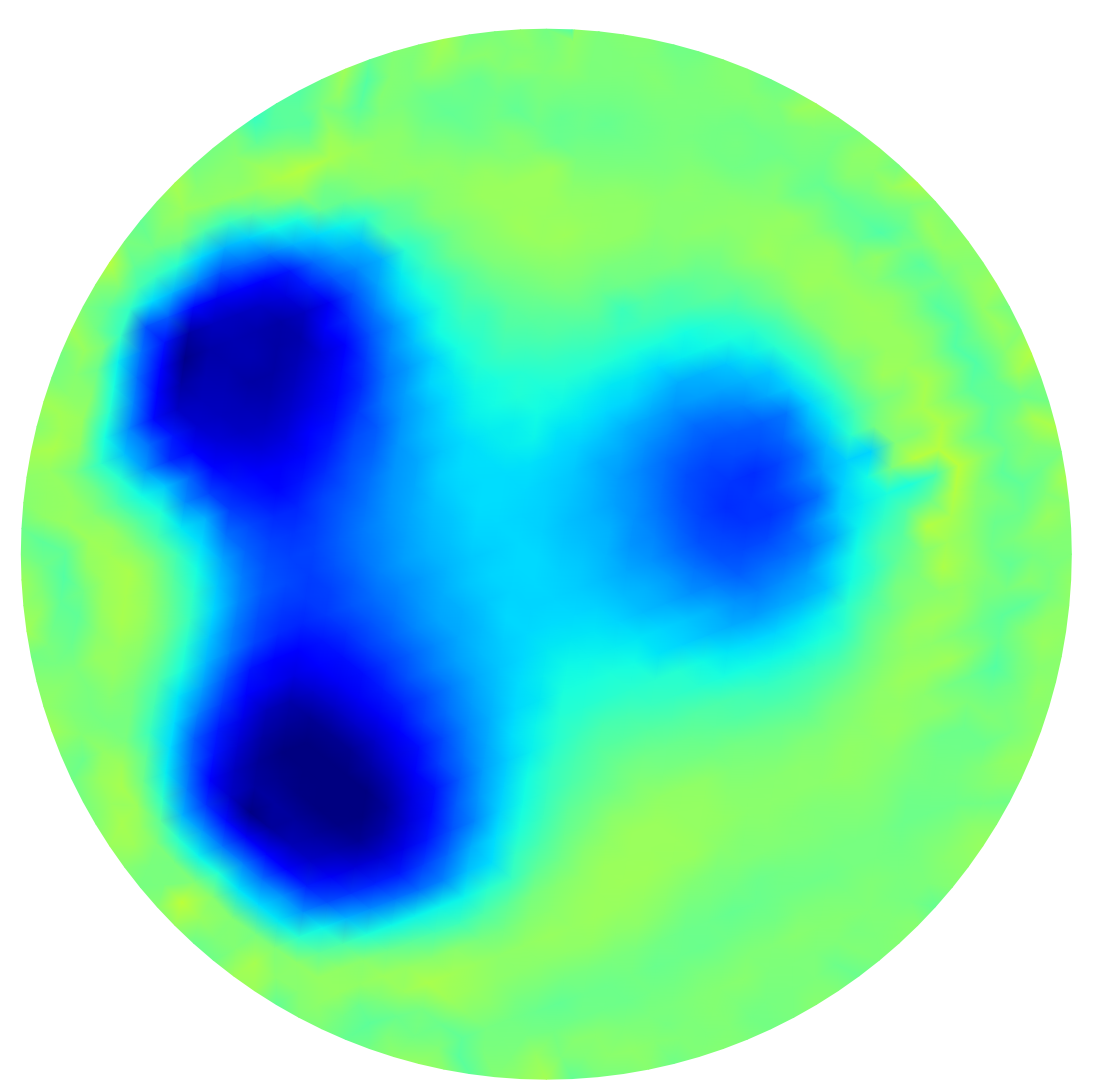}{(0.082, 0.6005)} &
\imgmetric{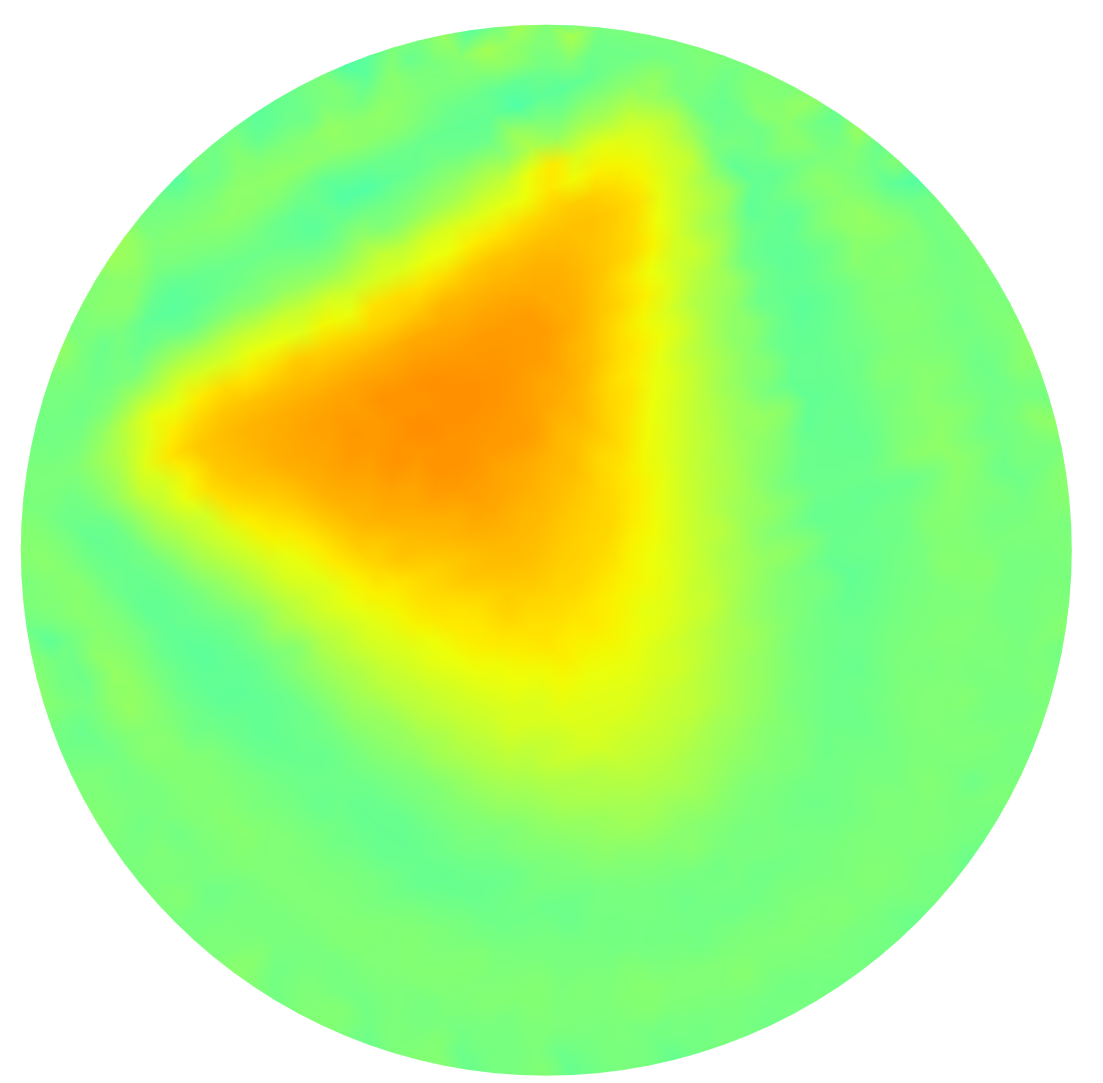}{(0.041, 0.5935)} &
\imgmetric{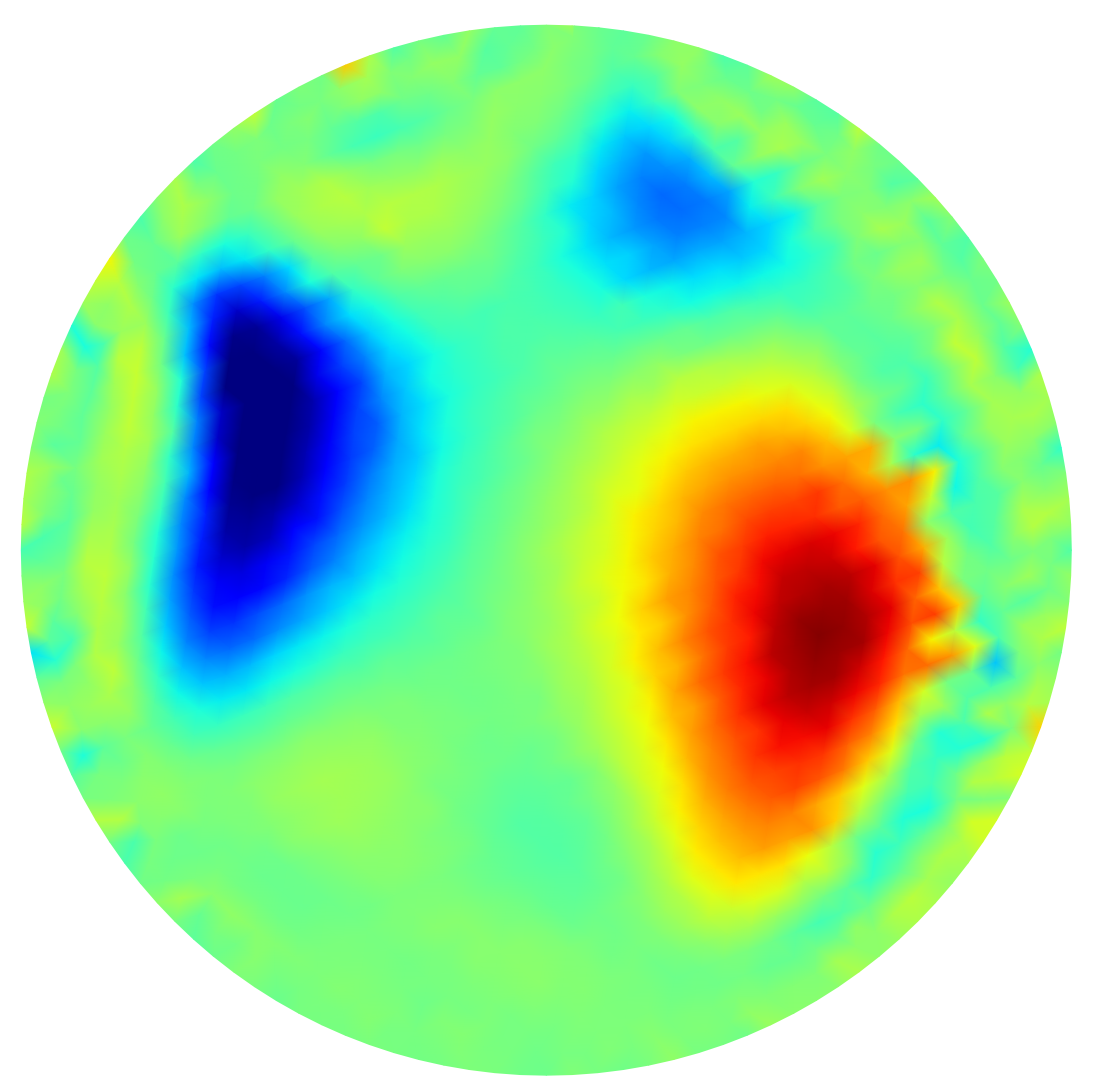}{(0.102, 0.5817)} &
\imgmetric{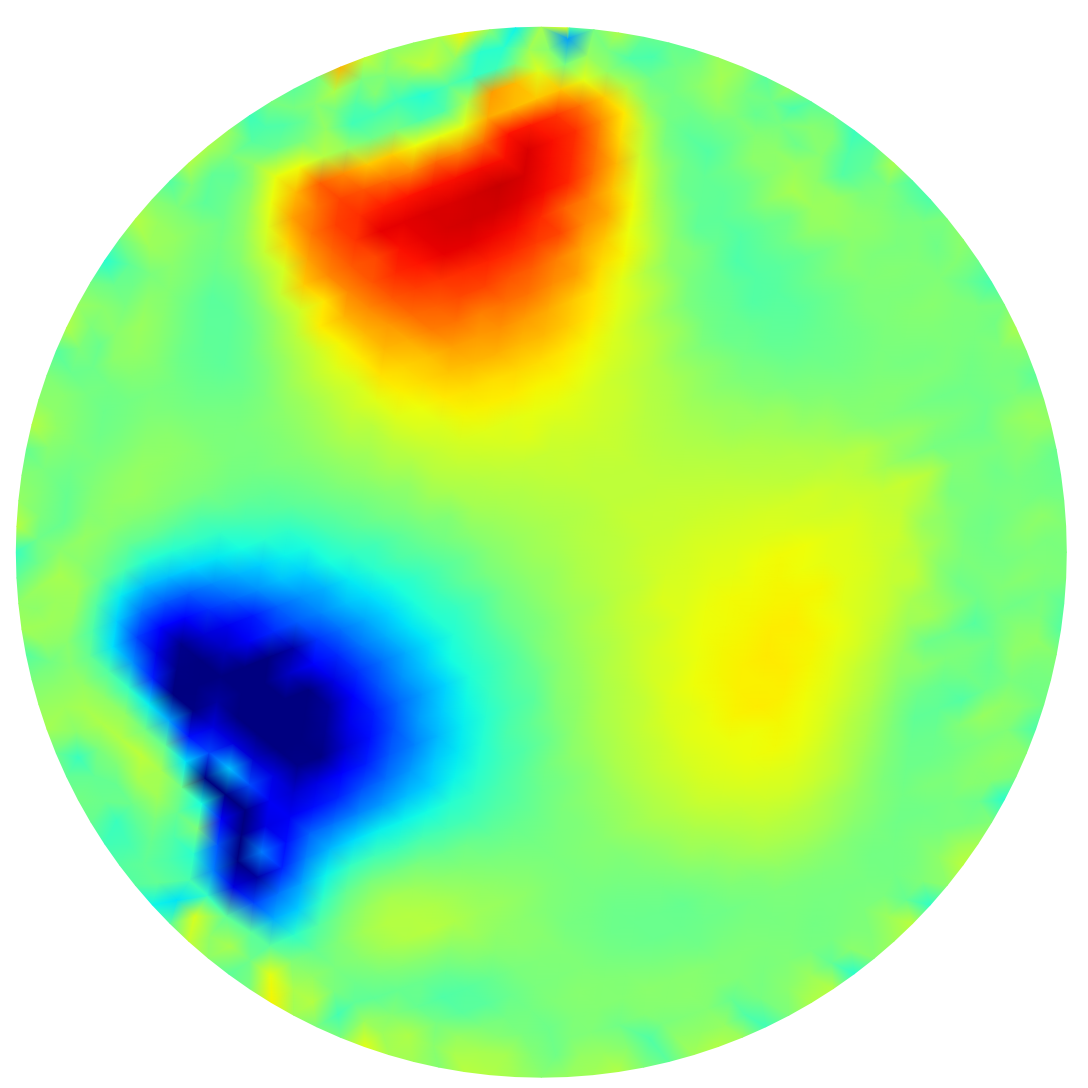}{(0.118, 0.6199)} &
\imgmetric{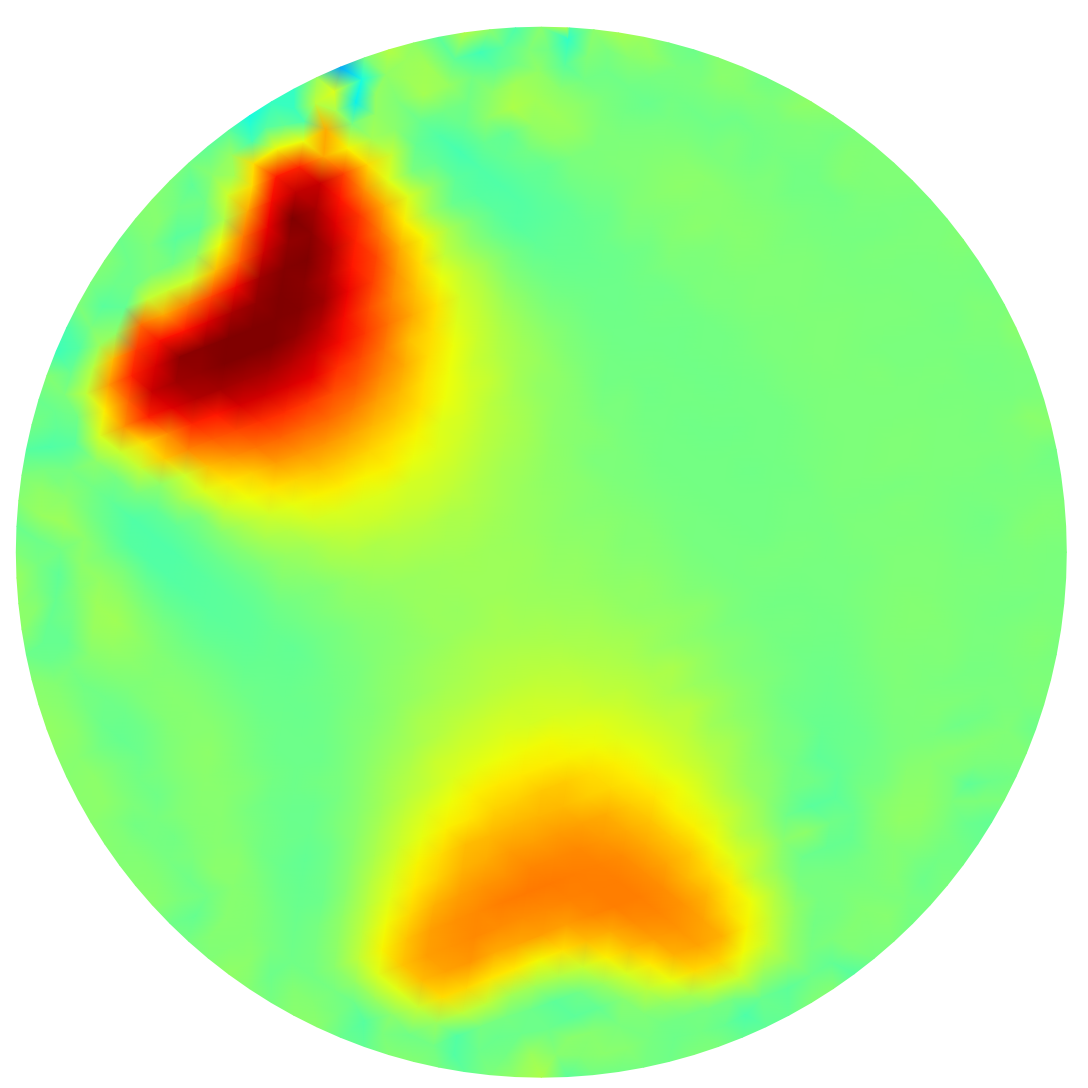}{(0.091, 0.6601)}\\[20mm]

DP-SGS~\cite{ling2025split} & & & & &\\
\imgmetric{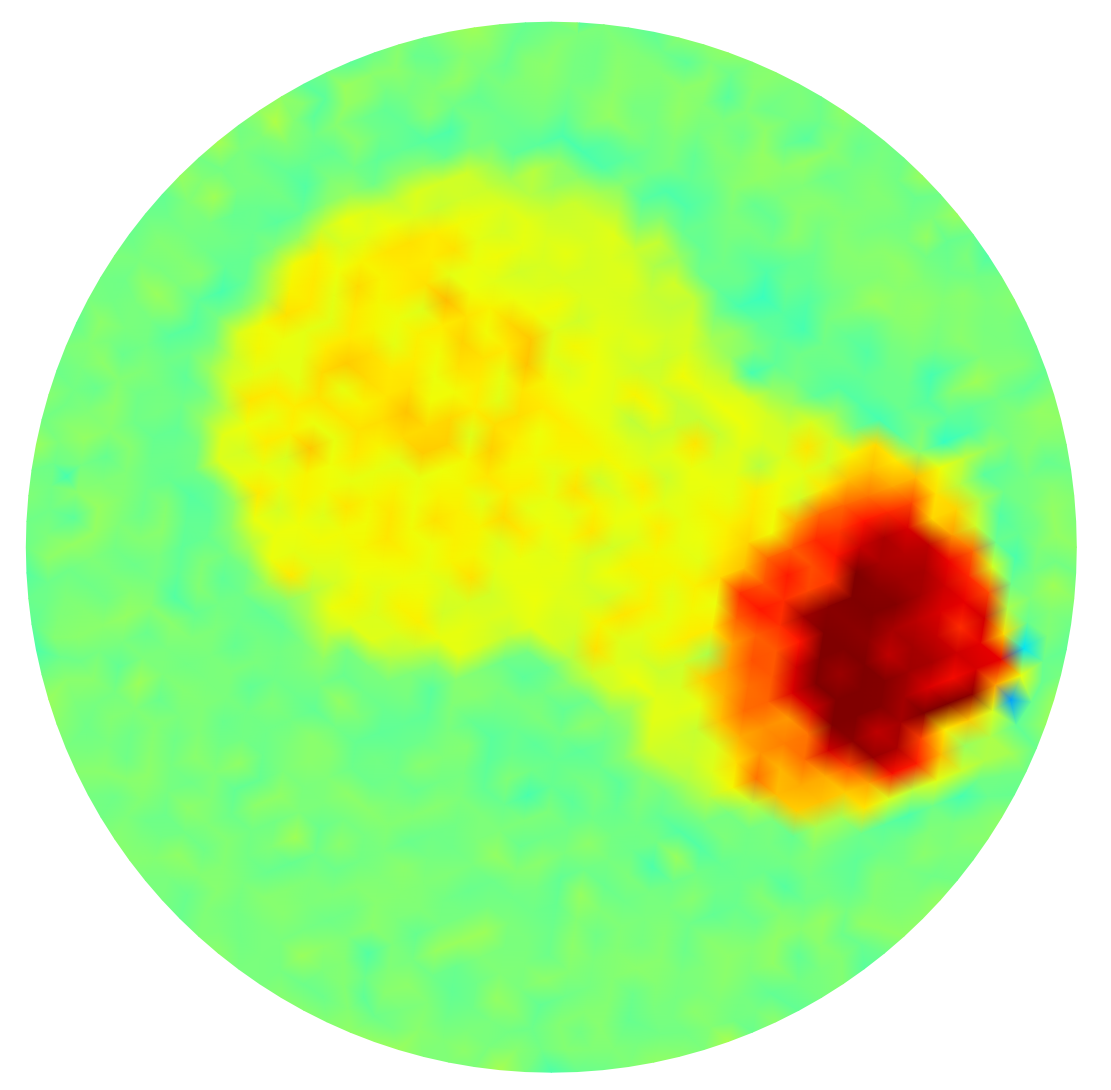}{(0.067, 0.6376)} &
\imgmetric{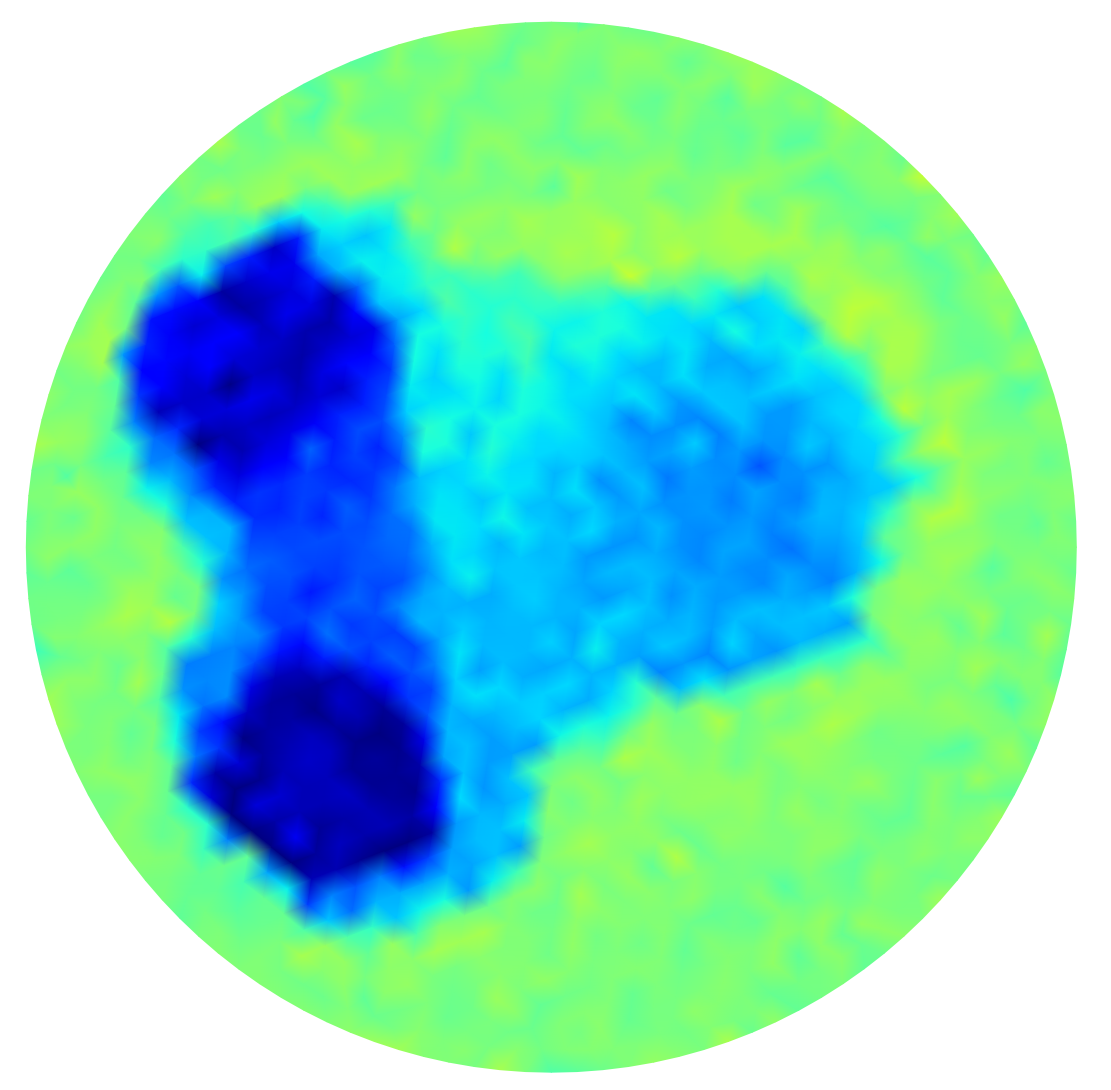}{(0.090, 0.5877)} &
\imgmetric{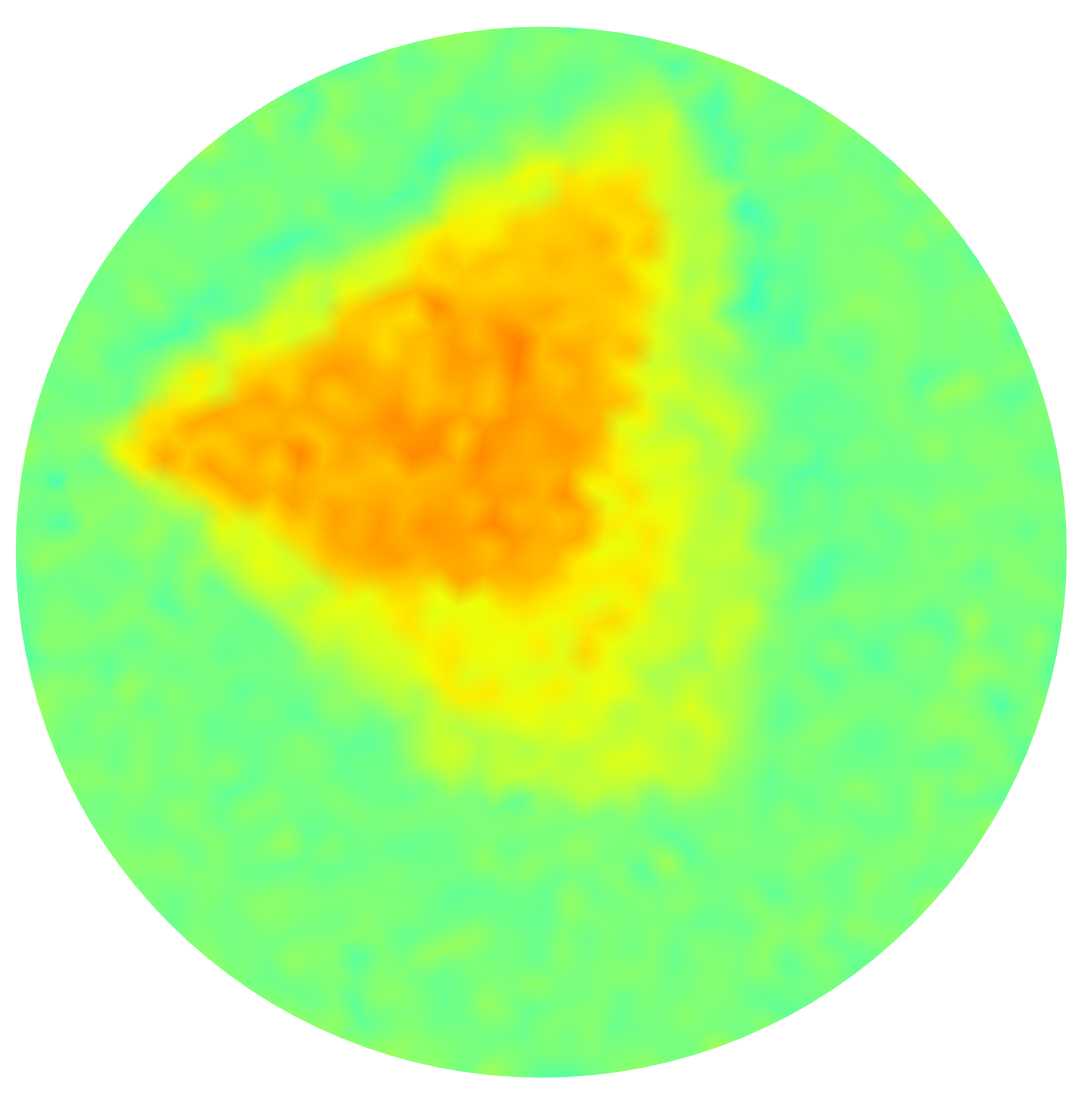}{(0.044, 0.5389)} &
\imgmetric{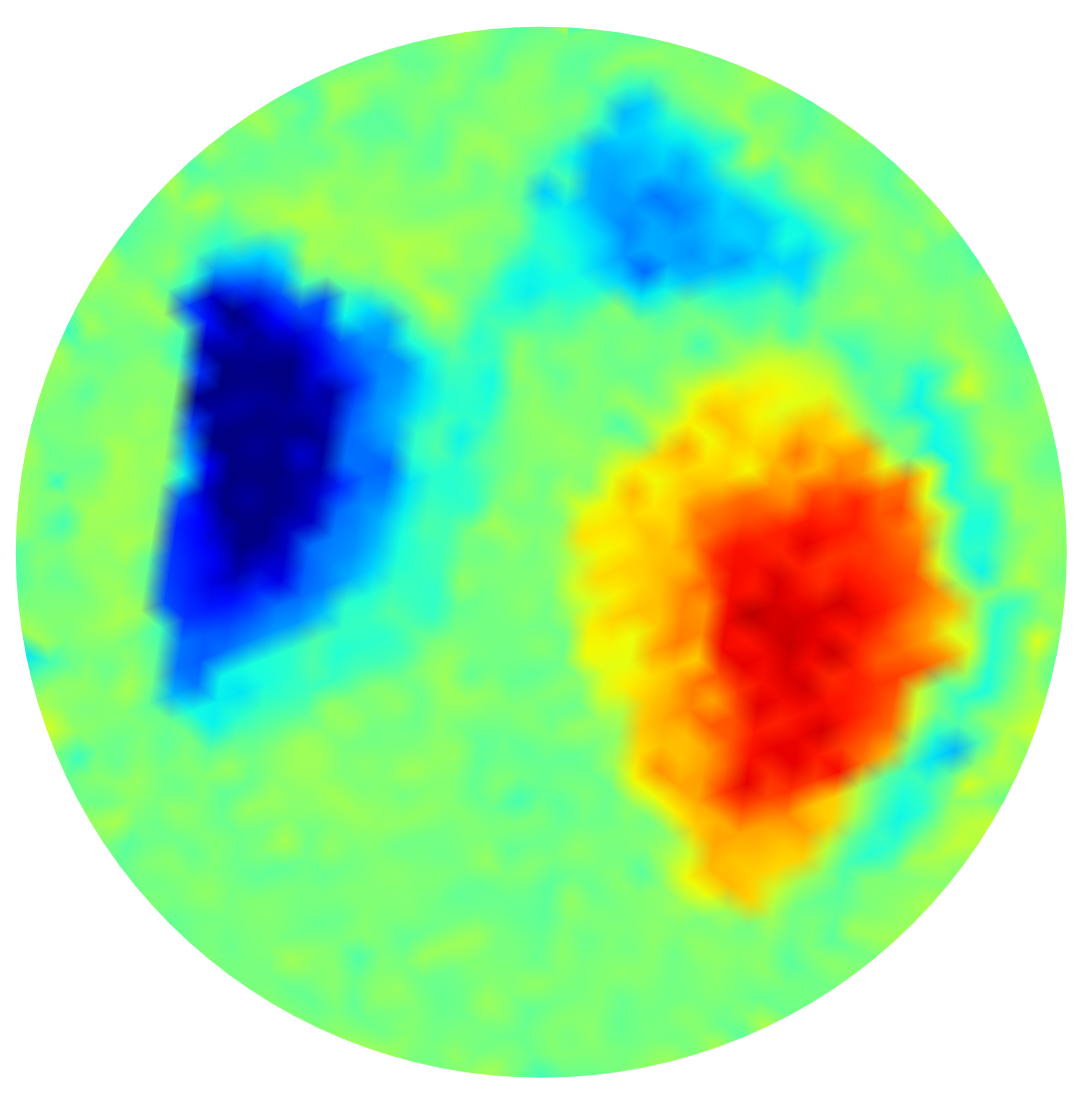}{(0.107, 0.5709)} &
\imgmetric{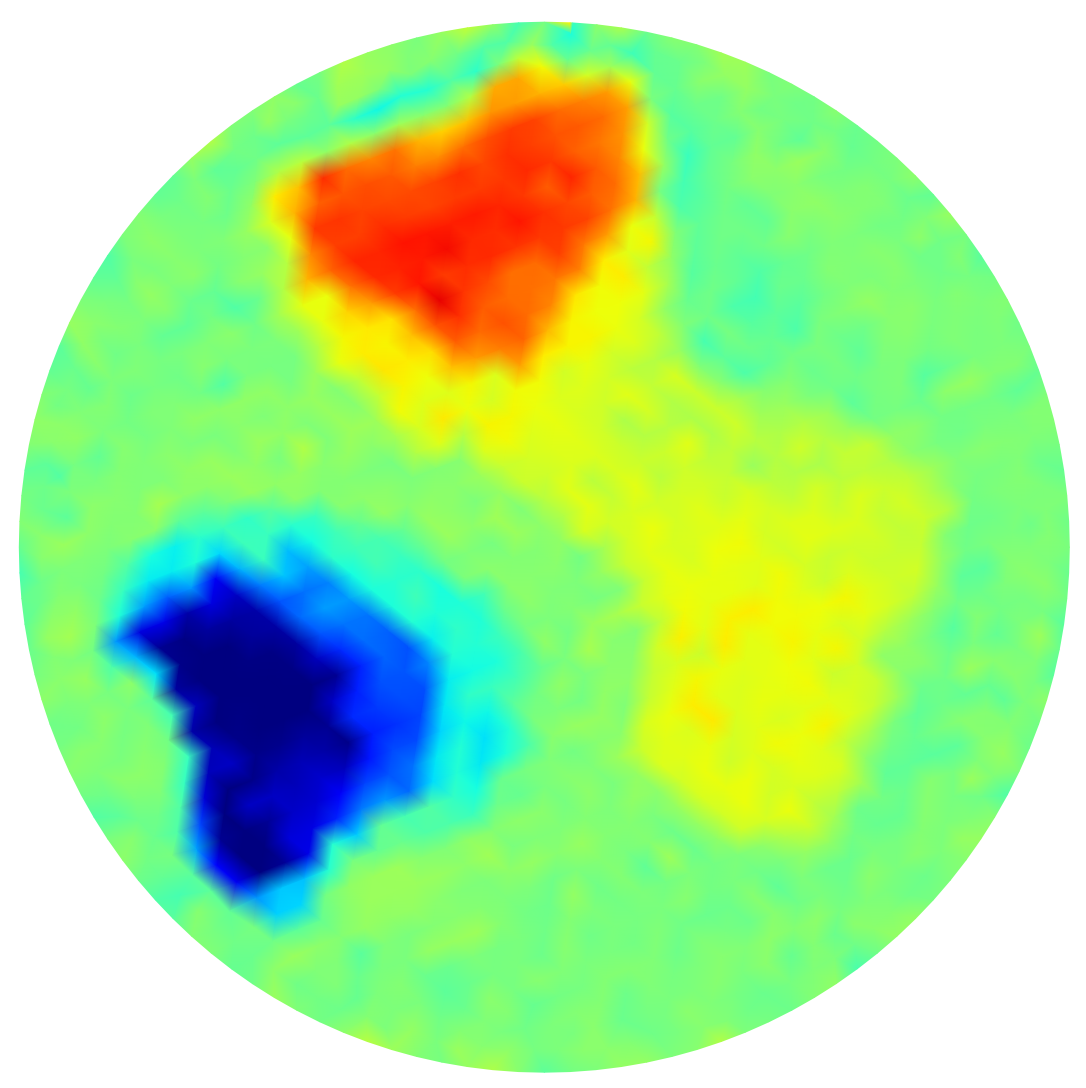}{(0.112, 0.6565)} &
\imgmetric{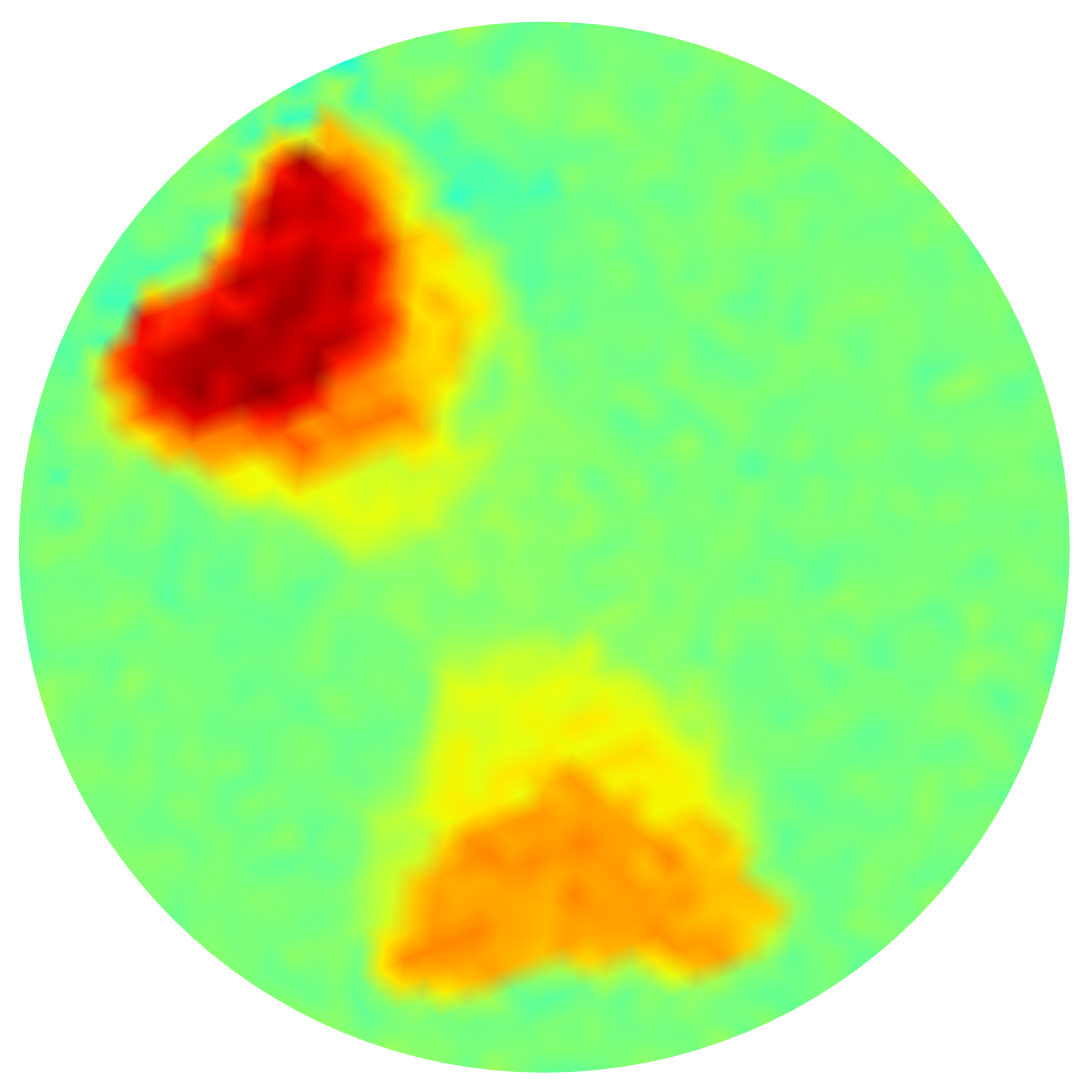}{(0.088, 0.6550)} \\[20mm]

RGN-TV~\cite{borsic2009vivo} & & & & &\\
\imgmetric{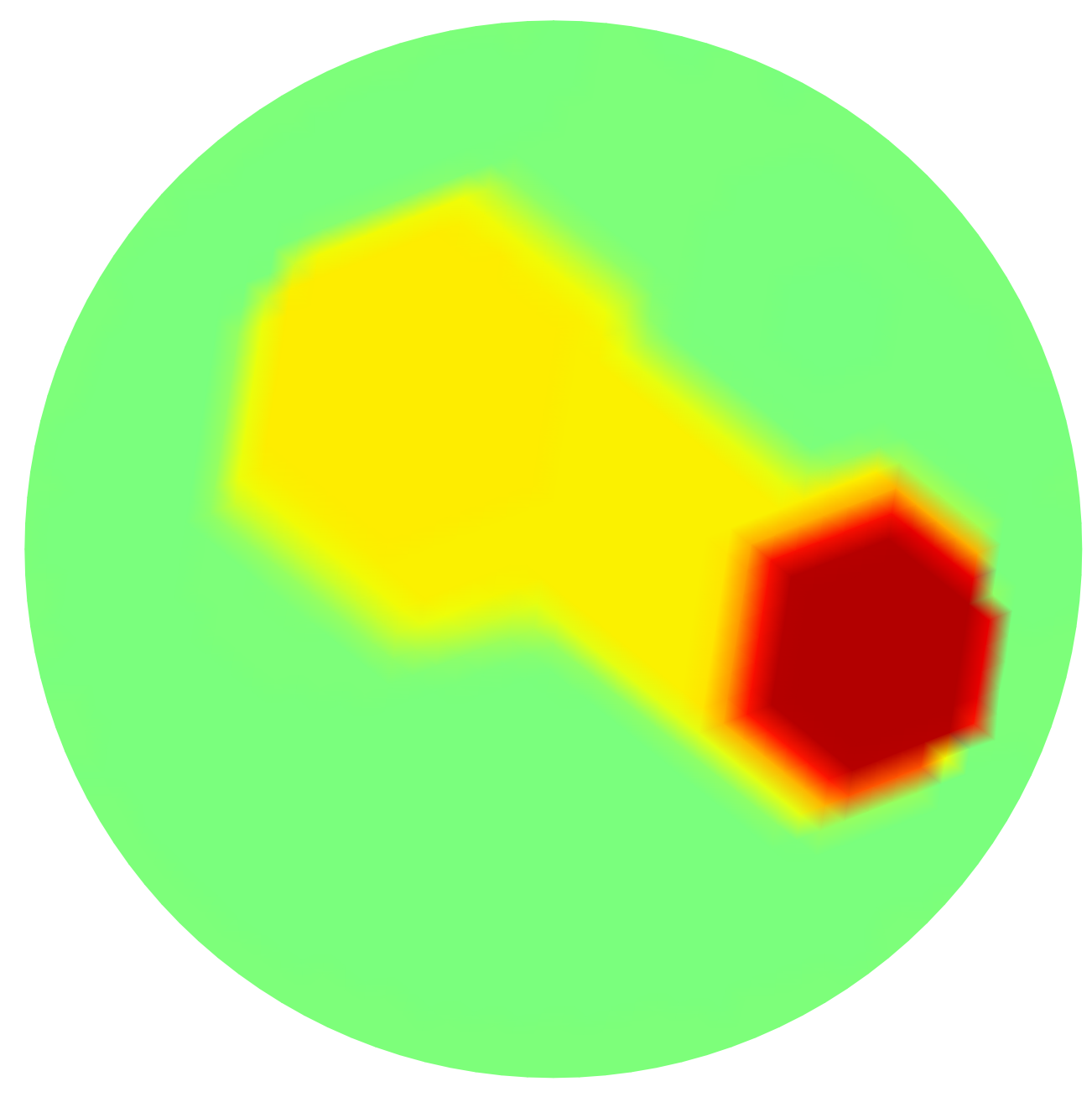}{(0.051, 0.8532)} &
\imgmetric{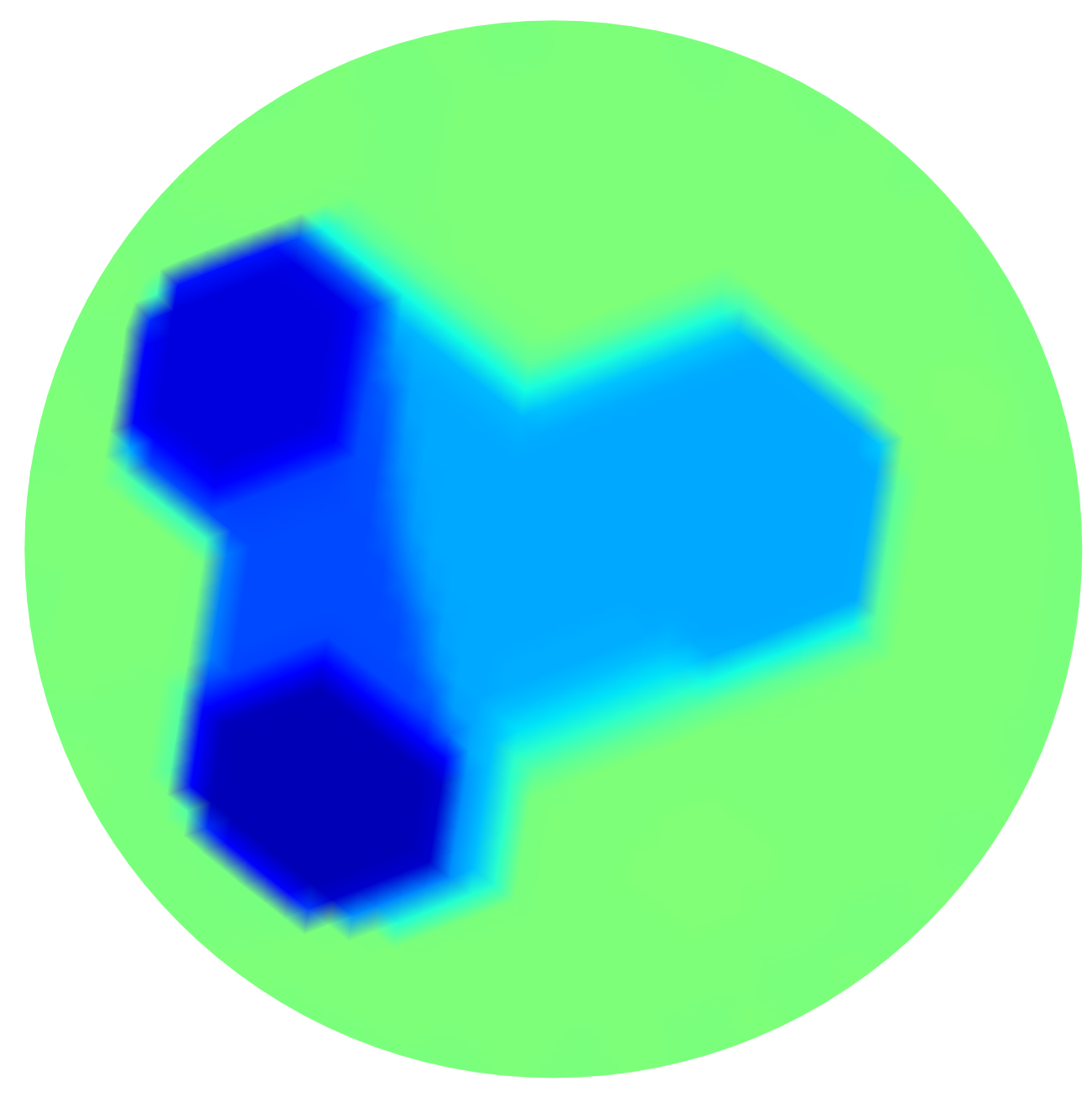}{(0.079, 0.8066)} &
\imgmetric{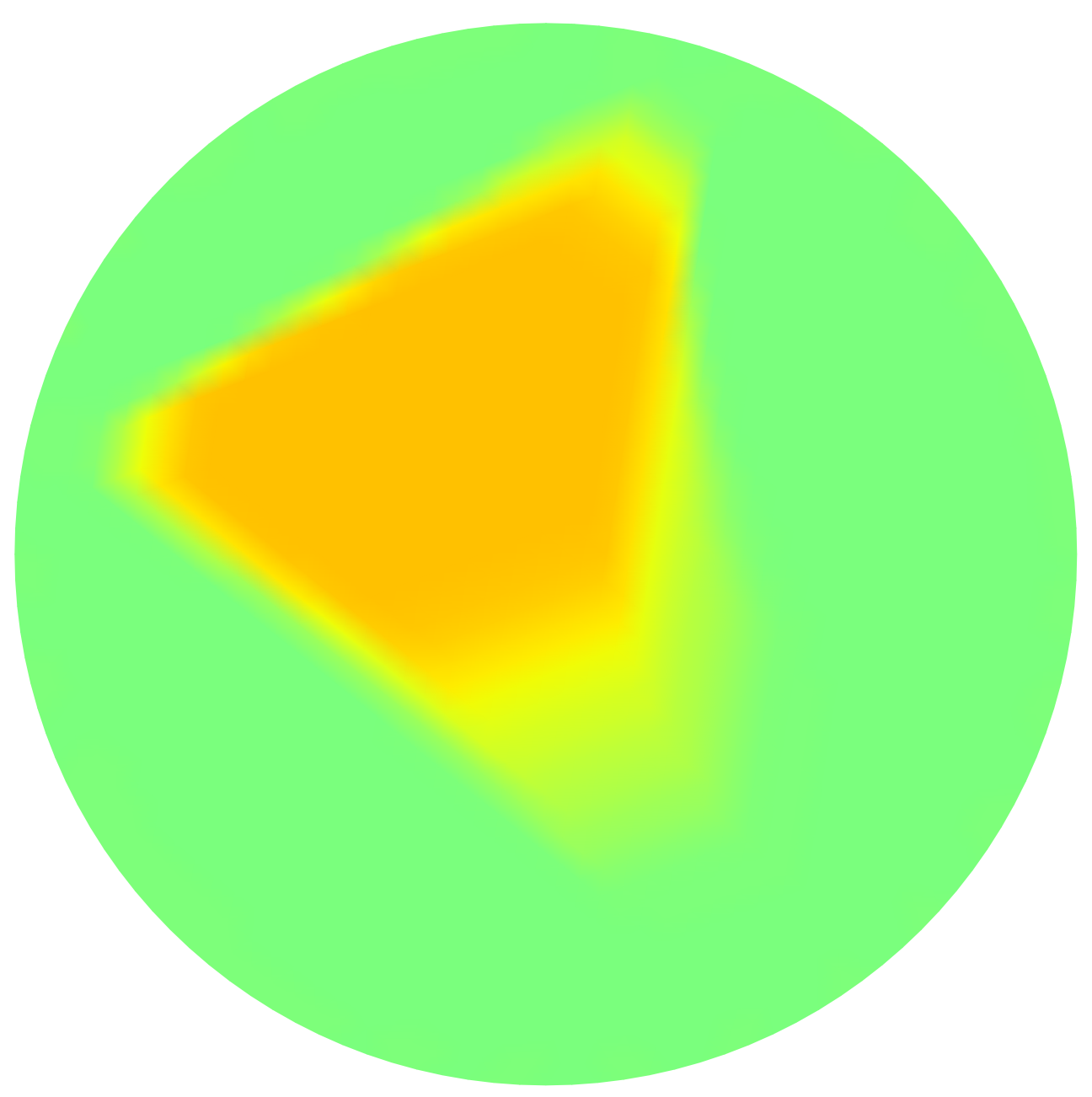}{(0.040, 0.8002)} &
\imgmetric{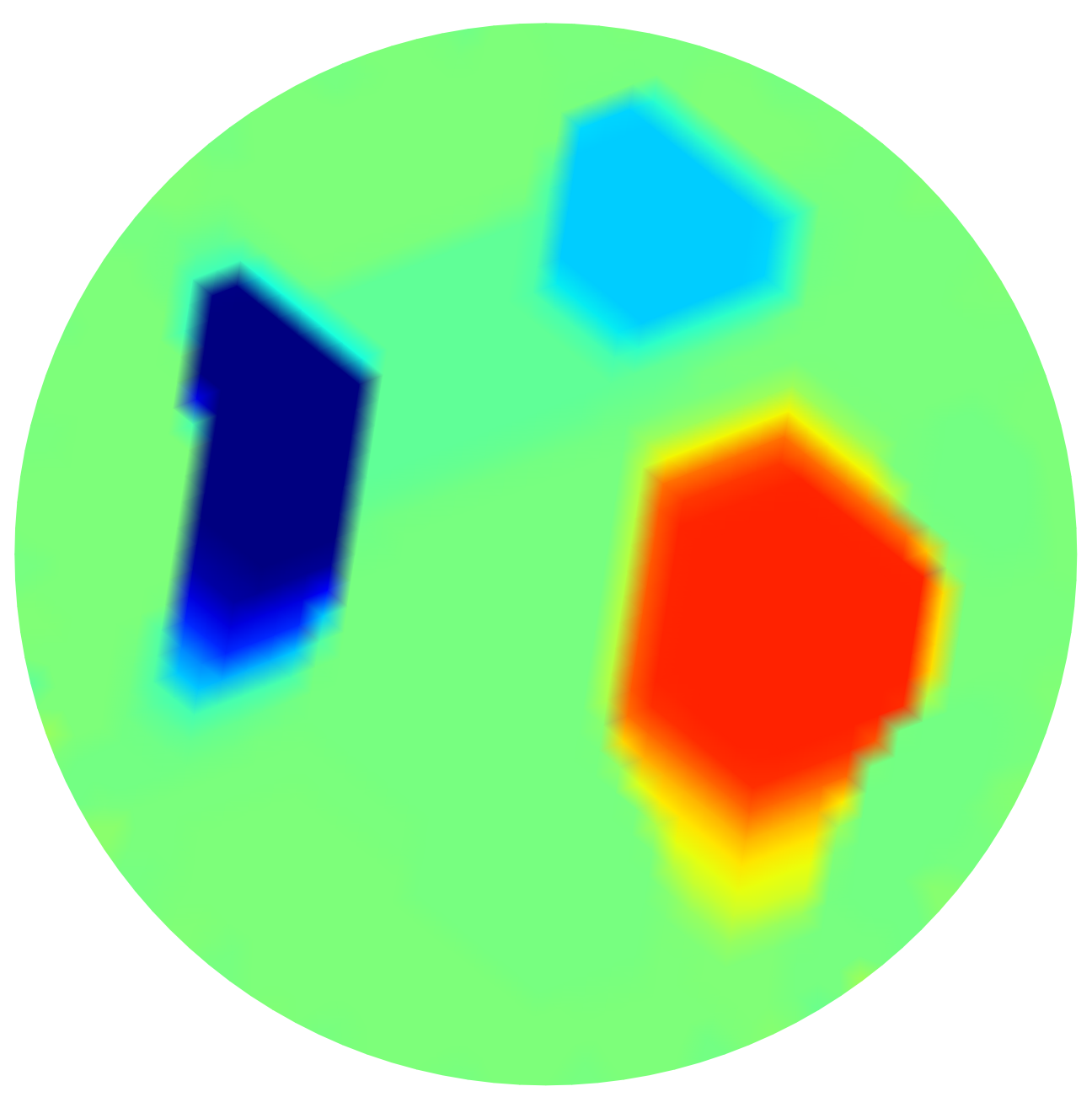}{(0.100, 0.7614)}&
\imgmetric{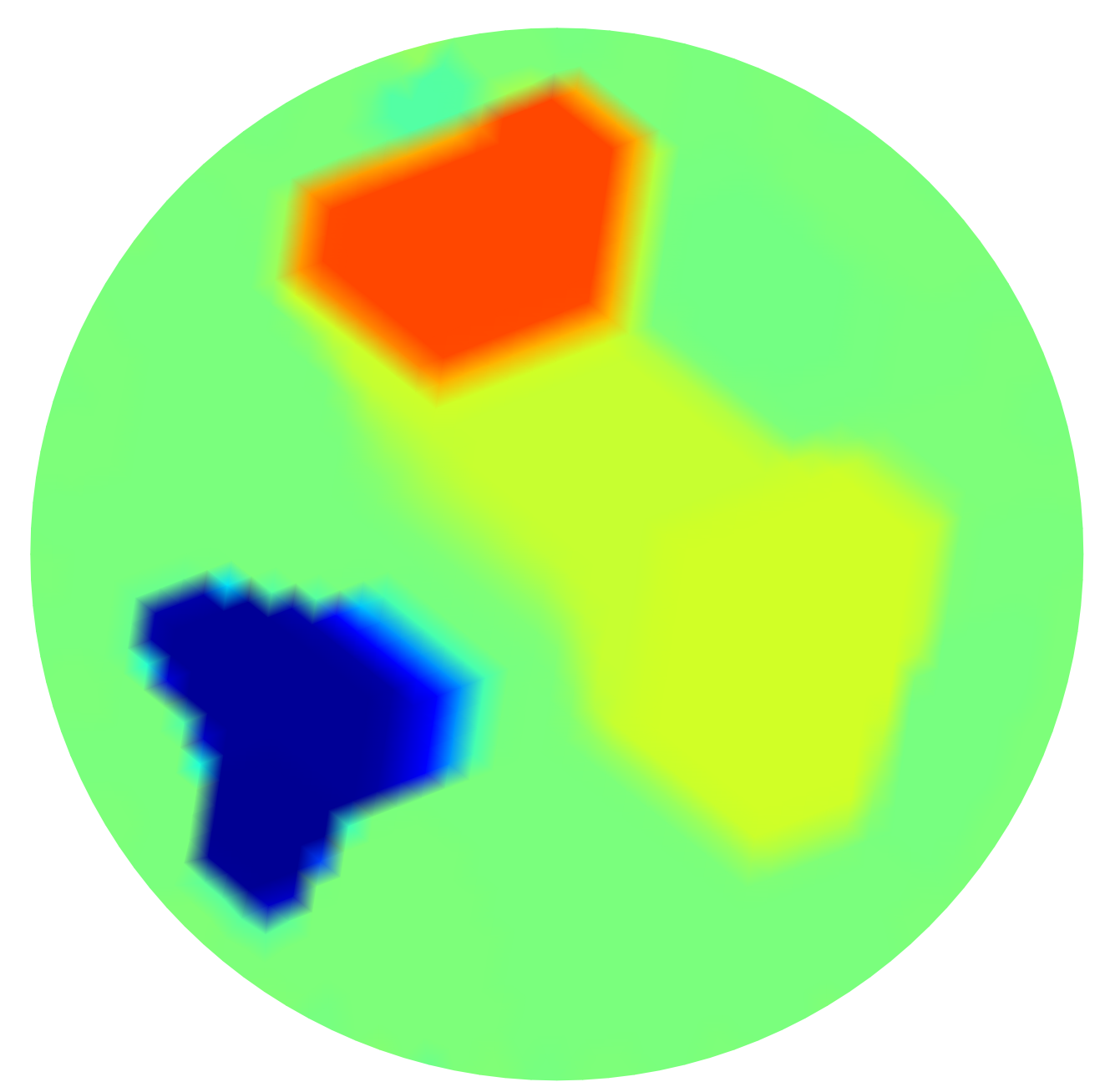}{(0.098, 0.8128)} &
\imgmetric{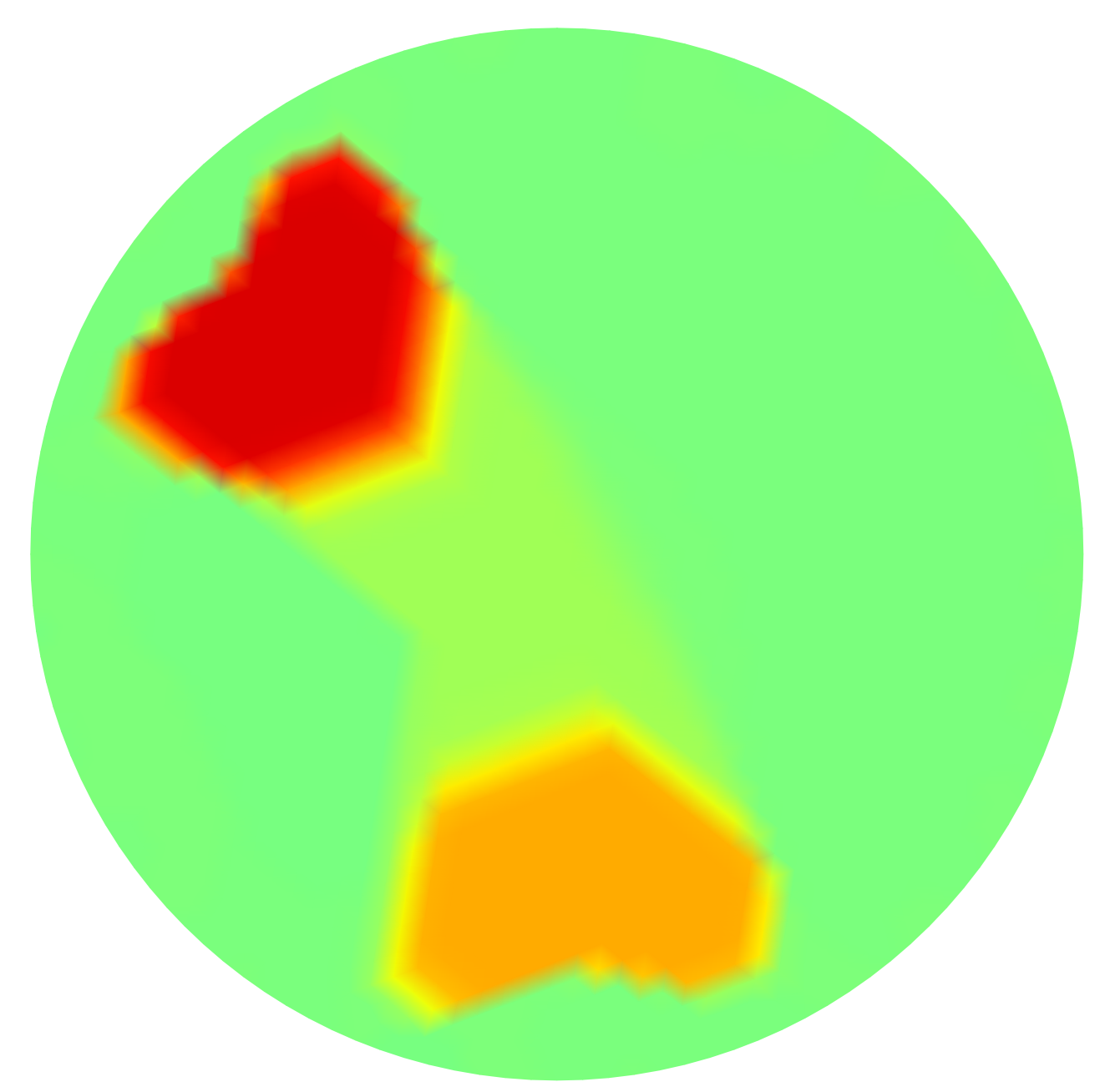}{(0.081, 0.8236)}\\[20mm]

Ours & & & & &\\
\imgmetric{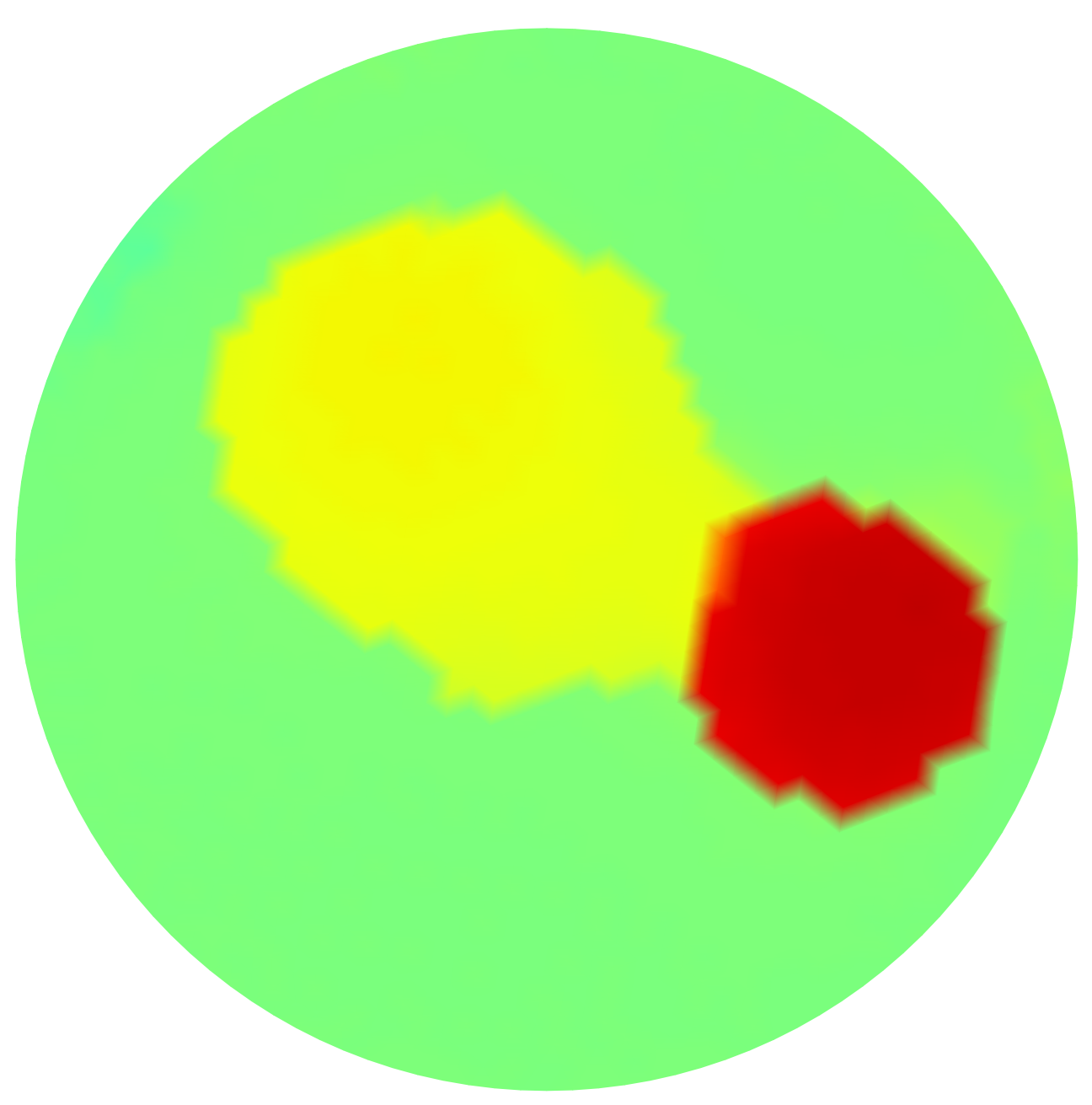}{(\textbf{0.043}, \textbf{0.8609})} &
\imgmetric{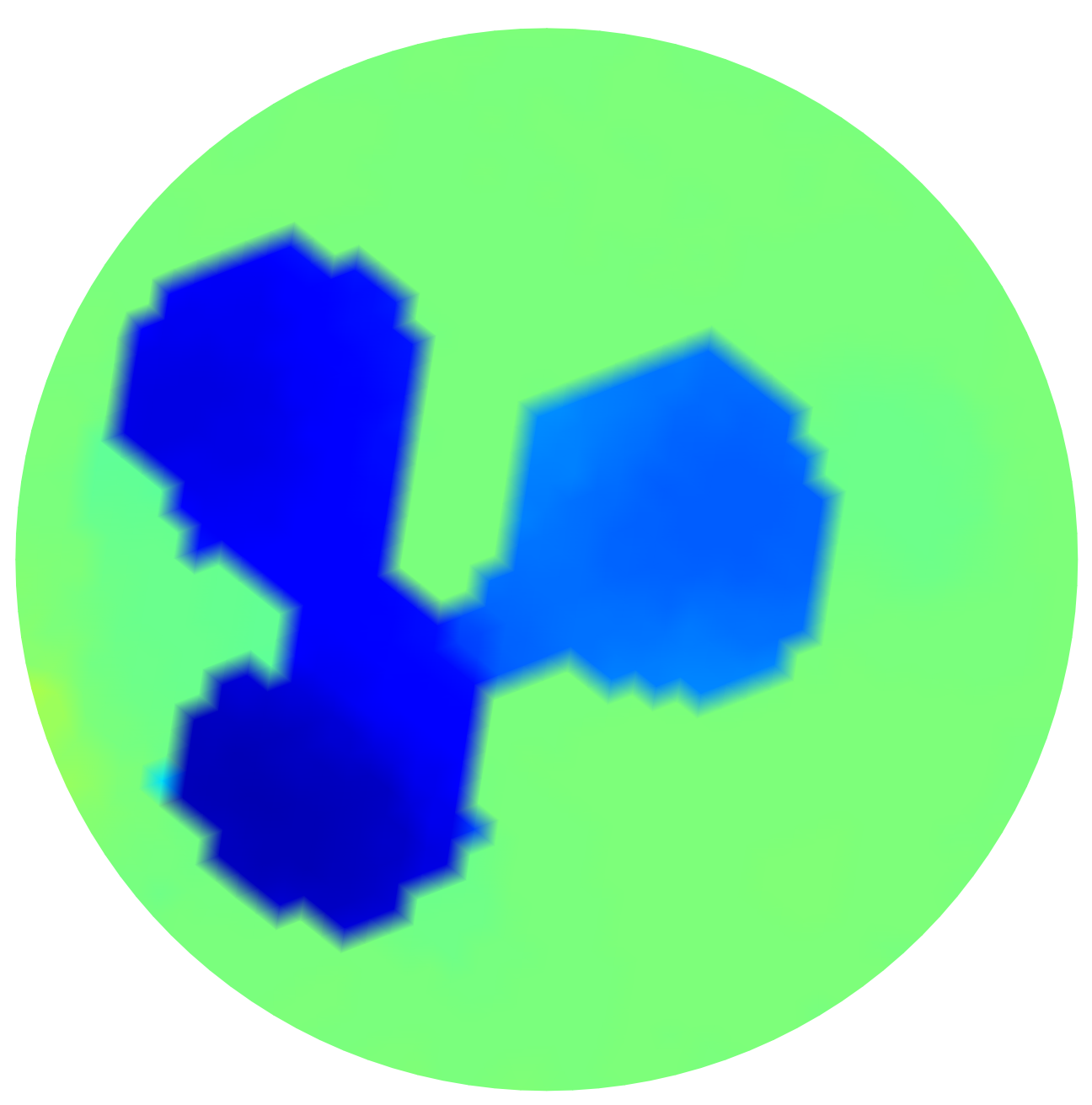}{(\textbf{0.075}, \textbf{0.8553})} &
\imgmetric{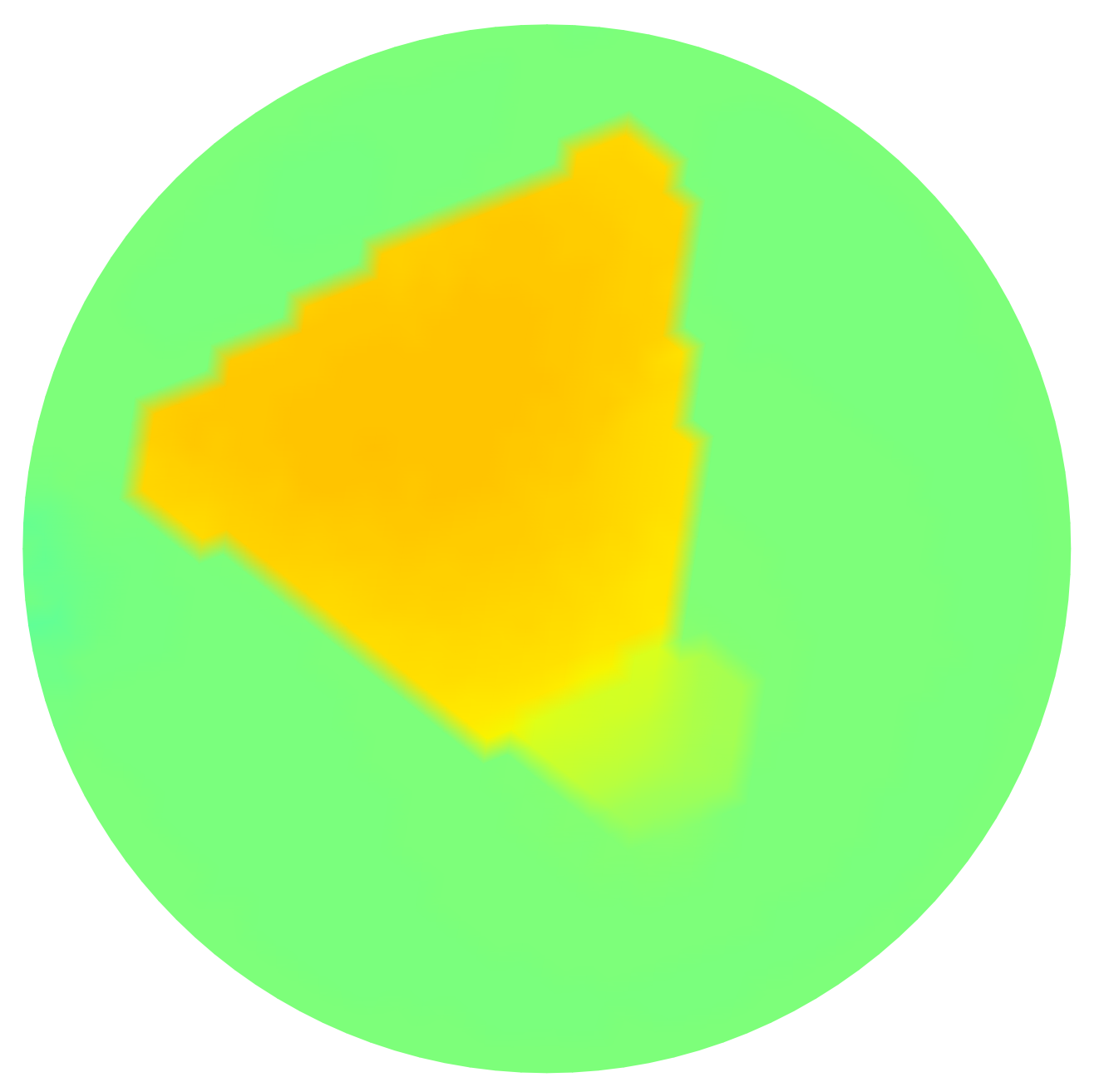}{(0.040, \textbf{0.8131})} &
\imgmetric{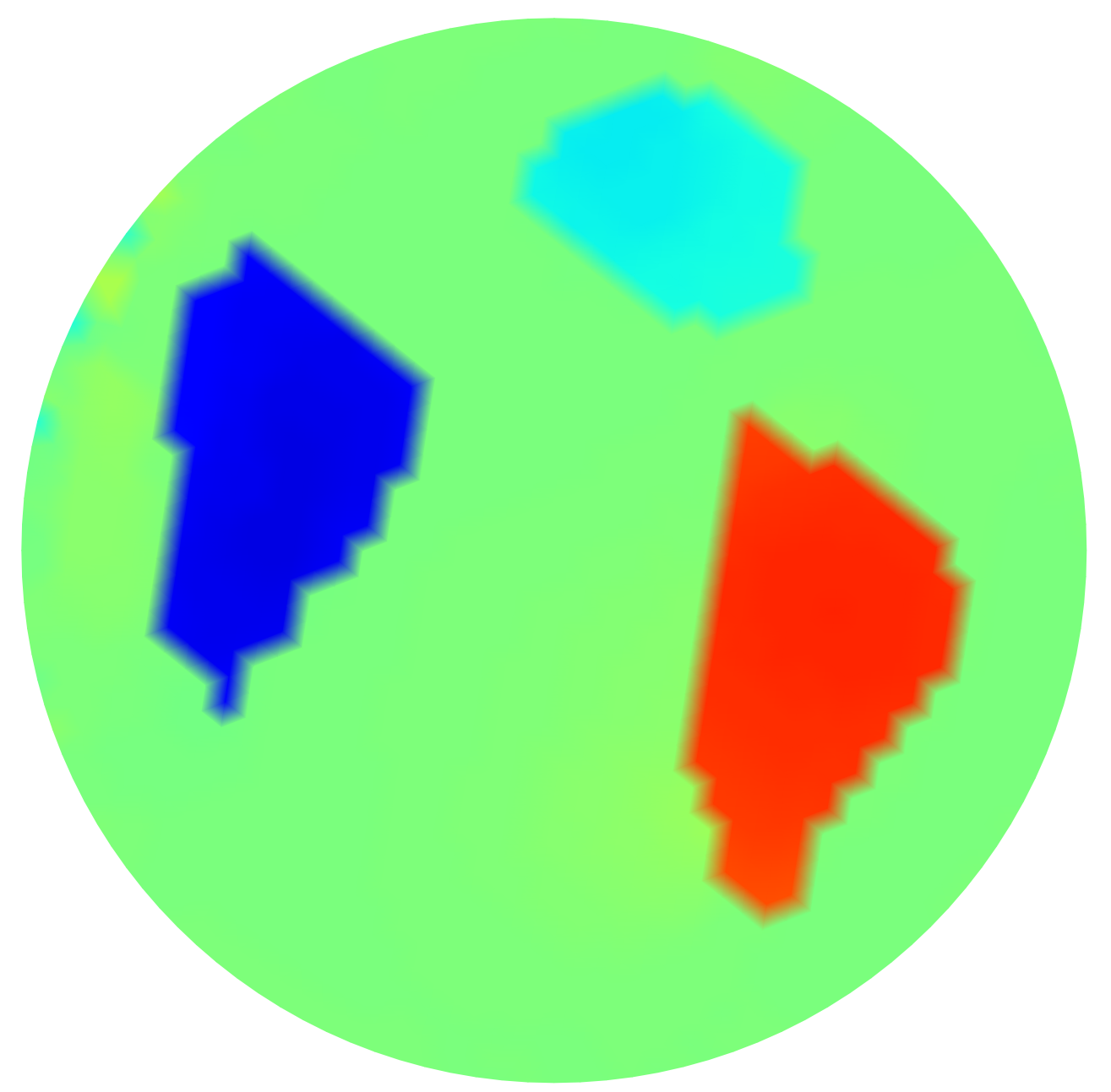}{(0.105, \textbf{0.7678})}&
\imgmetric{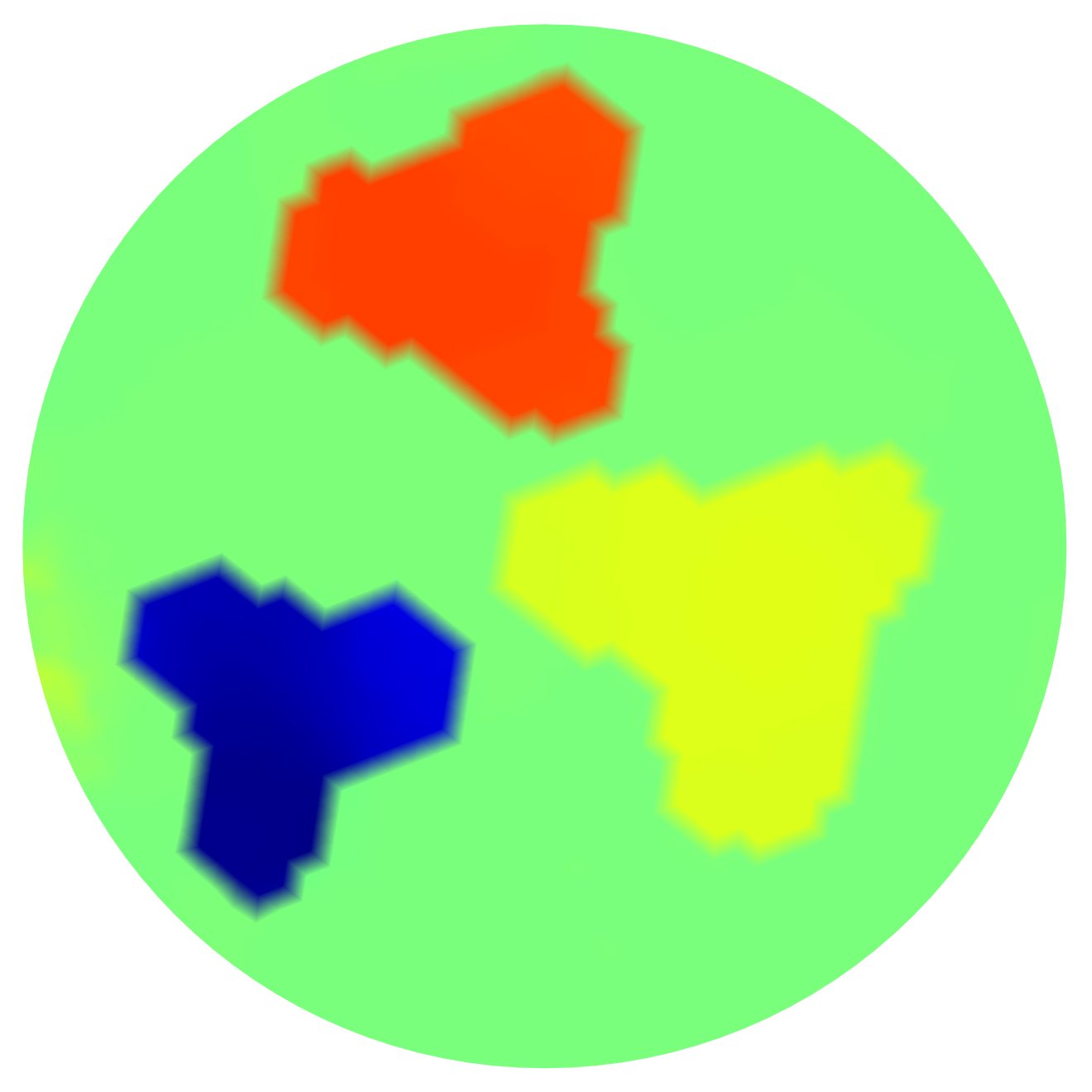}{(\textbf{0.084}, \textbf{0.8761})} &
\imgmetric{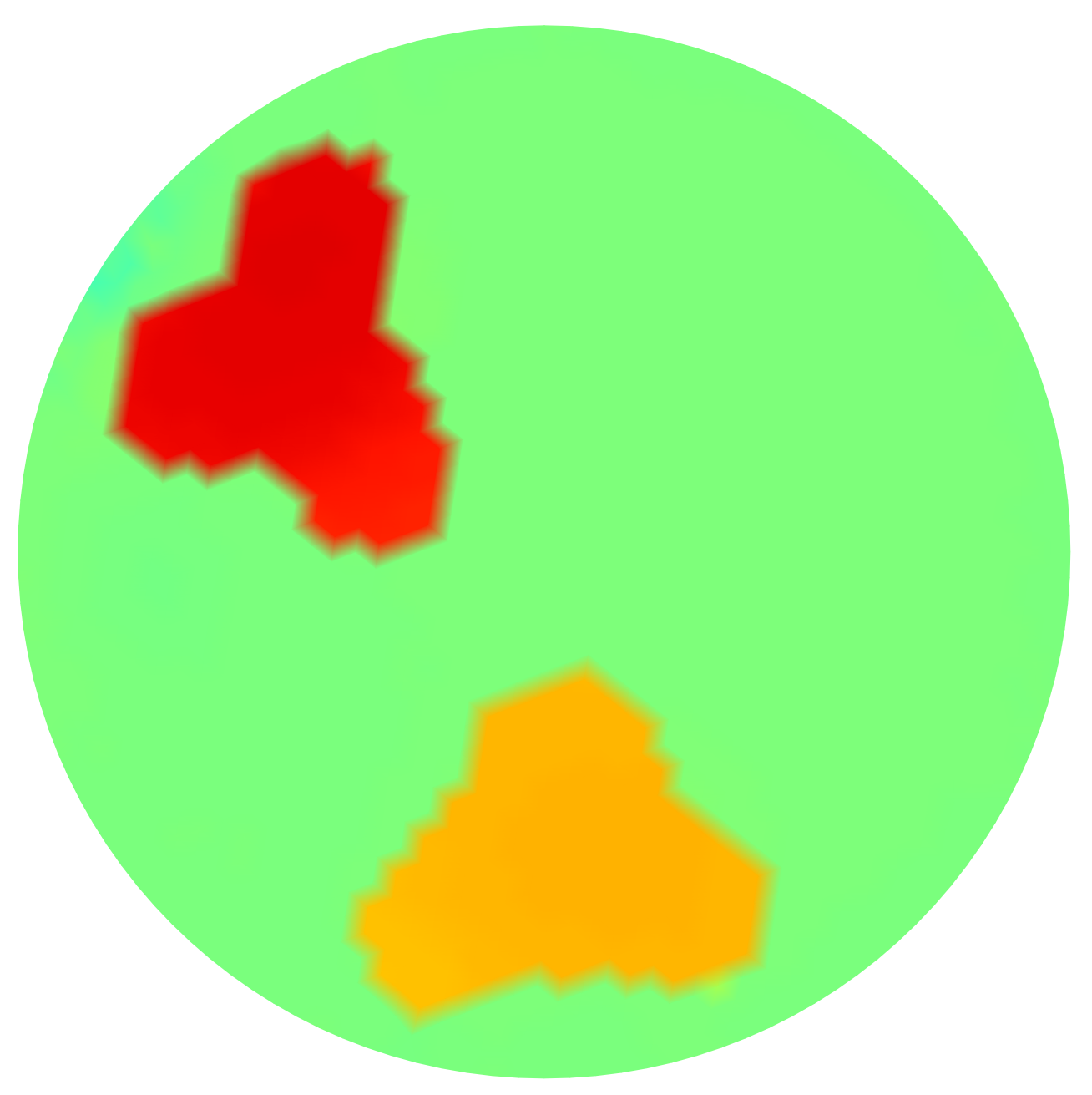}{(\textbf{0.080}, \textbf{0.8791})} \\
\end{tabular}
\end{adjustbox}
\end{table}
\caption{Example 4: Reconstructions obtained by (row-wise from 2-6) RTO-MH, GPnP-BM3D, DP-SGS, RGN-TV, and Ours; GT (first row) is the ground truth conductivity; (RMSE, SSIM) values are reported for each test case (in bold the lowest errors/largest SSIM).}
\label{tab:comp}
\end{figure}

\section{Conclusions}
\label{sec:concl}
 Diffusion posterior sampling represents the most recent and powerful evolution of the pre-trained Plug-and-Play hybrid model concept; it leverages unsupervised diffusion models as generative priors, ensuring superior image quality and total flexibility in solving inverse problems guided exclusively by signal physics. Our proposal introduces a guided target posterior into the DPS Bayesian framework, which successfully overcomes the instabilities caused by the severe ill-posedness of the EIT inverse problem.

Our core contribution is twofold. 
 We introduced
\emph{DDIM-RDPS}, a regularized variant of a standard DPS scheme that augments the learned
implicit diffusion prior with explicit regularization terms, namely total variation and generalized
Tikhonov, to counteract the severe ill-posedness and underdetermination characteristic of EIT. 
From a Bayesian perspective, this regularized formulation amounts to incorporating an additional
structured prior that guides the reverse diffusion trajectory toward physically plausible reconstructions
whenever the measurement data provides insufficient information.

In this framework, we extended the classical diffusion models to handle graph-structured
domains, enabling the method to operate directly on the unstructured triangular meshes that arise
naturally in finite-element discretizations of physical domains, without resorting to lossy grid
interpolations.
 
Extensive numerical experiments on both synthetic and real 2D EIT datasets demonstrated the
effectiveness and versatility of the proposed framework. In particular, DDIM-RDPS proved highly
robust to additive Gaussian and Laplacian measurement noise across multiple noise levels, and
consistently preserved geometric features and conductivity contrast across a wide variety of
inclusion morphologies. 

In the comparative benchmarking against state-of-the-art solvers, including RTO-MH, GPnP-BM3D,
DP-SGS, and RGN-TV, DDIM-RDPS achieved consistently both quantitatively and qualitatively better results across all test configurations.

The proposed framework opens several directions for future work. Extending the approach to
three-dimensional EIT or to other PDE-governed inverse problems, such as  optical diffusion tomography, represents a natural next step. Furthermore,  exploring the integration of more complex physical constraints within the DPS model would further strengthen the reliability of the reconstructions. We noticed that our unconditional graph-based diffusion model had better generation performances on some datasets, so it would be useful to improve its architecture and training to make it more robust and expressive. Another promising line of research is to develop completely mesh-free unconditional diffusion models, that would not depend on specific graphs/meshes.

\section*{Acknowledgments}
Co-funded by the European Union (ERC, SAMPDE, 101041040 – Next Generation EU, Missione 4 Componente 1 CUP D53D23005770006 and CUP D53D23016180001) and by the MIUR Excellence Department Project awarded to Dipartimento di Matematica, Università di Genova, CUP D33C23001110001. Views and opinions expressed are  those of the authors only and do not necessarily reflect those of the European Union or the European Research Council. Neither the European Union nor the granting authority can be held responsible for them. This material is based upon work supported by the Air Force Office of Scientific Research under award number FA8655-23-1-7083.
This work was supported in part by the National Group for Scientific Computation (GNCS-INDAM), Research Projects 2026.

\appendix

\section{Appendix}
\label{appendix}
	\begin{lemma}[Conditional Tweedie's formula]
		If the joint distribution between $x_0, y, x_t$ is given by $p_t(x_0, y, x_t) = p(x_0)p(y|x_0)p_t(x_t|x_0)$ with $p_t(x_t|x_0) = \mathcal{N}(x_t| \alpha_t x_0,  \sigma_t^2 I)$, then 
		\begin{equation}
			\nabla_{x_t} \log p_t(x_t|y) = \frac{1}{ \sigma_t^2}(\alpha_t \mathbb{E}[x_0|x_t, y] - x_t),
		\end{equation}
		where, according to \eqref{eq:forward_diffusion_closed},	$\alpha_t = \sqrt{\bar{\alpha_t}}$ and
		$\sigma_t^2 = 1-{\bar{\alpha_t}}$.
	\end{lemma}

	\begin{proof}
    We have
		\begin{align}
			\nabla_{x_t} \log p_t(x_t|y) &= \frac{\nabla_{x_t} p_t(x_t|y)}{p_t(x_t|y)} \\
			&= \frac{1}{p_t(x_t|y)} \nabla_{x_t} \int p_t(x_t|x_0, y) p(x_0|y) dx_0 \\
			&= \frac{1}{p_t(x_t|y)} \nabla_{x_t} \int p_t(x_t|x_0) p(x_0|y) dx_0 \\
			&= \frac{1}{p_t(x_t|y)} \int p(x_0|y) \nabla_{x_t} p_t(x_t|x_0) dx_0 \\
			&= \frac{1}{p_t(x_t|y)} \int p(x_0|y) p_t(x_t|x_0, y) \nabla_{x_t} \log p_t(x_t|x_0) dx_0 \\
			&= \int p_t(x_0|x_t, y) \nabla_{x_t} \log p_t(x_t|x_0) dx_0 \\
			&= \mathbb{E}_{p_t(x_0|x_t, y)}[\nabla_{x_t} \log p_t(x_t|x_0)],
		\end{align}
		where the steps  are due to the conditional independence between $x_t$ and $y$ given $x_0$, such that $p_t(x_t|x_0, y) = p_t(x_t|x_0)$. 
		
		For the Gaussian perturbation kernel $p_t(x_t|x_0) = \mathcal{N}(x_t|\alpha_t x_0, \sigma_t^2 I)$, we have:
		\begin{equation}
			\nabla_{x_t} \log p_t(x_t|x_0) = \frac{1}{ \sigma_t^2}(\alpha_t x_0 - x_t).
		\end{equation}
		Plugging this result into the expectation derived above, we conclude the proof.
	\end{proof}		

We are now able to prove Proposition~\ref{pro:cpm}.
    \begin{proof}[Proof of Proposition~\ref{pro:cpm}]
		Following Lemma 1, we have 
		\begin{equation}
			\label{eq:E}
		\mathbb{E}[x_0|x_t, y] = 	\nabla_{x_t} \log p_t(x_t|y) \frac{\sigma_t^2}{\alpha_t}  + \frac{x_t}{\alpha_t}.
		\end{equation}
		Then, by \eqref{eq:posterior_scoredps1}, we obtain
		\begin{equation}
			\mathbb{E}[x_0|x_t, y] = \frac{x_t}{\alpha_t} + \frac{\sigma_t^2}{\alpha_t}\Big(	\nabla_{x_t} \log p_t(x_t)
			+
			\nabla_{x_t} \log p_t(y \mid x_t)
			+
			\nabla_{x_t}  \log p_R(x_t)\Big) .
		\end{equation}
		According to $\mathbb{E}[x_0 \mid x_t] $ defined in \eqref{eq:twed_diff}, we obtain \eqref{eq:Ec}.
\end{proof}

\bibliographystyle{abbrv}
\bibliography{refs}

\end{document}